\documentclass{article}

\usepackage[letterpaper, margin=1in]{geometry}

\usepackage[utf8]{inputenc} 
\usepackage[T1]{fontenc} 
\usepackage{url} 
\usepackage{booktabs} 
\usepackage{nicefrac} 
\usepackage{microtype} 

\usepackage{amsmath,amsthm,amssymb}
\usepackage{authblk}
\usepackage{enumitem}
\usepackage{float}
\usepackage[pdftex]{graphicx}
\usepackage{mathtools}
\usepackage{multirow}
\usepackage{natbib}
\usepackage{placeins}  
\usepackage{tikz}

\usepackage{algorithmic}

\usepackage[colorlinks=true,
            linkcolor=blue,
            urlcolor=blue,
            citecolor=blue]{hyperref}

\usepackage{byol_definitions}
\allowdisplaybreaks

\title{Two-sample inference for high-dimensional Markov networks}

\author[1]{Byol Kim\thanks{byolkim@uchicago.edu}}
\affil[1]{Department of Statistics, The University of Chicago, Chicago, IL 60637,USA.}

\author[2]{Song Liu\thanks{song.liu@bristol.ac.uk}}
\affil{School of Mathematics, University of Bristol, Bristol, BS8 1UG, UK;\\
    Alan Turing Institute, London, NW1 2DB, UK.}

\author[3]{Mladen Kolar\thanks{mkolar@chicagobooth.edu}}
\affil{Booth School of Business, The University of Chicago, Chicago, IL 60637, USA.}

\date{}

\begin{document}

\maketitle

\begin{abstract}
Markov networks are frequently used in sciences to represent conditional independence relationships underlying observed variables arising from a complex system. It is often of interest to understand how an underlying network differs between two conditions. In this paper, we develop methods for comparing a pair of high-dimensional Markov networks where we allow the number of observed variables to increase with the sample sizes. By taking the density ratio approach, we are able to learn the network difference directly and avoid estimating the individual graphs. Our methods are thus applicable even when the individual networks are dense as long as their difference is sparse. We prove finite-sample Gaussian approximation error bounds for the estimator we construct under significantly weaker assumptions than are typically required for model selection consistency. Furthermore, we propose bootstrap procedures for estimating quantiles of a max-type statistics based on our estimator, and show how they can be used to test the equality of two Markov networks or construct simultaneous confidence intervals. The performance of our methods is demonstrated through extensive simulations. The scientific usefulness is illustrated with an analysis of a new fMRI dataset.
\end{abstract}


\section{Introduction}

Markov networks, also known as Markov random fields or undirected probabilistic graphical models, are successfully used in many application domains to represent interactions between measured components of a complex system and help scientists in uncovering structured information from large amounts of unstructured data \citep{Lauritzen1996Graphical,mackay2003information,Koller2009Probabilistic}. In genetics the graph structure can be used, for example, to model regulatory activities in gene expressions \citep{Hartemink2001Using,dobra2004sparse}, while in neuroscience it can be used to model brain network in order to identify features associated with different mental diseases \citep{Supekar2008Network}. Other successful application areas include social and political sciences \citep{Banerjee2008Model}, analysis of financial data \citep{Barber2015ROCKET}, and many others. One of the fundamental problems in statistics is that of learning the graph structure of a probabilistic graphical model based on independent and identically distributed (i.i.d.) samples. See \cite{Drton2017Structure} for a recent overview.

The focus of this paper is on developing a method for statistical
inference of parameters in a differential network.
For a recent survey, see \citet{Shojaie2020Differential}, and references therein.
In many
applications, interest centres not on a particular network, but rather
on whether and how the network changes between different states.  For
example, genes may regulate each other differently when the external
environment is altered. The way different regions of a brain interact
together may be altered depending on the activity that a patient is
performing.  A single graphical model lacks the ability to capture
such changes and cannot reflect the dynamic nature of such data,
therefore limiting our ability to gain key insights into the
underlying system under consideration.

We develop a collection of methods for performing statistical
inference on the difference of parameters in high-dimensional Markov
networks.  Subtleties arise when the target of inference is the difference
of parameters rather than the parameters themselves.  In
high-dimensional regimes, consistent estimation requires an assumption
of inherent low-dimensionality such as sparsity \citep{Yuan2007Model,
  Friedman2008Sparse, Yuan2010High, Cai2011Constrained,
  Ravikumar2011High}.  Therefore, a crude procedure that estimates the
network parameters separately, and then takes the difference, can only
work when \emph{all} the individual networks are sparse.  This is
quite restrictive for applications where the individual networks may
be dense, but the differences are expected to be sparse, say, due to
the experimental set-up.  Moreover, even when the assumption is
satisfied, many such methods have tuning parameters that have an
influence on the estimated structure, and it is unclear how they
should be combined in practice to yield a consistent estimate of the
difference.

This has led many researchers either to jointly estimate structurally
similar networks \citep{Chiquet2011Inferring, Danaher2011Joint,
  Guo2011Joint, Mohan2014Node, Ma2016Joint, Majumdar2018Joint} or to
directly estimate the difference \citep{Zhao2014Direct,
  Xu2016Semiparametric, Liu2017Support, Fazayeli2016Generalized}.  The
latter approaches tend to have better sample complexity as well as
greater applicability.  The methods we propose also belong to the
latter category.

Our proposal tries to fill two gaps in the existing literature on
differential network estimation.  First, the majority of the
literature on graphical models are developed assuming a particular observation model,
and the growing literature on difference estimation is no exception.
For example, \citet{Xia2015Testing} assume that the data are Gaussian,
whereas \citet{Cai2019Differential} use an Ising model.
By contrast, we work with \emph{general} Markov random fields; we present
  a unified framework for statistical inference in differential
  networks, without the need for developing separate methodology for
  different distributional assumptions.
Following the development of \citet{Sugiyama2008Direct, Sugiyama2012Density, Liu2014Direct,
  Liu2017Support, Fazayeli2016Generalized}, we take the density ratio
approach and estimate the difference directly.
The last three assume a high-dimensional regime, and study consistency of point estimators
defined as solutions to penalized procedures, but the question of
statistical inference is unaddressed.

This brings us to the second gap.  Most of the existing literature on
network difference estimation focuses on producing consistent point
estimates, leaving the question of quantifying uncertainty in those
estimates largely untouched.  The methods we develop in this paper can
be used to construct confidence intervals and carry out hypothesis
tests about the difference of networks parameters.  The theoretical
guarantees we provide hold under a fairly weak set of assumptions.  In
particular, they do not rely on perfect model selection at any stage,
which would have necessitated strong assumptions, e.g.,
incoherence and strong signal strength.
Certain features of our problem, e.g., nonlinearity,
introduce technical challenges in establishing our theoretical
results.

Our paper contributes to the growing literature on statistical
inference on high-dimensional parameter estimates.  Hypothesis testing
and confidence intervals for high-dimensional M-estimators are studied
in \citet{Zhang2011Confidence, Belloni2012Inference,
  Belloni2013Honest, Javanmard2013Confidence,
  Meinshausen2013Assumption, Geer2013asymptotically}. Related ideas
have been developed in the context of Gaussian graphical models
\citep{Ren2013Asymptotic, Jankova2014Confidence, Jankova2017Honest},
elliptical copula models \citep{Barber2015ROCKET, Lu2015Posta}, and
Markov networks \citep{Wang2016Inference, Yu2016Statistical}.
Existing inferential techniques for high-dimensional differential
networks rely on Gaussian observation model and separate estimation
\citep{Xia2015Testing, Belilovsky2016Testing, Liu2017Structural}.
By contrast, our methods also apply to non-Gaussian data and are based on direct difference estimation.

Our Gaussian bootstrap approximation results can be viewed as
another contribution along the lines of \citet{Chen2018Gaussian} and
\citet{Xue2020Distribution}, which build on the ideas of
\citet{Chernozhukov2013Gaussian, Chernozhukov2015Comparison,
  Chernozhukov2017Central}. In particular, the testing procedure
developed in \citet{Xue2020Distribution} relies on a Gaussian
approximation result of the difference of two independent sums in
high-dimensions. Our equal graph test is similar in flavor, but as our
estimator cannot be represented as an independent sum, the proof of
validity requires a careful control of the remainder. Furthermore, our
empirical bootstrap heuristic is an interesting generalization of
\citet{Dezeure2017High} to a non-linear and two-sample problem, the
theoretical exploration of which we leave up to future work.

The rest of this paper is organized as follows.
Section~\ref{sec:preliminaries} discusses some background.  Our
methods are presented in Section~\ref{sec:methodology}, and their
theoretical guarantees are given in Section~\ref{sec:theory}.  We
report the results of our extensive simulation study in
Section~\ref{sec:simulations}, and analyze a real fMRI dataset in
Section~\ref{sec:real_data}.  We conclude with a discussion of
alternatives and future directions in Section~\ref{sec:discussions}.
The proofs of the main results are found in Appendix.  A Julia
package implementing the proposed methods may be obtained from
\url{https://github.com/mlakolar/KLIEPInference.jl}, together with the code to
reproduce the results.


\section{\label{sec:preliminaries}Preliminaries}

We list notations that are used frequently throughout this paper.
Vectors are distinguished from scalars by bold font, e.g., $\vb$. Bold uppercase letters are reserved for matrices, e.g., $\Mb$.
For $d \in \NN$, $[d] = \{1, \dots, d\}$.
For $k \in [d]$, $\eb_k \in \RR^d$ is the $k$th standard basis vector.
For $\vb \in \RR^d$ and $k \in [d]$, we write $v_k$ for the $k$th component of $\vb$.
For $S \subseteq [d]$, $\vb_S \in \RR^d$ with $v_{S,k} = v_k$ for $k \in S$, $v_{S,k} = 0$ else.
Let $\Ical \subseteq [d]$ be an index set.
For a set of scalars $\{v_k\}_{k \in \Ical}$, $(v_k)_{k \in \Ical}$ denotes the $|\Ical|$-vector with the components given by the set.
Similarly, for a set of vectors $\{\vb_k\}_{k \in \Ical}$ with all $\vb_k \in \RR^d$, $[\vb_k]_{k \in \Ical}$ denotes the $d \times |\Ical|$-matrix with the columns given by the set.
Given $\vb_1 \in \RR^{d_1}$ and $\vb_2 \in \RR^{d_2}$, $\vb_1 \cup \vb_2 \in \RR^{d_1 + d_2}$ denotes their concatenation.
For $\vb \in \RR^d$, the partition of $\vb$ induced by a partition $d_1 + d_2 = d$ is denoted $\vb = \sbr{\begin{smallmatrix}\vb_1 \\ \vb_2\end{smallmatrix}}$.
Similarly, for $\Mb \in \RR^{d \times d}$, $\Mb = \sbr{\begin{smallmatrix}\Mb_{11} & \Mb_{12} \\ \Mb_{21} & \Mb_{22}\end{smallmatrix}}$.

The inner product is denoted as $\langle \ub, \vb \rangle = \ub^\top \vb = \sum_{k=1}^{d} u_k v_k$.
We use $\|\dummy\|$ to denote a norm on $\RR^d$, and $\|\dummy\|_*$ to denote its dual, $\|\vb\|_* = \sup_{\|\ub\| \leq 1} |\langle \ub, \vb\rangle|$.
For $p \in [1,\infty]$, $\|\vb\|_p = \rbr{\sum_{k=1}^d |v_k|^p}^{1/p}$ is the usual $\ell_p$-norm of $\vb$.
This is extended first to $q \in (0,1]$ as $\|\vb\|_q = \sum_{k'} |v_{k'}|^q$, and then to $q = 0$ by adopting the convention $0^0 \equiv 0$ so that $\|\vb\|_0 = |\supp(\vb)| = |\{ k : v_{k} \neq 0\}|$.
For $q \in [0,1)$, $\ell_q$-``norms" can be thought of as generalized sparsity measures.
For a matrix $\Mb \in \RR^{d \times d}$, $\|\Mb\| = \|\operatorname{vec}(\Mb)\|$, e.g., $\|\Mb\|_\infty = \max_{1 \leq k, k' \leq d} |M_{k k'}|$.
For $s > 0$, $\vertiii{\Mb}_s$ is the maximum $s$-sparse eigenvalue of $\Mb$: $\vertiii{\Mb}_s = \sup_{\|\vb\|_0 \leq s, \|\vb\| = 1} |\vb^\top \Mb \vb|$.
Also, $\topnorm{\Mb} = \sup_{\|\vb\| \leq 1} \|\Mb \vb\|_*$.

Let $(a_n)_{n \geq 1}$ and $(b_n)_{n \geq 1}$ be sequences.
$a_n \lesssim b_n$ or $a_n = O(b_n)$ whenever $|a_n / b_n| \leq M$ for some $M > 0$ for all sufficiently large $n$.
$a_n \asymp b_n$ when both $a_n \lesssim b_n$ and $b_n \lesssim a_n$.
$a_n = o(b_n)$ means $|a_n / b_n| \to 0$ as $n \to \infty$.
If such a relationship holds with probability approaching $1$, this is distinguished by $\PP$ in the subscript, for example, $O_{\PP}$, $o_{\PP}$, $\lesssim_{\PP}$, $\asymp_{\PP}$.

$\text{REM}$ is a catch-all symbol for the remainder of an
approximation, whose precise definition varies according to the
context and from line to line.  $\Phi$ is the cdf of the standard
Gaussian: $\Phi(z) = \int_{-\infty}^z \phi(t) \, dt$, where
$\phi(z) = e^{-z^2/2} / \sqrt{2\pi}$.  $\Phi^{-1}$ denotes its inverse
function, i.e., the standard Gaussian quantile function.

\subsection{\label{sec:expon_family_pairw}Statement of the problem}

A \emph{Markov network} describes conditional dependencies among a collection of random variables \citep{Lauritzen1996Graphical, Drton2017Structure}.
Let $\xb = (x_v)_{v = 1}^{m}$ be a random vector taking values in $\XX \subseteq \RR^{m}$.
Consider an \emph{undirected graph} $G$ on the node set $V = [m]$, i.e., a pair $G = (V, E)$, where $E$, called the edge set, contains unordered pairs of nodes.
We say that nodes $u, v \in V$ are connected by an edge if $\{u, v\} \in E$.
Formally, a Markov network associated with $G = ([m], E)$ is a collection of $m$-variate distributions such that $x_u \indep x_v \mid (x_w)_{w \neq u, v}$ if and only if $\{u, v\} \notin E$.
Thus, the edge set $E$ describes which pairs of random variables are conditionally independent given all the other variables.

Let $\mathcal{C}(G)$ denote the set of all \emph{cliques} of $G$, i.e., subsets of $V$ for which every pair of nodes is connected by an edge.
It is well-known that any such $\xb$ with a strictly positive density is an exponential family $f(\xb; \gammab) = \exp( \gammab^\top \psib(\xb) ) / Z(\gammab)$ for some $\gammab = (\gamma_C)_{C \in \mathcal{C}(G)}$, $\psib = (\psi_C)_{C \in \mathcal{C}(G)}$, and the normalizing constant $Z(\gammab) = \int \exp( \gammab^\top \psib(\xb) ) \, d\xb$, where  $\gamma_C \in \RR$ and $\psi_C$ is a function of the clique variables $(x_v)_{v \in C}$ only \citep{Hammersley1971Markov}.
A parametric class $\Fcal_\gamma$ of Markov networks is obtained by assuming a fixed $\psib$.

A special case of significance is the class of \emph{pairwise} Markov networks \citep{wainwright08graphical, Yang2013Graphical} that have densities of the form
\begin{equation}
\label{eq:family:pairwise}
	f(\xb; \gammab) = \frac{1}{Z(\gammab)} \exp\rbr{\sum_{v = 1}^{m} \gamma_v \psi_v(x_v) + \sum_{u = 1}^{m} \sum_{v = u+1}^{m} \gamma_{uv} \psi_{uv}(x_u, x_v)},
\end{equation}
where $\gammab = (\gamma_v)_{v=1}^{m} \cup (\gamma_{uv})_{1 \leq u < v \leq m}$, $\psib = (\psi_v)_{v = 1}^{m} \cup (\psi_{uv})_{1 \leq u < v \leq m}$.
For this class, each component function of $\psib$ is at most a function of two variables, and hence $x_u \indep x_v \mid (x_w)_{w \neq u, v}$ if and only if $\gamma_{uv} = 0$, $u \neq v$.
Thus, for a pairwise Markov network, the edge set $E$ dictates which of the pairwise parameters $(\gamma_{uv})_{1 \leq u < v \leq m}$ are nonzero.
A number of well-studied models belong to the pairwise class.

\begin{example}[\label{example:Ising}Ising models]
An Ising model is a family of discrete probability distribution on the vertices of the $m$-dimensional hypercube $\XX = \{\pm 1\}^m$ given by the probability mass function of the form \eqref{eq:family:pairwise} with $\psi_{v}(x_v) = x_v$, $\psi_{uv}(x_u, x_v) = x_u x_v$, and $\gamma_v, \gamma_{uv} \in \RR$.
Thus, the Markov network associated with $G = ([m], E)$ are all Ising models with $\gamma_{uv} \neq 0$ if and only if $\{u, v\} \in E$.
\end{example}

\begin{example}[Gaussian graphical models]
The most-studied example of a probabilistic graphical model is the case of the undirected Gaussian graphical model.
Suppose $\xb \sim \Ncal(\mub, \Sigmab)$. Then, $\xb$ has a density of the form \eqref{eq:family:pairwise} with $\psi_v(x_v) = x_v$, $\psi_{uv}(x_u, x_v) = x_u x_v$, $\gamma_v = (\Sigmab^{-1} \mub)_v$, and $\gamma_{uv} = -[\Sigmab^{-1}]_{uv} / 2$.
Thus, if $\xb$ is in a Gaussian graphical model with the graph $G = ([m], E)$, then the inverse covariance matrix satisfies $[\Sigmab^{-1}]_{uv} \neq 0$ if and only if $\{u, v\} \in E$.
\end{example}


Suppose $f_x = f(\dummy; \gammab_x)$ and $f_y = f(\dummy; \gammab_y)$ are two distributions from the same pairwise family \eqref{eq:family:pairwise} corresponding to parameters $\gammab_x$ and $\gammab_y$, respectively.
Then, the \emph{change} from $f_x = f(\dummy; \gammab_y)$ to $f_y = f(\dummy; \gammab_x)$ can be described by the \emph{difference} $\thetab^* = \gammab_x - \gammab_y$.
In particular, whenever $x_u \indep x_v \mid (x_w)_{w \neq u, v}$ is true for only one of $f_x$ or $f_y$, we have $\theta_{uv}^* = \gamma_{x, uv} - \gamma_{y, uv} \neq 0$.
More generally, the support of $\thetab^*$ gives the pairs of random variables for which the conditional dependence relationship has changed.

The \emph{differential network} is defined as the difference $\thetab^*$ of $\gammab_x$ and $\gammab_y$.
We represent the differential network with a graph $G = (V, E)$ where an edge $\{u, v\} \in E$ is drawn between vertices $u$ and $v$ if and only if $\theta_{uv}^* \neq 0$.
Our goal here is to learn the differential network given independent and identically distributed (i.i.d.) observations from each of $f_x = f(\dummy; \gammab_x)$ and $f_y = f(\dummy; \gammab_y)$.
More precisely, using $\xb^{(1)}, \dots, \xb^{(n_x)} \iidsim f_x$ and $\yb^{(1)}, \dots, \yb^{(n_y)} \iidsim f_y$, we would like to construct confidence intervals or conduct hypothesis tests over possibly high-dimensional sub-vectors of $\thetab^*$ with provably valid simultaneous guarantee at arbitrary user-specified confidence level of $1-\alpha$ for small $\alpha \in [0, 1]$.
This requires an estimate of $\thetab^*$, which we construct in the next section based on the density ratio $f_x / f_y$ without separately estimating the individual parameters $\gammab_x$ and $\gammab_y$.

\subsection{\label{sec:KLIEP}Direct difference estimation via density ratio}

We describe the Kullback-Leibler Importance Estimation Procedure
(KLIEP) \citep{Sugiyama2008Direct} and how it can be used to
directly estimate the differential network $\thetab^*$.

KLIEP is a framework for estimating the density ratio of two probability distributions based on i.i.d.~observations from each. When the distributions are from the same parametric exponential family, the density ratio depends on the underlying pair of parameters only through their \emph{difference} while maintaining the exponential form. Indeed, let $r_{\theta^*} = f_x / f_y$. Then,
\begin{equation*}
	r_{\theta^*}(\xb)
	= \frac{f_x(\xb)}{f_y(\xb)}
	= \frac{Z(\gammab_y) \exp\left( \gammab_x^\top \psib(\xb) \right)}{Z(\gammab_x) \exp\left( \gammab_y^\top \psib(\xb) \right)}
	= \frac{\exp\left( \thetab^{*\top} \psib(\xb) \right)}{Z_y(\thetab^*)},
\end{equation*}
where we have $Z_y(\thetab^*) = \EE_y[ \exp( \thetab^{*\top} \psib(\yb) ) ]$, because
\begin{multline*}
	Z_y(\thetab^*)
	= \frac{Z(\gammab_x)}{Z(\gammab_y)}
	= \frac{\int \exp\rbr{\gammab_x^\top \psib(\xb)} \, d\xb}{Z(\gammab_y)}\\
	= \int \exp\rbr{\thetab^{*\top} \psib(\xb)} \frac{\exp\rbr{\gammab_y^\top \psib(\xb)}}{Z(\gammab_y)} \, d\xb
	= \EE_y\sbr{\exp\rbr{\thetab^{*\top} \psib(\yb)}}.
\end{multline*}
This can be used to derive a procedure that directly learns $\thetab^*$ without having to learn either $\gammab_x$ or $\gammab_y$.
Let $D_{\textnormal{KL}}(f \| g)$ be the \emph{Kullback-Leibler (KL) divergence} for probability densities $f$ and $g$. Recall that $D_{\textnormal{KL}}(f \| g) \geq 0$ with equality if and only if $f = g$ almost everywhere.
Since $f_x = r_{\theta^*} f_y$, $\thetab^* = \arg\min_{\thetab} D_{\textnormal{KL}}(f_x \| r_\theta f_y)$.
Moreover, it is proved in Appendix~\ref{supp:lKLIEP:derivation} that
\begin{equation}
\label{eq:KLIEP:motivation}
\begin{aligned}
	\thetab^*
	&= \arg\min_{\thetab} D_{\textnormal{KL}}(f_x \| r_\theta f_y)\\
	&= \arg\min_{\thetab} \cbr{-\EE_x\sbr{\thetab^\top \psib(\xb)} + \log \EE_y\sbr{\exp\rbr{\thetab^\top \psib(\yb)}}},
\end{aligned}
\end{equation}
where $\EE_x$ denotes the expectation with respect to $f_x$ and $\EE_y$ the expectation with respect to $f_y$.
The \emph{empirical KLIEP loss} $\lKLIEP$ is obtained by replacing each expectation with the corresponding sample average:
\begin{equation}
\label{eq:lKLIEP}
\begin{aligned}
	\lKLIEP(\thetab)
	&= \lKLIEP(\thetab; \xb^{(1)}, \dots, \xb^{(n_x)}, \yb^{(1)}, \dots, \yb^{(n_y)})\\
	&= -\frac{1}{n_x} \sum_{i = 1}^{n_x} \thetab^\top \psib(\xb^{(i)}) + \log \cbr{\frac{1}{n_y} \sum_{j = 1}^{n_y} \exp \rbr{\thetab^\top \psib(\yb^{(j)})}}.
\end{aligned}
\end{equation}
Minimizing $\lKLIEP$ yields the KLIEP estimate $\hat\thetab_{\textnormal{KLIEP}} = \arg\min_{\thetab} \lKLIEP(\thetab)$ as a direct estimate of the differential network $\thetab^*$.
$\lKLIEP$ is convex in $\thetab$, and when it is strictly convex --- which requires $n_y > p$ --- the KLIEP estimate $\hat\thetab_{\textnormal{KLIEP}}$ is known to be approximately normal and unbiased \citep[Chapter 13]{Sugiyama2012Density}.

In the high-dimensional setting with $n_y \leq p$, the minimizer of $\lKLIEP$ is no longer unique, and regularization becomes necessary for consistent estimation.
In this setting, \cite{Liu2017Support} and \cite{Fazayeli2016Generalized} proposed regularized versions of KLIEP using norm penalties.
In particular, \cite{Liu2017Support} proposed the \emph{sparse KLIEP}
\begin{equation}
\label{eq:spKLIEP}
	\check\thetab
	= \check\thetab(\lambda)
	= \arg\min_{\thetab} \lKLIEP(\thetab; \xb^{(1)}, \dots, \xb^{(n_x)}, \yb^{(1)}, \dots, \yb^{(n_y)}) + \lambda \|\thetab\|_1,
\end{equation}
where $\lambda > 0$ is a regularization parameter to be chosen by the user.
They show that when $\thetab^*$ is sparse, the support of $\check\thetab$ consistently recovers the support of $\thetab^*$ for suitable choices of $\lambda$.
However, such results typically require additional conditions, e.g., a lower bound on the minimal signal strength and incoherence of the Hessian, which may be restrictive for many real data applications.
Furthermore, these are essentially results about the accuracy of the point estimates, whereas to construct confidence intervals or conduct hypothesis tests, one needs information about the distribution of the estimators.
This is difficult for regularized estimators, as we shall see next.

\subsection{\label{sec:stat-infer-high}De-biasing}

Challenges arise when a regularized estimator $\check\thetab$, e.g., the sparse KLIEP estimator \eqref{eq:spKLIEP}, is used for statistical inference. Regularization produces a non-negligible bias, and the distribution of the resulting estimator is typically intractable \citep[see][and references therein]{Ning2014General}.

We propose to deal with this issue by de-biasing each component of $\check\thetab$.
For convenience, adopt a linear indexing so that $\gammab = (\gamma_k)_{k = 1}^{p}$ and $\psib = (\psi_k)_{k = 1}^{p}$, where $p$ is the total number of parameters.
Suppose we wish to obtain a de-biased estimate of $\theta_k^*$ for some $k \in [p]$.
Let $\thetab_{k^c}^* = \thetab_{[p] \setminus \{k\}}^* \in \RR^{p-1}$ denote the vector of remaining $p-1$ parameters.
This is the nuisance parameter for carrying out statistical inference for $\theta_k^*$.
Abusing the notation somewhat, we write the resulting partition as $\thetab = (\theta_k,\thetab_{k^c})$.
Define $\omegab_k^*$ as the vector satisfying $\EE[\nabla^2 \lKLIEP(\thetab^*)] \omegab_k^* = \eb_k$, and let $\check\omegab_k$ be a consistent estimator of $\omegab_k^*$.

Our method offers two options for constructing an approximately normal and unbiased estimator $\hat\theta_k$ of $\theta_k^*$ that are asymptotically equivalent \citep{Chernozhukov2015Valid}.
The first option is to use the \emph{one-step} estimator \citep{Vaart1998Asymptotic,Zhang2011Confidence,Geer2013asymptotically}:
\begin{equation}
\label{eq:def_one_step}
	\hat\theta_k^{\textnormal{1+}} = \check\theta_k - \check\omegab_k^\top \nabla \lKLIEP(\check\thetab).
\end{equation}
This approximately solves a modified score equation $\check\omegab_k^\top \nabla \lKLIEP(\theta_k,\check\thetab_{k^c}) = 0$, where $\check\thetab_{k^c}$ is defined via $\check\thetab = (\check\theta_k, \check\thetab_{k^c})$, by taking one Newton iteration starting from $\check\theta_k$.
In Section~\ref{sec:theory:main}, we prove that the one-step estimator $\hat\theta_k^{\textnormal{1+}}$ is an approximately normal and unbiased estimator of $\theta_k^*$.

When $\check\thetab$ and $\check\omegab_k$ are both sparse vectors, de-biasing may be carried out via the so-called \emph{double selection} \citep{Chernozhukov2015Valid}.
Let $\tilde\thetab$ be the estimate obtained by re-fitting to the union of the supports of $\check\thetab$ and $\check\omegab_k$, i.e.,
\begin{equation}
\label{eq:def_double_selection}
	\tilde\thetab = \arg\min_\theta \lKLIEP(\thetab) \quad\text{subject to}\quad \supp(\thetab) \subseteq \{k\} \cup \supp(\check\thetab) \cup \supp(\check\omegab_k).
\end{equation}
Then, the {double-selection} estimator $\hat\theta_k^{\textnormal{2+}}$ is defined as the $k$th component of $\tilde\thetab$.
Intuitively, by including the estimated supports of both $\thetab^*$ and $\omegab_k^*$, the double selection procedure achieves robustness to errors from either model selection procedure.
Provided that $\tilde\thetab$ is as accurate as $\check\thetab$ --- which would be the case for sparse or approximately sparse $\thetab^*$ and $\omegab_k^*$ --- $\hat\theta_k^{\textnormal{2+}}$ is asymptotically equivalent to $\hat\theta_k^{\textnormal{1+}}$.

For a derivation of \eqref{eq:def_one_step} in the context of KLIEP, see Appendix~\ref{supp:proofs:debiasing}.
A general discussion of the relationship of one-step estimation and double selection may be found in \citet{Chernozhukov2015Valid}.



\section{\label{sec:methodology}Methodology}

We propose a procedure for constructing an approximately normal and unbiased estimator of the differential network (\Cref{sec:methodology:single}).
We then give two bootstrap sketching procedures for estimating the quantiles of a max-type statistic based
on the estimator from \Cref{sec:methodology:single}, and show how they can be used for simultaneous
inference (\Cref{sec:methodology:multiple}).

\subsection{\label{sec:methodology:single}Sparse Kullback-Leibler Importance Estimation With de-biasing (SparKLIE+)}

We present Procedure~\ref{procedure:KLIEP+}, which is a general recipe for de-biasing regularized KLIEP estimates for each $\theta_k^*$ in $k \in \Ical$, where $\Ical \subseteq [p]$ is the set of indices for the parameters of inferential interest. The procedure uses a general norm penalty for regularization.

\begin{procedure}[\label{procedure:KLIEP+}Kullback-Leibler Importance Estimation With de-biasing (KLIE+)]
\begin{framed}
\begin{algorithmic}
\REQUIRE {Data $\Xb_{n_x} = \{\xb^{(i)}\}_{i=1}^{n_x}$, $\Yb_{n_y} = \{\yb^{(j)}\}_{j=1}^{n_y}$, positive regularization parameters $\lambda_\theta, \lambda_k$, $k \in \Ical$}
\ENSURE {De-biased estimates $\hat\theta_k$, $k \in \Ical$}
\STATE {\textit{Step 1.} Find an initial estimate of $\thetab^*$
\begin{equation}
\label{eq:KLIEP+:1}
	\check\thetab = \arg\min_{\theta} \lKLIEP(\thetab;\Xb_{n_x},\Yb_{n_y}) + \lambda_\theta \|\thetab\|.
\end{equation}}
\FOR {$k \in \Ical$}
\STATE{\textit{Step 2.} Find an initial estimate of $\omegab_k^*$
\begin{equation}
\label{eq:KLIEP+:2}
	\check\omegab_k = \arg\min_{\omega} \frac 12 \omegab^\top \nabla^2 \lKLIEP(\check\thetab)\omegab - \omegab^\top \eb_k + \lambda_k \|\omegab\|.
\end{equation}}
\STATE{\textit{Step 3.} De-bias, either by \eqref{eq:def_one_step} or by \eqref{eq:def_double_selection}, to obtain $\hat\theta_k$.}
\ENDFOR
\RETURN {$\hat\theta_k$, $k \in \Ical$}
\end{algorithmic}
\end{framed}
\end{procedure}

A general Gaussian approximation bound for \Cref{procedure:KLIEP+} will be given below in \Cref{thm:main} in \Cref{sec:theory:main}. The result is valid as long as the initial estimators from \eqref{eq:KLIEP+:1} and \eqref{eq:KLIEP+:2} are sufficiently accurate.
For example, this is the case for sparse or approximately sparse $\thetab^*$ and $\omegab_k^*$ when the $\ell_1$-penalty is used (\Cref{lem:consistency:1,lem:consistency:2} in Appendix~\ref{supp:consistency:l1}).
We call this procedure Sparse Kullback-Leibler Importance Estimation with de-biasing (SparKLIE+), with SparKLIE+1 referring to SparKLIE+ that uses one-step \eqref{eq:def_one_step} for de-biasing and SparKLIE+2 referring to the double selection \eqref{eq:def_double_selection} option.

\begin{remark}[\label{remark:alternatives}Alternative procedures for initial estimation]
It is possible to use other procedures for either of the initial estimation steps as long as the errors satisfy $\|\check\thetab - \thetab^*\| \cdot \|\check\omegab_k - \omegab_k^*\| = o_{\PP}(n^{-1/2})$.
We give examples in the case of the $\ell_1$-penalty.
In Appendix~\ref{supp:implementation:autoscaling}, we detail an autoscaling procedure for each step that allows the regularization parameter to be chosen in a data-independent way while yielding consistent estimation.
We may also re-fit the model on the estimated support \citep{Belloni2013Least}.
Finally, it is also possible to use a constrained procedure, similar to the method of \citet{Ning2014General}, where instead of \eqref{eq:KLIEP+:2}, one solves
\[
	\min\ \|\omegab\|_1 \quad\text{subject to}\quad \|\nabla^2 \lKLIEP(\check\thetab) \omegab - \eb_k\|_\infty \leq \lambda_k.
\]
\end{remark}

\begin{remark}[\label{remark:lambda}Regularization parameters]
\Cref{procedure:KLIEP+} assumes that the user has already picked out the regularization parameters $\lambda_\theta, \lambda_k$, $k \in \Ical$.
However, the optimal choice, as dictated by \Cref{lem:grad:1,lem:grad:2} in Appendix~\ref{supp:grads}, depends on constants related to the regularity of the density ratio, which are typically unknown.
In Appendix~\ref{supp:exper5}, we empirically study the sensitivity of \Cref{procedure:KLIEP+} to the choice of regularization parameters and find that the performance is robust across a wide range of regularization levels.
Furthermore, in Appendix~\ref{supp:implementation:autoscaling}, we further provide alternative procedures for Steps 1 and 2 that allows for problem-independent choices of penalty levels.
This is the version of Procedure~\ref{procedure:KLIEP+} we use in \Cref{sec:simulations,sec:real_data}.
\end{remark}

\subsubsection{Variance of the SparKLIE+ estimator}

For statistical inference, we also need a consistent estimator of the variance of $\sqrt{n} \, \hat\theta_k$, $n = n_x + n_y$.
Define the \emph{empirical density ratio estimate}
\begin{equation}
	\hat r_\theta(\yb) = \exp\rbr{\thetab^\top \psib(\yb)} / \hat
Z_y(\thetab), \quad\text{where}\quad \hat Z_y(\thetab) = \frac{1}{n_y} \sum_{j=1}^{n_y} \exp\rbr{\thetab^\top \psib(\yb^{(j)})}.
\label{eq:rhat}
\end{equation}
Let $\hat\Sbb_\psi$ and $\hat\Sbb_{\psi\hat r}(\check\thetab)$ be the sample covariance matrices of $\{\psib(\xb^{(i)})\}_{i=1}^{n_x}$ and $\{\psib(\yb^{(j)}) \hat r_{\check\theta}(\yb^{(j)})\}_{j=1}^{n_y}$, i.e.,
\begin{gather*}
\hat\Sbb_\psi = \frac{1}{n_x} \sum_{i=1}^{n_x} \psib(\xb^{(i)}) \psib(\xb^{(i)})^\top - \overline{\psib} \overline{\psib}^\top,\\
\hat\Sbb_{\psi\hat r}(\thetab) = \frac{1}{n_y} \sum_{i=1}^{n_y} \hat r_\theta^2(\yb^{(j)}) \psib(\yb^{(j)}) \psib(\yb^{(j)})^\top - \hat\mub(\thetab) \hat\mub(\thetab)^\top,
\end{gather*}
where
\begin{align}
\overline{\psib} &= \frac{1}{n_x} \sum_{i=1}^{n_x} \psib(\xb^{(i)}), & \hat\mub(\thetab) &= \frac{1}{n_y} \sum_{j=1}^{n_y} \psib(\yb^{(j)}) \, \hat r_\theta(\yb^{(j)}).
\label{eq:psibar:muhat}
\end{align}
Let $\hat\Sbb_{\textnormal{pooled}}(\check\thetab)$ be the pooled covariance
\[
	\hat\Sbb_{\textnormal{pooed}}(\check\thetab) = \frac{n}{n_x} \, \hat\Sbb_\psi + \frac{n}{n_y} \, \hat\Sbb_{\psi\hat r}(\check\thetab).
\]
Finally, a consistent estimator of the variance of $\sqrt{n} \, \hat\theta_k$ is
\begin{equation}
\label{eq:varest}
	\hat\sigma_k^2 = \check\omegab_k^\top \hat\Sbb_{\textnormal{pooled}}(\check\thetab) \check\omegab_k.
\end{equation}
This estimates the variance of $\sqrt{n} \, \omegab_k^{*\top} \nabla\lKLIEP(\thetab^*)$, which we show is asymptotically equivalent to $\sqrt{n} \, (\hat\theta_k - \theta_k^*)$ in the proof of \Cref{thm:main} in Appendix~\ref{supp:proofs:thm:main}.
By \Cref{lem:varest} in Appendix~\ref{app:consistency_var_est}, $\hat\sigma_k^2$ is consistent if both $\check\thetab$ and $\check\omegab_k$ are.

\Cref{cor:main} in \Cref{sec:theory:main} implies that if $z_q = \Phi^{-1}(q)$ is the $q$-quantile of a standard Gaussian, then $\PP\{\sqrt{n} \, (\hat\theta_k - \theta_k^*) / \hat\sigma_k \leq z_q\} \approx \Phi^{-1}(z_q) = q$.
Thus, $\hat\theta_k \pm z_{1-\alpha/2} \times \hat\sigma_k / \sqrt{n}$ is an asymptotically valid $100 \times (1-\alpha) \%$ confidence interval (CI) for $\theta_k^*$.
Similarly, the test that rejects for $\sqrt{n} \, |\hat\theta_k - \theta_{0k}| / \hat\sigma_k > z_{1-\alpha/2}$ is asymptotically level-$\alpha$ for the one-dimensional null hypothesis $\Hcal_{0k}: \theta_k^* = \theta_{0k}$.
In Section~\ref{sec:simulations}, we verify with simulations that the approximations are fairly accurate and robust even at small sample sizes.

\subsection{\label{sec:methodology:multiple}High-dimensional inference via bootstrap sketched quantiles}

In Section~\ref{sec:methodology:single}, we proposed SparKLIE+, a procedure for obtaining an asymptotically unbiased estimator of a component of the differential network. Iterating Step~3 of SparKLIE+ over all edges yields an unbiased estimator $\hat\thetab$ of the differential network $\thetab^*$. To make inferences about the structure of $\thetab^*$ using $\hat\thetab$, one may construct a simultaneous confidence region or conduct a simultaneous hypothesis test. This raises issues of multiple comparisons.

We deal with this problem by a bootstrap approximation of the quantiles of the following statistic
\begin{equation}
\label{eq:def_max_statistic}
	T = T_{n_x, n_y} = \max_{k = 1, \dots, p} \sqrt{n} \, |\hat\theta_k - \theta_k^*|, \quad\text{where}\quad n = n_x + n_y.	
\end{equation}
Let $c_{T, q}$ be the $q$-quantile of $T$.
Then, it is easy to verify that $\hat\thetab \pm c_{T, 1-\alpha} / \sqrt{n}$ is a $100 \times (1-\alpha) \%$ confidence region for $\thetab^*$. Similarly, the test that rejects if $\max_k |\hat\theta_k| > c_{T, 1-\alpha} / \sqrt{n}$ controls the family-wise error rate at level $\alpha$ for the null hypothesis $H_0: \theta_k^* = 0$ for all $k \in [p]$. This approach has the advantage of adapting to the correlations among $\hat\theta_k$'s. Thus, given $c_{T, q}$ --- or an accurate estimator thereof --- we can learn the differential network structure while controlling the type I error rate.

However, in high-dimensions, it is itself a highly nontrivial problem to estimate $c_{T, q}$ with sufficient accuracy \citep[see][and references therein]{Chernozhukov2013Gaussian, Chernozhukov2017Central, Deng2020Beyond}. In this section, we present two bootstrap-based methods for estimating $c_{T, q}$.

Our first proposal employs the Gaussian multiplier bootstrap. Recall the definitions of $\hat r_{\theta}$ from \eqref{eq:rhat}, and of $\bar\psib$ and $\hat\mub(\thetab)$ from \eqref{eq:psibar:muhat}.

\begin{procedure}[\label{procedure:bootstrap:Gaussian}Gaussian multiplier bootstrap sketching for estimating quantiles of $T$ ]
\begin{framed}
\begin{algorithmic}
\REQUIRE {Data $\Xb_{n_x} = \{\xb^{(i)}\}_{i=1}^{n_x}$, $\Yb_{n_y} = \{\yb^{(j)}\}_{j=1}^{n_y}$; the outputs $\check\thetab$ and $\check\omegab_k$, $k \in \Ical$, of \eqref{eq:KLIEP+:1} and \eqref{eq:KLIEP+:2} from \Cref{procedure:KLIEP+}}
\ENSURE {A Gaussian bootstrap estimate $\hat c_{T, q}$ of $c_{T, q}$}
\FOR {$b = 1, \dots, n_b$}
\STATE {Draw Gaussian weights $\xi_x^{(b,1)}, \dots, \xi_x^{(b,n_x)}, \xi_y^{(b,1)}, \dots, \xi_y^{(b,n_y)} \iidsim \Ncal(0,1)$.}
\STATE {Compute
\begin{multline}
\label{eq:T-boot:Gaussian}
	\hat T^{(b)} = \max_k \sqrt{n} \Bigg| \Bigg\langle \check\omegab_k, \frac{1}{n_x} \sum_{i=1}^{n_x} \rbr{\psib(\xb^{(i)}) - \overline{\psib}} \xi_x^{(b,i)}\\
	- \frac{1}{n_y} \sum_{j=1}^{n_y} \rbr{\psib(\yb^{(j)}) \hat r_{\check\theta}(\yb^{(j)}) - \hat\mub(\check\thetab)} \xi_y^{(b,j)} \Bigg\rangle \Bigg|.
\end{multline}}
\ENDFOR
\RETURN {$\hat c_{T, q}$, the $q$ sample quantile of $\{\hat T^{(b)} : b = 1, \dots, n_b\}$.}
\end{algorithmic}
\end{framed}
\end{procedure}

\Cref{procedure:bootstrap:Gaussian} may be
  procedure for estimating the $(1-\alpha)$-quantile of the maximum of
  $|\Ncal(\zero, \hat\Sigmab)|$,
  where
  $\hat\Sigmab = \check\Omegab^\top \hat\Sbb_\text{pooled}
  \check\Omegab$, $\check\Omegab = [\check\omegab_k]_{k = 1}^{p}$, and
  $\hat\Sbb_\text{pooled}$ is defined in \eqref{eq:varest}.  Since we
  can show that $\hat\thetab - \thetab^* \approx \Ncal(\zero, \Sigmab^*)$
  for some fixed $\Sigmab^*$ and, moreover, $\hat\Sigmab \approx \Sigmab^*$, we claim
  that $\hat c_{T, q}$ is a good estimate of the
  $q$-quantile of $T$. This intuition is
  formally stated in \Cref{thm:bootstrap} in
  \Cref{sec:theory:bootstrap}.

  Although \Cref{procedure:bootstrap:Gaussian} is accurate for
  sufficiently large sample sizes, at smaller values of $n_x$ and
  $n_y$, empirical bootstrap tends to yield more robust estimates of
  the quantiles.  The procedure below, based on the empirical bootstrap,
  is what we recommend in practice.

\begin{procedure}[\label{procedure:bootstrap:empirical}Empirical bootstrap sketching for estimating quantiles of $T$]
\begin{framed}
\begin{algorithmic}
\REQUIRE {Data $\Xb_{n_x} = \{\xb^{(i)}\}_{i=1}^{n_x}$, $\Yb_{n_y} = \{\yb^{(j)}\}_{j=1}^{n_y}$; the outputs $\check\thetab$ and $\check\omegab_k$, $k \in \Ical$, of \eqref{eq:KLIEP+:1} and \eqref{eq:KLIEP+:2} from Procedure~\ref{procedure:KLIEP+}}
\ENSURE {An empirical bootstrap estimate $\hat c_{T, q}$ of $c_{T, q}$}
\FOR {$b = 1, \dots, n_b$}
\STATE {Re-sample $\Xb_{n_x}^{(b)} = \{\xb^{(b, 1)}, \dots, \xb^{(b, n_x)}\}$ and $\Yb_{n_y}^{(b)} = \{\yb^{(b, 1)}, \dots, \yb^{(b, n_y)}\}$ uniformly at random with replacement.}
\FOR {$k \in \Ical$}
\STATE {For replicating SparKLIE+1 estimate \eqref{eq:def_one_step}, $\hat\theta_k^{(b)} = \check\theta_k - \check\omegab_k^\top \nabla\lKLIEP(\check\thetab; \Xb_{n_x}^{(b)}, \Yb_{n_y}^{(b)})$.}
\STATE {For replicating SparKLIE+2 estimate \eqref{eq:def_double_selection}, $\hat\theta_k^{(b)}$, the $k$th component of
\[
	\arg\min_\theta \lKLIEP(\thetab; \Xb_{n_x}^{(b)}, \Yb_{n_y}^{(b)}) \text{ subject to } \supp(\thetab) \subseteq \{k\} \cup \supp(\check\thetab) \cup \supp (\check\omegab_k).
\]}
\ENDFOR
\STATE {Compute
\begin{equation}
\label{eq:T-boot:empirical}
	\hat T^{(b)} = \max_k \sqrt{n} \, |\hat\theta_k^{(b)} - \hat\theta_k|.
\end{equation}}
\ENDFOR
\RETURN {$\hat c_{T, q}$, the $q$ sample quantile of $\{\hat T^{(b)} : b = 1, \dots, n_b\}$.}
\end{algorithmic}
\end{framed}
\end{procedure}

Note that only Step 3 of Procedure~\ref{procedure:KLIEP+} is repeated in Procedure~\ref{procedure:bootstrap:empirical}. This is akin to the use of $\check\thetab$ and $\check\omegab_k$, $k \in \Ical$, in Procedure~\ref{procedure:bootstrap:Gaussian}.

We give a heuristic argument in support of Procedure~\ref{procedure:bootstrap:empirical}, leaving the formal proof to future work.
For the sake of argument, consider the infeasible estimator $\hat\theta_k^{1*} = \theta_k^* - \omegab_k^{*\top} \nabla \lKLIEP(\thetab^*)$ or $\hat\theta_k^{2*}$, the $k$th component of $\arg\min_{\thetab} \nabla \lKLIEP(\thetab)$ subject to $\supp(\thetab) \subseteq \{k\} \cup \supp(\thetab^*) \cup \supp(\omegab_k^*)$.
In other words, $\hat\theta_k^{1*}$ or $\hat\theta_k^{2*}$ is the result of applying \eqref{eq:def_one_step} or \eqref{eq:def_double_selection}, but with the true parameters $\thetab^*$ and $\omegab_k^*$ replacing the initial estimates $\check\thetab$ and $\check\omegab_k$.
It is easy to see that both $\hat\theta_k^{1*}$ and $\hat\theta_k^{2*}$ are approximately normal and unbiased estimators, and that making the same replacement in Procedure~\ref{procedure:bootstrap:empirical} would yield bootstrap replicates of $\hat\theta_k^{1*}$ and $\hat\theta_k^{2*}$.
Because $\check\thetab$ and $\check\omegab_k$ are consistent estimators, we expect Procedure~\ref{procedure:bootstrap:empirical} to be approximately valid for bootstrapping the SparKLIE+ estimator $\hat\theta_k^{1+}$ or $\hat\theta_k^{2+}$.
This intuition is verified in simulations in Section~\ref{sec:simulations:2}.


\section{\label{sec:theory}Theory}

In this section, we establish statistical validity of the inference procedures discussed in
\Cref{sec:methodology:single} and \Cref{sec:methodology:multiple} under two model assumptions introduced in \Cref{sec:conditions}.

\subsection{\label{sec:conditions}Assumptions}

We discuss two sufficient conditions that imply the accuracy of
Gaussian approximation.  The first is about the regularity of the
density ratio $r_\theta(\yb)$.

\begin{condition}[bounded density ratio model]
\label[condition]{cond:bdr}
There exist $\varrho > 0$ such that
\begin{equation*}
M_r^{-1} \leq r_\theta(\yb) \leq M_r \text{ a.s.~for all } \thetab \text{ with } \|\thetab - \thetab^*\| \leq \varrho
\end{equation*}
for some $M_r = M_r(\varrho) \geq 1$.
\end{condition}

For convenience, we fix $\varrho = \|\thetab^*\|$.  \Cref{prop:bss}
says that \Cref{cond:bdr} is equivalent to a boundedness condition on
the sufficient statistics, a claim that was stated without proof for
$\ell_2$-norm in \citet{Liu2017Support}.  We generalize the statement,
and prove it in Appendix~\ref{supp:pfs_for_bdr}.

\begin{proposition}[bounded sufficient statistics] \label[proposition]{prop:bss}
Condition 1 is satisfied if and only if $\|\psib(\xb)\|_* \leq M_\psi$ a.s.~for some $M_\psi < \infty$.
\end{proposition}

In general, regularity conditions on the density ratio tend to induce
even stronger regularity conditions on the sufficient statistics.  The
identity
$\hat Z_y(\thetab) / Z_y(\thetab) \equiv n_y^{-1} \sum_{j=1}^{n_y}
r_\theta(\yb^{(j)})$ implies
$\hat Z_y(\thetab) / Z_y(\thetab) \in [M_r^{-1},M_r]$.  Moreover,
$\hat r_\theta(\yb) \equiv (\hat Z_y(\thetab) / Z_y(\thetab))
r_\theta(\yb)$, so that
\begin{equation*}
M_r^{-2} \leq M_r^{-1} \rbr{1 - o_{\PP}(1)} \leq \hat r_\theta(\yb) \leq M_r \rbr{1 + o_{\PP}(1)} \leq M_r^2.
\end{equation*}
The outer bounds are obvious.  The inner bounds require a
concentration result (\Cref{lem:Hoeffding_r} in Appendix~\ref{supp:pfs_for_bdr}).

When \Cref{procedure:KLIEP+} is implemented with the $\ell_1$-penalty,
it is natural to impose Condition~1 with the $\ell_1$-norm, which by
\Cref{prop:bss} is equivalent to imposing an $\ell_\infty$-bound on
the sufficient statistics.  Thus, this choice of penalty works nicely
with models that take values on a bounded domain, such as Ising
models, Potts models, or truncated Gaussians with bounded support.
Indeed, for the Ising model defined in Example \ref{example:Ising},
$\|\psib(\xb)\|_\infty = 1$ but $\|\psib(\xb)\|_2^2 = p$.

The second are regularity conditions on the population covariances of
$\psib(\xb)$ under $f_x$ and $f_y$, as well as that of
$(\psib(\yb) - \mub_\psi) r_{\theta^*}(\yb)$ under $f_y$.  Recall
$\Sigmab_\psi = \Cov_x[ \psib(\xb) ]$, and let
$\Sigmab_{\psi r} = \Cov_y[ (\psib(\yb) - \mub_\psi) r_{\theta^*}(\yb)
]$, where
$\mub_\psi = \EE_x[\psib(\xb)] = \EE_y[\psib(\yb) r_{\theta^*}(\yb)]$.

\begin{condition}[bounded population eigenvalues] \label[condition]{cond:bddpopeigs}
There exist $0 < \underline{\kappa} \leq \bar\kappa < \infty$ such that
\begin{align*}
\underline\kappa \leq \min_{\|\vb\| = 1, \vb \neq \zero} \vb^\top \Sigmab_\psi \vb &\leq \max_{\|\vb\| = 1, \vb \neq \zero} \vb^\top \Sigmab_\psi \vb \leq \bar\kappa,\\
\underline\kappa \leq \min_{\|\vb\| = 1, \vb \neq \zero} \vb^\top \Sigmab_{\psi r} \vb &\leq \max_{\|\vb\| = 1, \vb \neq \zero} \vb^\top \Sigmab_{\psi r} \vb \leq \bar\kappa.
\end{align*}
\end{condition}

\Cref{cond:bddpopeigs} is a natural one, and ensures that the problem
is well behaved \citep{Liu2017Support}.  A lower bound on the minimum
eigenvalues ensures that the model is non-degenerate.  The upper bound
ensures that $\lKLIEP$ \eqref{eq:lKLIEP} is smooth, and can be regarded as
analogous to the assumption on the log-normalizing function in
\citet{Yang2013Graphical}.  These bounds will naturally appear in
bounding the convergence of $\nabla^2 \lKLIEP(\thetab)$ to
$\Sigmab_\psi$, as well as in bounding the variance of the estimator
$\sigma_k^2$.

Conditions imposed here are weaker than those in
\citet{Liu2017Support}, as we do not hope to correctly identify the
support of the parameter $\thetab^*$.  In particular, we do not need
to assume the incoherence condition, nor do we need to require that
the nonzero components of $\thetab^*$ be large enough.

Recall $\thetab^* = \gammab_x - \gammab_y$ and
$\omegab_k^* = \Sigmab_\psi^{-1} \eb_k$ where
$\Sigmab_\psi = \Cov_x[\psib(\xb)]$ and $k \in [p]$.  To facilitate
the discussion of rates in the next two sections, we introduce
additional notations.  Let $n = n_x + n_y$.  We view $n_x$, $n_y$,
$p$, $s_\theta = s_{\theta, q_\theta} = \|\thetab^*\|_{q_\theta}$,
$s_k = s_{k, q_k} = \|\omegab_k^*\|_{q_k}$ as sequences indexed by $n$
and possibly diverging to $\infty$.  $n_x$ and $n_y$ are characterized
by sequences $\eta_{x,n}$ and $\eta_{y,n}$ in $(0,1)$ such that
$\eta_{x,n} + \eta_{y,n} \equiv 1$, $n_x = n_{x,n} = \eta_{x,n} n$ and
$n_y = n_{y,n} = \eta_{y,n} n$.  In particular, this implies that
$n \asymp n_x \asymp n_y$.

The bounds we give below are finite-sample in the sense
  that they are given as functions of $n$, $p$, $s_\theta$, $s_k$.
  They can be used to study the asymptotic behavior as $n \to \infty$
  by considering a sequence of models
  $(\thetab^*, \Sigmab_\psi) = (\thetab^*_n, \Sigmab_{\psi,n})$ such
  that the induced sequence of $p$, $s_\theta$, $s_k$, etc.~satisfy
  the side conditions of each theorem.

\subsection{Finite-sample Gaussian approximation result for the SparKLIE+1}
\label{sec:theory:main}

\Cref{thm:main} gives a family of Gaussian approximation bounds for
\Cref{procedure:KLIEP+}.

Let $k \in [p]$.  Let $\check\thetab$ and $\check\omegab_k$ denote the
outputs of Steps~1 and 2 of \Cref{procedure:KLIEP+}.  For
$\lambda_\theta, \lambda_k, \delta_\theta, \delta_k, \delta_\sigma \in
[0,1)$, define an event
\begin{multline*}
\Ecal_\text{one}
= \Ecal_\text{one}(\lambda_\theta, \lambda_k, \delta_\theta, \delta_k, \delta_\sigma) =\\
\cbr{
\begin{array}{r c r c}
\textnormal{(G.1)} & 2 \|\nabla \lKLIEP(\thetab^*)\|_* \leq \lambda_\theta,&
\textnormal{(G.2)} & 2 \|\nabla^2 \lKLIEP(\thetab^*) \omegab_k^* - \eb_k\|_* \leq \lambda_k,\\
\textnormal{(E.1)} & \|\check\thetab - \thetab^*\| \leq \delta_\theta,&
\textnormal{(E.2)} & \|\check\omegab_k - \omegab_k^*\| \leq \delta_k,\\
\textnormal{(B.1)} & \abr{1 - \frac{\hat Z_y(\thetab^*)}{Z_y(\thetab^*)}} \lesssim \lambda_\theta,&
\textnormal{(B.2)} & \abr{\frac{1}{n_y} \sum_{j=1}^{n_y} \langle \omegab_k^*, \mub_\psi - \psib(\yb^{(j)}) \rangle \, r_{\theta^*}(\yb^{(j)})} \lesssim \lambda_k,\\
\textnormal{(V.1)} & 4 \|\hat\Sbb_\psi - \Sigmab_\psi\|_* \leq \delta_\sigma,&
\textnormal{(V.2)} & 4 \|\hat\Sbb_{\psi\hat r}(\thetab^*) - \Sigmab_{\psi r}\|_* \leq \delta_\sigma
\end{array}}.
\end{multline*}

\begin{theorem}
\label{thm:main}
Assume Conditions 1 and 2.
Let $\hat\theta_k$ be the estimator constructed by \Cref{procedure:KLIEP+} with one-step approximation as
\begin{equation*}
\hat\theta_k = \check\theta_k - \check\omegab_k^\top \nabla \lKLIEP(\check\thetab).
\end{equation*}
Suppose $\PP(\Ecal_\text{one}) \geq 1 - \varepsilon_{\text{one},n}$ for some $\lambda_\theta, \lambda_k, \delta_\theta, \delta_k, \delta_\sigma \in [0,1)$.
Then,
\begin{equation*}
\sup_{t \in \RR} \left| \PP\left\{ \sqrt{n} \, (\hat\theta_k - \theta_k^*) / \hat\sigma_k \leq t \right\} - \Phi(t) \right|\\
\leq \Delta_1 + \Delta_2 + \Delta_3 + \varepsilon_{\text{one},n},
\end{equation*}
where
\begin{gather*}
\Delta_1
\lesssim \sqrt{\frac{\bar\kappa^2 / \underline\kappa}{\eta_{x,n} \eta_{y,n}}} \frac{\|\omegab_k^*\|}{\sqrt{n}},\quad
\Delta_2
\lesssim \sqrt{\frac{\eta_{x,n} \eta_{y,n}}{\underline\kappa / \bar\kappa^2}} \rbr{(\delta_\theta+\lambda_\theta) (\delta_k+\lambda_k) + \|\omegab_k^*\| \delta_\theta^2} \sqrt{n},\\
\Delta_3
\lesssim (\bar\kappa^2 / \underline\kappa) \|\omegab_k^*\|^2 (\delta_\sigma+\delta_\theta) + \delta_k^2.
\end{gather*}
\end{theorem}

The proof is in Appendix~\ref{supp:proofs:thm:main}.  We highlight some of the main
technical difficulties.  To prove \Cref{thm:main}, we need to find a
linear approximation of $\sqrt{n} \, (\hat\theta_k - \theta_k^*)$ that
is easy to analyze.  This is not so obvious due to nonlinearity of
$\lKLIEP$.  Our results require a delicate control of the bias that
arises from using the empirical density ratio estimates, as we need to
make sure that the error terms are vanishing even after $\sqrt{n}$
scaling.  This is in contrast to \citet{Liu2017Support} or
\citet{Fazayeli2016Generalized} where it sufficed to control the
gradient in the dual norm.

\Cref{thm:main} gives a result for a general norm penalty in
\Cref{procedure:KLIEP+}.  When specialized to SparKLIE+1, we have the
following result.

\begin{theorem}
\label{cor:main}
Assume \Cref{cond:bdr} with the $\ell_1$-norm and
\Cref{cond:bddpopeigs}.  Let $\hat\theta_k$ be the SparKLIE+1
estimator with tuning parameters
\begin{equation*}
\lambda_\theta \asymp \rbr{\frac{\log p}{n}}^{1/2}
\quad\text{and}\quad
\lambda_k \asymp s_{k, q_k}^{1/(2-q_k)} \rbr{\frac{\log p}{n}}^{1/2}.
\end{equation*}
Let $s$ be a sequence of integers satisfying
$s \geq s_{\theta, 0}, s_{k, q_k} \lambda_k^{-q_k}$.  Let
$\varepsilon_{\text{RSC},n}$ be a sequence in $(0,1)$ decreasing to 0.
Then, subject to additional conditions on $n_y$ and the growth regime
detailed in Appendix~\ref{supp:l1proofs:cor:main},
\begin{multline*}
\sup_{t \in \RR} \left| \PP\left\{ \sqrt{n} \, (\hat\theta_k - \theta_k^*) / \hat\sigma_k \leq t \right\} - \Phi(t) \right|\\
\leq
O\rbr{s_{\theta, 0} s_{k, q_k}^{2+\tfrac{1-2q_k}{2-q_k}} \rbr{\frac{\log p}{n}}^{1-q_k} \sqrt{n}}
+ \varepsilon_{\text{RSC},n} + c\exp\rbr{-c' \log p},
\end{multline*}
where $c, c' > 0$ are constants that do not depend on $n$, $p$, $s_{\theta, 0}$ or $s_{k, q_k}$.
\end{theorem}

The proof in Appendix~\ref{supp:l1proofs:cor:main} relies on numerous technical lemmas to
derive the rates of $\check\thetab$ and $\check\omegab_k$.  In
particular, we prove a restricted strong convexity (RSC) of the
Hessian starting from a population-level assumption
(\Cref{cond:bddpopeigs}).  The proof is quite involved as the Hessian
is a weighted sample covariance where the weights are given by the
empirical density ratio estimates, which makes application of existing
results impossible.  The details are in Appendix~\ref{supp:auxiliary}.

\begin{remark}
\label{remark:sparsity}
\Cref{cor:main} gives a nontrivial bound only for sufficiently
(weakly) sparse $\thetab^*$ and $\omegab_k^*$.  The additional
condition on $n_y$ is a consequence of proving RSC from the
population-level assumptions.  In particular, it is linked to the
probability that the Hessian fails to satisfy RSC.  Analogous results
for other sparsity regimes can be obtained from Theorem~\ref{thm:main}
as well (see earlier version of this paper on arXiv).  Due to space
limitations, we have singled out this regime as being arguably the
most interesting.
\end{remark}

\begin{remark}
We note that the inverse of the Hessian
$\Sigmab_\psi^{-1}$ is determined by $\gammab_x$, since
$\Sigmab_\psi = \Cov_x[\psib(\xb)]$, and, therefore, the sparsity
of $\Sigmab_\psi^{-1}$ is related to that of $\gammab_x$.  In the
case of Gaussian graphical models, we can explicitly characterize
$\Sigmab_\psi^{-1}$ and we observe that the rows of the inverse of
the Hessian are sparse if the maximum degree of the underlying
graph is small. The proof strategy critically relies on the
properties of a Gaussian distribution and its log-partition
function and is intractable for general Markov random
fields. However, we further provide numerical evidence on the
relationship between the support of $\Sigmab_\psi^{-1}$ and that
of $\gammab_x$ for Ising models. For our method to perform well,
it suffices that the $\ell_q$-``norm'' is controlled for a small
$q \in [0,1)$, which we numerically verify.  See
Appendix~\ref{supp:rowsparsity}. Finally, we note that in some
cases the rows of $\Sigmab_\psi^{-1}$ are neither sparse nor
approximately sparse, but have bounded $\ell_1$ norm.  In this
case, a possible direction for developing a valid inference
procedure would be to modify the three step procedure in 
\citet{Ma2017Intersubject} or \citet{Yu2019Simultaneous}.
\end{remark}

\begin{remark}
\label{remark:asymmetry}
As pointed out by a reviewer, there is an inherent asymmetry in KLIEP,
and \Cref{cor:main} is one place where this can be observed.
Specifically, the quality of Gaussian approximation depends on which
set of observations is used as $\Xb$ and which as $\Yb$.  First,
$r_\theta$ may be more regular than $1 / r_\theta$ as measured by the
bounds.  This affects the magnitude of $\lambda_\theta$ or
$\lambda_k$.  Second, the larger sample will satisfy the sample
complexity condition with a smaller $\varepsilon_{\text{RSC}, n}$,
which is the probability that the Hessian fails to satisfy RSC.  For
the bounded sufficient statistics model we consider, we have found the
latter to have a larger impact on the results.  Therefore, we
recommend choosing $f_x$ and $f_y$ so that $n_x \leq n_y$.
In \Cref{sec:discussions}, we discuss alternative approaches
to differential network estimation that are not asymmetric in
nature. These, however, require imposing stronger conditions.
\end{remark}

\subsection{\label{sec:theory:bootstrap}Finite-sample consistency for Gaussian multiplier bootstrap sketched quantiles}

\Cref{thm:bootstrap} is a finite-sample consistency result for the Gaussian multiplier bootstrap.
Recall $T = \max_k \sqrt{n} \, |\hat\theta_k - \theta_k^*|$, and let $\hat c_{T, \alpha}$ denote the estimator of $(1-\alpha)$-quantile of $T$ from \Cref{procedure:bootstrap:Gaussian}.
Define $\Sigmab_\text{pooled}$ analogously as in \eqref{eq:varest}, and let
$\Omegab^* = \Sigmab_\psi^{-1}$. Recall that the $k$th column of
$\Omegab^*$ is $\omegab_k^*$.
For $\lambda_\theta, (\lambda_k)_{k=1}^p, \delta_\theta, (\delta_k)_{k=1}^p \in [0,1)$, define an event
\begin{multline*}
\Ecal_\text{all}
= \Ecal_\text{all}(\lambda_\theta, (\lambda_k)_{k=1}^p, \delta_\theta, (\delta_k)_{k=1}^p) =\\
\cbr{
\begin{array}{r c r c}
\textnormal{(G.1)} & 2 \|\nabla \lKLIEP(\thetab^*)\|_* \leq \lambda_\theta,&
\textnormal{(G.2)} & 2 \|\nabla^2 \lKLIEP(\thetab^*) \omegab_k^* - \eb_k\|_* \leq \lambda_k\ \forall\, k,\\
\textnormal{(E.1)} & \|\check\thetab - \thetab^*\| \leq \delta_\theta,&
\textnormal{(E.2)} & \|\check\omegab_k - \omegab_k^*\| \leq \delta_k\ \forall\, k,\\
\textnormal{(B.1)} & \abr{1 - \frac{\hat Z_y(\thetab^*)}{Z_y(\thetab^*)}} \lesssim \lambda_\theta,&
\textnormal{(B.2)} & \abr{\frac{1}{n_y} \sum_{j=1}^{n_y} \langle \omegab_k^*, \mub_\psi - \psib(\yb^{(j)}) \rangle \, r_{\theta^*}(\yb^{(j)})} \lesssim \lambda_k\ \forall\, k
\end{array}}.
\end{multline*}
Put $\nu_n = 1 \vee \max\{\|\omegab_k^*\| : k = 1, \dots, p\}$, and set
\begin{equation*}
B_n = \frac{(1 \vee \bar\kappa)^3 (1 \vee M_\psi)^3 M_r^3 \nu_n^{21/2}}{\sqrt{\underline\kappa^3 \eta_{x,n} \eta_{y,n}}}
\quad\text{and}\quad
\delta_n = \rbr{\frac{B_n^2 \log^7 (pn)}{n}}^{1/6}.
\end{equation*}

\begin{theorem}
\label{thm:bootstrap}
Assume Conditions 1 and 2.  Let $\hat\thetab$ be the estimator
constructed by \Cref{procedure:KLIEP+} with one-step approximation as
\begin{equation*}
\hat\thetab = \check\thetab - \check\Omegab^\top \nabla\lKLIEP(\check\thetab),
\end{equation*}
where $\check\Omegab = [\check\omegab_k]_{k=1}^p \in \RR^{p \times p}$
is the matrix with the $k$th column given by $\check\omegab_k$.
Suppose
\begin{align*}
D_1 & :=
\max_k \sqrt{\frac{\eta_{x,n} \eta_{y,n}}{\underline\kappa / \bar\kappa^2}} \rbr{(\delta_\theta+\lambda_\theta) (\delta_k+\lambda_k) + \|\omegab_k^*\| \delta_\theta^2} \sqrt{n}
\lesssim
\rbr{\frac{B_n^2 \log^4 (pn)}{n}}^{1/6},\\
D_2 & :=
\max_k \frac{\underline\kappa / \bar\kappa^2}{\eta_{x,n}^2 \eta_{y,n}^2} \rbr{\delta_k^2 + \eta_{x,n} \|\omegab_k\|^2 \rbr{\delta_\theta + \lambda_\theta}^2}
\lesssim
\rbr{\frac{B_{n}^2 \log (pn)}{n}}^{1/3}.
\end{align*}
If $\PP(\Ecal_\text{all}) \geq 1 - \varepsilon_{\text{all},n}$, then
\begin{equation}
\label{eq:bootstrap}
\sup_{\alpha \in (0,1)} \abr{\PP\cbr{T \leq \hat c_{T, \alpha}} - (1 - \alpha)}
= O(\delta_n + \varepsilon_{\text{all},n})
\end{equation}
with probability at least $1 - \varepsilon_{\text{all},n} - n^{-1}$.
\end{theorem}

The proof is in Appendix~\ref{supp:proofs:thm:bootstrap}. The bulk of hard work
  was done in establishing a linear approximation to
  $\sqrt{n} \, (\hat\theta_k - \theta_k^*)$ in the proof of
  \Cref{thm:main}. \Cref{thm:bootstrap} follows by showing that the
  error in the linear approximation can be controlled, allowing for
  application of results in \citet{Belloni2018High}. Due to the nonlinearity
  of $\lKLIEP$ \eqref{eq:lKLIEP} and the fact that we are using a two sample estimator,
  the detailed calculations are rather complicated.

As an application of \Cref{thm:bootstrap}, we evaluate the bound in \eqref{eq:bootstrap} in the case of SparKLIE+1 with $s_\theta = s_{\theta, 0} = \|\thetab^*\|_0$ and $s_k = s_{k, 0} = \|\omegab_k^*\|_0$.

\begin{theorem}
\label{cor:bootstrap}
Assume \Cref{cond:bdr} with $\ell_1$-norm and \Cref{cond:bddpopeigs}.
Suppose $T = \max_k \sqrt{n} \, |\hat\theta_k - \theta_k^*|$, where $\hat\thetab$ is the SparKLIE+1 estimator with tuning parameters
\begin{eqnarray*}
\lambda_\theta \asymp \rbr{\frac{\log p}{n}}^{1/2}
&\text{and}&
\lambda_k \asymp \rbr{\frac{s_{k, 0} \log p}{n}}^{1/2}, \quad k = 1, \dots, p.
\end{eqnarray*}
Let $s$ be a sequence of integers satisfying $s \geq s_{\theta, 0}, s_{k, 0}$, $k = 1, \dots, p$.
Let $\varepsilon_{\text{RSC},n}$ be a sequence in $(0,1)$ decreasing to 0.
Then, subject to an additional condition on $n_y$ detailed in Appendix~\ref{supp:l1proofs:cor:bootstrap},
\begin{equation*}
\sup_{\alpha \in (0,1)} \abr{\PP\cbr{T \leq \hat c_{T, \alpha}} - (1 - \alpha)}
= O(\delta_n + \varepsilon_{\text{RSC},n} + c\exp\rbr{-c' \log p})
\end{equation*}
with probability at least $1 - \varepsilon_{\text{RSC},n} - c\exp\rbr{-c' \log p} - n^{-1}$, where $c, c' > 0$ are constants that do not depend on $n$, $p$, $s_{\theta, 0}$ or $s_{k, 0}$.
\end{theorem}


\section{\label{sec:simulations}Simulation studies}

Through extensive simulations, we illustrate the finite-sample
performance of our methods: SparKLIE+ (\Cref{sec:simulations:1}) and
empirical bootstrap sketching (\Cref{sec:simulations:2}).

\subsection{\label{sec:simulations:1}Inference for a single edge via Gaussian approximation}

In Experiments~1 and 2, we look at the performance of statistical inference procedures based on Gaussian approximation when an edge has been fixed as a target of inferential interest.

\noindent \emph{Experiment~1.}
We check the coverage of the 95\% CI $\hat\theta_k \pm z_{0.975} \hat\sigma_k / \sqrt{n}$, where $k$ is a fixed edge of interest and $z_{0.975}$ is the 0.975-quantile of $\Ncal(0,1)$.
Here, SparKLIE+1 and +2 are compared with two other procedures: an oracle procedure with the knowledge of $\supp(\thetab^*)$ and a na\"{i}ve re-estimation procedure that re-fits the model based on the estimated support $\supp(\check\thetab)$, where $\check\thetab$ is a sparse KLIEP estimate.
See Appendix~\ref{supp:sec5:competing} for precise definitions.

The results below were obtained using autoscaling procedures for both initial estimation steps.
For each step, we used the canonical penalty level, which was $\lambda_{\theta 0} = 1.01 \Phi^{-1}(1 - 0.05 / p)$ for Step 1 and $\lambda_{0} = \sqrt{2 \log p / n_y}$ for Step 2.
However, we remark that even with the vanilla sparse KLIEP procedure \eqref{eq:spKLIEP} in Step 1, we have found the performance of \Cref{procedure:KLIEP+} to be robust to the choice of $\lambda_\theta$.
See \Cref{remark:alternatives,remark:lambda}, as well as Appendix~\ref{supp:exper5}.

The data are pairs of samples of i.i.d.~observations from a pair of Ising models $\gammab_x$ and $\gammab_y$. Eight pairs of $\gammab_x$ and $\gammab_y$ are compared, arising from all possible combinations of the number of nodes ($m = 25$ or $50$), the topology of $\gammab_x$ (a chain or a ternary tree), and two choices of $\thetab^*$ from which $\gamma_y = \gamma_x - \thetab^*$ is obtained. Each differential network has five nonzero edges, one of which has been fixed as the target of inference. For illustration, see Figures~\ref{fig:chain1} -- \ref{fig:tree2} in Appendix~\ref{supp:sec5:design:exper1}.

\begin{table}
  \label{table:exper1}
  \caption{Empirical coverage of the 95\% CI
    $\hat\theta_k \pm z_{0.975} \hat\sigma_k / \sqrt{n}$, where $k$ is
    a fixed edge of interest and $z_{0.975}$ is the 0.975-quantile of
    $\Ncal(0,1)$, of SparKLIE+1 and +2 estimators compared with the
    oracle and a na\"{i}ve re-fitted estimators. The results are
    averages over 1000 independent replications.}
\centering
\fbox{%
\begin{tabular}{c *{4}{c} *{4}{c}}
\multirow{2}{*}{$\gammab_x$} & \multirow{2}{*}{$\gammab_y$} & \multirow{2}{*}{$m$} & \multirow{2}{*}{$n_x$} & \multirow{2}{*}{$n_y$} & \multirow{2}{*}{oracle} & \multirow{2}{*}{na\"{i}ve} & \multicolumn{2}{c}{SparKLIE}\\
& & & & & & & {+1} & {+2}\\
\hline
\multirow{4}{*}{chain}
& \multirow{2}{*}{1} & 25 & 150 & 300 & 0.960 & 0.850 & 0.934 & 0.945\\
& & 50 & 300 & 600 & 0.946 & 0.822 & 0.943 & 0.948\\
& \multirow{2}{*}{2} & 25 & 150 & 300 & 0.962 & 0.907 & 0.948 & 0.948\\
& & 50 & 300 & 600 & 0.962 & 0.839 & 0.953 & 0.955\\
\multirow{4}{*}{\shortstack{ternary \\ tree}}
& \multirow{2}{*}{1} & 25 & 150 & 300 & 0.972 & 0.925 & 0.932 & 0.958\\
& & 50 & 300 & 600 & 0.976 & 0.874 & 0.973 & 0.979\\
& \multirow{2}{*}{2} & 25 & 150 & 300 & 0.972 & 0.946 & 0.957 & 0.977\\
& & 50 & 300 & 600 & 0.968 & 0.913 & 0.952 & 0.977
\end{tabular}}
\end{table}

Table~1 gives the proportions of successful coverage out of 1000 independent replications at the nominal confidence level of 95\%. In spite of the small sample sizes, the coverage of 95\% CIs based on either of the two SparKLIE+ estimators are close to the nominal level, and on par with the performance of the oracle procedure across all the data generating processes considered. By contrast, we see that the na\"{i}ve re-fitted estimator can undercover by as much as $\approx 13\%$.

In Appendix~\ref{supp:sec5:res:exper1}, we further provide normal Q-Q plots (Figures~\ref{fig:qqplots_chain1} -- \ref{fig:qqplots_tree2}) and empirical estimates of the biases (Table~\ref{table:bias}) for the four estimators.
These reveal that the inferior performance of the na\"{i}ve re-fitted estimator can be attributed to the larger bias.

In \emph{Experiment~2} in Appendix~\ref{supp:exper2}, we study the power SparKLIE+1 and +2 for testing the null hypothesis $\Hcal_0: \theta_k^* = 0$, where $k$ is a fixed edge of interest.

\subsection{\label{sec:simulations:2}Global inference with empirical bootstrap quantile estimates}

In Experiments~3 and 4, we look at the performance of \Cref{procedure:bootstrap:empirical} for making inferences about the entire differential network $\thetab^*$.

\noindent \emph{Experiment~3.}
We check that \Cref{procedure:bootstrap:empirical} produces consistent estimates of the quantiles $c_{T, 1-\alpha}$ of $T = \max_k \sqrt{n} \, |\hat\theta_k - \theta_k^*|$.
Here, we focus on the setting $\gammab = \gammab_x = \gammab_y$, i.e., $\thetab^* = \zero$. We generate a pair of samples of the same size $n_x = n_y = 500$ from the same Ising model with the parameter $\gammab$. The parameter $\gammab$ was generated as a disjoint union of $m / 5$ chains of length 5 for $m \in \cbr{25, 50, 100}$. The nonzero edge weights were drawn i.i.d.~from one of the three distributions: $\text{sgn} = 1$, $\Unif (0.2,0.4)$; $\text{sgn} = -1$, $\Unif (-0.4,-0.2)$; or $\text{sgn} = 0$, $\Unif (-0.4,-0.2) \cup (0.2,0.4)$.

For each draw of samples from $\gammab_x$ and $\gammab_y$, we use \Cref{procedure:bootstrap:empirical} with $n_b = 1000$ bootstrap replicates to estimate $\hat c_{T, 1-\alpha}$, and record $\ind\{T \leq \hat c_{T, 1-\alpha}\}$ for each $1-\alpha = 0.05, \dots, 0.95$.
Then, the results are averaged across 1000 independent draws of the pair of samples.
If \Cref{procedure:bootstrap:empirical} is consistent, $\ind\{T \leq \hat c_{T, 1-\alpha}\} \approx \ind\{T \leq c_{T, 1-\alpha}\}$, and hence the average over independent replicates would be close to $1-\alpha$.
This is indeed what we see in Figure~1.

\begin{figure}[!t]
  \label{fig:exper3}
  \caption{Consistency of the quantile estimates $\hat c_{T, 1-\alpha}$ from \Cref{procedure:bootstrap:empirical} in nine different settings, corresponding to all possible combinations of the number of nodes $m = 25$, $50$, or $100$ and the distribution of edge parameters $\text{sgn} = -1$, $0$, or $1$, where $\text{sgn} = 1$, $\Unif (0.2,0.4)$; $\text{sgn} = -1$, $\Unif (-0.4,-0.2)$; or $\text{sgn} = 0$, $\Unif (-0.4,-0.2) \cup (0.2,0.4)$. The blue line with $\bullet$ indicates SparKLIE+1. The orange line with $\blacktriangledown$ indicates SparKLIE+2. The $45^\circ$ line marks perfect calibration.}
\centering
\includegraphics[width=\linewidth]{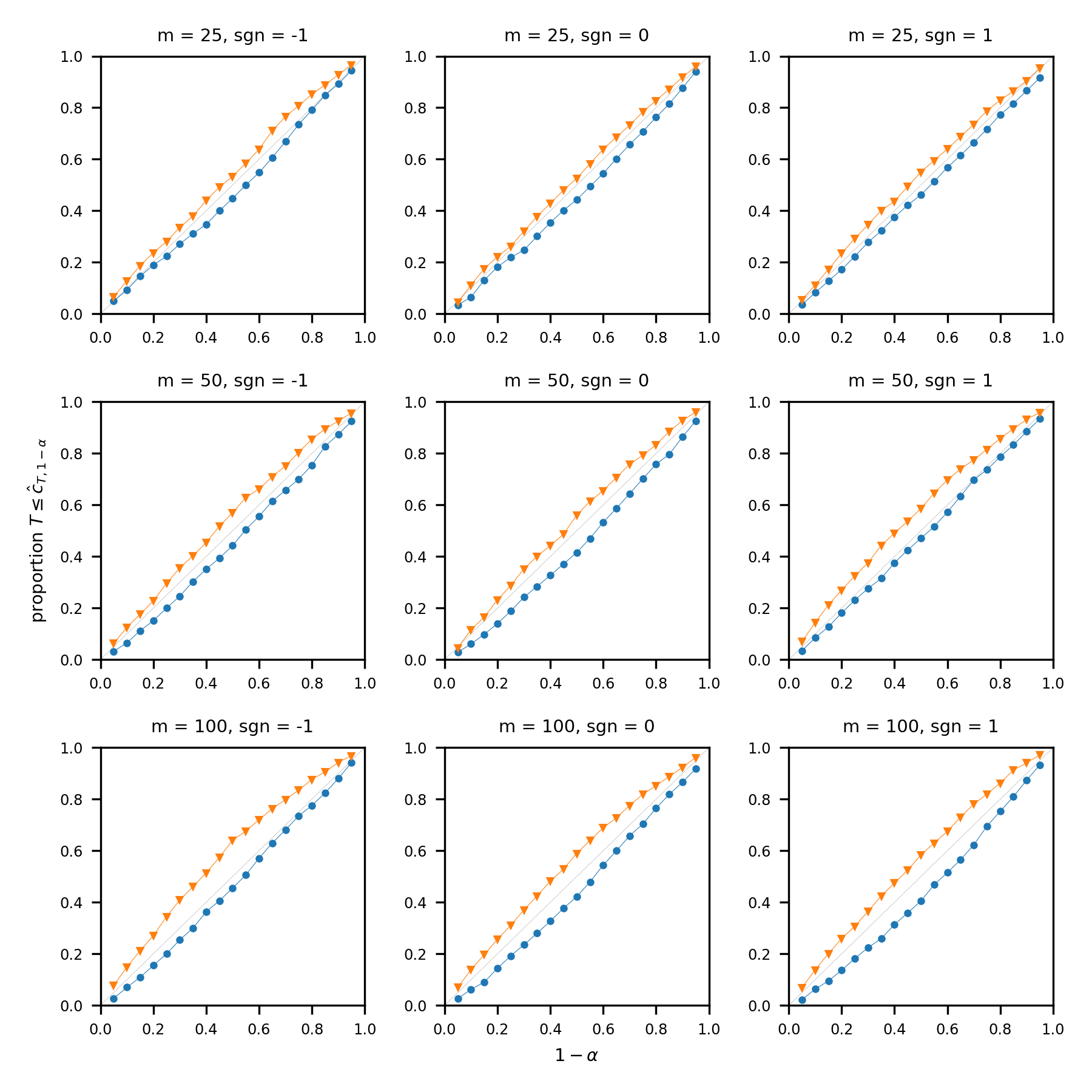}
\end{figure}

In \emph{Experiment~4} in Appendix~\ref{supp:exper4}, we study the power of the level-$\alpha$ test obtained by inverting the simultaneous confidence region $\hat\theta_k \pm \hat c_{T, 1-\alpha} / \sqrt{n}$ for testing the null hypothesis $\Hcal_0: \theta_k^* = 0$ for all $k$.


\section{\label{sec:real_data}Real data example: Alertness and motor control, an fMRI study}

We apply \Cref{procedure:KLIEP+} and
\Cref{procedure:bootstrap:empirical} to analyze a new fMRI dataset,
made available courtesy of Dr.~Jade Thai and Dr.~Christelle Langley at
the University of Bristol.
The dataset comes from a pilot study involving a multiple sclerosis
subject (MS) and a healthy control (HC) with the purpose of exploring
the relationship between alertness and motor control.  It consists of
two time series, one for each participant of the study, of fMRI
measurements at 0.906 second intervals from 116 regions of interest
(ROI) in the brain.  We further restrict to $m = 25$ ROIs
pre-specified by the neuroscientists.  The measurements were taken
while the participants were performing one of three types of tasks: a
sensorimotor task (T1), an intrinsic alertness task (T2), and an
extrinsic alertness task (T3).  For details concerning the study
design and data post-processing, see Appendix~\ref{supp:sec6}.

We model the fMRI measurements as independent observations from six
Gaussian graphical models, where the groups are given by the disease
status and the task type.  For example, the measurements collected
while the HC subject performed T1 are modeled as
\begin{equation*}
f_\text{HC,T1}(\xb) = \det(\Gb_\text{HC,T1} / (2 \pi))^{1/2} \exp\left( -(\xb - \mu_\text{HC,T1})^\top \Gb_\text{HC,T1} (\xb -  \mu_\text{HC,T1}) / 2\right).
\end{equation*}
Since we are interested in the difference in the graph structure, we
work with the data after centering by the group means.  The sample
sizes are given in Table~\ref{table:samplesizes}.

\begin{table}
\label{table:samplesizes}
\caption{Sample sizes by group}
\centering
\fbox{%
\begin{tabular}{c*{3} {c}}
& T1 & T2 & T3\\
\hline
HC & 342 & 300 & 306\\
MS & 342 & 300 & 311
\end{tabular}}
\end{table}

For either the HC or the MS subject, we study the pairwise differences
for the tasks.  Specifically, we would like to obtain, with FWER
control at $0.05$, six differential networks:
\begin{align*}
\Deltab_1^* &= \Gb_\text{HC, T1} - \Gb_\text{HC, T2}, &
\Deltab_2^* &= \Gb_\text{HC, T1} - \Gb_\text{HC, T3}, &
\Deltab_3^* &= \Gb_\text{HC, T2} - \Gb_\text{HC, T3}, \\
\Deltab_4^* &= \Gb_\text{MS, T1} - \Gb_\text{MS, T2}, &
\Deltab_5^* &= \Gb_\text{MS, T1} - \Gb_\text{MS, T3}, &
\Deltab_6^* &= \Gb_\text{MS, T2} - \Gb_\text{MS, T3}.
\end{align*}
The six differential networks $\hat\Deltab_g^*$, $g = 1, \dots, 6$, were estimated using \Cref{procedure:KLIEP+} with the autoscaling formulations for Steps 1 and 2 with the universal penalty levels as explained in Remark~\ref{remark:lambda} in Section~\ref{sec:methodology:single}.
The test statistic $T_0 = \max_{g = 1, \dots, 6} \max_{1 \leq u \leq v \leq 25} |\hat\Delta_{g, uv}|$ was used to test the null hypothesis $\Hcal_0: \Deltab_g^* = \zero$ for all $g = 1, \dots, 6$ at level 0.05 based on the rejection threshold $\hat c_{T_0, 0.05}$ obtained from \Cref{procedure:bootstrap:empirical}.
The test found \emph{no edges} to be statistically significant.
However, the conclusion is based on a pilot study from two individuals, and more data are needed.


\section{\label{sec:discussions}Discussion}

We have developed novel methods for making statistically valid comparisons of general Markov networks based on i.i.d.~observations from each. To our knowledge, this is the first work that allows one to conduct provably-valid inference using a direct estimate of the network difference for general Markov networks in high-dimensional settings. This means that our methods can deal with dense networks as long as their \emph{difference} is sparse. Also, our framework can easily handle non-Gaussian data. Furthermore, our theory does not require the conditions that are typically necessary to guarantee consistent support recovery, increasing applicability of our proposal. In addition, we develop the bootstrap sketching procedures to estimate the quantiles of extreme statistics accurately and in a computationally efficient manner even at large $p$.

As remarked by a reviewer, it is natural to ask whether it is possible to use other divergences to derive similar procedures. For closely-related varieties, such as the reverse and the symmetric KL, the answer is clearly yes. For arbitrary divergences, however, exact analogues may not exist. The derivation of KLIEP uses more than just the properties of a divergence. Indeed, the logarithm in KL plays an essential role in linearizing the ratio $f_x / (r_\theta f_y)$, yielding a population-level loss that involves expectations of only known functions of $\thetab$. In addition the loss is convex in $\thetab$, leading to a computationally attractive procedure. Using other divergences to measure discrepancy between $f_x$ and $r_\theta f_y$ would, to the best of our knowledge, lead to an estimator that is not convex in $\thetab$. Establishing statistical properties of such an estimator is beyond the scope of this paper.

It can be checked that the special case of the reverse KL reduces to KLIEP with the role of $f_x$ and $f_y$ swapped. The effect of switching the samples was discussed in \Cref{remark:asymmetry} in \Cref{sec:theory:main}. The symmetric KL leads to a procedure that minimizes the sum of the KLIEP and the reversed KLIEP loss functions. The theory developed in this paper extends in an obvious way to the symmetrized procedure. This means that the conditions that were previously imposed on only one of $f_x$ and $f_y$ now need to hold for both, reducing the applicability of our methods. Moreover, although the change is not expected to alter the order of error bounds, the constants are expected to be larger, and this is likely to result in a more brittle approximation at the same sample sizes, as corroborated by empirical evidence (Appendix~\ref{supp:exper5}).

In addition to the approach followed in this paper, where the density ratio is estimated by minimizing the divergence between one density and the product of the density ratio and another density, alternative approaches have been considered in the literature. For example, \citet{Nguyen2010Estimating} estimate the density ratio by maximizing a lower bound on an $f$-divergence. \citet{Kanamori2009least} estimate a density ratio by minimizing a squared loss between the true density ratio and the model of a density ratio. Developing inferential results for the parameters of differential networks obtained by such approaches is an interesting topic for future research.

Although we never place explicit assumptions on the form of dependence, some restraint is necessary in practice for good performance. This can already be seen from the results in \Cref{sec:theory}: the bounds deteriorate rapidly as $p$ increases to accommodate higher order dependencies. This is why we chose to focus exclusively on \emph{pairwise} models in our simulations and real data analysis. It is of future interest to develop an efficient search procedure to include only the relevant higher-order terms.

Finally, although it is a huge advantage of our methods that they can be used to compare general Markov networks, it may be possible to obtain more sample efficient procedures for particular models by utilizing distribution-specific properties. For example, it is of interest to develop inferential procedures for the network difference of Gaussian or Gaussian copula models.

\section*{Acknowledgments}

This work was completed in part with resources provided by the University of Chicago Research Computing Center. The fMRI dataset collection is funded by MS Research Treatment and Education whose Registered Charity Number is 1043280.

\bibliographystyle{my-plainnat}
\bibpunct{(}{)}{,}{a}{,}{,}
\bibliography{paper}

\newpage

\appendix

\section{The empirical KLIEP loss $\lKLIEP$}

\subsection{\label{supp:lKLIEP:derivation}Derivation of KLIEP}

We derive the form for $\thetab^*$ in Section~2.2. We have
\begin{align}
\thetab^*
&= \argmin_\theta D_\text{KL}(f_x \| f_y r_\theta)\\
&= \argmin_\theta \phantom{-}\int \log\rbr{\frac{f_x(\xb)}{r_\theta(\xb) f_y(\xb)}} f_x(\xb) \, d\xb\\
&= \argmin_\theta -\int \log r_\theta(\xb) \, f_x(\xb) \, d\xb \label{eq:derivation:1}\\
&= \argmin_\theta -\int \thetab^\top \psib(\xb) \, f_x(\xb) \, d\xb + \log Z_y(\thetab) \label{eq:derivation:2}\\
&= \argmin_\theta -\EE_x\sbr{\thetab^\top \psib(\xb)} + \log \EE_y \sbr{\exp\rbr{\thetab^\top \psib(\yb)}}
\label{eq:derivation:3}.
\end{align}
\eqref{eq:derivation:1} applies $\log$ to the ratio, and then notes that $\log r_\theta$ is the only term with dependency on $\thetab$.
\eqref{eq:derivation:2} follows from the definition of $r_\theta$.
\eqref{eq:derivation:3} is just the definition of expectation.
In particular, $Z_y(\thetab) = \EE_y[\exp( \thetab^\top \psib(\yb) )] = \int \exp( \thetab^\top \psib(\xb) ) f_y(\thetab) \, d\xb$.

\subsection{\label{supp:lKLIEP:derivatives}Derivatives and approximate moment-matching}

Recall
\begin{align*}
\hat Z_y(\thetab)
&= \frac{1}{n_y} \sum_{j=1}^{n_y} \exp\rbr{\thetab^\top \psib(\yb^{(j)})},&
\hat r_\theta(\yb)
&= \frac{\exp\rbr{\thetab^\top \psib(\yb)}}{\hat Z_y(\thetab)},&
\hat\mub(\thetab)
&= \frac{1}{n_y} \sum_{j=1}^{n_y} \psib(\yb^{(j)}) \hat r_\theta(\yb^{(j)}).
\end{align*}
The following identities hold:
\begingroup
\allowdisplaybreaks
\begin{align}
\frac{\partial \log \hat Z_y(\thetab) }{\partial \theta_k}
&= \hat\mu_k(\thetab),\\
\frac{\partial \hat r_\theta(\yb)}{\partial \theta_k}
&= \rbr{\psi_k(\yb_k) - \hat\mu_k(\thetab)} \hat r_\theta(\yb),\\
\frac{\partial \lKLIEP(\thetab)}{\partial \theta_k}
&= -\frac{1}{n_x} \sum_{i=1}^{n_x} \psi_k(\xb_k^{(i)}) + \hat\mu_k(\thetab),
\label{eq:lKLIEP:1st}\\
\frac{\partial^2 \lKLIEP(\thetab)}{\partial \theta_{k'} \partial \theta_k}
&= \frac{1}{n_y} \sum_{j=1}^{n_y} \psi_{k'}(\yb_{k'}^{(j)}) \psi_k(\yb_k^{(j)}) \hat r_\theta(\yb^{(j)}) - \hat\mu_{k'}(\thetab) \hat\mu_k(\thetab)
\label{eq:lKLIEP:2nd}\\
&= \frac{1}{n_y^2} \sum_{1 \leq j < j' \leq n_y} \left( \psi_{k'}(\yb_{k'}^{(j)}) - \psi_{k'}(\yb_{k'}^{(j')}) \right)\left( \psi_k(\yb_k^{(j)}) - \psi_k(\yb_k^{(j')}) \right) \hat r_\theta(\yb^{(j)}) \hat r_\theta(\yb^{(j')}),
\label{eq:lKLIEP:2nd:Ustat}\\
\frac{\partial^3 \lKLIEP(\thetab)}{\partial \theta_{k''} \partial \theta_{k'} \partial \theta_k}
&=
\begin{multlined}[t]
\frac{1}{n_y} \sum_{j=1}^{n_y} \psi_k(\yb_k^{(j)}) \psi_{k'}(\yb_{k'}^{(j)}) \psi_{k''}(\yb_{k''}^{(j)}) \hat r_\theta(\yb^{(j)})\\
- \hat\mu_{k''}(\thetab) \times \frac{1}{n_y} \sum_{j=1}^{n_y} \psi_k(\yb_k^{(j)}) \psi_{k'}(\yb_{k'}^{(j)}) \hat r_\theta(\yb^{(j)})\\
- \hat\mu_{k'} \nabla_{{k''}k}^2 \lKLIEP(\thetab) - \hat\mu_k(\thetab) \nabla_{{k''}k'}^2 \lKLIEP(\thetab).
\end{multlined}
\label{eq:lKLIEP:3rd}
\end{align}
\endgroup
These identities are useful in obtaining various uniform bounds.

The suggestive notation is by design: $\hat Z_y(\thetab) \approx Z_y(\thetab)$ and $\hat r_\theta(\yb) \approx r_\theta(\yb)$, obviously.
In fact, it is the message of \Cref{lem:moment_matching} below that
\begin{eqnarray*}
\hat\mu(\thetab) \approx \EE_{\theta+\gamma_y}[\psib(\xb)]
&\text{and}&
\EE_{\theta + \gamma_y}[ \nabla^2 \lKLIEP(\thetab) ] \approx \Cov_{\theta + \gamma_y}[ \psib(\xb) ].
\end{eqnarray*}
$\nabla^3 \lKLIEP(\thetab)$ also approximates the third central moment tensor of $\psib(\xb)$ under $\gammab = \thetab + \gammab_y$, but we will not need this fact.

\begin{lemma} \label[lemma]{lem:moment_matching}
\begin{equation*}
\EE_y[ \psib(\yb) r_\theta(\yb) ] = \EE_{\theta+\gamma_y}[\psib(\xb)].
\end{equation*}
and
\begin{equation} \label{eq:EyH}
\EE_y[\Hb(\thetab)]
:= \EE_y\left[ \frac{\hat Z_y(\thetab)^2}{Z_y(\thetab)^2} \nabla^2 \lKLIEP(\thetab) \right]\\
= \left( 1-\frac{1}{n_y} \right) \Cov_{\theta+\gamma_y}[\psib(\xb)].
\end{equation}
\end{lemma}

\begin{proof}
Let $f_\theta = f_{\theta + \gamma_y}$.
To prove the first identity,
\[
\EE_y[ \psi_k(\yb_k) r_\theta(\yb) ]
= \int \psi_k(\yb_k) r_\theta(\yb) f_y(\yb) \, d\yb
= \int \psi_k(\yb_k) f_\theta(\yb) \, d\yb
= \EE_{\theta + \gamma_y} \psi(\yb_k)
= \mu_k(\thetab).
\]
To prove the second identity, let $\yb, \yb' \sim f_y$ be independent, so that
\begin{multline*}
\EE_y\sbr{\rbr{\psi_{k'}(\yb_{k'}) - \psi_{k'}(\yb_{k'}')} \rbr{\psi_k(\yb_k) - \psi_k(\yb_k')} r_\theta(\yb) r_\theta(\yb')}\\
\begin{aligned}[t]
&= \iint \rbr{\psi_{k'}(\yb_{k'}) - \psi_{k'}(\yb_{k'}')} \rbr{\psi_k(\yb_k) - \psi_k(\yb_k')} r_\theta(\yb) r_\theta(\yb') f_y(\yb) f_y(\yb') \, d\yb \, d\yb'\\
&=
\begin{multlined}[t]
2 \iint \psi_{k'}(\yb_{k'}) \psi_k(\yb_k) r_\theta(\yb) r_\theta(\yb') f_y(\yb) f_y(\yb') \, d\yb \, d\yb'\\
- 2 \iint \psi_{k'}(\yb_{k'}) \psi_k(\yb_k') r_\theta(\yb) r_\theta(\yb') f_y(\yb) f_y(\yb') \, d\yb \, d\yb'.
\end{multlined}
\end{aligned}
\end{multline*}
The first integral is
\begin{multline*}
\iint \psi_{k'}(\yb_{k'}) \psi_k(\yb_k) r_\theta(\yb) r_\theta(\yb') f_y(\yb) f_y(\yb') \, d\yb \, d\yb'\\
\begin{aligned}[t]
&= \int \psi_{k'}(\yb_{k'}) \psi_k(\yb_k) r_\theta(\yb) f_y(\yb) \, d\yb \int r_\theta(\yb') f_y(\yb') d\yb'\\
&= \int \psi_{k'}(\yb_{k'}) \psi_k(\yb_k) \tilde f_{\theta + \gamma_y}(\yb) \, d\yb \int \tilde f_{\theta + \gamma_y}(\yb') \, d\yb'
\end{aligned}\\
= \EE_{\theta + \gamma_y}[ \psi_{k'}(\yb_{k'}) \psi_k(\yb_k) ].
\end{multline*}
As for the second integral,
\begin{multline*}
\iint \psi_{k'}(\yb_{k'}) \psi_k (\yb_k') r_\theta(\yb) r_\theta(\yb') f_y(\yb) f_y(\yb') \, d\yb \, d\yb'\\
\begin{aligned}[t]
&= \int \psi_{k'}(\yb_{k'}) r_\theta(\yb) f_y(\yb) \, d\yb \int \psi_k(\yb_k') r_\theta(\yb') f_y(\yb') \, d\yb'\\
&= \int \psi_{k'}(\yb_{k'}) \tilde f_{\theta + \gamma_y}(\yb) \, d\yb \int \psi_k(\yb_k') \tilde f_{\theta + \gamma_y}(\yb') \, d\yb'
\end{aligned}\\
= \EE_{\theta + \gamma_y} [ \psi_{k'}(\yb_{k'}) ] \EE_{\theta + \gamma_y} [ \psi_k(\yb_k) ]
= \mu_{k'}(\thetab) \mu_k(\thetab).
\end{multline*}
Thus,
\begin{multline*}
\EE_y \sbr{\rbr{\psi_{k'}(\yb_{k'}) - \psi_{k'}(\yb_{k'}')} \rbr{\psi_k(\yb_k) - \psi_k(\yb_k')} r_\theta(\yb) r_\theta(\yb')}\\
= \EE_{\theta + \gamma_y} [ \psi_{k'}(\yb_{k'})\psi_k(\yb_k) ] - \mu_{k'}(\thetab) \mu_k(\thetab)
= \Cov_{\theta + \gamma_y} [\psi_{k'}(\yb_{k'}),\psi_k(\yb_k)],
\end{multline*}
and therefore,
\begin{equation*}
\begin{aligned}
\EE_y [\Hb(\thetab)]
&= \EE_y \sbr{\frac{\hat Z_y^2(\thetab)}{Z_y^2(\thetab)} \nabla^2 \lKLIEP(\thetab)}\\
&=
\begin{multlined}[t]
\EE_y \sbr{\frac{1}{n_y^2} \sum_{1 \leq j < j' \leq n_y} \rbr{\psib(\yb^{(j)}) - \psib(\yb^{(j')})} \rbr{\psib(\yb^{(j)}) - \psib(\yb^{(j')})} r_\theta(\yb^{(j)}) r_\theta(\yb^{(j')})}\\
= \rbr{1-\frac{1}{n_y}} \Sigmab(\thetab).
\end{multlined}
\end{aligned}
\end{equation*}
\end{proof}

In KLIEP, the difference $\thetab^*$ of $\gammab_x$ from $\gammab_y$ is estimated by matching the moments of $\psib(\xb)$ with $\psib(\yb)$ by exponential tilting of the baseline pdf $f_y$.

KLIEP can be viewed as an approximate version of maximum-likelihood estimation.
Fixing $f_y \in \Fcal_\gamma$, define a new parametrization of the family by $\thetab = \gammab - \gammab_y$.
Abusing the notation somewhat,
\begin{equation*}
f_\theta(\xb)
= r_\theta(\xb) f_y(\xb)
= Z_y(\thetab)^{-1} \exp\left( \thetab^\top \psib(\xb) \right) f_y(\xb),
\end{equation*}
where $Z_y(\thetab)$ normalizes $f_\theta$ albeit with respect to the baseline $f_y$, that is to say,
\begin{equation*}
Z_y(\thetab)
= \int_{\XX} \exp\left( \thetab^\top \psib(\xb) \right) f_y(\xb) \, d\xb
= Z(\thetab +\gammab_y) / Z(\gammab_y).
\end{equation*}
Clearly, each $\gammab$ in the original parameter space corresponds to a unique $\thetab$ in the new parameter space.

Given $\xb^{(1)}, \dots, \xb^{(n_x)} \iidsim f_x$, the negative log-likelihood of the data with respect to the difference parametrization $\thetab$ is
\begin{equation*}
\ell_y(\thetab;\Xb_{n_x}) = -\frac{1}{n_x} \sum_{i=1}^{n_x} \thetab^\top \psib(\xb^{(i)}) + \log Z_y(\thetab).
\end{equation*}
Let $\mub(\thetab)$ and $\Sigmab(\thetab)$ be, respectively, the mean and the covariance of $\psib(\xb)$ under $f_\theta$:
\begin{equation*}
\mub(\thetab) = \int_{\XX} \psib(\xb) r_\theta(\xb) f_y(\xb) \, d\xb
\end{equation*}
and
\begin{equation*}
\Sigmab(\thetab)
= \int_{\XX} \psib(\xb) \psib(\xb)^\top r_\theta(\xb) f_y(\xb) \, d\xb - \left( \int_{\XX} \psib(\xb) r_\theta(\xb) f_y(\xb) \, d\xb \right) \left( \int_{\XX} \psib(\xb) r_\theta(\xb) f_y(\xb) \, d\xb \right)^\top.
\end{equation*}
It is straightforward to compute
\begin{equation*}
\nabla \ell_y(\thetab;\Xb_{n_x}) = -\frac{1}{n} \sum_{i=1}^n \psib(\xb^{(i)}) + \mub(\thetab)
\quad\text{and}\quad
\nabla^2 \ell_y(\thetab) = \Sigmab(\thetab).
\end{equation*}
Clearly, when $\xb^{(1)}, \dots, \xb^{(n_x)} \iidsim f_{\theta^*}$, $\thetab^* = \gammab_x - \gammab_y$ is the unique minimizer of $\EE_x[ \ell_y(\thetab;\Xb_{n_x}) ]$.
In this setting, $\nabla^2 \ell_y$ is a deterministic function of the parameter and thus does not depend on the data.
However, it is in general hard to minimize $\ell_y$ directly, so we look to minimize $\lKLIEP$ instead.

Using $\hat Z_y(\thetab)$ in place of $Z_y(\thetab)$ recovers $\lKLIEP$.
As
\begin{equation*}
\sup_{\thetab} |\lKLIEP(\thetab) - \ell_y(\thetab)|
= \left| \log \left\{ \frac{1}{n_y} \sum_{j=1}^{n_y} r_\theta(\yb^{(j)}) \right\} \right|,
\end{equation*}
$\lKLIEP$ converges to $\ell_y$ pointwise a.s.~as $n_y \to \infty$.
Consequently, the minimizer of $\lKLIEP$ and the minimizer of $\ell_y$ are asymptotically equivalent.

\section{Proofs of the general results}

In the below, positive constants that depend only on the fixed problem parameters are denoted as $c_0, c_1, \dots$, $c_0', c_1', \dots$, $K_0, K_1, \dots$, and their precise definitions may change from line to line. They are never allowed to depend on the sample sizes $n_x$ and $n_y$, the number of nodes $m$ or the number of parameters $p$ (usually $p = {m \choose 2}$), or the sparsity level of the true parameters $s_\theta = s_{\theta, q_\theta} = \|\thetab^*\|_{q_\theta}$ or $s_k = s_{k, q_k} = \|\omegab_k^*\|_{q_k}$ for $k \in [p]$ and for fixed $q_\theta, q_k \in [0,1)$.

\subsection{\label{supp:proofs:debiasing}Explaining de-biasing}

Suppose we wish to construct an unbiased estimator of $\theta_k^*$ for some $k \in [p]$.
Let $\thetab_{k^c}^* = \thetab_{[p] \setminus \{k\}}^* \in \RR^{p-1}$ denote the vector of remaining $p-1$ parameters; this is the nuisance parameter for carrying out statistical inference for $\theta_k^*$.
Abusing the notation somewhat, we write the resulting partition as $\thetab = (\theta_k, \thetab_{k^c})$.
Let $n = n_x + n_y$.

We consider estimators that arise as zeros of modified score functions of the form $g(\theta_k; \check\thetab_{k^c}, \omegab_k) = \omegab_k^\top \nabla \lKLIEP(\theta_k; \thetab_{k^c})$, where $\check\thetab_{k^c}$ is a consistent, but not necessarily $\sqrt{n}$-consistent estimator of $\thetab_{k^c}^*$ and $\omegab_k \in \RR^p$ is a fixed vector. $g$ has the first-order Taylor expansion
\begin{equation}
\label{eq:modified:Taylor}
	\omegab_k^\top \nabla \lKLIEP(\theta_k; \check\thetab_{k^c})
	= \omegab_k^\top \nabla \lKLIEP(\thetab^*) + \omegab_k^\top \nabla^2 \lKLIEP(\thetab^*) \begin{bmatrix} \theta_k - \theta_k^* \\ \check\thetab_{k^c} - \thetab_{k^c}^*\end{bmatrix} + \text{REM}.
\end{equation}
In an oracle setting where the true nuisance parameter $\thetab_{k^c}^*$ is known, the choice $\omegab_k = \eb_k$, which corresponds to the modified score function $g(\theta_k; \thetab_{k^c}^*, \eb_k) = \nabla_k \lKLIEP(\theta_k; \thetab_{k^c}^*)$, leads to an approximately normal and unbiased estimation of $\theta_k^*$.
When $\thetab_{k^c}^*$ is unknown and has to be estimated from the data as well, $\thetab_{k^c}^*$ is replaced by an estimate $\check\thetab_{k^c}$ in the na\"{i}ve plug-in approach.
This is acceptable when $\check\thetab_{k^c}$ is $\sqrt{n}$-consistent; for example, this would be the case in low dimensions with $n \gg p$.
Unfortunately, this na\"{i}ve plug-in approach ceases to work when $\check\thetab_{k^c} \to \thetab_{k^c}^*$ at a rate slower than $n^{-1/2}$, which is typically the case for regularized estimators used in high dimensions.
This is because estimation based on the modified score with $\omegab_k = \eb_k$ is in general \emph{not} insensitive to the error in $\check\thetab_{k^c}$.
To see why, plug in $\omegab_k = \eb_k$ in \eqref{eq:modified:Taylor}:
\begin{multline}
\label{eq:naive:Taylor}
	\nabla_k \lKLIEP(\theta_k; \check\thetab_{k^c})
	= \nabla_k \lKLIEP(\thetab^*) + \nabla_{kk}^2 \lKLIEP(\thetab^*) (\theta_k - \theta_k^*)\\
	+ \nabla_{kk^c}^2 \lKLIEP(\thetab^*) (\check\thetab_{k^c} - \thetab_{k^c}^*) + \text{REM}.
\end{multline}
As $n \to \infty$, $\nabla_{kk}^2 \lKLIEP(\thetab^*) \to \Sigmab_{\psi, kk}$ and $\nabla_{kk^c}^2 \lKLIEP(\thetab^*) \to \Sigmab_{\psi, kk^c}$, where $\Sigmab_\psi = \Cov_x[\psib(\xb)]$, by the Law of Large Numbers and \Cref{lem:moment_matching}.
Thus, the estimator $\tilde\theta_k$, defined as a root of $g(\theta_k; \check\thetab_{k^c}, \eb_k)$, has the asymptotic linear approximation
\begin{equation}
\label{eq:naive:linear}
	\sqrt{n} \, (\tilde\theta_k - \theta_k^*)
	= -\sqrt{n} \, \Sigmab_{\psi, kk}^{-1} \nabla_k \lKLIEP(\thetab^*) - \sqrt{n} \, \Sigmab_{\psi, kk}^{-1} \Sigmab_{\psi, kk^c} (\check\thetab_{k^c} - \thetab_{k^c}^*) + \text{REM}.
\end{equation}
Looking at the right-hand side, although the first term $\sqrt{n} \, \Sigmab_{\psi, kk}^{-1} \nabla_k \lKLIEP(\thetab^*)$ is approximately mean-zero and normal, the same cannot be guaranteed for the second term $\sqrt{n} \, \Sigmab_{\psi, kk}^{-1} \Sigmab_{\psi, kk^c} (\check\thetab_{k^c} - \thetab_{k^c}^*)$ in many high-dimensional settings. Indeed, when $\check\thetab_{k^c}$ is given by the sparse KLIEP \eqref{eq:spKLIEP}, we can only guarantee $\|\check \thetab_{k^c} - \thetab_{k^c}^*\|_1 \leq \sqrt{\|\thetab^*\|_0 \log p / n}$; this is insufficient if we are looking to use normal approximation as a basis of inference.

Comparing \eqref{eq:modified:Taylor} and \eqref{eq:naive:Taylor}, it is clear that the reason that $\omegab_k = \eb_k$ is a poor choice is because the nuisance components of $\Sigmab_{\psi} \eb_k$ are in general nonzero, and hence, the product $\sqrt{n} \, \Sigmab_{\psi, kk^c} (\check\thetab_{k^c} - \thetab_{k^c}^*)$ is non-negligible whenever $\sqrt{n} \, (\check\thetab_{k^c} - \thetab_{k^c}^*)$ is non-negligible.
Therefore, if $\omegab_k$ is chosen to satisfy $\Sigmab_{\psi} \omegab_k \approx \eb_k$ instead --- i.e., $\omegab_k$ approximates the $k$th column of the \emph{inverse} of $\Sigmab_{\psi}$ --- it would have the effect of counterbalancing the error in $\check\thetab_{k^c}$.

Of course, since $\Sigmab_{\psi}$ is itself an unknown parameter, $\omegab_k$ would have to be estimated from the data as well. However, the good news is that even in high-dimensions, it is often possible to find $\omegab_k$ satisfying
\begin{equation}
\label{eq:omega_consistency}
	\cbr{\nabla^2 \lKLIEP(\thetab^*) \omegab_k - \eb_k}^\top \begin{bmatrix}\hat \theta_k - \theta_k^* \\ \check \thetab_{k^c} - \thetab_{k^c}^*\end{bmatrix} = o_{\PP}\rbr{n^{-1/2}},
\end{equation}
provided that structural assumptions are reasonable for $\Sigmab_{\psi}^{-1}$.

With this $\omegab_k$, $g$ has the first-order Taylor expansion
\begin{multline}
\label{eq:Neymanized:Taylor}
	\omegab_k^\top \nabla \lKLIEP(\theta_k; \check\thetab_{k^c})\\
	= \omegab_k^\top \nabla \lKLIEP(\thetab^*) + (\theta_k - \theta_k^*) + \cbr{\nabla^2 \lKLIEP(\thetab^*) \omegab_k - \eb_k}^\top \begin{bmatrix}\theta_k - \theta_k^* \\ \check \thetab_{k^c} - \thetab_{k^c}^*\end{bmatrix} + \text{REM}.
\end{multline}
The resulting estimator $\hat\theta_k$ then has the asymptotic linear approximation
\begin{equation}
\label{eq:Neymanized:linear}
	\sqrt{n} \, (\hat\theta_k - \theta_k^*)
	= -\sqrt{n} \, \omegab_k^\top \nabla \lKLIEP(\thetab^*) - \sqrt{n} \cbr{\nabla^2 \lKLIEP(\thetab^*) \omegab_k - \eb_k}^\top \begin{bmatrix}\hat \theta_k - \theta_k^* \\ \check \thetab_{k^c} - \thetab_{k^c}^*\end{bmatrix} + \text{REM}.
\end{equation}
In contrast to \eqref{eq:naive:linear}, the bias terms in \eqref{eq:Neymanized:linear} are still vanishing after $\sqrt{n}$-scaling.

\subsection{\label{supp:proofs:thm:main}Proof of \Cref{thm:main}}

Let $k \in [p]$. Let $\check\thetab$ and $\check\omegab_k$ denote the
outputs of Steps~1 and 2 of \Cref{procedure:KLIEP+}. For
$\lambda_\theta, \lambda_k, \delta_\theta, \delta_k, \delta_\sigma \in
[0,1)$, define an event
\begin{multline*}
\Ecal_\text{one}
= \Ecal_\text{one}(\lambda_\theta, \lambda_k, \delta_\theta, \delta_k, \delta_\sigma) =\\
\cbr{
\begin{array}{r c r c}
\textnormal{(G.1)} & 2 \|\nabla \lKLIEP(\thetab^*)\|_* \leq \lambda_\theta,&
\textnormal{(G.2)} & 2 \|\nabla^2 \lKLIEP(\thetab^*) \omegab_k^* - \eb_k\|_* \leq \lambda_k,\\
\textnormal{(E.1)} & \|\check\thetab - \thetab^*\| \leq \delta_\theta,&
\textnormal{(E.2)} & \|\check\omegab_k - \omegab_k^*\| \leq \delta_k,\\
\textnormal{(B.1)} & \abr{1 - \frac{\hat Z_y(\thetab^*)}{Z_y(\thetab^*)}} \lesssim \lambda_\theta,&
\textnormal{(B.2)} & \abr{\frac{1}{n_y} \sum_{j=1}^{n_y} \langle \omegab_k^*, \mub_\psi - \psib(\yb^{(j)}) \rangle \, r_{\theta^*}(\yb^{(j)})} \lesssim \lambda_k,\\
\textnormal{(V.1)} & 4 \|\hat\Sbb_\psi - \Sigmab_\psi\|_* \leq \delta_\sigma,&
\textnormal{(V.2)} & 4 \|\hat\Sbb_{\psi\hat r}(\thetab^*) - \Sigmab_{\psi r}\|_* \leq \delta_\sigma
\end{array}}.
\end{multline*}

\begin{theorem}[Re-statement of \Cref{thm:main}]
Assume Conditions 1 and 2.
Let $\hat\theta_k$ be the estimator constructed by \Cref{procedure:KLIEP+} with one-step approximation as
\begin{equation*}
\hat\theta_k = \check\theta_k - \check\omegab_k^\top \nabla \lKLIEP(\check\thetab).
\end{equation*}
Suppose $\PP(\Ecal_\text{one}) \geq 1 - \varepsilon_{\text{one},n}$ for some $\lambda_\theta, \lambda_k, \delta_\theta, \delta_k, \delta_\sigma \in [0,1)$.
Then,
\begin{equation*}
\label{eq:main:final_bd}
\sup_{t \in \RR} \left| \PP\left\{ \sqrt{n} \, (\hat\theta_k - \theta_k^*) / \hat\sigma_k \leq t \right\} - \Phi(t) \right|\\
\leq \Delta_1 + \Delta_2 + \Delta_3 + \varepsilon_{\text{one},n},
\end{equation*}
where
\begin{gather*}
\Delta_1
\lesssim \sqrt{\frac{\bar\kappa^2 / \underline\kappa}{\eta_{x,n} \eta_{y,n}}} \frac{\|\omegab_k^*\|}{\sqrt{n}},\quad
\Delta_2
\lesssim \sqrt{\frac{\eta_{x,n} \eta_{y,n}}{\underline\kappa / \bar\kappa^2}} \rbr{(\delta_\theta+\lambda_\theta) (\delta_k+\lambda_k) + \|\omegab_k^*\| \delta_\theta^2} \sqrt{n}, \\
\Delta_3
\lesssim (\bar\kappa^2 / \underline\kappa) \|\omegab_k^*\|^2 (\delta_\sigma+\delta_\theta) + \delta_k^2.
\end{gather*}
\end{theorem}

\begin{proof}
The proof combines two lemmas: a Berry-Esseen-type result for the leading linear term in the decomposition of $\sqrt{n} \, (\hat\theta_k - \theta_k^*)$ (\Cref{lem:Berry-Esseen}) and a technical lemma for incorporating error bounds (\Cref{lem:GAR}).

Recall $\mub_\psi = \EE_x[\psib(\xb)] = \EE_y[\psib(\yb) r_{\theta^*}(\yb)]$.
To use the lemmas, we break up $\hat\theta_ k - \theta_k^*$ as
\begin{equation*}
\hat\theta_k - \theta_k^* = A + B,
\end{equation*}
where
\begin{equation*}
A = \frac{1}{n_x} \sum_{i=1}^{n_x} \langle \omegab_k^*, \psib(\xb^{(i)})-\mub_\psi \rangle + \frac{1}{n_y} \sum_{j=1}^{n_y} \langle \omegab_k^*, \mub_\psi-\psib(\yb^{(j)}) \rangle \, r_{\theta^*}(\yb^{(j)}),
\end{equation*}
and $B = (\hat\theta_k - \theta_k^*) - A$.
Also, we let $C = (\hat\sigma_k / \sigma_k) - 1$, so that
\begin{equation*}
\sqrt{n} \, (\hat\theta_k - \theta_k^*) / \hat\sigma_k = \sqrt{n} \, \{(A + B) / \sigma_k\} / (1 + C).
\end{equation*}
Since $A$ is a sum of two i.i.d.~sums, $\sqrt{n} \, A / \sigma_k$ is well-approximated by a Gaussian law.
Indeed, \Cref{lem:Berry-Esseen} says
\begin{equation}
\sup_{t \in \RR} \left| \PP\left\{ \sqrt{n} \, A / \sigma_k \leq t \right\} - \Phi(t) \right|
\lesssim \sqrt{\frac{\bar\kappa^2 / \underline\kappa}{\eta_{x,n} \eta_{y,n}}} \frac{\|\omegab_k^*\|}{\sqrt{n}}
:= \Delta_1.
\end{equation}
The remainder of the proof is about obtaining the bounds $\delta_B$,
$\delta_C$, and $\varepsilon_{BC}$ that can be used with
\Cref{lem:GAR}.

First, we need an exact expression for $B$.
By definition,
\begin{align}
\hat\theta_k
&= \check\theta_k - \check\omegab_k^\top \nabla \lKLIEP(\check\thetab) \nonumber\\
&= \check\theta_k - \omegab_k^{*\top} \nabla \lKLIEP(\check\thetab) - \rbr{\check\omegab_k - \omegab_k^*}^\top \nabla \lKLIEP(\check\thetab).
\label{eq:1Step:rewrite}
\end{align}
Expand $\check\theta_k - \omegab_k^{*\top} \nabla
\lKLIEP(\check\thetab)$ about $\thetab^*$:
\begin{equation}
\check\theta_k - \omegab_k^{*\top} \nabla \lKLIEP(\check\thetab)
= \theta_k^* - \omegab_k^{*\top} \nabla \lKLIEP(\thetab^*) - \cbr{\nabla^2 \lKLIEP(\thetab^*) \omegab_k^* - \eb_k}^\top \rbr{\check\thetab - \thetab^*} - \omegab_k^{*\top} \rb,
\label{eq:1Step:Taylor:1}
\end{equation}
where by Taylor's theorem, $\rb$ is given by
\begin{equation*}
r_k
= \frac{1}{2} \sum_{k'' = 1}^p \sum_{k' = 1}^p \left\{ \int_0^1 (1-t) \, \partial^3_{k'' k' k} \lKLIEP(\thetab^* + t(\check\thetab - \thetab^*)) \, dt \right\} (\check\theta_{k''} - \theta_{k''}^*) (\check\theta_{k'} - \theta_{k'}^*).
\end{equation*}
Combining \eqref{eq:1Step:rewrite} and \eqref{eq:1Step:Taylor:1}, and rearranging,
\begin{multline*}
\hat\theta_k - \theta_k^*
= -\omegab_k^{*\top} \nabla\lKLIEP(\thetab^*) \\
- \rbr{\check\omegab_k - \omegab_k^*}^\top \nabla\lKLIEP(\check\thetab)
- \cbr{\nabla^2\lKLIEP(\thetab^*) \omegab_k^* - \eb_k}^\top \rbr{\check\thetab - \thetab^*}
- \omegab_k^{*\top} \rb.
\end{multline*}
The leading term is
\begin{align*}
\omegab_k^{*\top} \nabla \lKLIEP(\thetab^*)
&= \left\langle \omegab_k^*, \frac{1}{n_x} \sum_{i=1}^{n_x} \psib(\xb^{(i)}) - \frac{1}{n_y} \sum_{j=1}^{n_y} \psib(\yb^{(j)}) \hat r_{\theta^*}(\yb^{(j)}) \right\rangle \nonumber\\
&= \left\langle \omegab_k^*, \frac{1}{n_x} \sum_{i=1}^{n_x} \rbr{\psib(\xb^{(i)})-\mub_\psi} + \frac{1}{n_y} \sum_{j=1}^{n_y} \rbr{\mub_\psi-\psib(\yb^{(j)})} \hat r_{\theta^*}(\yb^{(j)}) \right\rangle \nonumber\\
&= \left\langle \omegab_k^*, \frac{1}{n_x} \sum_{i=1}^{n_x} \rbr{\psib(\xb^{(i)})-\mub_\psi} + \frac{Z_y(\thetab^*)}{\hat Z_y(\thetab^*)} \cdot \frac{1}{n_y} \sum_{j=1}^{n_y} \rbr{\mub_\psi-\psib(\yb^{(j)})} \, r_{\theta^*}(\yb^{(j)}) \right\rangle \nonumber\\
&=
\begin{multlined}[t]
\left\langle \omegab_k^*, \frac{1}{n_x} \sum_{i=1}^{n_x} \rbr{\psib(\xb^{(i)})-\mub_\psi} + \frac{1}{n_y} \sum_{j=1}^{n_y} \rbr{\mub_\psi-\psib(\yb^{(j)})} r_{\theta^*}(\yb^{(j)}) \right\rangle \nonumber\\
+ \rbr{\frac{Z_y(\thetab^*)}{\hat Z_y(\thetab^*)} - 1} \frac{1}{n_y} \sum_{j=1}^{n_y} \langle \omegab_k^*,\mub_\psi-\psib(\yb^{(j)}) \rangle \, r_{\theta^*}(\yb^{(j)}),
\end{multlined}
\end{align*}
where in the second step, we have used $n_y^{-1} \sum_{j=1}^{n_y} \hat r_\theta(\yb^{(j)}) \equiv 1$ for any $\thetab$.
Recognizing $A$ as the first term of the last line, we have
\begin{multline*}
B
=
\underbrace{\rbr{\frac{Z_y(\thetab^*)}{\hat Z_y(\thetab^*)} - 1} \frac{1}{n_y} \sum_{j=1}^{n_y} \langle \omegab_k^*,\mub_\psi-\psib(\yb^{(j)}) \rangle \, r_{\theta^*}(\yb^{(j)})}_{B_0} \\
- \underbrace{\rbr{\check\omegab_k - \omegab_k^*}^\top \nabla \lKLIEP(\check\thetab)}_{B_1}
- \underbrace{\cbr{\nabla^2 \lKLIEP(\thetab^*)\omegab_k^* - \eb_k}^\top \rbr{\check\thetab - \thetab^*}}_{B_2}
- \underbrace{\omegab_k^{*\top} \rb}_{B_3}.
\end{multline*}

We proceed to bound $|B|$ on $\Ecal_\text{one}$ using the defining conditions.
(B.1) and (B.2) imply a bound on $B_0$:
\begin{multline} \label{eq:B0}
|B_0|
= \abr{\rbr{\frac{Z_y(\thetab^*)}{\hat Z_y(\thetab^*)} - 1} \frac{1}{n_y} \sum_{j=1}^{n_y} \langle \omegab_k^*,\mub_\psi-\psib(\yb^{(j)}) \rangle \, r_{\theta^*}(\yb^{(j)})}\\
= \abr{\frac{Z_y(\thetab^*)}{\hat Z_y(\thetab^*)}} \abr{1 - \frac{\hat Z_y(\thetab^*)}{Z_y(\thetab^*)}} \abr{\frac{1}{n_y} \sum_{j=1}^{n_y} \langle \omegab_k^*,\mub_\psi-\psib(\yb^{(j)}) \rangle \, r_{\theta^*}(\yb^{(j)})}
\leq K_1 \lambda_\theta \lambda_k,
\end{multline}
because $Z_y(\thetab^*) / \hat Z_y(\thetab^*) \in [M_r^{-1}, M_r]$ under \Cref{cond:bdr}.
$B_1$ is further decomposed as
\begin{equation*}
B_1
= \underbrace{\rbr{\check\omegab_k - \omegab_k^*}^\top \nabla \lKLIEP(\thetab^*)}_{B_{11}}
+ \underbrace{\rbr{\check\omegab_k - \omegab_k^*}^\top \rbr{\nabla \lKLIEP(\check\thetab) - \nabla \lKLIEP(\thetab^*)}}_{B_{12}}.
\end{equation*}
Using (G.1) and (E.2) for \eqref{eq:B11}, and (G.2) and (E.1) for \eqref{eq:B2},
\begin{gather}
|B_{11}|
\leq \|\check\omegab_k - \omegab_k^*\| \|\nabla \lKLIEP(\thetab^*)\|_*
\leq \lambda_\theta \delta_k,
\label{eq:B11}\\
|B_2|
\leq \|\nabla^2\lKLIEP(\thetab^*) \omegab_k^* - \eb_k\|_* \|\check\thetab - \thetab^*\|
\leq \lambda_k \delta_\theta.
\label{eq:B2}
\end{gather}
For $B_{12}$, we use the mean value theorem to express each component of $\nabla \lKLIEP(\check\thetab) - \nabla \lKLIEP(\thetab^*)$ as
\begin{equation*}
\partial_k \lKLIEP(\check\thetab) - \partial_k \lKLIEP(\thetab^*)
= \sum_{k' = 1}^p \partial^2_{k k'} \lKLIEP(\bar\thetab_k) (\check\theta_{k'} - \theta_{k'}^*),
\end{equation*}
where $\bar\thetab_k$ is on the line segment connecting $\check\thetab$ and $\thetab^*$.
Using \eqref{eq:lKLIEP:2nd}, this can be written as
\begin{multline*}
\partial_k \lKLIEP(\check\thetab) - \partial_k \lKLIEP(\thetab^*)\\
\begin{aligned}[t]
&= \sum_{k' = 1}^p \cbr{\frac{1}{n_y} \sum_{j=1}^{n_y} \hat r_{\bar\theta_k}(\yb^{(j)}) \psi_k(\yb_k^{(j)}) \psi_{k'}(\yb_{k'}^{(j)}) - \hat\mu_k(\bar\thetab_k) \hat\mu_{k'}(\bar\thetab_k)} (\check\theta_{k'} - \theta_{k'}^*)\\
&= \frac{1}{n_y} \sum_{j=1}^{n_y} \hat r_{\bar\theta_k}(\yb^{(j)}) \psi_k(\yb_k^{(j)}) \cbr{\sum_{k' = 1}^p \psi_{k'}(\yb_{k'}^{(j)}) (\check\theta_{k'} - \theta_{k'}^*)} - \hat\mu_k(\bar\thetab_k) \cbr{\sum_{k' = 1}^p \hat\mu_{k'}(\bar\thetab_k) (\check\theta_{k'} - \theta_{k'}^*)}.
\end{aligned}
\end{multline*}
Under \Cref{cond:bdr},
\begin{eqnarray*}
\sum_{k' = 1}^p \psi_{k'}(\yb_{k'}^{(j)}) (\check\theta_{k'} - \theta_{k'}^*)
\leq M_\psi \|\check\thetab - \thetab^*\|
&\text{and}&
\sum_{k' = 1}^p \hat\mu_{k'}(\bar\thetab_k) (\check\theta_{k'} - \theta_{k'}^*)
\leq M_\psi M_r^2 \|\check\thetab - \thetab^*\|,
\end{eqnarray*}
so that
\begin{equation*}
\|\nabla \lKLIEP(\check\thetab) - \nabla \lKLIEP(\thetab^*)\|_*
\leq K_2 \|\check\thetab - \thetab^*\|.
\end{equation*}
With (E.1) and (E.2),
\begin{equation}
|B_{12}|
\leq \|\check\omegab_k - \omegab_k^*\| \|\nabla \lKLIEP(\check\thetab) - \nabla \lKLIEP(\thetab^*)\|_*
\leq K_2 \delta_\theta \delta_k.
\label{eq:B12}
\end{equation}
We turn to $B_3$.
Under \Cref{cond:bdr}, \eqref{eq:lKLIEP:3rd} implies a uniform bound on the third-order tensor, so
\begin{equation}
|B_3|
\leq \|\omegab_k^*\| \|\rb\|_*
\leq K_3 \|\omegab_k^*\| \delta_\theta^2.
\label{eq:B3}
\end{equation}
Combining \eqref{eq:B0} to \eqref{eq:B3},
\begin{equation}
\sqrt{n} \, |B| / \sigma_k
\lesssim \sqrt{\frac{\eta_{x,n} \eta_{y,n}}{\underline\kappa / \bar\kappa^2}} \rbr{(\delta_\theta+\lambda_\theta) (\delta_k+\lambda_k) + \|\omegab_k^*\| \delta_\theta^2} \sqrt{n}
:= \Delta_2.
\end{equation}
Next, we bound $|C|$ on $\Ecal_\text{one}$.
Using (E.1), (E.2), (V.1), (V.2),
\begin{equation}
\abr{\frac{\hat\sigma_k}{\sigma_k} - 1}
\leq \abr{\frac{\hat\sigma_k^2 - \sigma_k^2}{\sigma_k^2}}\\
\lesssim (\bar\kappa^2 / \underline\kappa) \|\omegab_k^*\|^2 (\delta_\sigma+\delta_\theta) + \delta_k^2
:= \Delta_3
\end{equation}
by \Cref{lem:varest}.

Taking $A = \sqrt{n} \, A / \sigma_k$, $B = \sqrt{n} \, B / \sigma_k$ $C = (\hat\sigma_k / \sigma_k) - 1$, $\varepsilon_A = \Delta_1$, $\delta_B = \Delta_2$, $\delta_C = \Delta_3$, $\varepsilon_{BC} = \varepsilon_\text{one}$ in \Cref{lem:GAR} concludes the proof.
\end{proof}

\begin{remark}
In the last step, one could just as well apply \Cref{lem:GAR} with $A = \sqrt{n} \, A / \sigma_k$, $B = \sqrt{n} \, B / \sigma_k$ $C = 0$, $\varepsilon_A = \Delta_1$, $\delta_B = \Delta_2$, $\delta_C = 0$, $\varepsilon_{BC} = \varepsilon_\text{one}$ to end up with
\begin{equation*}
\sup_{t \in \RR} \left| \PP\left\{ \sqrt{n} \, (\hat\theta_k - \theta_k^*) / \sigma_k \leq t \right\} - \Phi(t) \right|\\
\leq \Delta_1 + \Delta_2 + \varepsilon_{\text{one}},
\end{equation*}
but this result is not as useful.
\end{remark}

\subsection{\label{supp:proofs:thm:bootstrap}Proof of \Cref{thm:bootstrap}}

For $k = 1, \dots, p$, let
\begin{equation} \label{eq:def_L_hat_boot}
\hat L^B_{n_x,n_y,k}
= -\frac{1}{\sqrt{n}} \, \check\omegab_k^\top \left\{ \eta_{x,n}^{-1} \sum_{i=1}^{n_x} \rbr{\psib(\xb^{(i)}) - \overline{\psib}} \xi_x^{(i)} - \eta_{y,n}^{-1} \sum_{j=1}^{n_y} \rbr{\psib(\yb^{(j)}) \hat r_{\check\theta}(\yb^{(j)}) - \hat\mub(\check\thetab)}\xi_y^{(j)} \right\},
\end{equation}
where $\check\thetab$ is a consistent estimator of $\thetab$, and
\begin{equation*}
\xi_x^{(1)}, \dots, \xi_x^{(n_x)}, \xi_y^{(1)}, \dots, \xi_y^{(n_y)} \iidsim \Ncal(0,1).
\end{equation*}
Note that
\begin{equation*}
\overline{\psib} = \frac{1}{n_x} \sum_{i=1}^{n_x} \psib(\xb^{(i)})
\quad\text{and}\quad
\hat\mub(\check\thetab) = \frac{1}{n_y} \sum_{j=1}^{n_y} \psib(\yb^{(j)}) \hat r_{\check\theta}(\yb^{(j)})
\end{equation*}
are the two components of $\nabla \lKLIEP(\check\thetab)$. The
centering is necessary for variance-matching, because
$\nabla \lKLIEP(\check\thetab) \not\equiv \zero$ for high-dimensional
estimators. In terms of \eqref{eq:def_L_hat_boot}, the statistic to be bootstrapped is written as
\begin{equation*}
\hat T_{n_x,n_y} = \max_k |\hat L^B_{n_x,n_y,k}|.
\end{equation*}

The multiplier bootstrap scheme presupposes that the conditional
distribution of $\hat L^B_{n_x,n_y,k}$ is a good proxy for the
distribution of $\sqrt{n} \, (\hat\theta_k - \theta_k^*)$. It is not
difficult to imagine that the conditional distribution of
$\hat L^B_{n_x,n_y,k}$ is a good proxy for the distribution of
$L_{n_x,n_y,k}$, where
\begin{equation*}
L_{n_x,n_y,k}
= -\frac{1}{\sqrt{n}} \, \omegab_k^{*\top} \left\{ \eta_{x,n}^{-1} \sum_{i=1}^{n_x} \rbr{\psib(\xb^{(i)}) - \mub_\psi} - \eta_{y,n}^{-1} \sum_{j=1}^{n_y} \rbr{\psib(\yb^{(j)}) r_{\theta^*}(\yb^{(j)}) - \mub_\psi} \right\}.
\end{equation*}
Because $L_{n_x,n_y,k}$ is the leading term in the first-order Taylor
approximation of $\sqrt{n} \, (\hat\theta_k - \theta_k^*)$, the
distribution of the former is also close to the distribution of the
latter, and hence, the conditional distribution of
$\hat L^B_{n_x,n_y,k}$ can be used to estimate the quantiles of
$\sqrt{n} \, (\hat\theta_k - \theta_k^*)$.

Let $\Sigmab_\text{pooled}$ be defined as in \eqref{eq:varest}, and let
$\Omegab^* = \Sigmab_\psi^{-1}$. Recall that the $k$th column of
$\Omegab^*$ is $\omegab_k^*$.
For $\lambda_\theta, (\lambda_k)_{k=1}^p, \delta_\theta, (\delta_k)_{k=1}^p \in [0,1)$, define an event
\begin{multline*}
\Ecal_\text{all}
= \Ecal_\text{all}(\lambda_\theta, (\lambda_k)_{k=1}^p, \delta_\theta, (\delta_k)_{k=1}^p) =\\
\cbr{
\begin{array}{r c r c}
\textnormal{(G.1)} & 2 \|\nabla \lKLIEP(\thetab^*)\|_* \leq \lambda_\theta,&
\textnormal{(G.2)} & 2 \|\nabla^2 \lKLIEP(\thetab^*) \omegab_k^* - \eb_k\|_* \leq \lambda_k\ \forall\, k,\\
\textnormal{(E.1)} & \|\check\thetab - \thetab^*\| \leq \delta_\theta,&
\textnormal{(E.2)} & \|\check\omegab_k - \omegab_k^*\| \leq \delta_k\ \forall\, k,\\
\textnormal{(B.1)} & \abr{1 - \frac{\hat Z_y(\thetab^*)}{Z_y(\thetab^*)}} \lesssim \lambda_\theta,&
\textnormal{(B.2)} & \abr{\frac{1}{n_y} \sum_{j=1}^{n_y} \langle \omegab_k^*, \mub_\psi - \psib(\yb^{(j)}) \rangle \, r_{\theta^*}(\yb^{(j)})} \lesssim \lambda_k\ \forall\, k
\end{array}}.
\end{multline*}
Put $\nu_n = 1 \vee \max\{\|\omegab_k^*\| : k = 1, \dots, p\}$, and set
\begin{equation*}
B_n = \frac{(1 \vee \bar\kappa)^3 (1 \vee M_\psi)^3 M_r^3 \nu_n^{21/2}}{\sqrt{\underline\kappa^3 \eta_{x,n} \eta_{y,n}}}
\quad\text{and}\quad
\delta_n = \rbr{\frac{B_n^2 \log^7 (pn)}{n}}^{1/6}.
\end{equation*}

\begin{theorem}[Re-statement of \Cref{thm:bootstrap}]
Assume Conditions 1 and 2. Let $\hat\thetab$ be the estimator
constructed by \Cref{procedure:KLIEP+} with one-step approximation as
\begin{equation*}
\hat\thetab = \check\thetab - \check\Omegab^\top \nabla\lKLIEP(\check\thetab),
\end{equation*}
where $\check\Omegab = [\check\omegab_k]_{k=1}^p \in \RR^{p \times p}$
is the matrix with the $k$th column given by $\check\omegab_k$.
Suppose
\begin{align*}
D_1 & :=
\max_k \sqrt{\frac{\eta_{x,n} \eta_{y,n}}{\underline\kappa / \bar\kappa^2}} \rbr{(\delta_\theta+\lambda_\theta) (\delta_k+\lambda_k) + \|\omegab_k^*\| \delta_\theta^2} \sqrt{n}
\lesssim
\rbr{\frac{B_n^2 \log^4 (pn)}{n}}^{1/6},\\
D_2 & :=
\max_k \frac{\underline\kappa / \bar\kappa^2}{\eta_{x,n}^2 \eta_{y,n}^2} \rbr{\delta_k^2 + \eta_{x,n} \|\omegab_k\|^2 \rbr{\delta_\theta + \lambda_\theta}^2}
\lesssim
\rbr{\frac{B_{n}^2 \log (pn)}{n}}^{1/3}.
\end{align*}
If $\PP(\Ecal_\text{all}) \geq 1 - \varepsilon_{\text{all},n}$, then
\begin{equation*}
\sup_{\alpha \in (0,1)} \abr{\PP\cbr{T_{n_x,n_y} \leq \hat c_{T, \alpha}} - (1 - \alpha)}
= O(\delta_n + \varepsilon_{\text{all},n})
\end{equation*}
with probability at least $1 - \varepsilon_{\text{all},n} - n^{-1}$.
\end{theorem}

\begin{proof}
We prove the result for the case where $\mub_\psi = \EE_x[ \psib(\xb) ] = \EE_y[ \psib(\yb) r_{\theta^*}(\yb) ] = \zero$.
The general result follows by the consistency of empirical averages.

The proof is by Theorems 2.1 and 2.2 of \cite{Belloni2018High}.
The two theorems are Gaussian approximation results for approximate means over the class $\Acal$ of hyper-rectangles in $\RR^p$, in other words, $\Acal$ contains sets of the form
\begin{equation*}
A = \cbr{\vb \in \RR^p : l_k \leq v_k \leq u_k \text{ for all } k = 1, \dots, p},
\end{equation*}
where $-\infty \leq l_k \leq u_k \leq +\infty$ for all $k$.
In the proof of \Cref{thm:main}, we saw that $\sqrt{n} \, (\hat\thetab - \thetab^*)$ may be decomposed as
\begin{equation*}
\sqrt{n} \, (\hat\thetab - \thetab^*) = L_{n} + R_{n},
\end{equation*}
where the leading linear term has the form
\begin{align*}
L_{n}
&= -\frac{1}{\sqrt{n}} \, \Omegab^{*\top} \cbr{\eta_{x,n}^{-1} \sum_{i=1}^{n_x} \psib(\xb^{(i)}) - \eta_{y,n}^{-1} \sum_{j=1}^{n_y} \psib(\yb^{(j)}) r_{\theta^*}(\yb^{(j)})}
\intertext{and the remainder is given by}
R_{n}
&=
\begin{multlined}[t]
\sqrt{n} \, \Bigg\{ \Omegab^{*\top} \rbr{\frac{Z_y(\thetab^*)}{\hat Z_y(\thetab^*)} - 1} \frac{1}{n_y} \sum_{j=1}^{n_y} \psib(\yb^{(j)}) r_{\theta^*}(\yb^{(j)})\\
- \cbr{\check\Omegab - \Omegab^*}^\top \nabla \lKLIEP(\check\thetab) - \cbr{\nabla^2 \lKLIEP(\thetab^*) \Omegab^* - \Ib}^\top \rbr{\check\thetab - \thetab^*} + \Omegab^{*\top} \rb \Bigg\}.
\end{multlined}
\end{align*}
This demonstrates that our problem also falls under the approximate means framework.

Let $\Pb = \PP[\dummy \mid \Xb_{n_x}, \Yb_{n_y}]$ denote the conditional probability given the data.
If applicable, Theorem 2.1 would give us
\begin{equation*}
\sup_{A \in \Acal} \abr{\PP\cbr{\sqrt{n} \, (\hat\thetab - \thetab^*) \in A} - \PP\cbr{\Ncal(\zero, \Omegab^{*\top} \Sigmab_\text{pooled} \Omegab^*) \in A}}
= O\rbr{\delta_n + \varepsilon_{\text{all},n}},
\end{equation*}
and Theorem 2.2 would give us
\begin{equation*}
\sup_{A \in \Acal} \abr{\Pb\cbr{\hat L_n^B \in A} - \PP\cbr{\Ncal(\zero, \Omegab^{*\top} \Sigmab_\text{pooled} \Omegab^*) \in A}}
= O\rbr{\delta_n}
\end{equation*}
with probability at least $1 - \varepsilon_{\text{all},n} - n^{-1}$.
Combining the two statements,
\begin{equation} \label{eq:comparison}
\sup_{A \in \Acal} \abr{\PP\cbr{\sqrt{n} \, (\hat\thetab - \thetab^*) \in A}- \Pb\cbr{\hat L_n^B \in A}}
= O\rbr{\delta_n + \varepsilon_{\text{all},n}}
\end{equation}
with probability at least $1 - \varepsilon_{\text{all},n} - n^{-1}$.
Once \eqref{eq:comparison} is established for $\Acal$, then \emph{a fortiori} \eqref{eq:comparison} is established for the sub-collection
\begin{equation*}
A = \cbr{\vb \in \RR^p : \max_k |v_k| \leq t \text{ for all } k = 1, \dots, p},
\end{equation*}
so that in particular
\begin{equation}
\sup_{A \in \Acal} \left| \PP\left\{ T_{n_x,n_y} \leq \hat c_{T,\alpha} \right\} - (1-\alpha) \right|
= O\rbr{\delta_n + \varepsilon_{\text{all},n}},
\end{equation}
which is the statement of the theorem.

Thus, in a nutshell, our work here boils down to checking that our problem satisfies the conditions of Theorems 2.1 and 2.2 of \cite{Belloni2018High} --- Conditions M, E, and A --- which we restate in context below.

Before we proceed, let
\begin{equation*}
\hat L_{n}
= -\frac{1}{\sqrt{n}} \, \check\Omegab^{\top} \cbr{\eta_{x,n}^{-1} \sum_{i=1}^{n_x} \psib(\xb^{(i)}) - \eta_{y,n}^{-1} \sum_{j=1}^{n_y} \psib(\yb^{(j)}) \hat r_{\check\theta}(\yb^{(j)})}.
\end{equation*}
This is a feasible approximation to $L_{n}$, and this is what we actually bootstrap as $\hat L_{n}^B$.

\paragraph*{Condition M.}
Translated to our problem, Condition M of \cite{Belloni2018High} is
\begin{gather}
\Var[ L_{n,k} ]
= \omegab_k^{*\top} \left\{ \eta_{x,n}^{-1} \Sigmab_\psi + \eta_{y,n}^{-1} \Sigmab_{\psi r} \right\} \omegab_k^*
\geq c \text{ for some } c > 0, \label{eq:cond_M1}\\
\eta_{x,n}^{-2} \EE_x\sbr{|\omegab_k^{*\top} \psib(\xb)|^3} + \eta_{y,n}^{-2} \EE_y\sbr{|\omegab_k^{*\top} \psib(\yb) r_{\theta^*}(\yb)|^3} \leq c^{3/2} B_{n}, \label{eq:cond_M2}\\
\eta_{x,n}^{-3} \EE_x\sbr{|\omegab_k^{*\top} \psib(\xb)|^4} + \eta_{y,n}^{-3} \EE_y\sbr{|\omegab_k^{*\top} \psib(\yb) r_{\theta^*}(\yb)|^4} \leq c^{2\phantom{/3}} B_{n}^2\phantom{,} \label{eq:cond_M3}
\end{gather}
for each $k \in [p]$.

Under \Cref{cond:bddpopeigs}, \eqref{eq:bd_on_var} says
\begin{equation*}
\Var[ L_{n,k} ]
= \omegab_k^{*\top} \left\{ \eta_{x,n}^{-1} \Sigmab_\psi + \eta_{y,n}^{-1} \Sigmab_{\psi r} \right\} \omegab_k^*
\geq \underline{\kappa} / (\bar\kappa^2 \eta_{x,n} \eta_{y,n}) \quad \forall \, k.
\end{equation*}
Thus, \eqref{eq:cond_M1} is satisfied with $c = \underline{\kappa} / (\bar\kappa^2 \eta_{x,n} \eta_{y,n})$.

By \eqref{eq:omega_k:l1}, for all $k$,
\begin{equation} \label{eq:bd_on_x_summand}
|\omegab_k^{*\top} \psib(\xb)|
\leq M_\psi \|\omegab_k^*\|
\end{equation}
and
\begin{equation} \label{eq:bd_on_y_summand}
|\omegab_k^{*\top} \psib(\yb) r_{\theta^*}(\yb)|
\leq M_r M_\psi \|\omegab_k^*\|.
\end{equation}
So,
\begin{equation*}
c^{-3/2} \rbr{\eta_{x,n}^{-2} \EE_x\sbr{|\omegab_k^{*\top} \psib(\xb)|^3} + \eta_{y,n}^{-2} \EE_y\sbr{|\omegab_k^{*\top} \psib(\yb) r_{\theta^*}(\yb)|^3}}\\
\leq \frac{\bar\kappa^3 M_r^3 M_\psi^3 \nu_n^3}{\sqrt{\underline\kappa^3 \eta_{x,n} \eta_{y,n}}}
\leq B_{n}
\end{equation*}
and
\begin{equation*}
c^{-2} \rbr{\eta_{x,n}^{-3} \EE_x\sbr{|\omegab_k^{*\top} \psib(\xb)|^4} + \eta_{y,n}^{-3} \EE_y\sbr{|\omegab_k^{*\top} \psib(\yb) r_{\theta^*}(\yb)|^4}}\\
\leq \frac{\bar\kappa^4 M_r^4 M_\psi^4 \nu_n^4}{\underline\kappa^2 \eta_{x,n} \eta_{y,n}}
\leq B_{n}^2.
\end{equation*}
Thus, both \eqref{eq:cond_M2} and \eqref{eq:cond_M3} are satisfied with $B_n$ as defined in Section~\ref{sec:theory:bootstrap}.

\paragraph*{Condition E.}
Translated to our problem, Condition E of \cite{Belloni2018High} is
\begin{equation*}
\EE_x\sbr{\exp\cbr{\abr{\omegab_k^{*\top} \psib(\xb)} \left/ \rbr{\eta_{x,n} c^{1/2} B_{n} \right.}}}
\leq 2
\end{equation*}
and
\begin{equation*}
\EE_y\sbr{\exp\cbr{\abr{\omegab_k^{*\top} \psib(\yb) r_{\theta^*}(\yb)} \left/ \rbr{\eta_{y,n} c^{1/2} B_{n} \right.}}}
\leq 2
\end{equation*}
with
\begin{equation*}
\rbr{\frac{B_{n}^2 \log^7 (pn)}{n}}^{1/6} \leq \delta_{n}.
\end{equation*}
But these are all immediate by \eqref{eq:bd_on_x_summand}, \eqref{eq:bd_on_y_summand}, and how we defined $B_{n}$ and $\delta_{n}$ in Section~\ref{sec:theory:bootstrap}.

\paragraph*{Condition A.}
Translated to our problem, Condition A of \cite{Belloni2018High} is
\begin{equation} \label{eq:cond_A1}
\PP\cbr{\max_k |R_{n,k}| > c^{1/2} \delta_{n} / \sqrt{\log (pn)}}\leq \varepsilon_{\text{all},n}
\end{equation}
and
\begin{equation} \label{eq:cond_A2}
\PP\cbr{\max_k v_k^2 > c \, \delta_{n}^2 / \log^2 (pn)} \leq \varepsilon_{\text{all},n}
\end{equation}
where
\begin{multline*}
v_k^2
= v_{x,k}^2 + v_{y,k}^2
= \frac{\eta_{x,n}^{-1}}{n_x} \sum_{i=1}^{n_x} \langle \check\omegab_k - \omegab_k^*, \psib(\xb^{(i)}) \rangle^2\\
+ \frac{\eta_{y,n}^{-1}}{n_y} \sum_{j=1}^{n_y} \left( \langle \check\omegab_k, \psib(\yb^{(j)}) \hat r_{\check\theta}(\yb^{(j)}) \rangle - \langle \omegab_k^*, \psib(\yb^{(j)}) r_{\theta^*}(\yb^{(j)}) \rangle \right)^2.
\end{multline*}

We saw in the proof of \Cref{thm:main} that on $\Ecal_\text{all}$,
\begin{equation*}
c^{-1/2} |R_{n,k}|
\lesssim \sqrt{\frac{\eta_{x,n} \eta_{y,n}}{\underline\kappa / \bar\kappa^2}} \rbr{(\delta_\theta+\lambda_\theta) (\delta_k+\lambda_k) + \|\omegab_k^*\| \delta_\theta^2} \sqrt{n}
\quad \forall \, k.
\end{equation*}
Under the conditions of the theorem,
\begin{equation*}
c^{-1/2} |R_{n,k}|
\lesssim \rbr{\frac{B_n^2 \log^4 (pn)}{n}}^{1/6}
= \left. \rbr{\frac{B_n^2 \log^7 (pn)}{n}}^{1/6} \right/ \sqrt{\log (pn)}
\lesssim \delta_n / \sqrt{\log (pn)}
\quad \forall \, k.
\end{equation*}
$v_k^2$ is controlled by obtaining separate bounds for $v_{x,k}^2$ and $v_{y,k}^2$.
For the former,
\begin{equation*}
v_{x,k}^2
= \frac{1}{\eta_{x,n} n_x} \sum_{i=1}^{n_x} \langle \check\omegab_k - \omegab_k^*, \psib(\xb^{(i)}) \rangle^2
\leq \eta_{x,n}^{-1} M_\psi^2 \|\check\omegab_k-\omegab_k^*\|^2
\lesssim \eta_{x,n}^{-1} \delta_k^2
\end{equation*}
In the case of the latter, we first decompose each summand using
\begin{multline*}
\langle \check\omegab_k, \psib(\yb^{(j)}) \hat r_{\check\theta}(\yb^{(j)}) \rangle - \langle \omegab_k^*, \psib(\yb^{(j)}) r_{\theta^*}(\yb^{(j)}) \rangle \\
= \langle \check\omegab_k - \omegab_k^*, \psib(\yb^{(j)}) \hat r_{\check\theta}(\yb^{(j)}) \rangle + \langle \omegab_k^*, \psib(\yb^{(j)}) \rangle \, \left( \hat r_{\check\theta}(\yb^{(j)}) - r_{\theta^*}(\yb^{(j)}) \right).
\end{multline*}
Then,
\begin{equation*}
|\langle \check\omegab_k - \omegab_k^*, \psib(\yb^{(j)}) \hat r_{\check\theta}(\yb^{(j)}) \rangle|
\leq M_\psi M_r^2 \|\check\omegab_k - \omegab_k^*\|,
\end{equation*}
and
\begin{multline*}
\left| \langle \omegab_k^*, \psib(\yb^{(j)}) \rangle \, \left( \hat r_{\check\theta}(\yb^{(j)}) - r_{\theta^*}(\yb^{(j)}) \right) \right| \\
= \left| \langle \omegab_k^*, \psib(\yb^{(j)}) \rangle \left\{ \left( \hat r_{\check\theta}(\yb^{(j)}) - \hat r_{\theta^*}(\yb^{(j)}) \right) + \left( \hat r_{\theta^*}(\yb^{(j)}) - r_{\theta^*}(\yb^{(j)}) \right) \right\} \right| \\
\leq M_\psi \|\omegab_k^*\| \left( L_1 \|\check\thetab - \thetab^*\| + M_r^2 \left| 1 - \frac{\hat Z_y(\thetab^*)}{Z_y(\thetab^*)} \right| \right),
\end{multline*}
where we have used \Cref{lem:Lips}, as well as \eqref{eq:Zhat_over_Z} and \eqref{eq:omega_k:l1}.
Hence,
\begin{align*}
v_{y,k}^2
&= \frac{n}{n_y^2} \sum_{j=1}^{n_y} \left( \langle \check\omegab_k, \psib(\yb^{(j)}) \hat r_{\check\theta}(\yb^{(j)}) \rangle - \langle \omegab_k^*, \psib(\yb^{(j)}) r_{\theta^*}(\yb^{(j)}) \rangle \right)^2\\
&\leq \eta_{y,n}^{-1} \Bigg\{ M_\psi M_r^2 \|\check\omegab_k - \omegab_k^*\| + M_\psi \|\omegab_k^*\| \left( L_1 \|\check\thetab - \thetab^*\| + M_r^2 \left| 1 - \frac{\hat Z_y(\thetab^*)}{Z_y(\thetab^*)} \right| \right) \Bigg\}^2.\\
&\lesssim \eta_{y,n}^{-1} \cbr{\delta_k + \|\omegab_k^*\| \rbr{\delta_\theta+\lambda_\theta}}^2.
\end{align*}
Thus,
\begin{equation*}
v_k^2
\lesssim (\eta_{x,n} \eta_{y,n})^{-1} \delta_k^2 + \eta_{y,n}^{-1} \|\omegab_k\|^2 \rbr{\delta_\theta + \lambda_\theta}^2.
\end{equation*}
Under the conditions of the theorem,
\begin{equation*}
c \, v_k^2
\lesssim \rbr{\frac{B_{n}^2 \log (pn)}{n}}^{1/3}\\
= \left. \rbr{\frac{B_{n}^2 \log^7 (pn)}{n}}^{1/3} \right/ \log^2 (pn)
\lesssim \delta_{n}^2 / \log (pn)
\quad \forall \, k.
\end{equation*}
Clearly,
\begin{equation}
\PP\cbr{\max_k |R_{n,k}| > c^{1/2} \delta_{n} / \sqrt{\log (pn)}} \leq \PP(\Ecal^c) \leq \varepsilon_{\text{all},n}
\end{equation}
and
\begin{equation}
\PP\cbr{\max_k v_k^2 > c \, \delta_{n}^2 / \log^2 (pn)} \leq \PP(\Ecal^c) \leq \varepsilon_{\text{all},n}.
\end{equation}

\paragraph*{Conclusion.}
Subject to some growth constraints, all three of Conditions M, E, and A are satisfied by our problem.
The result follows by the discussion at the start of the proof.
\end{proof}

\section{Proofs for the $\ell_1$-penalty case}

\subsection{\label{supp:l1proofs:cor:main}Proof of \Cref{cor:main}}

\begin{theorem}[Re-statement of \Cref{cor:main}]
Assume \Cref{cond:bdr} with $\ell_1$-norm and \Cref{cond:bddpopeigs}.
Assume additionally that
\begin{eqnarray}
\label{cond:main:lasso:sparsity}
\frac{s_{\theta, 0}}{s_{k, q_k}} \rbr{\frac{n}{\log p}}^{\tfrac{q_k}{4}} \lesssim 1
&\text{and}&
\frac{1}{s_{k, q_k}} \rbr{\frac{\log p}{n}}^{\tfrac{q_k}{4}\tfrac{2-q_k}{1-q_k}} \lesssim 1.
\end{eqnarray}
Let $\hat\theta_k$ be the SparKLIE+1 estimator with tuning parameters
\begin{eqnarray}
\label{cond:main:lasso:lambda}
\lambda_\theta \asymp \rbr{\frac{\log p}{n}}^{1/2}
&\text{and}&
\lambda_k \asymp s_{k, q_k}^{1/(2-q_k)} \rbr{\frac{\log p}{n}}^{1/2}.
\end{eqnarray}
Let $s$ be a sequence of integers satisfying
\begin{equation*}
s \geq s_{\theta, 0} \vee s_{k, q_k} \lambda_k^{-q_k}.
\end{equation*}
Let $\varepsilon_{\text{RSC},n}$ be a sequence in $(0,1)$ decreasing
to 0.
Then, provided that
\begin{equation}
\label{cond:main:lasso:sample}
n_y \geq C' (\bar\kappa / \underline\kappa^2) M_\psi^2 M_r^2 s \log^2 (s) \log (p \vee n_y) \log (n_y) / \varepsilon_{\text{RSC},n}^2,
\end{equation}
where $C' > 0$ is the known, absolute constant determined in \Cref{lem:sparse_norm_H_minus_EH}, we have
\begin{multline*}
\sup_{t \in \RR} \left| \PP\left\{ \sqrt{n} \, (\hat\theta_k - \theta_k^*) / \hat\sigma_k \leq t \right\} - \Phi(t) \right|\\
\leq
O\rbr{s_{\theta, 0} s_{k, q_k}^{2+\tfrac{1-2q_k}{2-q_k}} \rbr{\frac{\log p}{n}}^{1-q_k} \sqrt{n}}
+ \varepsilon_{\text{RSC},n} + c\exp\rbr{-c' \log p}.
\end{multline*}
\end{theorem}

\begin{proof}
For the sake of clarity, we ignore the factors of $\bar\kappa$, $\underline\kappa$, $\eta_{x,n}$, and $\eta_{y,n}$ in calculations.
Detailed bounds are, albeit tedious, not difficult to derive.

By \Cref{thm:main}, it suffices to find an event $\Ecal \subseteq \Ecal_\text{one}$ with $\PP(\Ecal^c) \searrow 0$.
Let
\begin{align*}
\Hb(\thetab):
&= \frac{\hat Z_y^2(\thetab)}{Z_y^2(\thetab)} \nabla^2 \lKLIEP(\thetab)\\
&= \frac{1}{n_y^2} \sum_{1 \leq j < j' \leq n_y} \rbr{\psib(\yb^{(j)}) - \psib(\yb^{(j')})} \rbr{\psib(\yb^{(j)}) - \psib(\yb^{(j')})}^\top r_\theta(\yb^{(j)}) r_\theta(\yb^{(j')}).
\end{align*}
Consider the event
\begin{multline*}
\Ecal_\text{one}^\text{L} =\\
\cbr{
\begin{array}{c}
\begin{array}{r c r c}
\textnormal{(G.1)} & 2 \|\nabla \lKLIEP(\thetab^*)\|_\infty \leq \lambda_\theta,&
\textnormal{(G.2)} & 2 \|\nabla^2 \lKLIEP(\thetab^*) \omegab_k^* - \eb_k\|_\infty \leq \lambda_k,
\end{array}\\
\begin{array}{r c}
\textnormal{(B.1)} & \abr{1 - \frac{\hat Z_y(\thetab^*)}{Z_y(\thetab^*)}} \lesssim \lambda_\theta,\\
\textnormal{(B.2)} & \abr{\frac{1}{n_y} \sum_{j=1}^{n_y} \langle \omegab_k^*, \mub_\psi - \psib(\yb^{(j)}) \rangle r_{\theta^*}(\yb^{(j)})} \lesssim \lambda_k,
\end{array}\\
\begin{array}{r c r c}
\textnormal{(V.1)} & \|\hat\Sbb_\psi - \Sigmab_\psi\|_\infty \lesssim s_{\theta, 0} \lambda_\theta&
\textnormal{(V.2)} & \|\hat\Sbb_{\psi\hat r}(\thetab^*) - \Sigmab_{\psi r}\|_\infty \lesssim s_{\theta, 0} \lambda_\theta,
\end{array}\\
\begin{array}{r c}
\textnormal{(SE)} & \vertiii{\Hb(\thetab^*) - \EE_y\Hb(\thetab^*)}_s \leq \underline\kappa / 128
\end{array}
\end{array}}.
\end{multline*}
Note that (SE) replaces (E.1) and (E.2) in the definition of $\Ecal_\text{one}$.
We shall show
\begin{itemize}[nosep]
\item (G.1) and (SE) imply (E.1), and
\item (G.2) and (SE) in conjunction with (E.1) imply (E.2),
\end{itemize}
so that $\Ecal_\text{one}^\text{L} \subseteq \Ecal_\text{one}$.

Define
\begin{equation*}
\Kcal(S,\beta,\rho) = \cbr{\vb \in \RR^p : \|\vb_{S^c}\|_1 \leq \beta \|\vb_S\|_1 + (1+\beta) \rho,\ \|\vb\| \leq 1}
\end{equation*}
for any $S \subseteq [p]$, $S \neq \varnothing$, $\beta \geq 0$, $\rho \geq 0$.
We shall use this with
\begin{eqnarray*}
S_\theta = \cbr{k' : |\theta^*_{k'}| > \lambda_\theta},&
s_\theta = |S_\theta|,&
\rho_\theta = \|\thetab^*_{S_\theta^c}\|_1
\end{eqnarray*}
and
\begin{eqnarray*}
S_k = \cbr{k' : |\omega^*_{k,k'}| > \lambda_k},&
s_k = |S_k|,&
\rho_k = \|\omegab_{k,S_k^c}^*\|_1.
\end{eqnarray*}

By the first part of \Cref{lem:RSC}, (B.1) and (SE) imply
\begin{equation*}
\vb^\top \nabla^2 \lKLIEP(\thetab^*) \vb \geq c_1\underline{\kappa} \|\vb\|^2 - c_2 \rho_\theta^2 / s_\theta
\quad\text{for all}\quad \vb \in \Kcal(S_\theta,3,\rho_\theta).
\end{equation*}
Combining this with (G.1), \Cref{lem:consistency:1} gives us
\begin{equation} \label{eq:Step1:l1rate}
\|\check\thetab - \thetab^*\|_1
\lesssim s_{\theta, 0} \lambda_\theta
\asymp s_{\theta, 0} \rbr{\frac{\log p}{n}}^{1/2},
\end{equation}
where we have used the condition on $\lambda_\theta$ \eqref{cond:main:lasso:lambda}.
Under the conditions of the corollary, the second part of \Cref{lem:RSC} imply
\begin{equation*}
\vb^\top \nabla^2 \lKLIEP(\thetab^*) \vb \geq c_3\underline{\kappa} \|\vb\|^2
\quad\text{for all}\quad \vb \in \Kcal(S_k,6,0).
\end{equation*}
Combining this with (G.2), \Cref{lem:consistency:2} gives us
\begin{align*}
\|\check\omegab_k - \omegab_k^*\|_1
&\lesssim \|\check\thetab - \thetab^*\|_1^2 s_{k, q_k} \lambda_k^{-1-q_k}
+ s_{k, q_k}^2 \lambda_k^{1-2q_k}
+ s_{k, q_k} \lambda_k^{1-q_k}\\
&\lesssim s_{\theta, 0}^2 \lambda_\theta^{2} s_{k, q_k} \lambda_k^{-1-q_k}
+ s_{k, q_k}^2 \lambda_k^{1-2q_k}
+ s_{k, q_k} \lambda_k^{1-q_k}\\
&\lesssim
\begin{multlined}[t]
s_{\theta, 0}^2 s_{k, q_k}^{1-\tfrac{1+q_k}{2-q_k}} \rbr{\frac{\log p}{n}}^{(1-q_k)/2}\\
+ s_{k, q_k}^{2+\tfrac{1-2q_k}{2-q_k}} \rbr{\frac{\log p}{n}}^{(1-2q_k)/2}
+ s_{k, q_k}^{1+\tfrac{1-q_k}{2-q_k}} \rbr{\frac{\log p}{n}}^{(1-q_k)/2}
\end{multlined}\\
&\lesssim s_{k, q_k}^{2+\tfrac{1-2q_k}{2-q_k}} \rbr{\frac{\log p}{n}}^{(1-2q_k)/2}.
\end{align*}
where we have used the assumptions \eqref{cond:main:lasso:sparsity} and \eqref{cond:main:lasso:lambda}, as well as \eqref{eq:Step1:l1rate}.
Thus,
\begin{equation}
\Delta_2 \lesssim
s_{\theta, 0} s_{k, q_k}^{2+\tfrac{1-2q_k}{2-q_k}} \rbr{\frac{\log p}{n}}^{1-q_k} \sqrt{n}.
\label{eq:main:lasso:Delta2:case2}
\end{equation}
The terms corresponding to $\Delta_1$ and $\Delta_3$ are of smaller order, so we ignore them.

Next, we bound $\PP({\Ecal_\text{one}^\text{L}}^c)$.
Let
\begin{align*}
\Ecal_1 &= \cbr{2 \|\nabla \lKLIEP(\thetab^*)\|_\infty \leq \lambda_\theta},\\
\Ecal_2 &= \cbr{2 \|\nabla^2 \lKLIEP(\thetab^*) \omegab_k^* - \eb_k\|_\infty \leq \lambda_k},\\
\Ecal_3 &= \cbr{\abr{1 - \frac{\hat Z_y(\thetab^*)}{Z_y(\thetab^*)}} \lesssim \lambda_\theta},\\
\Ecal_4 &= \cbr{\abr{\frac{1}{n_y} \sum_{j=1}^{n_y} \langle \omegab_k^*, \mub_\psi - \psib(\yb^{(j)}) \rangle r_{\theta^*}(\yb^{(j)})} \lesssim \lambda_k},\\
\Ecal_5 &= \cbr{\|\hat\Sbb_\psi - \Sigmab_\psi\|_\infty \lesssim s_{\theta, 0} \lambda_\theta},\\
\Ecal_6 &= \cbr{\|\hat\Sbb_{\psi\hat r}(\thetab^*) - \Sigmab_{\psi r}\|_\infty \lesssim s_{\theta, 0} \lambda_\theta},\\
\Ecal_7 &= \cbr{\vertiii{\Hb(\thetab^*) - \EE_y\Hb(\thetab^*)}_2 \leq \underline\kappa / 128}.
\end{align*}
Clearly,
\begin{equation*}
\PP({\Ecal_\text{one}^\text{L}}^c) \leq \sum_{\ell=1}^7 \PP(\Ecal_\ell^c).
\end{equation*}
Under the conditions of the corollary, \Cref{lem:grad:1} and \Cref{lem:grad:2} indicate that
\begin{gather*}
\PP(\Ecal_1^c)
= \PP\left\{ 2\|\nabla \lKLIEP(\thetab^*)\|_\infty > \lambda_\theta \right\}
\leq c_4\exp\rbr{-c_4' \log p},\\
\PP(\Ecal_2^c)
= \PP\left\{ 2\|\hat\Hb(\thetab^*) \omegab_k^* - \eb_k\|_\infty > \lambda_k \right\}
\leq c_5\exp\rbr{-c_5' \log p}.
\end{gather*}
\Cref{lem:Hoeffding_r} says
\begin{equation*}
\PP(\Ecal_3^c)
= \PP\cbr{\abr{\frac{\hat Z_y(\thetab^*)}{Z_y(\thetab^*)} - 1} \gtrsim \lambda_\theta}
\leq c_6\exp\rbr{-c_6' \log p}.
\end{equation*}
Because $\{\langle \omegab_k^*, \mub_\psi-\psib(\yb^{(j)}) \rangle \, r_{\theta^*}(\yb^{(j)})\}_{j=1}^{n_y}$ are bounded mean-zero i.i.d.~random variables, we also have the following Hoeffding bound
\begin{equation*}
\PP(\Ecal_4^c)
= \PP\cbr{\abr{\frac{1}{n_y} \sum_{j=1}^{n_y} \langle \omegab_k^*, \mub_\psi - \psib(\yb^{(j)}) \rangle r_{\theta^*}(\yb^{(j)})} \gtrsim \lambda_k}
\leq c_7\exp\rbr{-c_7' \log p}.
\end{equation*}
\Cref{lem:sample_cov:psi} and \Cref{lem:sample_cov:psi_rhat} indicate that
\begin{gather*}
\PP(\Ecal_5^c)
= \PP\left\{ \|\Sbb_\psi - \Sigmab_\psi\|_\infty \gtrsim s_{\theta, 0} \lambda_\theta \right\}
\leq c_8\exp\rbr{-c_8' \log p},\\
\PP(\Ecal_6^c)
= \PP\left\{ \|\hat\Sbb_{\psi\hat r}(\thetab^*) - \Sigmab_{\psi r}\|_\infty \gtrsim s_{\theta, 0} \lambda_\theta \right\}
\leq c_9\exp\rbr{-c_9' \log p}.
\end{gather*}
Furthermore, \Cref{lem:sparse_norm_H_minus_EH} gives
\begin{equation*}
\PP(\Ecal_7^c) \leq \varepsilon_{\text{RSC},n}.
\end{equation*}
Therefore,
\begin{equation} \label{eq:main:lasso:eps}
\PP({\Ecal_\text{one}^\text{L}}^c) \leq \varepsilon_{\text{RSC},n} + c\exp\rbr{-c' \log p}
\end{equation}
for some constants $c, c' > 0$.

We complete the proof by combining the bound from \eqref{eq:main:lasso:Delta2:case2} and the bound from \eqref{eq:main:lasso:eps} with \eqref{eq:main:final_bd}:
\begin{multline*}
\sup_{t \in \RR} \left| \PP\left\{ \sqrt{n} \, (\hat\theta_k - \theta_k^*) / \hat\sigma_k \leq t \right\} - \Phi(t) \right|\\
\leq
O\rbr{s_{\theta, 0} s_{k, q_k}^{2+\tfrac{1-2q_k}{2-q_k}} \rbr{\frac{\log p}{n}}^{1-q_k} \sqrt{n}}
+ \varepsilon_{\text{RSC},n} + c\exp\rbr{-c' \log p}.
\end{multline*}
\end{proof}

\subsection{\label{supp:l1proofs:cor:bootstrap}Proof of \Cref{cor:bootstrap}}

\begin{theorem}[Re-statement of \Cref{cor:bootstrap}]
Assume \Cref{cond:bdr} with $\ell_1$-norm and \Cref{cond:bddpopeigs}.
Suppose $T_{n_x, n_y} = \max_k \sqrt{n} \, |\hat\theta_k - \theta_k^*|$, where $\hat\thetab$ is the SparKLIE+1 estimator with tuning parameters
\begin{eqnarray*}
\lambda_\theta \asymp \rbr{\frac{\log p}{n}}^{1/2}
&\text{and}&
\lambda_k \asymp \rbr{\frac{s_{k, 0} \log p}{n}}^{1/2}, \quad k = 1, \dots, p.
\end{eqnarray*}
Let $s$ be a sequence of integers satisfying $s \geq s_{\theta, 0}, s_{k, 0}$, $k = 1, \dots, p$.
Let $\varepsilon_{\text{RSC},n}$ be a sequence in $(0,1)$ decreasing to 0.
Then, provided that
\begin{equation*}
n_y \geq C' (\bar\kappa / \underline\kappa^2) M_\psi^2 M_r^2 s \log^2 (s) \log (p \vee n_y) \log (n_y) / \varepsilon_{\text{RSC},n}^2,
\end{equation*}
where $C' > 0$ is the known, absolute constant determined in \Cref{lem:sparse_norm_H_minus_EH}, we have
\begin{equation*}
\sup_{\alpha \in (0,1)} \abr{\PP\cbr{T_{n_x,n_y} \leq \hat c_{T, \alpha}} - (1 - \alpha)}
= O(\delta_n + \varepsilon_{\text{RSC},n} + c\exp\rbr{-c' \log p})
\end{equation*}
with probability at least $1 - \varepsilon_{\text{RSC},n} - c\exp\rbr{-c' \log p} - n^{-1}$.
\end{theorem}

\begin{proof}
For the sake of clarity, we ignore the factors of $\bar\kappa$, $\underline\kappa$, $\eta_{x,n}$, and $\eta_{y,n}$ in calculations.
Detailed bounds are, albeit tedious, not difficult to derive.

As in the proof of \Cref{cor:main}, the key to the proof is in finding an event $\Ecal \subseteq \Ecal_\text{all}$ with $\PP(\Ecal^c) \searrow 0$.
Let $\Hb(\thetab) = (\hat Z_y^2(\thetab) / Z_y^2(\thetab)) \nabla^2 \lKLIEP(\thetab)$.
Consider
\begin{multline*}
\Ecal_\text{all}^\text{L} =\\
\cbr{
\begin{array}{c}
\begin{array}{r c r c}
\textnormal{(G.1)} & 2 \|\nabla \lKLIEP(\thetab^*)\|_\infty \leq \lambda_\theta,&
\textnormal{(G.2)} & 2 \|\nabla^2 \lKLIEP(\thetab^*) \omegab_k^* - \eb_k\|_\infty \leq \lambda_k \ \forall \, k,
\end{array}\\
\begin{array}{r c}
\textnormal{(B.1)} & \abr{1 - \frac{\hat Z_y(\thetab^*)}{Z_y(\thetab^*)}} \lesssim \lambda_\theta,\\
\textnormal{(B.2)} & \abr{\frac{1}{n_y} \sum_{j=1}^{n_y} \langle \omegab_k^*, \mub_\psi - \psib(\yb^{(j)}) \rangle r_{\theta^*}(\yb^{(j)})} \lesssim \lambda_k\ \forall \, k,\\
\textnormal{(SE)} & \vertiii{\Hb(\thetab^*) - \EE_y\Hb(\thetab^*)}_s \leq \underline\kappa / 128
\end{array}
\end{array}}.
\end{multline*}
Following the argument of the proof of \Cref{cor:main}, on $\Ecal_\text{all}^\text{L}$,
\begin{eqnarray*}
\delta_\theta \lesssim \rbr{\frac{s^2 \log p}{n}}^{1/2}
&\text{and}&
\delta_k \lesssim \rbr{\frac{s^5 \log p}{n}}^{1/2}
\quad \forall \, k,
\end{eqnarray*}
and hence,
\begin{eqnarray*}
D_1
\lesssim \frac{s^{7/2} \log p}{\sqrt{n}}
\lesssim \rbr{\frac{B_n^2 \log^4 (pn)}{n}}^{1/6}
&\text{and}&
D_2
\lesssim \frac{s^5 \log p}{n}
\lesssim \rbr{\frac{B_n^2 \log (pn)}{n}}^{1/3}.
\end{eqnarray*}

We finish the proof by finding a bound for $\varepsilon_{\text{all},n}$.
Let
\begin{align*}
\Ecal_1 &= \cbr{2 \|\nabla \lKLIEP(\thetab^*)\|_\infty \leq \lambda_\theta},\\
\Ecal_{2k} &= \cbr{2 \|\nabla^2 \lKLIEP(\thetab^*) \omegab_k^* - \eb_k\|_\infty \leq \lambda_k},\\
\Ecal_3 &= \cbr{\abr{1 - \frac{\hat Z_y(\thetab^*)}{Z_y(\thetab^*)}} \lesssim \lambda_\theta},\\
\Ecal_{4k} &= \cbr{\abr{\frac{1}{n_y} \sum_{j=1}^{n_y} \langle \omegab_k^*, \mub_\psi - \psib(\yb^{(j)}) \rangle r_{\theta^*}(\yb^{(j)})} \lesssim \lambda_k},\\
\Ecal_5 &= \cbr{\vertiii{\Hb(\thetab^*) - \EE_y\Hb(\thetab^*)}_s \leq \underline\kappa / 128},
\end{align*}
so that
\begin{equation*}
\varepsilon_{\text{all},n}
\leq \PP({\Ecal_\text{all}^\text{L}}^c)
\leq \PP(\Ecal_1^c) + \sum_{k=1}^p \PP(\Ecal_{2k}^c) + \PP(\Ecal_3^c) + \sum_{k=1}^p \PP(\Ecal_{4k}^c) + \PP(\Ecal_5^c).
\end{equation*}
By a sequence of arguments similar to that in the proof of \Cref{cor:main},
\begin{equation*}
\varepsilon_{\text{all},n} \leq \varepsilon_{\text{RSC},n} + c\exp\rbr{-c' \log p}.
\end{equation*}
\end{proof}

\subsection{\label{supp:consistency:l1}Consistency of $\ell_1$-penalized estimators}

In the following,
\begin{equation*}
\Kcal(S,\beta,\rho) = \{ \vb \in \RR^p : \|\vb_{S^c}\|_1 \leq \beta \|\vb_S\|_1 + (1+\beta) \rho,\ \|\vb\| \leq 1\},
\end{equation*}
where $S \subseteq [p]$ is nonempty, $\beta \geq 0$, and $\rho \geq 0$.

\begin{lemma} \label[lemma]{lem:consistency:1}
Consider the optimization problem \eqref{eq:KLIEP+:1} using $\ell_1$-penalty and a regularization parameter $\lambda_\theta$ satisfying
\begin{equation*}
\lambda_\theta \geq 2\|\nabla \lKLIEP(\thetab^*)\|_\infty.
\end{equation*}
Suppose, in addition, it holds that
\begin{equation*}
\vb^\top \nabla^2 \lKLIEP(\thetab^*) \vb \geq c \underline\kappa \|\vb\|_2^2 - c' \frac{\rho_\theta^2}{s_{\theta, 0}} \quad\text{for}\quad \vb \in \Kcal(S_\theta,3,\rho_\theta),
\end{equation*}
for some $c, c' > 0$, where
\begin{eqnarray*}
S_\theta = \cbr{k' : |\theta^*_{k'}| > \lambda_\theta}, &
s_\theta = |S_\theta|, &
\rho_\theta = \|\thetab^*_{S_\theta^c}\|_1.
\end{eqnarray*}
Then any solution $\check\thetab$ satisfies
\begin{equation*}
\|\check\thetab - \thetab^*\|_1 \lesssim
(1+\underline\kappa^{-1}) \|\thetab^*\|_{q_\theta} \lambda_\theta^{1-q_\theta}.
\end{equation*}
\end{lemma}

\begin{proof}
By a direct application of Theorem 1 of \cite{negahban2010unified}.,
\begin{equation} \label{eq:consistency1}
\|\check\thetab - \thetab^*\|_2^2
\leq \frac{9 s_\theta \lambda_\theta^2}{c^2 \underline\kappa^2} + \frac{4 \lambda_\theta \rho_\theta}{c \underline\kappa} + \frac{2 c' \lambda_\theta \rho_\theta^2}{c \underline\kappa s_\theta}.
\end{equation}
By \eqref{eq:bd_on_strong_signals} and \eqref{eq:bd_on_weak_signals},
\begin{equation*}
s_\theta \leq \|\thetab^*\|_{q_\theta} \lambda_\theta^{-q_\theta}
\quad\text{and}\quad
\rho_\theta \leq \|\thetab^*\|_{q_\theta} \lambda_\theta^{1-q_\theta},
\end{equation*}
so that
\begin{multline*}
\|\check\thetab - \thetab^*\|_2^2
\leq \frac{9 \|\thetab^*\|_{q_\theta} \lambda_\theta^{2-q_\theta}}{c^2 \underline\kappa^2}
+ \frac{4 \|\thetab^*\|_{q_\theta} \lambda_\theta^{2-q_\theta}}{c \underline\kappa}
+ \frac{2 c' \|\thetab^*\|_{q_\theta}^2 \lambda_\theta^{3-2q_\theta}}{c \underline\kappa s_\theta}\\
= \underline\kappa^{-2} \|\thetab^*\|_{q_\theta} \lambda_\theta^{2-q_\theta} \rbr{\frac{9}{c^2} + \frac{4}{c} \underline\kappa +\frac{2 c'}{c} \underline\kappa \|\thetab^*\|_{q_\theta} \lambda^{1-q_\theta}}
\leq K_1 \underline\kappa^{-2} \|\thetab^*\|_{q_\theta} \lambda_\theta^{2-q_\theta}
\end{multline*}
for an appropriate choice of $K_1 > 0$.
Therefore,
\begin{equation} \label{event:consistency1}
\|\check\thetab - \thetab^*\|_1
\leq 4 \sqrt{s_\theta} \|\check\thetab - \thetab^*\| + 4 \rho_\theta
\leq K_2 \underline\kappa^{-1} \|\thetab^*\|_{q_\theta} \lambda_\theta^{1-q_\theta} + 4\|\thetab^*\|_{q_\theta} \lambda_\theta^{1-q_\theta}
\leq K_3 (1+\underline\kappa^{-1}) \|\thetab^*\|_{q_\theta} \lambda_\theta^{1-q_\theta}.
\end{equation}
\end{proof}

\begin{lemma} \label[lemma]{lem:consistency:2}
Assume \Cref{cond:bdr}.
Consider the optimization problem \eqref{eq:KLIEP+:2} using $\ell_1$-penalty and a regularization parameter $\lambda_k$ satisfying
\begin{equation*}
\lambda_k \geq 2\|\nabla^2 \lKLIEP(\thetab^*) \omegab_k^* - \eb_k\|_\infty.
\end{equation*}
Suppose, in addition, it holds that
\begin{equation*}
\vb^\top \nabla^2 \lKLIEP(\check\thetab) \vb \geq c \underline\kappa \|\vb\|_2^2 \quad\text{for}\quad \vb \in \Kcal(S_k,6,0),
\end{equation*}
for some $c > 0$, where $S_k = \cbr{k' : |\omega^*_{k'}| > \lambda_k}$.
Then any solution $\check\omegab_k$ satisfies
\begin{equation*}
\|\check\omegab_k - \omegab_k^*\|_1 \lesssim
\underline\kappa^{-2} \|\check\thetab - \thetab^*\|_1^2 s_{k, q_k} \lambda_k^{-1-q_k}
+ s_{k, q_k}^2 \lambda_k^{1-2q_k}
+ \underline\kappa^{-1} s_{k, q_k} \lambda_k^{1-q_k}.
\end{equation*}
\end{lemma}

\begin{proof}
Put $\hat\Hb(\thetab) = \nabla^2 \lKLIEP(\thetab)$.
The objective function is
\begin{equation*}
\frac{1}{2} \omegab^\top \hat\Hb(\check\thetab) \omegab - \omegab^\top \eb_k + \lambda_k \|\omegab\|_1.
\end{equation*}
For $S_k$ in the statement of the theorem, set
\begin{eqnarray*}
s_k = |S_k|
&\text{and}&
\rho_k = \|\omegab^*_{S_k^c}\|_1.
\end{eqnarray*}
Since $\check\omegab_k$ is the solution to \eqref{eq:KLIEP+:2} using $\ell_1$-penalty,
\begin{equation*}
\frac{1}{2} \check\omegab_k^\top \hat\Hb(\check\thetab) \check\omegab_k - \check\omegab_k^\top \eb_k + \lambda_k \|\check\omegab_k\|_1
\leq \frac{1}{2} \omegab_{k,S_k}^{*\top} \hat\Hb(\check\thetab) \omegab_{k,S_k}^* - \omegab_{k,S_k}^{*\top} \eb_k + \lambda_k \|\omegab_{k,S_k}^*\|_1.
\end{equation*}
Setting $\db = \check\omegab_k - \omegab_{k,S_k}^*$, the above can be rearranged as
\begin{multline}
\frac{1}{2} \db^\top \hat \Hb(\check\thetab) \db
\leq
\lambda_k \rbr{\|\omegab_{k,S_k}^*\|_1 - \|\check\omegab_k\|_1}
- \db^\top \{\hat \Hb(\thetab^*) \omegab_k^* - \eb_k\}\\
- \db^\top \{\hat \Hb(\check\thetab) - \hat \Hb(\thetab^*)\} \omegab_{k,S_k}^*
+ \db^\top \hat\Hb(\thetab^*) \omegab_{k,S_k^c}^*.
\label{eq:step2:basic_ineq:1}
\end{multline}
By Cauchy-Schwarz, the condition of the lemma implies
\begin{equation}
|\db^\top \{\hat \Hb(\thetab^*) \omegab_k^* - \eb_k\}|
\leq \|\db\|_1 \|\hat \Hb(\thetab^*) \omegab_k^* - \eb_k\|_\infty
\leq \frac{\lambda_k}{2} \|\db\|_1.
\label{eq:step2:rhs:2}
\end{equation}
\eqref{eq:additive_quadform_bd:Hess_diff} of \Cref{lem:tilting} yields
\begin{equation}
|\db^\top \{\hat \Hb(\check\thetab) - \hat \Hb(\thetab^*)\} \omegab_{k,S_k}^*|
\leq \frac18 \db^\top \hat \Hb(\check\thetab) \db + K_1 \|\check\thetab - \thetab^*\|_1^2 \|\omegab_{k,S_k}^*\|_1^2.
\label{eq:step2:rhs:3}
\end{equation}
\eqref{eq:additive_quadform_bd:Hess} of \Cref{lem:tilting} yields
\begin{equation}
|\db^\top \hat \Hb(\thetab^*) \omegab_{k,S_k^c}^*|
\leq \frac18 \db^\top \hat \Hb(\check\thetab) \db + K_2 \rho_k^2.
\label{eq:step2:rhs:4}
\end{equation}
Combining \eqref{eq:step2:rhs:2} to \eqref{eq:step2:rhs:4} with \eqref{eq:step2:basic_ineq:1}, and noting $\|\omegab_{k,S_k}^*\|_1 - \|\check\omegab_k\|_1 \leq \|\db_{S_k}\|_1 - \|\db_{S_k^c}\|_1$,
\begin{equation}
\frac14 \db^\top \hat \Hb(\thetab^*) \db
+ \frac{\lambda_k}{2} \|\db_{S_k^c}\|_1
\leq
\frac{3\lambda_k}{2} \|\db_{S_k}\|_1
+ K_1 \|\check\thetab - \thetab^*\|_1^2 \|\omegab_{k,S_k}^*\|_1^2 + K_2 \rho_k^2.
\label{eq:step2:basic_ineq:2}
\end{equation}

We consider two cases.
First, suppose that
\begin{equation*}
\frac{3 \lambda_k}{2} \|\db_{S_k}\|_1 \leq K_1 \|\check\thetab - \thetab^*\|_1^2 \|\omegab_{k,S_k}^*\|_1^2 + K_2 \rho_k^2.
\end{equation*}
Then,
\begin{equation*}
\frac{\lambda_k}{2} \|\db_{S_k^c}\|_1
\leq 2\rbr{K_1 \|\check\thetab - \thetab^*\|_1^2 \|\omegab_{k,S_k}^*\|_1^2 + K_2 \rho_k^2}.
\end{equation*}
easily, and hence
\begin{equation}
\|\db\|_1
\leq K_3 \|\check\thetab - \thetab^*\|_1^2 \|\omegab_{k,S_k}^*\|_1^2 \lambda_k^{-1} + K_4 \rho_k^2 \lambda_k^{-1}.
\label{eq:error_bd:2}
\end{equation}
in the this case.

Next, suppose that
\begin{equation*}
\frac{3 \lambda_k}{2} \|\db_{S_k}\|_1 \geq K_1 \|\check\thetab - \thetab^*\|_1^2 \|\omegab_{k,S_k}^*\|_1^2 + K_2 \rho_k^2.
\end{equation*}
Then, \eqref{eq:step2:basic_ineq:2} yields $\db \in \Kcal(S_k, 6, 0)$, and hence
\[
\|\db\|_1 \leq 7 \|\db_{S_k}\|_1 \leq 7 \sqrt{s_k} \|\db\|.
\]
We are able to apply the restricted strong convexity assumption to \eqref{eq:step2:basic_ineq:2}, which yields
\begin{equation}
\|\db\|_1
\leq K_5 \underline\kappa^{-1} s_k \lambda_k.
\label{eq:error_bd:1}
\end{equation}

Finally, combining the two error bounds \eqref{eq:error_bd:1} and \eqref{eq:error_bd:2},
\begin{align*}
\|\check\omegab_k - \omegab_k^*\|_1
&\leq \|\db\|_1 + \rho_k\\
&\leq
K_3 \|\check\thetab - \thetab^*\|_1^2 \|\omegab_{k,S_k}^*\|_1^2 \lambda_k^{-1}
+ K_4 \rho_k^2 \lambda_k^{-1}
+ K_5 \underline\kappa^{-1} s_k \lambda_k + \rho_k.
\end{align*}
By \eqref{eq:bd_on_strong_signals} and \eqref{eq:bd_on_weak_signals},
\begin{eqnarray}
s_k \leq s_{k, q_k} \lambda_k^{-q_k}
& \text{and} &
\rho_k \leq s_{k, q_k} \lambda_k^{1-q_k}.
\end{eqnarray}
Thus,
\begin{equation*}
\|\check\omegab_k - \omegab_k^*\|_1
\leq K_6 \underline\kappa^{-2} \|\check\thetab - \thetab^*\|_1^2 s_{k, q_k} \lambda_k^{-1-q_k}
+ K_7 s_{k, q_k}^2 \lambda_k^{1-2q_k}
+ K_8 \underline\kappa^{-1} s_{k, q_k} \lambda_k^{1-q_k}.
\end{equation*}
\end{proof}

\begin{lemma} \label[lemma]{lem:tilting}
Let $\thetab \in \bar\Bcal_\varrho(\thetab^*)$, $c > 0$.
Under \Cref{cond:bdr},
\begin{equation}
|\db^\top \hat \Hb(\thetab^*) \vb|
\leq \frac 1 {2c} \, \db^\top \hat \Hb(\thetab) \db + c M_\psi^2 M_r^{16} \|\vb\|_1^2
\label{eq:additive_quadform_bd:Hess}
\end{equation}
and
\begin{equation}
|\db^\top \{\hat \Hb(\check\thetab) - \hat \Hb(\thetab^*)\} \vb|
\leq \frac 1 {2c} \, \db^\top \hat \Hb(\thetab) \db + 4c {L_1}^2 M_\psi^2 M_r^{12} \|\check\thetab - \thetab\|_1^2 \|\vb\|_1^2.
\label{eq:additive_quadform_bd:Hess_diff}
\end{equation}
\end{lemma}

\begin{proof}
Because the geometric mean of nonnegative numbers is dominated by the arithmetic mean,
\begin{align*}
|\db^\top \hat \Hb(\thetab^*) \vb|
&\leq
\rbr{\phantom{c^{-2}} \db^\top \hat \Hb(\thetab) \db}^{1/2}
\rbr{\phantom{c^2} \max_{j,j'} \rbr{\frac{\hat r_{\theta^*}(\yb^{(j)}) \hat r_{\theta^*}(\yb^{(j')})}{\hat r_{\theta}(\yb^{(j)}) \hat r_{\theta}(\yb^{(j')})}}^2 \vb^\top \hat \Hb(\thetab) \vb}^{1/2} \nonumber\\
&=
\rbr{c^{-2} \db^\top \hat \Hb(\thetab) \db}^{1/2}
\rbr{c^2 \max_{j,j'} \rbr{\frac{\hat r_{\theta^*}(\yb^{(j)}) \hat r_{\theta^*}(\yb^{(j')})}{\hat r_{\theta}(\yb^{(j)}) \hat r_{\theta}(\yb^{(j')})}}^2 \frac{Z_y^2(\thetab)}{\hat Z_y^2(\thetab)} \, \vb^\top \Hb(\thetab) \vb}^{1/2} \nonumber\\
&\leq \frac 1 {2c} \, \db^\top \hat \Hb(\thetab) \db + \frac c 2 \max_{j,j'} \rbr{\frac{\hat r_{\theta^*}(\yb^{(j)}) \hat r_{\theta^*}(\yb^{(j')})}{r_{\theta}(\yb^{(j)}) r_{\theta}(\yb^{(j')})}}^2 \frac{\hat Z_y^2(\thetab)}{Z_y^2(\thetab)} \|\Hb(\thetab)\|_\infty \|\vb\|_1^2
\end{align*}
and
\begin{multline*}
|\db^\top \{\hat \Hb(\check\thetab) - \hat \Hb(\thetab^*)\} \vb|\\
\begin{aligned}[t]
&\leq
\rbr{\db^\top \hat \Hb(\thetab) \db}^{1/2}
\rbr{\max_{j,j'} \rbr{\frac{\hat r_{\check\theta}(\yb^{(j)}) \hat r_{\check\theta}(\yb^{(j')}) - \hat r_{\theta^*}(\yb^{(j)}) \hat r_{\theta^*}(\yb^{(j')})}{\hat r_{\theta}(\yb^{(j)}) \hat r_{\theta}(\yb^{(j')})}}^2
\frac{Z_y^2(\thetab)}{\hat Z_y^2(\thetab)} \, \vb^\top \Hb(\thetab) \vb}^{1/2}\\
&\leq \frac 1 {2c} \, \db^\top \hat \Hb(\check\thetab) \db + \frac c 2 \max_{j,j'} \rbr{\frac{\hat r_{\check\theta}(\yb^{(j)}) \hat r_{\check\theta}(\yb^{(j')}) - \hat r_{\theta^*}(\yb^{(j)}) \hat r_{\theta^*}(\yb^{(j')})}{r_{\theta}(\yb^{(j)}) r_{\theta}(\yb^{(j')})}}^2 \frac{\hat Z_y^2(\thetab)}{Z_y^2(\thetab)} \|\Hb(\thetab)\|_\infty \|\vb\|_1^2.
\end{aligned}
\end{multline*}
Under \Cref{cond:bdr}, $\|\Hb(\thetab)\|_\infty \leq 2 M_\psi^2 M_r^2$ for all $\thetab \in \bar\Bcal_\varrho(\thetab^*)$.
Furthermore,
\begin{equation*}
M_r^{-6} \leq \frac{\hat r_{\theta^*}(\yb^{(j)}) \hat r_{\theta^*}(\yb^{(j')})}{r_{\theta}(\yb^{(j)}) r_{\theta}(\yb^{(j')})} \leq M_r^6,
\end{equation*}
and
\begin{multline*}
\abr{\frac{\hat r_{\check\theta}(\yb^{(j)}) \hat r_{\check\theta}(\yb^{(j')}) - \hat r_{\theta^*}(\yb^{(j)}) \hat r_{\theta^*}(\yb^{(j')})}{r_{\theta}(\yb^{(j)}) r_{\theta}(\yb^{(j')})}}\\
\begin{aligned}[t]
&= \frac{\hat r_{\check\theta}(\yb^{(j)}) \Big| \hat r_{\check\theta}(\yb^{(j')}) - \hat r_{\theta^*}(\yb^{(j')}) \Big| + \Big| \hat r_{\check\theta}(\yb^{(j)}) - \hat r_{\theta^*}(\yb^{(j)}) \Big| \hat r_{\theta^*}(\yb^{(j')})}{r_{\theta}(\yb^{(j)}) r_{\theta}(\yb^{(j')})}\\
&\leq 2 L_1 M_r^4 \|\check\thetab - \thetab\|_1.
\end{aligned}
\end{multline*}
The inequalities follow.
\end{proof}

\section{Model assumptions}

In this section, we go over some of the implications of the assumptions in \Cref{sec:conditions}.
Appendix~\ref{supp:pfs_for_bdr} discusses the properties of the bounded density ratio model of \Cref{cond:bdr}.
In Appendix~\ref{supp:bddpopeigs}, we derive bounds on the $\ell_2$- and $\ell_1$-norms of $\omegab_k^* = \Sigmab_\psi^{-1} \eb_k$, as well as lower- and upper-bounds on the variance of the linearization $\sigma_{n,k}^2$, as direct consequences of \Cref{cond:bddpopeigs}.
In Appendix~\ref{supp:rowsparsity} we characterize the sparsity of the rows of $\Sigmab_\psi^{-1}$.

\subsection{\label{supp:pfs_for_bdr}Properties of the bounded density ratio model}

\begin{proposition}[Re-statement of \Cref{prop:bss}]
Condition 1 is satisfied if and only if $\|\psib(\xb)\|_* \leq M_\psi$ a.s.~for some $M_\psi < \infty$.
\end{proposition}

\begin{proof}
We shall first treat the case $\thetab^* = \zero$, and then show how the general case follows from the special one.
Assume $\|\psib(\xb)\|_* \leq M_\psi$ for some $M_\psi < \infty$.
For each $\xb$, by the definition of the dual norm,
\begin{equation*}
|\langle \psib(\xb), \thetab \rangle|
= |\langle \psib(\xb), \thetab/\|\thetab\| \rangle| \|\thetab\|
\leq \|\psib(\xb)\|_* \|\thetab\|
\leq \varrho M_\psi.
\end{equation*}
It is easy to see that for each $\thetab \in \bar\Bcal_\varrho(\thetab^*)$,
\begin{eqnarray*}
e^{-\varrho M_\psi} \leq e^{\langle \psib(\xb), \thetab \rangle} \leq e^{\varrho M_\psi}
&\text{and}&
e^{-\varrho M_\psi} \leq Z_y(\thetab) \leq e^{\varrho M_\psi},
\end{eqnarray*}
and hence,
\begin{equation*}
e^{-2\varrho M_\psi} \leq r_\theta(\xb) \leq e^{2\varrho M_\psi}.
\end{equation*}
In particular, one may choose $M_r = M_r(\varrho) = e^{2\varrho M_\psi}$.

This proves one direction of the claim.
For the other direction, first note that \Cref{cond:bdr} implies
\begin{equation*}
\langle \psib(\xb), \thetab \rangle \leq \log M_r(\varrho) + \log Z_y(\thetab) \quad\text{for all}\quad \thetab \in \bar\Bcal_\varrho(\thetab^*).
\end{equation*}
For each $\xb$, $\varrho \|\psib(\xb)\|_* = \langle \psib(\xb),\thetab_x\rangle$ for some $\thetab_x \in \bar\Bcal_\varrho(\thetab^*)$ by compactness, so
\begin{equation*}
\|\psib(\xb)\|_* \leq \varrho^{-1} \big( \log M_r(\varrho) + \log Z_y(\thetab_x) \big).
\end{equation*}
Using compactness again,
\begin{equation*}
\|\psib(\xb)\|_* \leq \varrho^{-1} \left( \log M_r(\varrho) + \max_{\|\thetab\| \leq \varrho} \log Z_y(\thetab) \right),
\end{equation*}
and the bound is finite by assumption.
Now, the right-hand side is a function of $\varrho$ only, whereas the left-hand side is independent of $\varrho$.
Thus,
\begin{equation*}
\|\psib(\xb)\|_* \leq \inf_{\varrho > 0} \varrho^{-1} \left( \log M_r(\varrho) + \max_{\|\thetab\| \leq \varrho} \log Z_y(\thetab) \right).\end{equation*}

This completes the proof for the case $\thetab^* = \zero$.
For general $\thetab^*$,
\begin{equation*}
|\langle \psib(\xb),\thetab \rangle|
\leq |\langle \psib(\xb), \thetab-\thetab^* \rangle| + |\langle \psib(\xb), \thetab^* \rangle|
\leq \|\psib\|_* (\varrho + \|\thetab^*\|),
\end{equation*}
and
\begin{equation*}
\langle \psib(\xb), \thetab-\thetab^* \rangle \leq \log \big( M_r^2 Z_y(\thetab) / Z_y(\thetab^*) \big),
\end{equation*}
and the proof goes through as before.
\end{proof}

Under the bounded density ratio model, $\hat Z_y(\thetab)$, $\hat r_\theta(\yb)$, and $\mub(\thetab)$ are all locally Lipschitz continuous in $\thetab$.

\begin{lemma} \label[lemma]{lem:Lips}
There exist $L_0, L_1, L_2 > 0$ such that for all $\thetab \in \bar\Bcal_\varrho(\thetab^*)$,
\begin{align}
|\hat Z_y(\thetab) - \hat Z_y(\thetab^*)| &\leq L_0\|\thetab - \thetab^*\|, \\
\label{eq:Lips_rhat} |\hat r_{\theta}(\yb) - \hat r_{\theta^*}(\yb)| &\leq L_1 \|\thetab - \thetab^*\|, \\
\label{eq:Lips_muhat} \left\| \hat\mub(\thetab) - \hat\mub(\thetab^*) \right\|_* &\leq L_2 \|\thetab - \thetab^*\|.
\end{align}
\end{lemma}

\begin{proof}
$\hat Z_y(\thetab)$, $\hat r_\theta(\yb)$, and $\hat\mub(\thetab)$ are all differentiable functions of $\thetab$, and hence the mean value theorem and the boundedness assumption can be used to derive the required bounds.
\end{proof}

It is not difficult to imagine that under the bounded density ratio model, all the relevant sample quantities concentrate sufficiently fast.
The following lemma proves this intuition.
It is always true that for any $\thetab$,
\begin{equation} \label{eq:Zhat_over_Z}
\frac{r_\theta(\yb)}{\hat r_\theta(\yb)}
= \frac{\hat Z_y(\thetab)}{Z_y(\thetab)}
= \frac{1}{n_y} \sum_{j=1}^{n_y} \frac{\exp\left( \thetab^\top \psib(\yb^{(j)}) \right)}{Z_y(\thetab)}
= \frac{1}{n_y} \sum_{j=1}^{n_y} r_\theta(\yb^{(j)}),
\end{equation}
and
\begin{equation} \label{eq:expectation_of_r}
\EE_y [r_\theta(\yb)]
= \int r_\theta(\yb) f_y(\yb) \, d\yb
= \int f(\yb; \thetab + \gammab_y) \, d\yb
= 1.
\end{equation}
If, in addition, $r_\theta(\yb)$ is bounded, then \eqref{eq:Zhat_over_Z} and \eqref{eq:expectation_of_r} can be used to derive the following results.

\begin{lemma} \label[lemma]{lem:Hoeffding_r}
Suppose $\thetab \in \bar\Bcal_\varrho(\thetab^*)$.
For any $t > 0$,
\begin{align*}
\PP\left\{ \frac{\hat Z_y(\thetab)}{Z_y(\thetab)} - 1 > \phantom{-}t \right\}
&\leq \exp\left( -\frac{2 t^2 n_y}{(M_r-M_r^{-1})^2} \right)
\intertext{and}
\PP\left\{ \frac{\hat Z_y(\thetab)}{Z_y(\thetab)} - 1 < -t \right\}
&\leq \exp\left( -\frac{2 t^2 n_y}{(M_r-M_r^{-1})^2} \right).
\end{align*}
\end{lemma}

\begin{proof}
Apply Hoeffding's inequality to the random variable $r_\theta(\yb) \in [M_r^{-1},M_r]$, $\EE_y[r_\theta(\yb)] = 1$.
\end{proof}

Having highlighted a few of the features of the
bounded density ratio model, we proceed to explain why
\eqref{eq:KLIEP+:1} or \eqref{eq:KLIEP+:2} are expected to yield consistent estimators of $\thetab^*$ or $\omegab_k^*$ under \Cref{cond:bdr}.

The optimization problem described by \eqref{eq:KLIEP+:1} or \eqref{eq:KLIEP+:2} has a convex objective with $\ell_1$-penalty.
It is well-understood that given a regularization level $\lambda > 0$, a minimizer of the corresponding regularized objective is consistent for the population optimum,
provided that the gradient at the population optimum is bounded by $\lambda / 2$ in $\ell_\infty$-norm
(the dual norm of the $\ell_1$-norm), and the Hessian behaves like a
positive definite matrix when restricted to the right set. The
boundedness of the density ratio and sufficient statistics help guarantee both.

The gradient of $\lKLIEP$ at $\thetab^*$ is
\begin{equation} \label{eq:gradient}
\nabla\lKLIEP(\thetab^*)
= -\frac{1}{n_x} \sum_{i=1}^{n_x} \psib(\xb^{(i)}) + \frac{1}{n_y} \sum_{j=1}^{n_y} \psib(\yb^{(j)}) \hat r_{\theta^*}(\yb^{(j)}).
\end{equation}
Since $\mub_\psi = \EE_x[ \psib(\xb) ] = \EE_y[ \psib(\yb) r_{\theta^*}(\yb) ]$, $\hat r_{\theta^*}(\yb) = (Z_y(\thetab^*) / \hat Z_y(\thetab^*)) r_{\theta^*}(\yb)$, and $\hat Z_y(\thetab^*) / Z_y(\thetab^*) \xrightarrow{P} 1$, each average in the gradient is a consistent estimator of $\mub_\psi$, so that the gradient as a whole is converging to a zero vector.
Because both $\psib(\xb^{(i)})$'s and $\psib(\yb^{(j)}) \hat r_{\theta^*}(\yb^{(j)})$'s are bounded, a Hoeffding-type bound can be used to control the gradient.

The gradient of the quadratic part of \eqref{eq:KLIEP+:2}, as well as the curvature of both \eqref{eq:KLIEP+:1} and \eqref{eq:KLIEP+:2}, involves the Hessian of $\lKLIEP$:
\begin{equation*}
\nabla^2\lKLIEP(\thetab)
= \frac{1}{n_y^2} \sum_{1 \leq j < j' \leq n_y} \rbr{\psib(\yb^{(j)}) - \psib(\yb^{(j')})} \rbr{\psib(\yb^{(j)}) - \psib(\yb^{(j')})}^\top \hat r_\theta(\yb^{(j)}) \hat r_\theta(\yb^{(j')}).
\end{equation*}
Note that the above only uses the samples from $f_y$.
The form of the Hessian makes it clear that if too many of $\hat r_\theta(\yb^{(j)})$'s are small, this results in a loss of curvature.
Moreover, when many $\hat r_\theta(\yb^{(j)})$'s are small, the identity $n_y^{-1} \sum_{j=1}^{n_y} \hat r_\theta(\yb^{(j)}) \equiv 1$ makes it likely that many $\hat r_\theta(\yb^{(j)})$'s are also large to balance the sum.
This is likely to lead to the Hessian becoming ill-conditioned.
As before, the boundedness of the density ratio provides a protection against this kind of degeneracy.

\subsection{\label{supp:bddpopeigs}Consequences of the bounds on the population eigenvalues}

\subsubsection{\label{supp:bds_on_omega_k}Bounds on $\omega_k^*$}

It is an easy consequence of the definitions of $\omegab_k^*$, $\underline{\kappa}$, and $\bar\kappa$ that
\begin{equation} \label{eq:omega_k:l2}
\bar{\kappa}^{-1} \leq \|\omegab_k^*\|_2 \leq \underline{\kappa}^{-1} \quad\text{for all}\quad k = 1, \dots, p.
\end{equation}

Before we turn to bounding the $\ell_1$-norm of $\omegab_k^*$ in terms of its $\ell_{q_k}$-``norm", we look at some useful inequalities related to $\ell_q$-``norms".
Fix $\lambda > 0$, and let $S_\lambda = \{ k : |v_k| > \lambda \}$ and $s_\lambda = |S_\lambda|$.
Then,
\begin{equation*}
\|\vb\|_q \geq \sum_{k \in S_\lambda} |v_k|^q \geq s_\lambda \lambda^q,
\end{equation*}
so that
\begin{equation} \label{eq:bd_on_strong_signals}
s_\lambda \leq \lambda^{-q} \|\vb\|_q.
\end{equation}
Moreover,
\begin{equation} \label{eq:bd_on_weak_signals}
\|\vb_{S_\lambda^c}\|_1
= \sum_{k \notin S_\lambda} |v_k|
= \sum_{k \notin S_\lambda} |v_k|^{1-q} |v_k|^q
\leq \lambda^{1-q} \|\vb\|_q.
\end{equation}
Thus,
\begin{equation} \label{eq:l1bd_initial_form}
\|\vb\|_1
= \|\vb_{S_\lambda}\|_1 + \|\vb_{S_\lambda^c}\|_1
\leq \sqrt{s_\lambda} \|\vb\|_2 + \|\vb_{S_\lambda^c}\|_1
\leq \lambda^{-q/2} \|\vb\|_q^{1/2} \|\vb\|_2 + \lambda^{1-q} \|\vb\|_q.
\end{equation}
To simplify the form of the upper bound, we balance the two terms by seeking $r \in \RR$ such that
\begin{equation*}
\lambda \asymp \|\vb\|_q^r \quad\text{and}\quad \lambda^{-q/2} \|\vb\|_q^{1/2} \asymp \lambda^{1-q} \|\vb\|_q.
\end{equation*}
This is solved by $r = -1/(2-q)$.
Substituting this into \cref{eq:l1bd_initial_form},
\begin{equation} \label{eq:l1bd_final_form}
\|\vb\|_1 \leq (1+\|\vb\|_2) \|\vb\|_q^{1/(2-q)}.
\end{equation}
Applying \eqref{eq:l1bd_final_form} to $\omegab_k^*$,
\begin{equation} \label{eq:omega_k:l1}
\|\omegab_k^*\|_1 \leq (1+\|\omegab_k^*\|_2) s_{k, q_k}^{1/(2-q_k)} \leq (1+\underline{\kappa}^{-1}) s_{k, q_k}^{1/(2-q_k)}
\quad\text{for}\quad k = 1, \dots, p.
\end{equation}

\subsubsection{\label{supp:bd_var}Bounds on $\sigma_k^2$}

Define
\begin{multline*}
\sigma_{n,k}^2
= \Var\left[ \sqrt{n} \, \left\langle \omegab_k^*, \frac{1}{n_x} \sum_{i=1}^{n_x} \psib(\xb^{(i)}) - \frac{1}{n_y} \sum_{j=1}^{n_y} \psib(\yb^{(j)}) r_{\theta^*}(\yb^{(j)}) \right\rangle \right]\\
= \omegab_k^{*\top} \left\{ \eta_{x,n}^{-1} \Sigmab_\psi + \eta_{y,n}^{-1} \Sigmab_{\psi r} \right\} \omegab_k^*,
\end{multline*}
where $\Sigmab_\psi = \Cov_x[\psib(\xb)]$ and $\Sigmab_{\psi r} = \Cov_y[(\psib(\yb) - \mub_\psi) r_{\theta^*}(\yb)]$.
Since $\Sigmab_\psi$ and $\Sigmab_{\psi r}$ are symmetric and positive
definite by \Cref{cond:bddpopeigs},
we have
\begin{equation*}
\lambda_{\max}\left( \eta_{x,n}^{-1} \Sigmab_\psi + \eta_{y,n}^{-1} \Sigmab_{\psi r} \right)
\leq \eta_{x,n}^{-1} \lambda_{\max}(\Sigmab_\psi) + \eta_{y,n}^{-1} \lambda_{\max}(\Sigmab_{\psi r})
\leq \bar{\kappa} / (\eta_{x,n} \eta_{y,n}),
\end{equation*}
and, similarly,
\begin{equation*}
\lambda_{\min}\left( \eta_{x,n}^{-1} \Sigmab_\psi + \eta_{y,n}^{-1} \Sigmab_{\psi r} \right)
\geq \underline{\kappa} / (\eta_{x,n} \eta_{y,n}).
\end{equation*}
Thus,
\begin{equation} \label{eq:bd_on_var}
\frac{\underline{\kappa}}{\bar\kappa^2 \eta_{x,n} \eta_{y,n}}
\leq \frac{\underline{\kappa} \|\omegab_k^*\|_2^2}{\eta_{x,n} \eta_{y,n}}
\leq \sigma_k^2
\leq \frac{\bar{\kappa} \|\omegab_k^*\|_2^2}{\eta_{x,n} \eta_{y,n}}
\leq \frac{\bar{\kappa}}{\underline{\kappa}^2 \eta_{x,n} \eta_{y,n}},
\end{equation}
where the outer-most pair of inequalities use \eqref{eq:omega_k:l2}.

\subsection{\label{supp:rowsparsity}When are the rows of ${\Sigma}_\psi^{-1}$ sparse?}

\begin{figure}[p]
	\centering
	\caption{\label{fig:rowsparsity:ising}\textbf{The row sparsity of $\Sigmab_{\psi}^{-1}$ for some Ising models.} We examine the sparsity patterns of $\Gammab_x$ (first row), $\Sigmab_{\psi}^{-1}$ (second row), and the symmetric difference of supports of $\Sigmab_{\psi, \textnormal{Gaussian}}^{-1}$ and the edge interaction diagonal block of $\Sigmab_\psi^{-1}$ (third row) for some Ising models for $m = 5, 6, \dots, 12$. $\Sigmab_{\psi}^{-1}$ is observed to have sparse rows when the maximum degree of $\Gammab_x$ is small. In addition, the edge interaction diagonal block of $\Sigmab_{\psi}^{-1}$ is observed to be structurally similar to that of $\Sigmab_{\psi, \textnormal{Gaussian}}^{-1}$. The results here suggest that some form of row sparsity assumption on $\Sigmab_{\psi}^{-1}$ is reasonable even for Ising models if the maximum degree is expected to be small.}
	
	\vspace{\baselineskip}
	
	\begin{minipage}[b]{\linewidth}
	\centering
	{\small (a) chains}
	\includegraphics[width=\linewidth]{invhess_ising_chains.png}
	\end{minipage}
	
	\begin{minipage}[b]{\linewidth}
	\centering
	{\small (b) cycles}
	\includegraphics[width=\linewidth]{invhess_ising_cycles.png}
	\end{minipage}
	
	\begin{minipage}[b]{\linewidth}
	\centering
	{\small (c) stars}
	\includegraphics[width=\linewidth]{invhess_ising_stars.png}
	\end{minipage}
\end{figure}

For our method, one sufficient condition for theoretical validity is consistent estimation of both $\thetab^*$ and $\Sigmab_\psi^{-1}$. It is well-understood that when parameters satisfy structural assumptions, they can be estimated consistently even in high-dimensional settings; this is what motivated us to use $\ell_1$-regularized procedures for sparse or approximately sparse $\thetab^*$ and $\Sigmab_\psi^{-1}$. However, we have $\Sigmab_\psi^{-1} = \Cov_{x}[\psib(\xb)]^{-1}$, and hence $\Sigmab_\psi^{-1}$ is determined by $\gammab_x$. Therefore, to see whether it is plausible to assume $\Sigmab_\psi^{-1}$ is a row-sparse matrix, it is helpful to understand how $\Sigmab_\psi^{-1}$ is related to $\gammab_x$.

Recall that $f_x$ is an exponential family. Lemma~\ref{lemma:lw} gives the map $\gammab_x \mapsto \Sigmab_\psi^{-1}(\gammab_x)$ under the condition of regularity and minimality.

\begin{lemma}[Essentially Lemma 1 in \cite{Loh2013Structure}\label{lemma:lw}]
Consider a regular, minimal exponential family
\begin{align*}
	f_x(\xb) &= \exp\rbr{\langle \gammab_x, \psib(\xb) \rangle - A(\gammab_x)}, &
	A(\gammab_x) &= \log \rbr{\int \exp\rbr{\langle \gammab_x, \psib(\xb) \rangle} d\xb}.
\end{align*}
Then,
\[
	\rbr{\Cov_{x}[\psib(\xb)]}^{-1} = \nabla^2 A^* \circ \nabla A(\gammab_x),
\]
where $A^*$ is the convex dual function to $A$
\[
	A^*(\mub) = \sup_{\gammab \in \Omega} \cbr{\langle \mub, \gammab \rangle - A(\gammab)}.
\]
\end{lemma}

\begin{proof}
The proof in \citet{Loh2013Structure} is a direct consequence of combining Proposition B.2 and Theorem 3.4 in \citet{Wainwright2008Graphical}; the former holds for \emph{any} regular, minimal exponential family, and the latter, more generally.
\end{proof}

Lemma~\ref{lemma:lw} can be used to show that in the case of Gaussian graphical models, $\Sigmab_\psi^{-1}$ has sparse rows when the maximum degree of the underlying graph is small.

\begin{example}[Gaussian graphical models]
Suppose $\xb \sim \Ncal(\zero, \Sigmab)$ for some covariance matrix $\Sigmab \in \RR^{m \times m}$. Then, the probability density function is given by $f_x(\xb) = \exp(\tr[\Gammab_x \Psib(\xb)] - A(\Gammab_x))$, where $\Gammab_x = 2^{-1} \Sigmab^{-1}$, $\Psib(\xb) = \xb \xb^\top$, and
\[
	A(\Gammab_x) = \log Z(\Gammab_x) = \frac m2 \log (2 \pi) - \frac 12 \log \det(-2 \Gammab_x).
\]
By direct computation,
\[
	\nabla A(\Gammab_x) = \frac 12 \Gammab_x^{-1} = \Sigmab,
\]
and
\begin{gather*}
	A^*(\Mb) = -\frac m2 \log(2 \pi e) - \frac 12 \log \det (\Mb),\\
	\nabla^2 A^*(\Mb) = \frac 12 \Db_m^\top \rbr{\Mb^{-1} \otimes \Mb^{-1}} \Db_m.
\end{gather*}
where $\Db_m: \RR_{{m+1 \choose 2}} \to \RR^{m^2}$ is the \emph{duplication matrix}, which is defined by the property
\[
	\Db_m \vech(\Mb) = \vect(\Mb).
\]
(Here, $\vech: \mathbb{S}^{m} \to \RR^{{m+1 \choose 2}}$ is the \emph{half-vectorization} map that vectorizes only the lower-triangular part of a matrix.) Thus,
\[
	\Sigmab_\psi^{-1} = \Sigmab_\psi^{-1}(\Gammab) = 2 \Db_m^\top \rbr{\Gammab \otimes \Gammab} \Db_m,
\]
so that $\Sigmab_\psi^{-1}$ is row sparse if $\Sigmab^{-1}$ is row sparse. In particular, the $(ab,cd)$-th component of $\Sigmab_\psi^{-1}$ is nonzero if and only if both $\gamma_{x,ab}$ and $\gamma_{x,cd}$ are nonzero.
\end{example}

For general Markov random fields, the usefulness of Lemma~\ref{lemma:lw} is limited due to intractability of $A$. For the case of \emph{discrete} Markov random fields, \citet{Loh2013Structure} study sufficient conditions under which the inverse of a \emph{submatrix} of $\Sigmab_\psi$ reflects the structure of the underlying graph, but their proof techniques do not apply to the inverse of the full matrix.

Thus, we turn to numerical tools for verifying the plausibility of the row-sparsity assumption in the case of Ising models. For small values of the number of nodes $m = 5, 6, \dots, 12$, we first generate a graph by fixing a topology and drawing weights $\iidsim \Unif([-0.5,-0.2] \cup [0.2,0.5])$. We then explicitly evaluate the population $\Sigmab_\psi^{-1}$ under an Ising model. We looked at three different topologies: a chain, a cycle, or a star.

The graph structures are displayed in the first rows of Figure~\ref{fig:rowsparsity:ising} (a - c). The sparsity patterns of $\Sigmab_\psi^{-1}$'s are in the second rows. Note that here, the sufficient statistics include the node potentials; the edge interaction parameters are associated with the last ${m \choose 2}$ rows of $\Sigmab_\psi^{-1}$. For ease of comparison, in the third rows, we also plot the symmetric differences of the support of $\Sigmab_{\psi, \textnormal{Gaussian}}^{-1}$, which is computed assuming a Gaussian model, and the support of the lower diagonal block of $\Sigmab_{\psi}^{-1}$. (We ignored entries with magnitudes $< 10^{-10}$.) It is clear from the plots in the last rows that the edge interaction diagonal block of $\Sigmab_\psi^{-1}$ has a structure similar to that of $\Sigmab_{\psi, \textnormal{Gaussian}}^{-1}$. $\Sigmab_\psi^{-1}$ is typically denser compared to $\Sigmab_{\psi, \textnormal{Gaussian}}^{-1}$, but some form of row sparsity assumption still appears to be quite reasonable, at least for the examples we have considered.

As a further check, we tracked the evolution of $\max_k \|\omegab_k^*\|_{q_k}$ over the edge interaction rows of $\Sigmab_\psi^{-1}$ as $m$ was increased. (No thresholding was applied.) This resulted in Figure~\ref{fig:lqnorms:ising}. We observe that although $\Sigmab_{\psi}^{-1}$ may violate exact sparsity --- as evidenced by the curve corresponding to $q = 0$ --- many sparse ``norms" remain well-controlled even as $m$ is increased. In fact, for chains and cycles, the plots are flat for $q = 0.5, 0.25, 0.125$.

Finally, following the ideas in \citet{Ma2017Intersubject} and \citet{Yu2019Simultaneous}, we remark that a modified procedure that uses sample splitting can be used to construct provably de-biased and asymptotically Gaussian estimators of the difference in situations when the rows of $\Sigmab_\psi^{-1}$ are only bounded in $\ell_1$-norm (without being sparse or approximately sparse). The modified procedure first splits the $f_y$-sample into two, and then uses only one part to obtain $\check\thetab$, and the other part to obtain $\check\Omegab$.

\begin{figure}
	\centering
	\caption{\label{fig:lqnorms:ising}\textbf{The growth of $\max_k \|\omegab_k^*\|_{q_k}$ for some Ising models.} We plot $\max_k \|\omegab_k^*\|_{q}$ as a function of the number of nodes $m$ for $q = 1, 0.5, 0.25, 0.125, 0$. Except for $q = 0$, most sparse ``norms" remain well-controlled even as $m$ is increased. Thus, the assumption that $\max_k \|\omegab_k^*\|_{q_k}$ is small appears to be reasonable for many Ising models.}
	
	\vspace{\baselineskip}
	
	\begin{minipage}[b]{0.30\linewidth}
	\centering
	{\small (a) chains}
	\includegraphics[width=\linewidth]{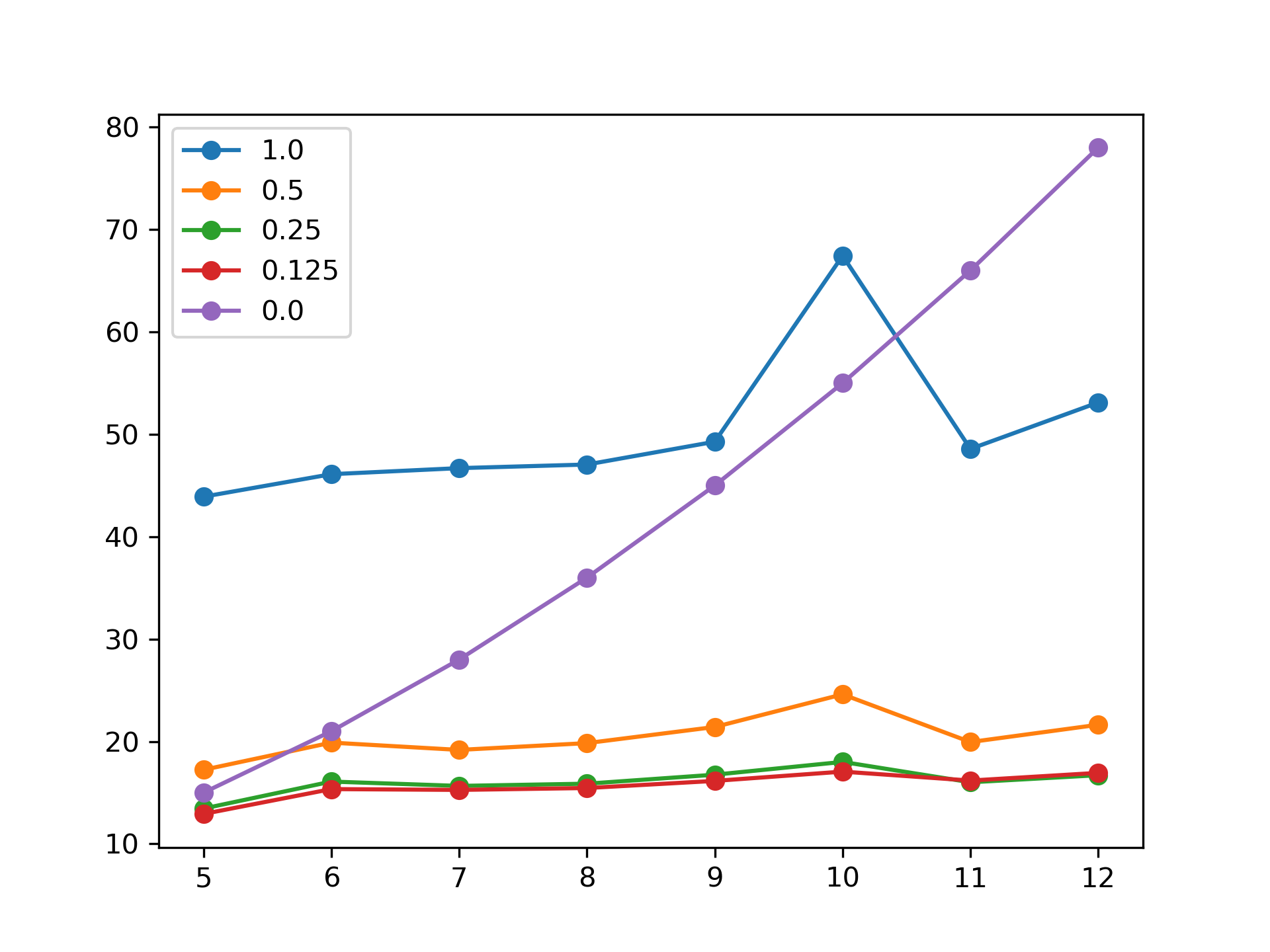}
	\end{minipage}
	\hfill
	\begin{minipage}[b]{0.30\linewidth}
	\centering
	{\small (b) cycles}
	\includegraphics[width=\linewidth]{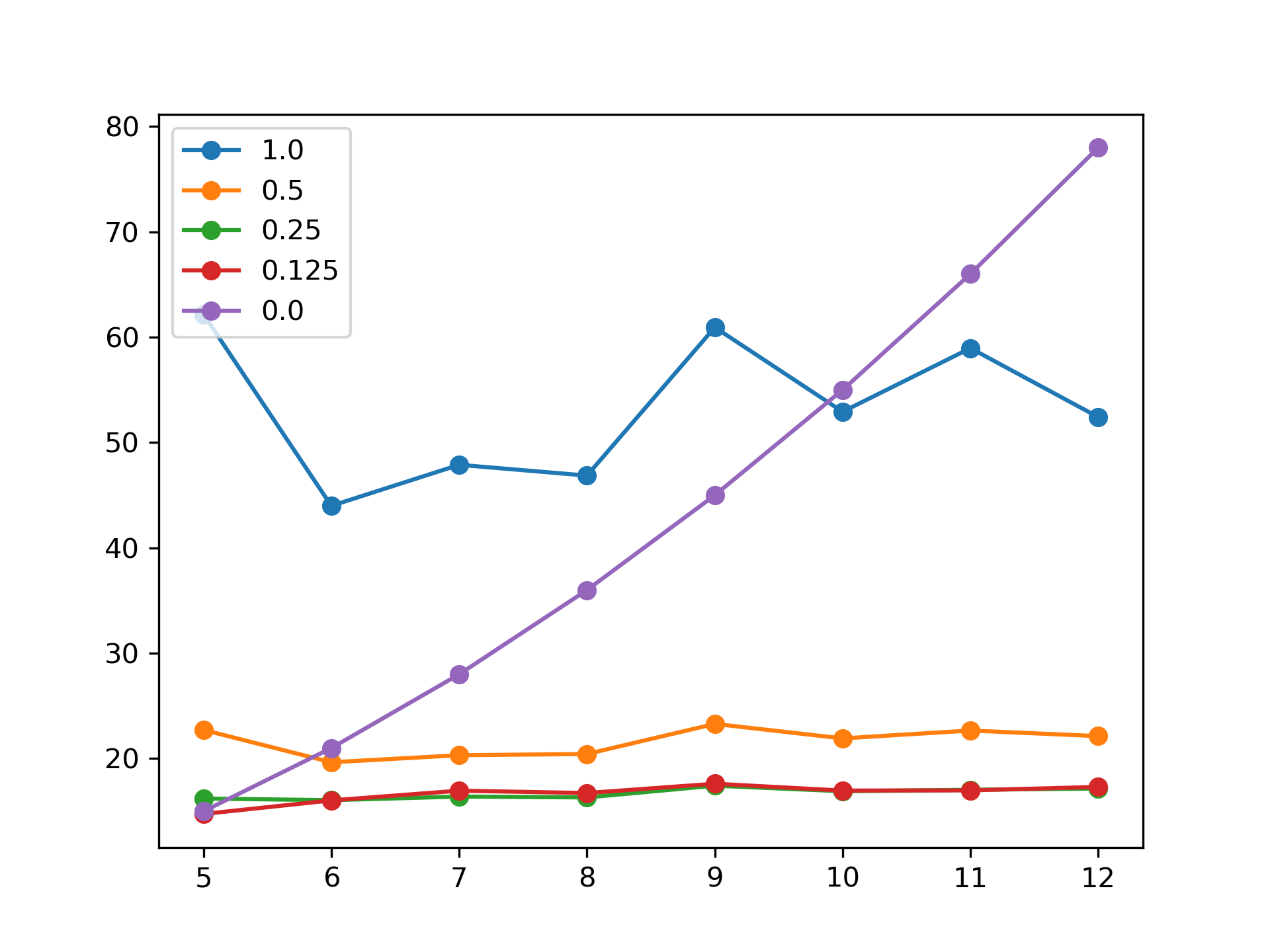}
	\end{minipage}
	\hfill
	\begin{minipage}[b]{0.30\linewidth}
	\centering
	{\small (c) stars}
	\includegraphics[width=\linewidth]{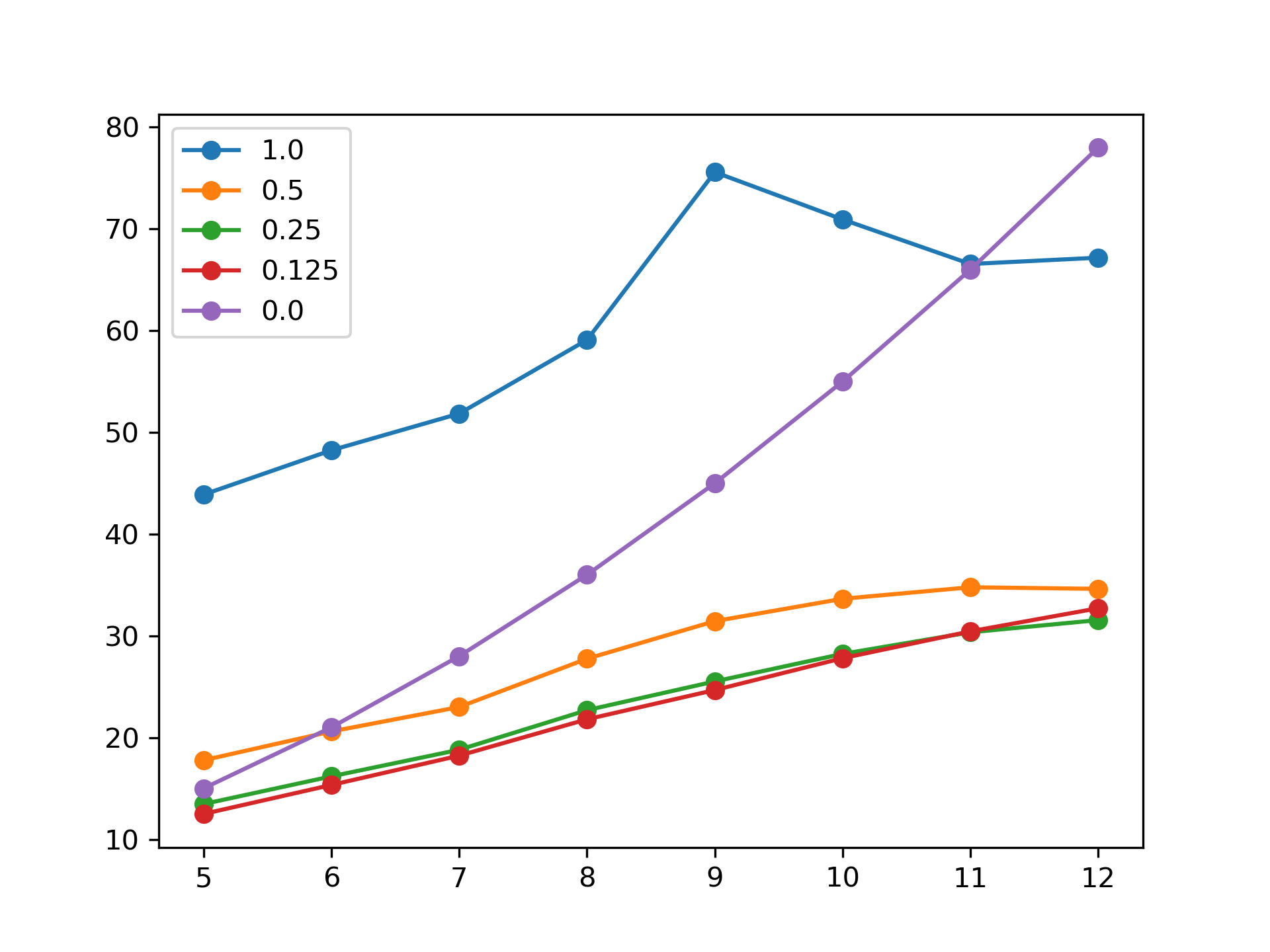}
	\end{minipage}
\end{figure}

\FloatBarrier
\section{\label{supp:auxiliary}Auxiliary results for the $\ell_1$-penalty case}

\subsection{\label{supp:grads}Bounds on the gradients}

The two lemmas in this section bound the gradients of the loss functions in \eqref{eq:KLIEP+:1} and \eqref{eq:KLIEP+:2}.

\begin{lemma} \label[lemma]{lem:grad:1}
Under \Cref{cond:bdr} with $\ell_1$-norm,
\begin{equation*}
\PP\left\{ \|\nabla \lKLIEP(\thetab^*)\|_\infty > t \right\}
\leq 4p \exp(-c t^2 n)
\end{equation*}
for some $c > 0$ depending on $M_r, M_\psi$ only.
In particular, if
\begin{equation*}
\lambda_\theta \geq K \sqrt{\frac{\log p}{n}},
\end{equation*}
for some $K \geq \sqrt{2 / c}$,
then
\begin{equation*}
\PP\left\{ 2\|\nabla \lKLIEP(\thetab^*)\|_\infty > \lambda_\theta \right\} \leq 4\exp(-c' \lambda_\theta^2 n),
\end{equation*}
for some $c' > 0$.
\end{lemma}

\begin{proof}
Let $\mub_\psi = (\mu_{\psi_k})_{k=1}^p = \EE_x[ \psib(\xb) ] = \EE_y[ \psib(\yb) r_{\theta^*}(\yb) ]$.
Using $n_y^{-1} \sum_{j=1}^{n_y} \hat r_\theta(\yb^{(j)}) = 1$,
\begin{equation*}
\begin{aligned}
\nabla \lKLIEP(\thetab^*)
&= -\frac{1}{n_x} \sum_{i=1}^{n_x} \psib(\xb^{(i)}) + \frac{1}{n_y} \sum_{j=1}^{n_y} \psib(\yb^{(j)}) \hat r_{\theta^*}(\yb^{(j)}) \\
&= -\frac{1}{n_x} \sum_{i=1}^{n_x} \psib(\xb^{(i)}) + \mub_\psi + \frac{1}{n_y}\sum_{j=1}^{n_y} \{\psib(\yb^{(j)}) - \mub_\psi\} \, \hat r_{\theta^*}(\yb^{(j)}) \\
&= -\frac{1}{n_x} \sum_{i=1}^{n_x} \psib(\xb^{(i)}) + \mub_\psi + \frac{Z_y(\thetab^*)}{\hat Z_y(\thetab^*)} \frac{1}{n_y}\sum_{j=1}^{n_y} \{\psib(\yb^{(j)}) - \mub_\psi\} \, r_{\theta^*}(\yb^{(j)}).
\end{aligned}
\end{equation*}
\Cref{cond:bdr} implies that
$Z_y(\thetab^*) / \hat Z_y(\thetab^*) \in [M_r^{-1},M_r]$.
For any $t > 0$,
\begin{multline*}
\PP \left\{ \|\nabla \lKLIEP(\thetab^*)\|_\infty > t \right\}\\
\begin{aligned}[t]
&\leq \PP\left\{ \left\| \frac{1}{n_x} \sum_{i=1}^{n_x} \psib(\xb^{(i)}) - \mub_\psi \right\|_{\infty} > \frac{t}{2} \right\} + \PP\left\{ M_r \left\| \frac{1}{n_y}\sum_{j=1}^{n_y} \{\psib(\yb^{(j)})-\mub_\psi\} \, r_{\theta^*}(\yb^{(j)}) \right\|_{\infty} > \frac{t}{2} \right\} \\
&\leq
\begin{multlined}[t]
\sum_{k=1}^p \PP\left\{ \left| \frac{1}{n_x} \sum_{i=1}^{n_x} \psi_k(\xb_k^{(i)}) - \mu_{\psi_k} \right| > \frac{t}{2} \right\}\\
+ \sum_{k=1}^p \PP\left\{ M_r \left| \frac{1}{n_y}\sum_{j=1}^{n_y} \{\psi_k(\yb_k^{(j)})-\mu_{\psi_k}\} \, r_{\theta^*}(\yb^{(j)}) \right| > \frac{t}{2} \right\}.
\end{multlined}
\end{aligned}
\end{multline*}
Since $\{(\psi_k(\xb_k^{(i)})-\mu_{\psi_k})\}_{i=1}^{n_x}$ and $\{(\psi_k(\yb_k^{(j)}) - \mu_{\psi_k}) \,r_{\theta^*}(\yb^{(j)})\}_{j=1}^{n_y}$ are each i.i.d.~bounded and mean zero random variables,
\begin{equation*}
\PP\left\{ \left| \frac{1}{n_x} \sum_{i=1}^{n_x} \psi_k(\xb_k^{(i)}) - \mu_k^* \right| > \frac{t}{2} \right\}
\leq 2\exp(-c_1 t^2 n_x)
\end{equation*}
and
\begin{equation*}
\PP\left\{ M_r \left| \frac{1}{n_y}\sum_{j=1}^{n_y} \{\psi_k(\yb_k^{(j)}) - \mu_{\psi_k}\} \, r_{\theta^*}(\yb^{(j)}) \right| > \frac{t}{2} \right\}
\leq 2\exp(-c_2 t^2 n_y)
\end{equation*}
by Hoeffding's inequality, where $c_1, c_2 > 0$ are constants depending on $M_r, M_\psi$ only.
Thus,
\begin{equation*}
\PP\left\{ \|\nabla \lKLIEP(\thetab^*)\|_\infty > t \right\}
\leq 2p \exp(-c_1 t^2 n_x) + 2p \exp(-c_2 t^2 n_y)
\leq 4p \exp(-c t^2 n)
\end{equation*}
for some $c > 0$.
\end{proof}

\begin{lemma} \label[lemma]{lem:grad:2}
For $t \geq 2 / n_y$,
\begin{multline*}
\PP\left\{ \|\hat\Hb(\thetab^*) \omegab_k^* - \eb_k\|_\infty > t \right\}\\
\leq 2\exp\rbr{-\frac{c t^2 n_y}{(1+\underline\kappa^{-1})^2 s_{k, q_k}^{2/(2-q_k)}}}
+ 2p \exp\rbr{-\frac{c' t^2 n_y}{(1+\underline\kappa^{-1})^2 s_{k, q_k}^{2/(2-q_k)}}}
\end{multline*}
for some $c, c' > 0$ depending on $M_r, M_\psi$ only.
In particular, if
\begin{equation*}
\lambda_k \geq K (1+\underline\kappa^{-1}) s_{k, q_k}^{1/(2-q_k)} \sqrt{\frac{\log p}{n_y}},
\end{equation*}
for some $K \geq \sqrt{2 / (c \wedge c')}$,
then
\begin{equation*}
\PP\left\{ 2\|\hat\Hb(\thetab^*) \omegab_k^* - \eb_k\|_\infty > \lambda_k \right\}
\leq 4\exp\rbr{-\frac{c'' \lambda_k^{*2} n_y}{(1+\underline\kappa^{-1})^2 s_{k, q_k}^{2/(2-q_k)}}}.
\end{equation*}
for some $c'' > 0$.
\end{lemma}

\begin{proof}
Let $\hat\Hb(\thetab) = \nabla^2 \lKLIEP(\thetab)$, and $\Hb(\thetab) = (\hat Z_y^2(\thetab) / Z_y^2(\thetab)) \hat\Hb(\thetab)$.
We have $\Sigmab_\psi \omegab_k^* = \eb_k$ by definition,
and $\EE_y \Hb(\thetab^*) = (1-n_y^{-1}) \Sigmab_\psi$ by \eqref{eq:EyH}.
Therefore,
\begin{equation*}
\hat\Hb(\thetab^*) \omegab_k^* - \eb_k
= \{\hat\Hb(\thetab^*) - \Hb(\thetab^*)\} \omegab_k^* + \{\Hb(\thetab^*) - \EE_y \Hb(\thetab^*)\} \omegab_k^* - n_y^{-1} \eb_k.
\end{equation*}
For $t \geq 2/n_y$,
\begin{multline*}
\PP\left\{ \|\hat\Hb(\thetab^*) \omegab_k^* - \eb_k\|_\infty > t \right\}
\leq \PP\left\{ \|\hat\Hb(\thetab^*) \omegab_k^* - (1-n_y^{-1}) \eb_k\|_\infty > \frac{t}{2} \right\} \\
\leq \PP\left\{ \| \{\hat\Hb(\thetab^*) - \Hb(\thetab^*)\} \omegab_k^* \|_\infty > \frac{t}{4} \right\} + \PP\left\{ \| \{\Hb(\thetab^*) - \EE_y \Hb(\thetab^*)\} \omegab_k^* \|_\infty > \frac{t}{4} \right\}.
\end{multline*}
By \Cref{lem:Hhat_minus_H},
\begin{equation*}
\PP\left\{ \| \{\hat\Hb(\thetab^*) - \Hb(\thetab^*)\} \omegab_k^* \|_\infty > \frac{t}{4} \right\}
\leq 2\exp\rbr{-\frac{c t^2 n_y}{(1+\underline\kappa^{-1})^2 s_{k, q_k}^{2/(2-q_k)}}},
\end{equation*}
where $c  > 0$ is a constant depending only on $M_r, M_\psi$.
By \Cref{lem:H_Chernoff_bd},
\begin{equation*}
\PP\left\{ \| \{\Hb(\thetab^*) - \EE_y \Hb(\thetab^*)\} \omegab_k^* \|_\infty > \frac{t}{4} \right\}
\leq 2p \exp\rbr{-\frac{c' t^2 n_y}{(1+\underline\kappa^{-1})^2 s_{k, q_k}^{2/(2-q_k)}}},
\end{equation*}
where $c' > 0$ is a constant depending only on $M_r, M_\psi$.
Thus,
\begin{multline*}
\PP\left\{ \|\hat\Hb(\thetab^*) \omegab_k^* - \eb_k\|_\infty > t \right\}\\
\leq 2\exp\rbr{-\frac{c t^2 n_y}{(1+\underline\kappa^{-1})^2 s_{k, q_k}^{2/(2-q_k)}}}
+ 2p \exp\rbr{-\frac{c' t^2 n_y}{(1+\underline\kappa^{-1})^2 s_{k, q_k}^{2/(2-q_k)}}}.
\end{multline*}
\end{proof}

\subsection{\label{supp:Hessian}Bounds on the Hessian}

This section contains the technical lemmas that go into bounding the $\ell_1 \to \ell_\infty$ operator norm --- a.k.a.~the maximum magnitude component --- of the Hessian.
The ultimate goal is to control the $\ell_\infty$-norm of the matrix-vector product $\nabla^2 \lKLIEP(\thetab^*) \, \omegab_k^*$.
Since a bound on the $\ell_1$-norm of $\omegab_k^*$ is easily implied by our structural assumptions on $\omegab_k^*$, it is natural to consider the $\ell_1 \to \ell_\infty$ operator norm of the Hessian in bounding the matrix-vector product.

To compute the bound, we first observe that $\nabla^2 \lKLIEP(\thetab^*) \approx \Sigmab_\psi$, and decompose the Hessian into a sum of three terms:
\begin{equation*}
\hat\Hb(\thetab^*)
=
\underbrace{\{\hat\Hb(\thetab^*) - \Hb(\thetab^*)\}}_\text{\Cref{lem:Hhat_minus_H}}
+ \underbrace{\{\Hb(\thetab^*) - \EE_y \Hb(\thetab^*)\}}_\text{\Cref{lem:H_Chernoff_bd}}
+ (1 - n_y^{-1}) \Sigmab_\psi,
\end{equation*}
where $\hat\Hb(\thetab) = \nabla^2 \lKLIEP(\thetab)$, and $\Hb(\thetab) = (\hat Z_y^2(\thetab) / Z_y^2(\thetab)) \hat\Hb(\thetab)$.

\Cref{lem:Hhat_minus_H} reduces the difference $\hat\Hb(\thetab^*) - \Hb(\thetab^*)$ to the deviation of the sample average of the ratios from their expectation.
\Cref{lem:H_Chernoff_bd} is the usual concentration bound for U-statistics applied to our problem.

\begin{lemma} \label[lemma]{lem:Hhat_minus_H}
Suppose \Cref{cond:bdr} holds, and let $\thetab \in \bar\Bcal_\varrho(\thetab^*)$.
For any $\vb \in \RR^p$,
\begin{equation*}
\PP\{ \|\{\hat\Hb(\thetab) - \Hb(\thetab)\} \vb\|_\infty > t \}
\leq 2\exp\left( -\frac{t^2 n_y}{2M_\psi^4 M_r^8 (M_r+1)^2 (M_r-M_r^{-1})^2 \|\vb\|_1^2} \right).
\end{equation*}
In particular,
\begin{equation*}
\PP\{ \|\hat\Hb(\thetab) - \Hb(\thetab)\|_\infty > t \}
\leq 2\exp\left( -\frac{t^2 n_y}{2M_\psi^4 M_r^8 (M_r+1)^2 (M_r-M_r^{-1})^2} \right).
\end{equation*}
\end{lemma}

\begin{proof}
\Cref{cond:bdr} implies that $\hat Z_y(\thetab) / Z_y(\thetab) \in [M_r^{-1},M_r]$, and that $\hat\Hb(\thetab)$ has uniformly bounded components. In particular, on $\bar\Bcal_\varrho(\thetab^*)$, for any $k, \ell \in [p]$,
\begin{equation*}
\begin{aligned}
|\hat H_{k\ell}(\thetab)|
&= \left| \frac{1}{n_y^2} \sum_{1 \leq j < j' \leq n_y} \left( \psi_k(\yb_k^{(j)}) - \psi_k(\yb_k^{(j')}) \right) \left( \psi_\ell(\yb_\ell^{(j)}) - \psi_\ell(\yb_\ell^{(j')}) \right) \hat r_\theta(\yb^{(j)}) \hat r_\theta(\yb^{(j')}) \right| \\
&\leq \frac{1}{n_y^2} \sum_{1 \leq j < j' \leq n_y} \left| \psi_k(\yb_k^{(j)}) - \psi_k(\yb_k^{(j')}) \right| \left| \psi_\ell(\yb_\ell^{(j)}) - \psi_\ell(\yb_\ell^{(j')}) \right| \hat r_\theta(\yb^{(j)}) \hat r_\theta(\yb^{(j')})
\leq 2M_\psi^2 M_r^4.
\end{aligned}
\end{equation*}
Now,
\begin{equation*}
\hat\Hb(\thetab) - \Hb(\thetab)
= \left( 1 - \frac{\hat Z_y^2(\thetab)}{Z_y^2(\thetab)} \right) \hat\Hb(\thetab)
= \left( 1 - \frac{\hat Z_y(\thetab)}{Z_y(\thetab)} \right) \left( 1 + \frac{\hat Z_y(\thetab)}{Z_y(\thetab)} \right) \hat\Hb(\thetab),
\end{equation*}
so that
\begin{align*}
 \PP\{ \|\{\hat\Hb(\thetab) - \Hb(\thetab)\} \vb\|_\infty > t \}
 &
\leq \PP\left\{ \|\hat\Hb(\thetab)\|_\infty \|\vb\|_1 \left| \frac{\hat Z_y(\thetab)}{Z_y(\thetab)} + 1 \right| \left| \frac{\hat Z_y(\thetab)}{Z_y(\thetab)} - 1 \right| > t \right\} \\
& \leq \PP\left\{ 2M_\psi^2 M_r^4 (M_r+1) \|\vb\|_1 \left| \frac{\hat Z_y(\thetab)}{Z_y(\thetab)} - 1 \right| > t \right\}.
\end{align*}
It then follows by \Cref{lem:Hoeffding_r} that
\begin{equation*}
\PP\{ \|\{\hat\Hb(\thetab) - \Hb(\thetab)\} \vb\|_\infty > t \}
\leq 2\exp\left( -\frac{t^2 n_y}{2 M_\psi^4 M_r^8 (M_r+1)^2 (M_r-M_r^{-1})^2 \|\vb\|_1^2} \right).
\end{equation*}
\end{proof}

\begin{lemma} \label[lemma]{lem:H_Chernoff_bd}
 Suppose \Cref{cond:bdr} holds, and let $\thetab \in \bar\Bcal_\varrho(\thetab^*)$.
For any $\vb \in \RR^p$ and any $k \in [p]$,
\begin{equation*}
\PP\left\{ \left| \eb_k^\top \{\Hb(\thetab) - \EE_y\Hb(\thetab)\} \vb \right| > t \right\}
\leq 2\exp\left( -\frac{t^2 n_y}{16 M_\psi^4 M_r^4 \|\vb\|_1^2} \right).
\end{equation*}
In particular,
\begin{equation*}
\PP\left\{ \|\{\Hb(\thetab) - \EE_y\Hb(\thetab)\} \vb\|_\infty > t \right\}
\leq 2\exp\left( -\frac{t^2 n_y}{16 M_\psi^4 M_r^4 \|\vb\|_1^2} + \log p \right).
\end{equation*}
and
\begin{equation*}
\PP\left\{ \|\{\Hb(\thetab) - \EE_y\Hb(\thetab)\}\|_\infty > t \right\}
\leq 2\exp\left( -\frac{t^2 n_y}{16 M_\psi^4 M_r^4} + \log p \right).
\end{equation*}
\end{lemma}

\begin{proof}
For any $k \in [p]$ and for any $a > 0$,
\begin{equation*}
\begin{aligned}
\PP\left\{ \eb_k^\top \{\Hb(\thetab) - \EE_y\Hb(\thetab)\} \vb > t \right\}
&= \PP\left\{ \|\vb\|_1 a \cdot \eb_k^\top \{\Hb(\thetab) - \EE_y\Hb(\thetab)\} (\vb / \|\vb\|_1) > a t \right\} \\
&\leq \PP\left\{ \exp\left( \|\vb\|_1 a \cdot \eb_k^\top \{\Hb(\thetab) - \EE_y\Hb(\thetab)\} (\vb / \|\vb\|_1) \right) > \exp(at) \right\} \\
&\leq \exp(-at) \, \EE_y \left[ \exp\left( \|\vb\|_1 a \cdot \eb_k^\top \{\Hb(\thetab) - \EE_y\Hb(\thetab)\} (\vb / \|\vb\|_1) \right) \right] \\
&\leq \exp\left( -a t + 4 M_\psi^4 M_r^4 \|\vb\|_1^2 a^2 / n_y \right),
\end{aligned}
\end{equation*}
where in the last line, we have used \Cref{lem:subG_quadform}.
Optimizing the bound, we get
\begin{equation*}
\PP\left\{ \eb_k^\top \{\Hb(\thetab) - \EE_y\Hb(\thetab)\} \vb > t \right\}
\leq \exp\left( -\frac{t^2 n_y}{16 M_\psi^4 M_r^4 \|\vb\|_1^2} \right).
\end{equation*}
A similar argument applied to the other side gives us
\begin{equation*}
\PP\left\{ \left| \eb_k^\top \{\Hb(\thetab) - \EE_y\Hb(\thetab)\} \vb \right| > t \right\}
\leq 2\exp\left( -\frac{t^2 n_y}{16 M_\psi^4 M_r^4 \|\vb\|_1^2} \right).
\end{equation*}
Taking the union bound over all $k \in [p]$,
\begin{equation*}
\PP\left\{ \|\{\Hb(\thetab) - \EE_y\Hb(\thetab)\} \vb\|_\infty > t \right\}
\leq 2\exp\left( -\frac{t^2 n_y}{16 M_\psi^4 M_r^4 \|\vb\|_1^2} + \log p \right).
\end{equation*}
\end{proof}

\begin{lemma} \label[lemma]{lem:subG_quadform}
Suppose \Cref{cond:bdr} holds, and let $\thetab \in \bar\Bcal_\varrho(\thetab^*)$.
For any $\ub, \vb \in \RR^p$ with $\|\ub\|_1=\|\vb\|_1=1$ and any $t \in \RR$,
\begin{equation*}
\EE_y \left[ \exp\left( t\cdot \ub^\top \{\Hb(\thetab)-\EE_y\Hb(\thetab)\} \vb \right) \right]
\leq \exp(4 M_\psi^4 M_r^4 t^2 / n_y).
\end{equation*}
\end{lemma}

\begin{proof}
Define
\begin{equation*}
U
:= \frac{2}{1-1/n_y} \ub^\top \Hb(\thetab) \vb
= \frac{2}{n_y(n_y-1)} \sum_{1 \leq j < j' \leq n_y} g(\yb^{(j)},\yb^{(j')}),
\end{equation*}
where
\begin{equation*}
g(\yb,\yb')
= \langle \psib(\yb) - \psib(\yb'), \ub \rangle \, \langle \psib(\yb) - \psib(\yb'), \vb \rangle \, r_\theta(\yb) \, r_\theta(\yb').
\end{equation*}
Let
\begin{equation*}
V(\yb^{(1)},\dots,\yb^{(n_y)})
:= \frac{1}{\lfloor n_y/2\rfloor} \left( g(\yb^{(1)},\yb^{(2)}) + g(\yb^{(3)},\yb^{(4)}) + \cdots + g(\yb^{(2\lfloor n_y/2\rfloor-1)},\yb^{(2\lfloor n_y/2\rfloor)}) \right)
\end{equation*}
and write
\begin{equation*}
U
= \frac{1}{n_y!} \sum_{\sigma \in \mathfrak{S}_{n_y}} V(\yb^{(\sigma(1))},\dots,\yb^{(\sigma(n_y))}),
\end{equation*}
where $\mathfrak{S}_{n_y}$ is the group of permutations on $[n_y]$.
For any $t \in \RR$,
\begin{multline*}
\EE_y \left[ \exp\left( t\cdot \ub^\top \{\Hb(\thetab) - \EE_y\Hb(\thetab)\} \vb \right) \right]\\
\begin{aligned}[t]
&= \EE_y \left[ \exp\left( \frac{1-1/n_y}{2} \, t \cdot \left( U - \EE_y U \right) \right) \right] \\
&=
\begin{multlined}[t]
\EE_y \left[ \exp\left( \frac{1-1/n_y}{2} \, t \right.\right.\\
\left.\left. \times \frac{1}{n_y!} \left( \sum_{\sigma \in \mathfrak{S}_{n_y}} \left( V(\yb^{(\sigma(1))},\dots,\yb^{(\sigma(n_y))}) - \EE_y \left[ V(\yb^{(\sigma(1))},\dots,\yb^{(\sigma(n_y))}) \right] \right) \right) \right) \right]
\end{multlined}\\
&\leq
\begin{multlined}
\frac{1}{n_y!} \sum_{\sigma \in \mathfrak{S}_{n_y}} \EE_y \left[ \exp\left( \frac{1-1/n_y}{2} t \right.\right.\\
\left.\left. \times \left( V(\yb^{(\sigma(1))},\dots,\yb^{(\sigma(n_y))}) - \EE_y \left[ V(\yb^{(\sigma(1))},\dots,\yb^{(\sigma(n_y))}) \right] \right) \right) \right]
\end{multlined}\\
&\leq \exp(4 M_\psi^4 M_r^4 t^2 / n_y),
\end{aligned}
\end{multline*}
where the second-to-last inequality follows from the Jensen's inequality and
the last inequality follows from \Cref{lem:subGV}.
\end{proof}

\begin{lemma} \label[lemma]{lem:subGV}
Let $V(\yb^{(1)},\dots,\yb^{(n_y)})$ be as in the proof of \Cref{lem:subG_quadform}. For any $t \in \RR$,
\begin{equation*}
\EE_y \left[ \exp\left( t\cdot\left(V(\yb^{(1)},\dots,\yb^{(n_y)})-\EE_y \left[ V(\yb^{(1)},\dots,\yb^{(n_y)}) \right] \right) \right) \right]
\leq \exp(16 M_\psi^4 M_r^4 t^2 / n_y).
\end{equation*}
\end{lemma}

\begin{proof}
Consider a random variable $G$ with $|G|\leq D$ and $\EE G=g$.
Using the convexity of the exponential function,
\begin{equation*}
e^{tG} \leq \frac{D-G}{2D} e^{-Dt} + \frac{G+D}{2D} e^{Dt},
\end{equation*}
so that
\begin{equation*}
\begin{aligned}
\EE[ e^{t(G-g)} ]
&\leq e^{-tg} \frac{(D-g) e^{-Dt} + (D+g) e^{Dt}}{2D} \\
&= e^{-tg} \frac{e^{-Dt} (D - g + (D + g) e^{2Dt})}{2D} \\
&=\exp\left( -(D+g) t + \log\left( 1 - \frac{D+g}{2D} + \frac{D+g}{2D} e^{2Dt} \right) \right).
\end{aligned}
\end{equation*}
Put $\tilde t = 2Dt$ and $p = (D+g) / 2D$, and write
\begin{equation*}
h(\tilde t) = -p \tilde t + \log(1 - p + pe^{\tilde t}).
\end{equation*}
Then,
\begin{equation*}
h'(\tilde t) = -p + \frac{p e^{\tilde t}}{1 - p + p e^{\tilde t}}
\end{equation*}
and
\begin{equation*}
h''(\tilde t)
= \frac{(1-p) p e^{\tilde t}}{(1 - p + p e^{\tilde t})^2}
= \left( \frac{p e^{\tilde t}}{1 - p + p e^{\tilde t}} \right) \left( 1 - \frac{p e^{\tilde t}}{1 - p + p e^{\tilde t}} \right)
\leq \frac{1}{4},
\end{equation*}
since $p \exp(\tilde t) / (1 - p + p \exp(\tilde t)) \in (0,1)$.
By Taylor's theorem,
\begin{equation*}
h(\tilde t) \leq h(0) + h'(0) \tilde t + \frac{1}{8} \tilde t^2 = \frac{1}{8} \tilde t^2,
\end{equation*}
so that
\begin{equation} \label{ineq:mgf_bdd_rv}
\EE[e^{t(G-g)}] \leq e^{D^2t^2/2}.
\end{equation}
Now, $g(\yb^{(j)},\yb^{(j')})$'s occurring in $V(\yb^{(1)},\dots\yb^{(n_y)})$ are i.i.d.~with
\begin{multline} \label{ineq:g}
|g(\yb^{(j)},\yb^{(j')})|
= \left| \left\langle \psib(\yb^{(j)}) - \psib(\yb^{(j')}), \ub \right\rangle \left\langle \psib(\yb^{(j)}) - \psib(\yb^{(j')}), \vb \right\rangle r_\theta(\yb^{(j)}) r_\theta(\yb^{(j')}) \right| \\
\leq \left\| \psib(\yb^{(j)}) - \psib(\yb^{(j')}) \right\|_\infty^2 r_\theta(\yb^{(j)}) r_\theta(\yb^{(j')})
\leq 4 M_\psi^2 M_r^2,
\end{multline}
since $\|\ub\|_1=\|\vb\|_1=1$.
Applying \eqref{ineq:mgf_bdd_rv} to the random variable $g(\yb^{(1)},\yb^{(2)})$,
\begin{equation*}
\EE_y \left[ \exp\left( \frac{t}{\lfloor n_y/2 \rfloor}\cdot\left(g(\yb^{(1)},\yb^{(2)}) - \EE_y \left[ g(\yb^{(1)},\yb^{(2)}) \right] \right) \right) \right]
\leq \exp(32 M_\psi^4 M_r^4 t^2 / n_y^2).
\end{equation*}
By independence,
\begin{multline*}
\EE_y \left[ \exp\left( t \cdot \left( V(\yb^{(1)},\dots,\yb^{(n_y)}) - \EE_y \left[ V(\yb^{(1)},\dots,\yb^{(n_y)}) \right] \right) \right) \right] \\
= \EE_y \left[ \exp\left( \frac{t}{\lfloor n_y/2 \rfloor} \cdot \left( g(\yb^{(1)},\yb^{(2)}) - \EE_y \left[ g(\yb^{(1)},\yb^{(2)}) \right] \right) \right) \right]^{\lfloor n_y/2 \rfloor}
\leq \exp(16 M_\psi^4 M_r^4 t^2 / n_y).
\end{multline*}
\end{proof}

\subsection{\label{supp:RSC}Restricted strong convexity}

In the following,
\begin{equation*}
\Kcal(S,\beta,\rho) = \{ \vb \in \RR^p : \|\vb_{S^c}\|_1 \leq \beta \|\vb_S\|_1 + (1+\beta) \rho,\ \|\vb\| \leq 1\},
\end{equation*}
where $S \subseteq [p]$ is nonempty, $\beta \geq 0$, and $\rho \geq 0$.

\begin{lemma} \label[lemma]{lem:RSC}
Suppose $Z_y^2(\thetab^*) / \hat Z_y^2(\thetab^*) \geq c$ for some $c > 0$, and
\begin{equation*}
\vertiii{\Hb(\thetab^*) - \EE_y\Hb(\thetab^*)}_s \leq \underline\kappa / (2 (2+\beta)^2)
\end{equation*}
for some $s \in [p]$ and $\beta \geq 0$.
Then for all nonempty $S \subseteq [p]$ with $|S| \leq s$ and for all $\rho \geq 0$,
\begin{equation*}
\vb^\top \hat \Hb(\thetab^*) \vb
\geq \frac{c \underline\kappa}{2} \rbr{\|\vb\|^2 - \frac{\rho^2}{s}}
\quad\text{for all}\quad \vb \in \Kcal(S,\beta,\rho),
\end{equation*}
as well as
\begin{equation*}
\vb^\top \hat \Hb(\thetab) \vb
\geq \exp\rbr{-2M_\psi (M_r^2+1) \|\thetab - \thetab^*\|_1} \cdot \frac{c \underline\kappa}{2} \rbr{\|\vb\|^2 - \frac{\rho^2}{s}}
\quad\text{for all}\quad \vb \in \Kcal(S,\beta,\rho).
\end{equation*}
\end{lemma}

\begin{proof}
We have
\begin{align*}
\vb^\top \hat \Hb(\thetab^*) \vb
= \frac{Z_y^2(\thetab^*)}{\hat Z_y^2(\thetab^*)} \vb^\top \Hb(\thetab^*) \vb
= \sbr{\rbr{1-\frac{1}{n_y}} \vb^\top \Sigmab_\psi \vb + \vb^\top \{\Hb(\thetab^*)-\EE_y\Hb(\thetab^*)\} \vb}.
\end{align*}
For $n_y$ large enough, under the conditions of the lemma and applying \Cref{lem:sparse_norm},
\begin{align}
\vb^\top \hat \Hb(\thetab^*) \vb
&\geq c \rbr{\underline \kappa \|\vb\|^2 - \frac{\underline \kappa}{2 (2+\beta)^2} \rbr{\|\vb\| + \frac{\|\vb\|_1}{\sqrt{s}}}^2}
\nonumber\\
&\geq c \rbr{\underline \kappa \|\vb\|^2 - \frac{\underline \kappa}{2} \rbr{\|\vb\| + \frac{\rho}{\sqrt{s}}}^2}
\nonumber\\
&\geq \frac{c \underline\kappa}{2} \rbr{\|\vb\|^2 - \frac{\rho^2}{s}}.
\label{eq:RSC_at_theta_star}
\end{align}
For the second part of the statement, first note
\begin{align*}
\vb^\top \hat \Hb(\thetab) \vb
&\geq
\min_{j,j'} \frac{\hat r_{\theta}(\yb^{(j)}) \hat r_{\theta} (\yb^{(j')})}{\hat r_{\theta^*}(\yb^{(j)}) \hat r_{\theta^*} (\yb^{(j')})}
\vb^\top \hat\Hb(\thetab^*) \vb\\
&=
\min_{j,j'} \exp\cbr{\rbr{\psi(\yb^{(j)}) + \psi(\yb^{(j')})}^\top \rbr{\thetab-\thetab^*} - 2 \log \frac{\hat Z_y(\thetab\phantom{^*})}{\hat Z_y(\thetab^*)}}
\vb^\top \hat\Hb(\thetab^*) \vb.
\end{align*}
By convexity of LogSumExp,
\begin{align*}
-\log \hat Z_y(\thetab) + \log \hat Z_y(\thetab^*)
&\geq -\nabla[\log \hat Z_y(\thetab)]^\top \rbr{\thetab-\thetab^*}\\
&= -\frac{1}{n_y} \sum_{j=1}^{n_y} \hat r_\theta(\yb^{(j)}) \psib(\yb^{(j)})^\top \rbr{\thetab-\thetab^*}
\geq -M_\psi M_r^2 \|\thetab - \thetab^*\|_1,
\end{align*}
so that
\begin{equation*}
\exp\cbr{\rbr{\thetab-\thetab^*}^\top \rbr{\psi(\yb^{(j)}) + \psi(\yb^{(j')})} - 2 \log \frac{\hat Z_y(\thetab\phantom{^*})}{\hat Z_y(\thetab^*)}}
\geq -2M_\psi (M_r^2+1) \|\thetab - \thetab^*\|_1,
\end{equation*}
and hence,
\begin{equation*}
\vb^\top \hat \Hb(\thetab) \vb
\geq \exp\rbr{-2M_\psi (M_r^2+1)\|\thetab - \thetab^*\|_1} \vb^\top \hat\Hb(\thetab^*) \vb.
\end{equation*}
Combining with \eqref{eq:RSC_at_theta_star} from the first part finishes the proof.
\end{proof}

\begin{lemma} \label[lemma]{lem:sparse_norm_H_minus_EH}
For $c > 0$, $\beta \geq 0$, $\varepsilon \in (0,1)$, whenever
\[
n_y \geq C (\bar\kappa / \underline\kappa^2) M_\psi^2 M_r^2 s \log^2 (s) \log (p \vee n_y) \log (n_y) c^2 (2+\beta)^4 / \varepsilon^2,
\]
where $C > 0$ denotes a known, absolute constant, we have
\[
\vertiii{\Hb(\thetab^*) - \EE_y\Hb(\thetab^*)}_s
= \sup_{\|\vb\|_0 \leq s, \|\vb\| = 1} |\vb^\top \{\Hb(\thetab^*)-\EE_y\Hb(\thetab^*)\} \vb|
\leq \underline\kappa / (c (2+\beta)^2)
\]
with probability $1-\varepsilon$.
\end{lemma}

\begin{proof}
Similar to the proof of \Cref{lem:subG_quadform}, let
\begin{equation*}
U_{\vb}
:= \frac{2}{1-1/n_y} \vb^\top \Hb(\thetab^*) \vb
= \frac{2}{n_y(n_y-1)} \sum_{1 \leq j < j' \leq n_y} g_{\vb}(\yb^{(j)},\yb^{(j')}),
\end{equation*}
where
\begin{equation*}
g_{\vb}(\yb,\yb')
= \langle \psib(\yb) - \psib(\yb'), \vb \rangle \, \langle \psib(\yb) - \psib(\yb'), \vb \rangle \, r_\theta(\yb) \, r_\theta(\yb').
\end{equation*}
Let
\begin{multline*}
V_{\vb}(\yb^{(1)}, \dots, \yb^{(n_y)})\\
:= \frac{1}{\lfloor n_y/2\rfloor} \left( g_{\vb}(\yb^{(1)},\yb^{(2)}) + g_{\vb}(\yb^{(3)},\yb^{(4)}) + \cdots + g_{\vb}(\yb^{(2\lfloor n_y/2\rfloor-1)},\yb^{(2\lfloor n_y/2\rfloor)}) \right),
\end{multline*}
and write
\begin{equation*}
U_{\vb}
= \frac{1}{n_y!} \sum_{\sigma \in \mathfrak{S}_{n_y}} V_{\vb}(\yb^{(\sigma(1))},\dots,\yb^{(\sigma(n_y))}),
\end{equation*}
where $\mathfrak{S}_{n_y}$ is the group of permutations on $[n_y]$.
Then
\begin{multline*}
\EE_y\sbr{
\sup_{\substack{\|\vb\|_0 \leq s \\ \|\vb\| = 1}}
\abr{U_{\vb} - \EE_y U_{\vb} } }\\
\begin{aligned}
&= \EE_y\sbr{
\sup_{\substack{\|\vb\|_0 \leq s \\ \|\vb\| = 1}}
\abr{
\frac{1}{n_y!} \sum_{\sigma \in \mathfrak{S}_{n_y}} V_{\vb}(\yb^{(\sigma(1))},\dots,\yb^{(\sigma(n_y))}) - \EE_y V_{\vb}(\yb^{(\sigma(1))},\dots,\yb^{(\sigma(n_y))}) } }\\
&\leq \EE_y\sbr{\sup_{\substack{\|\vb\|_0 \leq s \\ \|\vb\| = 1}}
\abr{
V_{\vb}(\yb^{(1)},\dots,\yb^{(n_y)}) - \EE_y V_{\vb}(\yb^{(1)},\dots,\yb^{(n_y)}) } }.
\end{aligned}
\end{multline*}
Denoting $\zb^{(i)} = \rbr{\psib(\yb^{(2i-1)}) - \psib(\yb^{(2i)})}\sqrt{r_\theta(\yb^{(2i-1)}) r_\theta(\yb^{(2i)})}$, we have
\begin{multline*}
\EE_y\sbr{
\sup_{\substack{\|\vb\|_0 \leq s \\ \|\vb\| = 1}} \abr{\vb^\top \{\Hb(\thetab^*)-\EE_y\Hb(\thetab^*)\} \vb } }\\
\leq
\frac{1-1/n_y}{2}
\EE_y\sbr{
\sup_{\substack{\|\vb\|_0 \leq s \\ \|\vb\| = 1}}
\abr{\vb^\top \rbr{\sum_{i \in [\lfloor n_y/2\rfloor]} \zb^i\zb^{i\top} - \EE_y\sbr{ \zb^i\zb^{i\top} } } \vb } }.
\end{multline*}
Note that $\|\zb^{i}\|_{\infty} \leq 2 M_\psi M_r$.
Then an application of Lemma 11 of \cite{Belloni2013Least} gives us
\begin{equation*}
\EE_y\sbr{
\sup_{\substack{\|\vb\|_0 \leq s \\ \|\vb\| = 1}}
\abr{\vb^\top \{\Hb(\thetab^*)-\EE_y\Hb(\thetab^*)\} \vb } }
\leq a_n^2 + a_n \sqrt{\bar\kappa},
\end{equation*}
where $a_n^2 = C M_\psi^2 M_r^2 s \log^2 (s) \log (p\vee n_y) \log(n_y) / n_y$, $C > 0$ is a known, absolute constant inherited from the lemma.
Using Markov's inequality, we get that
\[
\sup_{\substack{\|\vb\|_0 \leq s \\ \|\vb\| = 1}}
\abr{\vb^\top \{\Hb(\thetab^*)-\EE_y\Hb(\thetab^*)\} \vb }
\leq \underline\kappa / (c (2+\beta)^2)
\]
with probability $1-\varepsilon$.
\end{proof}

\begin{lemma}[Lemma 4.9 of \cite{Barber2015ROCKET}] \label[lemma]{lem:sparse_norm}
For any $\Mb \in \RR^{p \times p}$ and $s \geq 1$,
\begin{equation*}
\vb^\top \Mb \vb
\leq \vertiii{\Mb}_s \rbr{\|\vb\| + \frac{\|\vb\|_1}{\sqrt{s}}}^2
\quad\text{for all}\quad \vb \in \RR^P.
\end{equation*}
\end{lemma}

\section{Auxiliary results}

\subsection{\label{app:GAR}Gaussian approximation lemmas}

\begin{lemma} \label[lemma]{lem:Berry-Esseen}
For $\omegab \in \RR^p$, let
\begin{equation*}
A_n = A_n(\omegab) = \left\langle \omegab, \frac{1}{n_x} \sum_{i=1}^{n_x} \rbr{\psib(\xb^{(i)})-\mub_\psi} + \frac{1}{n_y} \sum_{j=1}^{n_y} \rbr{\mub_\psi-\psib(\yb^{(j)})} r_{\theta^*}(\yb^{(j)}) \right\rangle,
\end{equation*}
and
\begin{equation*}
\sigma_n^2
= \sigma_n^2(\omegab)
= \Var\sbr{\sqrt{n} \, A_n(\omegab)}.
\end{equation*}
Then,
\begin{equation*}
\sup_{t \in \RR} \left| \PP\left\{ \sqrt{n} \, A_n / \sigma_n \leq t \right\} - \Phi(t) \right|
\leq
\frac{2 C M_r M_\psi \|\omegab\|}{\eta_{x,n} \eta_{y,n} \sigma_n \sqrt{n}},
\end{equation*}
where $C > 0$ denotes a known, absolute constant.
\end{lemma}

\begin{proof}
Write
\begin{equation*}
\sqrt{n} \, A_n / \sigma_n
= \frac{1}{\sqrt{n}} \left\{ \sum_{i=1}^{n_x} \frac{\langle \omegab, \psib(\xb^{(i)})-\mub_\psi \rangle}{\eta_{x,n} \sigma_n} + \sum_{j=1}^{n_y} \frac{\langle \omegab, \mub_\psi-\psib(\yb^{(j)}) \rangle \, r_{\theta^*}(\yb^{(j)})}{\eta_{y,n} \sigma_n} \right\}.
\end{equation*}
Now,
\begin{eqnarray*}
\frac{\abr{\langle \omegab, \psib(\xb)-\mub_\psi \rangle}}{\eta_{x,n} \sigma_n}
\leq \frac{2M_\psi \|\omegab\|}{\eta_{x,n} \sigma_n}
&\text{and}&
\frac{\abr{\langle \omegab, \mub_\psi-\psib(\yb) \rangle \, r_{\theta^*}(\yb)}}{\eta_{y,n} \sigma_n}
\leq \frac{2M_r M_\psi \|\omegab\|}{\eta_{y,n} \sigma_n}.
\end{eqnarray*}
Noting that $M_r \geq 1$, the Berry-Esseen inequality (Theorem 3.4 of \cite{Chen2011Normal}) yields
\begin{equation*}
\sup_{t \in \RR} \left| \PP\left\{ \sqrt{n} \, A_n / \sigma_n \leq t \right\} - \Phi(t) \right|
\leq \frac{2 C M_r M_\psi \|\omegab\|}{\eta_{x,n} \eta_{y,n} \sigma_n \sqrt{n}},
\end{equation*}
where $C > 0$ is a known, absolute constant from the theorem.
\end{proof}

\begin{lemma}[Lemma D.3 of \cite{Barber2015ROCKET}] \label[lemma]{lem:GAR}
If
\begin{eqnarray*}
\sup_{z \in \RR} \left| \PP\{ A \leq z \} - \Phi(z) \right| \leq \varepsilon_A
& \text{and} &
\PP\{ |B| \leq \delta_B, |C| \leq \delta_C \} \geq 1 - \varepsilon_{BC}
\end{eqnarray*}
for some $\delta_B, \delta_C, \varepsilon_A, \varepsilon_{BC} \in [0,1)$, then
\begin{equation*}
\sup_{z \in \RR} \left| \PP\{ (A + B) / (1 + C) \leq z \} - \Phi(z) \right| \leq \delta_B + \frac{\delta_C}{1-\delta_C} + \varepsilon_A + \varepsilon_{BC}.
\end{equation*}
\end{lemma}

\subsection{\label{app:consistency_var_est}Consistency of the variance estimator}

\begin{lemma} \label[lemma]{lem:varest}
On the event that
\begin{align*}
\|\thetab - \thetab^*\| &\leq \delta_\theta,&
\|\check\omegab_k - \omegab_k^*\| &\leq \delta_k,&
&\text{and}&
\topnorm{\hat\Sbb_\psi - \Sigmab_\psi},\ \topnorm{\hat\Sbb_{\psi\hat r}(\thetab^*) - \Sigmab_{\psi r}} \leq \delta_\sigma / 4,
\end{align*}
the variance estimate \eqref{eq:varest} satisfies
\begin{equation*}
|\hat\sigma_k^2 - \sigma_k^2|
\leq \rbr{\eta_{x,n} \eta_{y,n}}^{-1}
\cbr{\|\omegab_k^*\|^2 \rbr{\delta_\sigma + 2 L_3 \, \delta_\theta}
+ \rbr{\delta_\sigma + 2 L_3 \, \delta_\theta + \topnorm{\Sigmab_\psi} + \topnorm{\Sigmab_{\psi r}}} \delta_k^2}.
\end{equation*}
\end{lemma}

\begin{proof}
Let
\begin{equation*}
\Sigmab_\text{pooled} = \eta_{x,n}^{-1} \Sigmab_\psi + \eta_{y,n}^{-1} \Sigmab_{\psi r}.
\end{equation*}
We have
\begin{equation*}
\begin{aligned}
\hat\sigma_k^2 - \sigma_k^2
&= \check\omegab_k^\top \hat\Sbb_\text{pooled} \check\omegab_k - \omegab_k^{*\top} \Sigmab_\text{pooled} \omegab_k^* \\
&= \check\omegab_k^\top \left\{ \eta_{x,n}^{-1} \hat\Sbb_\psi + \eta_{y,n}^{-1} \hat\Sbb_{\psi\hat r}(\thetab) \right\} \check\omegab_k - \omegab_k^{*\top} \left\{ \eta_{x,n}^{-1} \Sigmab_\psi + \eta_{y,n}^{-1} \Sigmab_{\psi r} \right\} \omegab_k^* \\
&= \eta_{x,n}^{-1} \left( \check\omegab_k^\top \hat\Sbb_\psi \check\omegab_k - \omegab_k^{*\top} \Sigmab_\psi \omegab_k^* \right) + \eta_{y,n}^{-1} \left( \check\omegab_k^\top \hat\Sbb_{\psi\hat r}(\thetab) \check\omegab_k - \omegab_k^{*\top} \Sigmab_{\psi r} \omegab_k^* \right).
\end{aligned}
\end{equation*}
The first term is bounded as
\begin{equation*}
\begin{aligned}
\abr{\check\omegab_k^\top \hat\Sbb_\psi \check\omegab_k - \omegab_k^{*\top} \Sigmab_\psi \omegab_k^*}
&\leq \abr{\check\omegab_k^\top \{\hat\Sbb_\psi - \Sigmab_\psi\} \check\omegab_k} + \abr{(\check\omegab_k - \omegab_k^*)^\top \Sigmab_\psi (\check\omegab_k - \omegab_k^*)}\\
&\leq \topnorm{\hat\Sbb_\psi - \Sigmab_\psi} \|\check\omegab_k\|^2 + \topnorm{\Sigmab_\psi} \|\check\omegab_k - \omegab_k^*\|^2\\
&\leq \tfrac 12 {\delta_\sigma} \rbr{\|\omegab_k^*\|^2 + \delta_k^2} + \topnorm{\Sigmab_\psi} \delta_k^2.
\end{aligned}
\end{equation*}
Similarly,
\begin{multline*}
\abr{\check\omegab_k^\top \hat\Sbb_{\psi\hat r}(\thetab) \check\omegab_k - \omegab_k^{*\top} \Sigmab_{\psi r} \omegab_k^*}\\
\begin{aligned}
&\leq \abr{\check\omegab_k^\top \{\hat\Sbb_{\psi\hat r}(\thetab) - \Sigmab_{\psi r}\} \check\omegab_k} + \abr{(\check\omegab_k - \omegab_k^*)^\top \Sigmab_{\psi r} (\check\omegab_k - \omegab_k^*)}\\
&\leq \topnorm{\hat\Sbb_{\psi\hat r}(\thetab) - \Sigmab_{\psi r}} \|\check\omegab_k\|^2 + \topnorm{\Sigmab_{\psi r}} \|\check\omegab_k - \omegab_k^*\|^2\\
&\leq \rbr{\topnorm{\hat\Sbb_{\psi\hat r}(\thetab) - \hat\Sbb_{\psi\hat r}(\thetab^*)} + \topnorm{\hat\Sbb_{\psi\hat r}(\thetab^*) - \Sigmab_{\psi r}}} \|\check\omegab_k\|^2 + \topnorm{\Sigmab_{\psi r}} \|\check\omegab_k - \omegab_k^*\|^2\\
&\leq \rbr{L_3 \, \|\thetab - \thetab^*\| + \topnorm{\hat\Sbb_{\psi\hat r}(\thetab^*) - \Sigmab_{\psi r}}} \|\check\omegab_k\|^2 + \topnorm{\Sigmab_{\psi r}} \|\check\omegab_k - \omegab_k^*\|^2\\
&\leq \rbr{2L_3 \, \delta_\theta + \tfrac 12 {\delta_\sigma}} \rbr{\|\omegab_k^*\|^2 + \delta_k^2} + \topnorm{\Sigmab_{\psi r}} \delta_k^2,
\end{aligned}
\end{multline*}
where the penultimate line is by \Cref{lem:Lips_sample_cov}.
Thus,
\begin{equation*}
|\hat\sigma_k^2 - \sigma_k^2|
\leq \rbr{\eta_{x,n} \eta_{y,n}}^{-1}
\cbr{\|\omegab_k^*\|^2 \rbr{\delta_\sigma + 2 L_3 \, \delta_\theta}
+ \rbr{\delta_\sigma + 2 L_3 \, \delta_\theta + \topnorm{\Sigmab_\psi} + \topnorm{\Sigmab_{\psi r}}} \delta_k^2}.
\end{equation*}
\end{proof}

\begin{lemma} \label[lemma]{lem:Lips_sample_cov}
There exists $L_3 > 0$ depending on $M_r, M_\psi$ only such that
\begin{equation*}
\topnorm{\hat\Sbb_{\psi\hat r}(\thetab) - \hat\Sbb_{\psi\hat r}(\thetab^*)} \leq L_3 \|\thetab - \thetab^*\| \quad\text{for all}\quad \thetab \in \bar\Bcal_\varrho(\thetab^*).
\end{equation*}
\end{lemma}

\begin{proof}
By applying \Cref{lem:Lips} after computing the form of each $\hat S_{\psi \hat r_{k' k}}(\thetab) - \hat S_{\psi \hat r_{k' k}}(\thetab^*)$.
\end{proof}

\begin{lemma} \label[lemma]{lem:sample_cov:psi}
Under \Cref{cond:bdr} with $\ell_1$-norm, there exist constants $K, c, c' > 0$ depending on $M_\psi$ only such that for any $t \in [K \sqrt{\log p / n_x}, 1]$,
\begin{equation*}
\PP\cbr{\|\hat\Sbb_\psi - \Sigmab_\psi\|_\infty > t} \leq c\exp(-c' t^2 n_x).
\end{equation*}
\end{lemma}

\begin{proof}
Let $k, k' \in [p]$.
\begin{multline*}
\hat S_{\psi_{k' k}} - \Sigma_{\psi_{k' k}}
= \frac{1}{n_x} \sum_{i=1}^{n_x} \left( \psi_{k'}(\xb_{k'}^{(i)}) - \mu_{\psi_{k'}} \right) \left( \psi_k(\xb_k^{(i)}) - \mu_{\psi_k} \right) - \Sigma_{\psi_{k' k}}\\
- \left\{ \frac{1}{n_x} \sum_{i=1}^{n_x} \psi_{k'}(\xb_{k'}^{(i)}) - \mu_{\psi_{k'}} \right\} \left\{ \frac{1}{n_x} \sum_{i=1}^{n_x} \psi_k(\xb_k^{(i)}) - \mu_{\psi_k} \right\}.
\end{multline*}

Suppose $t$ satisfies the conditions of the lemma, and suppose
\begin{gather*}
\left| \frac{1}{n_x} \sum_{i=1}^{n_x} \psi_k(\xb_k^{(i)}) - \mu_{\psi_k} \right| \leq t \quad \forall \, k,\\
\left| \frac{1}{n_x} \sum_{i=1}^{n_x} \left( \psi_{k'}(\xb_{k'}^{(i)}) - \mu_{\psi_{k'}} \right) \left( \psi_k(\xb_k^{(i)}) - \mu_{\psi_k} \right) - \Sigma_{\psi_{k' k}} \right| \leq t \quad \forall \, k, k'.
\end{gather*}
On this event,
\begin{multline*}
\|\hat\Sbb_\psi - \Sigmab_\psi\|_\infty\\
\begin{aligned}[t]
&= \max_{k,k'} |\hat S_{\psi_{k' k}} - \Sigma_{\psi_{k' k}}|\\
&\leq \max_{k,k'} \left| \frac{1}{n_x} \sum_{i=1}^{n_x} \left( \psi_{k'}(\xb_{k'}^{(i)}) - \mu_{\psi_{k'}} \right) \left( \psi_k(\xb_k^{(i)}) - \mu_{\psi_k} \right) - \Sigma_{\psi_{k' k}} \right| + \max_k \left| \frac{1}{n_x} \sum_{i=1}^{n_x} \psi_k(\xb_k^{(i)}) - \mu_{\psi_k} \right|^2
\end{aligned}\\
\leq t + t^2
\leq 2t,
\end{multline*}
using the upper bound on $t$.

Now, the boundedness of $\psib(\xb)$ implies
\begin{gather*}
\PP\left\{ \left| \frac{1}{n_x} \sum_{i=1}^{n_x} \psi_k(\xb_k^{(i)}) - \mu_{\psi_k} \right| > t \right\} \leq 2\exp(-c_1 t^2 n_x),\\
\PP\left\{ \left| \frac{1}{n_x} \sum_{i=1}^{n_x} \left( \psi_{k'}(\xb_{k'}^{(i)}) - \mu_{\psi_{k'}} \right) \left( \psi_k(\xb_k^{(i)}) - \mu_{\psi_k} \right) - \Sigma_{\psi_{k' k}} \right| > t \right\} \leq 2\exp(-c_2 t^2 n_x),
\end{gather*}
where $c_1, c_2 > 0$ are constants depending on $M_\psi$ only.

Thus,
\begin{equation} \label{eq:sample_cov:psi:primitive}
\PP\cbr{\|\hat\Sbb_\psi - \Sigmab_\psi\|_\infty > t}
\leq 2p \exp(-c_1 t^2 n_x) + 2p^2 \exp(-c_2 t^2 n_x)
\leq 4p^2 \exp(-c_3 t^2 n_x),
\end{equation}
where $c_3 > 0$ is another constant depending on $M_\psi$ only.
\eqref{eq:sample_cov:psi:primitive} can be simplified by using the lower bound on $t$:
\begin{equation*}
\PP\cbr{\|\hat\Sbb_\psi - \Sigmab_\psi\|_\infty > t} \leq c\exp(-c' t^2 n_x),
\end{equation*}
where $c, c' > 0$ are constants depending on $M_\psi$ only.
\end{proof}

\begin{lemma} \label[lemma]{lem:sample_cov:psi_rhat}
Under the bounded density ratio model (\Cref{cond:bdr}), there exist constants $K, c, c' > 0$ depending on $M_r, M_\psi$ only such that for any $t \in [K \sqrt{\log p / n_y}, 1]$,
\begin{equation*}
\PP\cbr{\|\hat\Sbb_{\psi r}(\thetab^*) - \Sigmab_{\psi r}\|_\infty > t} \leq c\exp(-c' t^2 n_y).
\end{equation*}
\end{lemma}

\begin{proof}
Let $k, k' \in [p]$.
We have
\begin{equation*}
\hat S_{\psi\hat r_{k' k}}(\thetab^*) - \Sigma_{\psi r_{k' k}}
= \left\{ \hat S_{\psi\hat r_{k' k}}(\thetab^*) - \frac{Z_y^2(\thetab^*)}{\hat Z_y^2(\thetab^*)} \Sigma_{\psi r_{k' k}} \right\} + \left( \frac{Z_y^2(\thetab^*)}{\hat Z_y^2(\thetab^*)} - 1 \right) \Sigma_{\psi r_{k' k}}
\end{equation*}
with
\begin{equation*}
\begin{aligned}
\hat S_{\psi\hat r_{k' k}}(\thetab^*) - & \frac{Z_y^2(\thetab^*)}{\hat Z_y^2(\thetab^*)} \Sigma_{\psi r_{k' k}}\\
= &
\begin{multlined}[t]
\frac{Z_y^2(\thetab^*)}{\hat Z_y^2(\thetab^*)} \, \Bigg[ \frac{1}{n_y} \sum_{j=1}^{n_y} \left( \psi_{k'}(\yb_{k'}^{(j)}) r_{\theta^*}(\yb^{(j)}) - \mu_{\psi_{k'}}\right) \left( \psi_k(\yb_k^{(j)}) r_{\theta^*}(\yb^{(j)}) - \mu_{\psi_k} \right) - \Sigma_{\psi r_{k' k}}\\
- \Bigg\{ \frac{1}{n_y} \sum_{j=1}^{n_y} \psi_{k'}(\yb_{k'}^{(j)}) r_\theta(\yb^{(j)}) - \mu_{\psi_{k'}} \Bigg\} \Bigg\{ \frac{1}{n_y} \sum_{j=1}^{n_y} \psi_k(\yb_k^{(j)}) r_\theta(\yb^{(j)}) - \mu_{\psi_k} \Bigg\} \Bigg]
\end{multlined}
\end{aligned}
\end{equation*}
and
\begin{equation*}
\frac{Z_y^2(\thetab^*)}{\hat Z_y^2(\thetab^*)} - 1
= \frac{Z_y^2(\thetab^*)}{\hat Z_y^2(\thetab^*)} \left( 1 + \frac{\hat Z_y(\thetab^*)}{Z_y(\thetab^*)} \right) \left( 1 - \frac{\hat Z_y(\thetab^*)}{Z_y(\thetab^*)} \right).
\end{equation*}
\Cref{cond:bdr} implies that $Z_y(\thetab^*) / \hat Z_y(\thetab^*) \in [M_r^{-1}, M_r]$, as well as that $\|\Sigmab_{\psi r}\|_\infty$ is bounded by some constant.
So,
\begin{equation*}
\begin{aligned}
\Bigg| \hat S_{\psi\hat r_{k' k}}(\thetab^*) - & \frac{Z_y^2(\thetab^*)}{\hat Z_y^2(\thetab^*)} \Sigma_{\psi r_{k' k}} \Bigg|\\
\leq &
\begin{multlined}[t]
M_r^2 \, \Bigg[ \Bigg| \frac{1}{n_y} \sum_{j=1}^{n_y} \left( \psi_{k'}(\yb_{k'}^{(j)}) r_{\theta^*}(\yb^{(j)}) - \mu_{\psi_{k'}}\right) \left( \psi_k(\yb_k^{(j)}) r_{\theta^*}(\yb^{(j)}) - \mu_{\psi_k} \right) - \Sigma_{\psi r_{k' k}}\Bigg|\\
+ \Bigg| \frac{1}{n_y} \sum_{j=1}^{n_y} \psi_{k'}(\yb_{k'}^{(j)}) r_\theta(\yb^{(j)}) - \mu_{\psi_{k'}} \Bigg| \Bigg| \frac{1}{n_y} \sum_{j=1}^{n_y} \psi_k(\yb_k^{(j)}) r_\theta(\yb^{(j)}) - \mu_{\psi_k} \Bigg| \Bigg]
\end{multlined}
\end{aligned}
\end{equation*}
and
\begin{equation*}
\abr{\rbr{\frac{Z_y^2(\thetab^*)}{\hat Z_y^2(\thetab^*)} - 1} \Sigma_{\psi r_{k' k}}}
\leq M_r^2 (1+M_r) \|\Sigmab_{\psi r}\|_\infty \abr{1 - \frac{\hat Z_y(\thetab^*)}{Z_y(\thetab^*)}}.
\end{equation*}

Suppose $t$ satisfies the conditions of the lemma, and suppose
\begin{gather*}
\abr{\frac{\hat Z_y(\thetab)}{Z_y(\thetab)} - 1} \leq t,\\
\abr{\frac{1}{n_y} \sum_{j=1}^{n_y} \psi_k(\yb_k^{(j)}) r_\theta(\yb^{(j)}) - \mu_{\psi_k}} \leq t \quad \forall \, k,\\
\abr{\frac{1}{n_y} \sum_{j=1}^{n_y} \left( \psi_{k'}(\yb_{k'}^{(j)}) r_{\theta^*}(\yb^{(j)}) - \mu_{\psi_{k'}}\right) \left( \psi_k(\yb_k^{(j)}) r_{\theta^*}(\yb^{(j)}) - \mu_{\psi_k} \right) - \Sigma_{\psi r_{k' k}}} \leq t \quad \forall \, k, k'.
\end{gather*}
On this event,
\begin{equation*}
\abr{\hat S_{\psi\hat r_{k' k}}(\thetab^*) - \frac{Z_y^2(\thetab^*)}{\hat Z_y^2(\thetab^*)} \Sigma_{\psi r_{k' k}}}
\leq M_r^2 (t + t^2)
\leq 2 M_r^2 \, t
\end{equation*}
and
\begin{equation*}
\abr{\rbr{\frac{Z_y^2(\thetab^*)}{\hat Z_y^2(\thetab^*)} - 1} \Sigma_{\psi r_{k' k}}}
\leq M_r^2 (1+M_r) \|\Sigmab_{\psi r}\|_\infty \, t,
\end{equation*}
and hence,
\begin{equation*}
\|\hat\Sbb_{\psi\hat r}(\thetab^*) - \Sigmab_{\psi r}\|_\infty \leq K t
\end{equation*}
for some constant $K > 0$.

We finish the proof by bounding the probability of the complementary event.
By \Cref{lem:Hoeffding_r},
\begin{equation*}
\PP\left\{ \left| \frac{\hat Z_y(\thetab)}{Z_y(\thetab)} - 1 \right| > t \right\} \leq 2\exp( -c_1 t^2 n_y),
\end{equation*}
for some constant $c_1 > 0$ depending on $M_r$ only.
On the other hand, the boundedness of $\psi(\yb) r_{\theta^*}(\yb)$ implies
\begin{gather*}
\PP\cbr{\abr{\frac{1}{n_y} \sum_{j=1}^{n_y} \psi_k(\yb_k^{(j)}) r_{\theta^*}(\yb^{(j)}) - \mu_{\psi_k}} > t} \leq 2\exp(-c_2 t^2 n_y),\\
\PP\cbr{\abr{\frac{1}{n_y} \rbr{\psi_{k'}(\yb_{k'}^{(j)}) r_{\theta^*}(\yb^{(j)}) - \mu_{\psi_{k'}}} \rbr{\psi_{k}(\yb_{k}^{(j)}) r_{\theta^*}(\yb^{(j)}) - \mu_{\psi_{k}}} - \Sigma_{\psi r_{k' k}}} > t}
\leq 2\exp(-c_3 t^2 n_y),
\end{gather*}
where $c_2, c_3 > 0$ are constants depending on $M_r, M_\psi$ only.

Thus,
\begin{multline} \label{eq:sample_cov:psi_rhat:primitive}
\PP\cbr{\|\hat\Sbb_{\psi\hat r}(\thetab^*) - \Sigmab_{\psi r}\|_\infty > t}\\
\leq 2\exp( -c_1 t^2 n_y) + 2p \exp(-c_2 t^2 n_y) + 2p^2 \exp(-c_3 t^2 n_y)
\leq 6p^2 \exp(-c_4 t^2 n_y),
\end{multline}
where $c_4 > 0$ is another constant depending on $M_r, M_\psi$ only.
\eqref{eq:sample_cov:psi_rhat:primitive} can be simplified by using the lower bound on $t$:
\begin{equation*}
\PP\cbr{\|\hat\Sbb_{\psi r}(\thetab^*) - \Sigmab_{\psi r}\|_\infty > t} \leq c\exp(-c' t^2 n_y),
\end{equation*}
where $c, c' > 0$ are constants depending on $M_r, M_\psi$ only.
\end{proof}

\section{Implementation details}

\subsection{\label{supp:implementation:autoscaling}Autoscaling}

\subsubsection{Sparse KLIEP with self-normalizing penalty}

The default option in KLIEPInference.jl (\url{https://github.com/mlakolar/KLIEPInference.jl}) replaces \eqref{eq:KLIEP+:1} in the initial KLIEP estimation step with the following modified version
\begin{flalign}
\label{eq:KLIEP+:1:selfnorm}
	\check\thetab &\leftarrow \arg\min_{\thetab} \lKLIEP(\thetab;\Xb_{n_x},\Yb_{n_y}) + \lambda_{\theta 0} \sum_{k = 1}^{p} \tau_k |\theta_k|,
\end{flalign}
where $\lambda_{\theta 0} = (1 + a) \Phi^{-1}(1 - b / p)$ for some small $a, b > 0$ is the universal penalty and $\tau_k > 0$ is the $k$th penalty loading.
For $\lambda_{\theta 0}$, we used $a = 0.01$ and $b = 0.05$.
The $k$th penalty loading is chosen to match the sample standard deviation of $\partial_k \lKLIEP(\thetab^*)$; this has the effect of penalizing components with larger variance more.

As $\thetab^*$ is unavailable to us,
we take the following two-step approach:
\begin{procedure}[\label{procedure:step1}Two-step procedure for minimizing \eqref{eq:KLIEP+:1:selfnorm}]
\begin{framed}
\begin{algorithmic}
\STATE {Initialize $\check\thetab \leftarrow \zero$.}
\STATE {Compute the initial penalty loadings: for $k = 1, \dots, p$,
\begin{flalign*}
	\mathindent
	\tau_k &\leftarrow \hat\Sbb_{\text{pooled}, kk}(\check\thetab).&&
\end{flalign*}}
\STATE {Compute $\check\thetab$:
\begin{flalign*}
	\mathindent
	\check\thetab &\leftarrow \arg\min_{\thetab} \lKLIEP(\thetab;\Xb_{n_x},\Yb_{n_y}) + \lambda_{\theta 0} \sum_{k = 1}^{p} \tau_k |\theta_k|.&&
\end{flalign*}}
\STATE {Update the penalty loadings: for $k = 1, \dots, p$,
\begin{flalign*}
	\mathindent
	\tau_k &\leftarrow \hat\Sbb_{\text{pooled}, kk}(\check\thetab).&&
\end{flalign*}}
\STATE {Estimate $\check\thetab$ with the updated penalty loadings.}
\end{algorithmic}
\end{framed}
\end{procedure}

The intuition behind $\lambda_{\theta 0} = (1 + a) \Phi^{-1}(1 - b / p)$
and $\tau_k \approx \sqrt{\hat\Var[\partial_k \lKLIEP(\thetab^*)]}$ is as
follows. Estimation using \eqref{eq:KLIEP+:1:selfnorm} is consistent provided that
\begin{equation}
\label{eq:consistency:selfnorm}
	\PP\cbr{\max_k |\partial_k \lKLIEP(\thetab^*) / \tau_k| > \lambda_{\theta 0}} \text{ is small.}
\end{equation}
For sufficiently large sample sizes, we would
have $\partial_k \lKLIEP(\thetab^*) / \sqrt{\hat\Var[\partial_k \lKLIEP(\thetab^*)]} \approx \Ncal(0, 1)$,
and hence for
$\lambda_{\theta 0} = (1 + a) \Phi^{-1}(1 - b / p)$,
an upper bound for the probability in \eqref{eq:consistency:selfnorm} is about $b > 0$.
Thus, $b$ can be interpreted as a tolerance parameter that controls 
the probability of the undesirable event.
Similar approach was taken in
\citet{Belloni2011Square, Belloni2014Pivotal, Belloni2019Pivotal}
in the context of linear regression, nonparametric regression,
and error-in-variables regression problems. For detailed
discussions of the motivation and the relationship
to the moderate deviations theory, we refer the reader to
these works and the references therein.
In particular, a rigorous proof in the context of our
problem would involve establishing a moderate deviation
bound \citep{delaPena2009Self, Jing2003Self} for
the self-normalized gradient
$[\partial_k \lKLIEP(\thetab^*) / \sqrt{\hat\Var[\partial_k \lKLIEP(\thetab^*)]}]_{k = 1}^{p}$,
which we leave up to future work.

\subsubsection{Sparse Hessian inversion via the scaled lasso}

The default option in KLIEPInference.jl (\url{https://github.com/mlakolar/KLIEPInference.jl}) replaces \eqref{eq:KLIEP+:2} in the Hessian inversion step with a scaled lasso formulation \citep{Sun2012Scaled}.
In particular, we use the approach described in \citet{Sun2012Sparse} that allows us to estimate a sparse inverse of the Hessian without hyperparameter tuning. This implementation is used for all of our
experiments.

In the below, we describe the procedure in more detail.
The equation \eqref{eq:KLIEP+:2} is modified so that $\check\omegab_k = -\check\tau_k \check\db_k$, where
\begin{flalign}
\label{eq:KLIEP+:2:scaled}
	\mathindent
	\check\db_k, \check\tau_k
	&\leftarrow \arg\min_{\db, \tau: d_k = -1} \frac{\db^\top \nabla^2 \lKLIEP(\check\thetab) \db}{2 \tau} + \frac{\tau}{2} + \lambda_0 \sum_{k' = 1}^{p} \partial_{k'k'}^2 \lKLIEP(\check\thetab) |d_{k'}|&&
\end{flalign}
and the universal penalty level $\lambda_0 = \sqrt{2 \log p / n_y}$
does not depend on the unknown problem specific parameters.
Following \citet{Sun2012Sparse}, the solution $(\check\db_k, \check\tau_k)$ is obtained from the following iterative procedure:
\begin{procedure}[\label{procedure:step2}Iterative procedure for solving \eqref{eq:KLIEP+:2:scaled}]
\begin{framed}
\begin{algorithmic}
\STATE {Initialize $\check\db_k \leftarrow \eb_k$.}
\REPEAT
\STATE {\begin{flalign*}
	\mathindent
	\check\tau_k &\leftarrow \check\db_k^\top \nabla^2 \lKLIEP(\check\thetab) \check\db_k,&&\\
	\lambda &\leftarrow \lambda_0 \check\tau_k,\\
	\check\db_k &\leftarrow \arg\min_{\db} \frac 12 \db^\top \nabla^2 \lKLIEP(\check\thetab) \db - \db^\top \eb_k + \lambda \|\db\|_1.
\end{flalign*}}
\UNTIL {converged}
\end{algorithmic}
\end{framed}
\end{procedure}
For a detailed discussion of the procedure and its theoretical properties,
the reader is referred to \citet{Sun2012Sparse}.

\subsection{\label{supp:implementation:hyperparameter}Regularization parametter tuning}

In all our experiments, including the experiments published only in Appendix, we used \Cref{procedure:step2} for Step 2 with the universal penalty level $\lambda_0 = \sqrt{2 \log p / n_y}$.
For Experiments~1 -- 3, we use \Cref{procedure:step1} for Step 1 with the universal penalty level $\lambda_{\theta 0} = 1.01 \Phi^{-1}(1 - 0.05 / p)$.
Experiments~4 -- 5 use the original sparse KLIEP formulation \citep{Liu2017Support} \emph{without} autoscaling. For Experiment 4, we used $\lambda_\theta = \sqrt{4 \log p / n_x}$, because for Ising models, the components of the gradient $\nabla \lKLIEP(\thetab^*)$ are bounded by $2$ when $\thetab^* \approx \zero$.


Parameter tuning is an issue for most, if not all, high-dimensional estimation procedures, and ours is no exception. As noted by one reviewer, it is at least unclear how the regularization parameter pair can be chosen to achieve the best performance. In the case of the bounded model, it is possible to make an educated guess for the first-stage regularization parameter $\lambda_\theta$ (\Cref{lem:grad:1}), and this is what we do in our experiments. Choosing the second-stage regularization parameters $\lambda_k$ is a more delicate matter.

One heuristic is to cross-validate the \emph{three-stage procedure in its entirety} over a 2D grid of $(\lambda_\theta, \lambda_k)$ pairs using the empirical KLIEP loss.
A clear drawback of this strategy is that it is computationally intensive.
It also has very little theory.

A good alternative is to use autoscaling procedures for the initial estimation steps. In our simulations, the combination of \Cref{procedure:step1} and \Cref{procedure:step2} has been seen to yield excellent performance while removing the need for hyperparameter tuning. For theory, we need the initial estimates obtained using \Cref{procedure:step1} and \Cref{procedure:step2} to be consistent. While we leave this up for future work, theoretical results for similar problems (e.g., \citet{Belloni2011Square} in the case of Step 1 and \citet{Sun2012Sparse} in the case of Step 2) lend support to our claim.

Additionally, to study the sensitivity of the overall procedure to the choice of the regularization parameter when the original sparse KLIEP formulation \citep{Liu2017Support} is used, we ran additional experiments where we varied $\lambda_\theta$ on a grid of five values under the same set-up as that of Experiment 1. For Step 2, we still use \Cref{procedure:step2} with the universal penalty level $\lambda_0 = \sqrt{2 \log p / n_y}$. We record the coverage and the median width of the 95\% confidence intervals as well as the bias of the final estimate over 1000 independent replications. The regularization parameter settings are detailed in Table~\ref{table:lambda_values:additional}. The results are shown in Tables~\ref{table:divergences:coverage:chain} to \ref{table:divergences:bias:tree}. The coverage, the median width, and the bias are all stable for both SparKLIE+ procedures. The reversed and the symmetrized procedures do show some instability, but it is likely that this has more to do with the fact that both procedures have a larger sample complexity relative to KLIEP. See \Cref{remark:asymmetry} in \Cref{sec:theory:main}.

\subsection{\label{supp:implementation:studentization}Studentized bootstrap}

Consider the studentized analogue of the statistic in \eqref{eq:def_max_statistic}
\begin{equation}
	W = W_{n_x, n_y} = \max_{k = 1, \dots, p} \sqrt{n} \, |\hat\theta_k - \theta_k^*| / \hat\sigma_k,
\end{equation}
where $\hat\theta_k$ is either SparKLIE+1 or SparKLIE+2 estimator and $\hat\sigma_k$ is the estimator of the standard error from \eqref{eq:varest}. $W$ can replace $T$ as the reference distribution in carrying out statistical inference. Letting $c_{W, q}$ be the $q$-quantile of $T$, $\hat\thetab \pm (c_{W, 1-\alpha} / \sqrt{n}) \hat\sigmab$, where $\hat\sigmab = (\hat\sigma_k)_{k = 1}^{p}$, is a $100 \times (1-\alpha) \%$ confidence region for $\thetab^*$. Similarly, the test that rejects if $\max_k |\hat\theta_k | / \hat\sigma_k > c_{W, 1-\alpha} / \sqrt{n}$ controls the family-wise error rate at level $\alpha$ for the null hypothesis $H_0: \theta_k^* = 0$ for all $k \in [p]$. This approach has the advantage of being adaptive to the heterogeneity in variance across multiple components.

The bootstrap procedures of Section~\ref{sec:methodology:multiple} can be easily modified to yield estimates of the quantiles of $W$. In Procedure~\ref{procedure:bootstrap:Gaussian}, this is accomplished by replacing \eqref{eq:T-boot:Gaussian} with
\begin{multline}
\label{eq:W-boot:Gaussian}
	\hat W^{(b)} = \max_k \frac{1}{\hat\sigma_k \sqrt{n}} \Bigg| \Bigg\langle \check\omegab_k, \frac{n}{n_x} \sum_{i=1}^{n_x} \rbr{\psib(\xb^{(i)}) - \overline{\psib}} \xi_x^{(b,i)}\\
	- \frac{n}{n_y} \sum_{j=1}^{n_y} \rbr{\psib(\yb^{(j)}) \hat r_{\check\theta}(\yb^{(j)}) - \hat\mub(\check\thetab)} \xi_y^{(b,j)} \Bigg\rangle \Bigg|.
\end{multline}
In the case of Procedure~\ref{procedure:bootstrap:empirical}, one replaces \eqref{eq:T-boot:empirical} with
\begin{equation}
\label{eq:W-boot:empirical}
	\hat W^{(b)} = \max_k \sqrt{n} \, |\hat\theta_k^{(b)} - \hat\theta_k| / \hat\sigma_k.
\end{equation}

\section{Supplementary material for Section 5}

\subsection{\label{supp:sec5:competing}Competing procedures}

The \emph{oracle} refers to the following procedure:
\begin{flalign*}
\mathindent
\hat\thetab^\text{oracle}
&\leftarrow \argmin_{\thetab} \lKLIEP(\thetab;\Xb_{n_x},\Yb_{n_y})
\text{ subject to } \supp(\thetab) \subseteq \{k\} \cup \supp(\thetab^*).&&
\end{flalign*}
This is clearly infeasible due to the occurrence of $\thetab^*$ in the constraint.
It is meant to be a performance benchmark rather than an actual alternative.

The \emph{na\"{i}ve} re-estimation is the procedure obtained by replacing the unknown $\thetab^*$ in the constraint with a sample estimate $\check\thetab$:
\begin{flalign*}
\mathindent
\hat\thetab^\text{na\"{i}ve}
\leftarrow \argmin_{\thetab} \lKLIEP(\thetab;\Xb_{n_x},\Yb_{n_y})
\text{ subject to } \supp(\thetab) \subseteq \{k\} \cup \supp(\check\thetab).&&
\end{flalign*}
This can have a near oracle behavior if $\check\thetab$ recovers the true support with high probability.
Unfortunately, the sufficient conditions are often not met for many interesting applications; they are also notoriously difficult to check from the data \citep{Liu2017Support}.
As such, the procedure is expected to be brittle to errors in model selection.

Finally, \emph{SparKLIE+2} is the procedure obtained by choosing double-selection rather than one-step approximation in Step~3 of SparKLIE+1 (Procedure~2):
\begin{enumerate}[label={\it Step \arabic*.}, wide=0pt]
\setcounter{enumi}{2}
\item (Re-estimation on the combined support)
\begin{flalign*}
\mathindent
\hat\thetab^\text{2+}
\leftarrow \argmin_{\thetab} \lKLIEP(\thetab;\Xb_{n_x},\Yb_{n_y})
\text{ subject to }
\supp(\thetab) \subseteq \{k\} \cup \supp(\check\thetab) \cup \supp(\check\omegab_k).&&
\end{flalign*}
\end{enumerate}
This looks deceptively like the na\"{i}ve re-estimation, but the inclusion of the coordinates with large correlations with $k$ makes the procedure robust to model selection mistakes.
SparKLIE+2 is first-order equivalent to SparKLIE+1 \citep{Chernozhukov2015Valid}.

\subsection{\label{supp:sec5:design:exper1}Parameter generation for Experiment~1}

\tikzstyle{vertex} = [circle, draw, inner sep = 1.5, fill = black]
\tikzstyle{dotdotdot} = [circle, draw, inner sep = 0.1, fill = black]

\begin{figure}[tp]
    \caption{\label{fig:chain1}\textbf{Chain 1 pair, realized edge weights.} The displayed weights are the actual values used in the experiments. Excluding the difference graph, all the weights, including the ones not shown here, were generated i.i.d.~$\Unif(-1,1)$. The target of inference is marked in red.}
    \centering
    
    \vspace{\baselineskip}
    
    \begin{minipage}[t][0.1\textheight][c]{\linewidth}
    \centering
    {\small (a) $\gammab_x$}
    
    \vspace{\baselineskip}
    
    \begin{tikzpicture}[shorten >=1pt]
    \foreach \x in {1, 2, ..., 10} \node[vertex, label=below:$\scriptstyle\x$] (G-\x) at (1.25*\x,0) {};
    \foreach \x in {1, 2, 3} \node[dotdotdot] (E-\x) at (13.75+0.2*\x,0) {};
    \draw (G-1) -- node[anchor=south] {$\scriptscriptstyle-0.54$} (G-2);
    \draw (G-2) -- node[anchor=south] {$\scriptscriptstyle-0.85$} (G-3);
    \draw (G-3) -- node[anchor=south] {$\scriptscriptstyle 0.74$} (G-4);
    \draw (G-4) -- node[anchor=south] {$\scriptscriptstyle 0.56$} (G-5);
    \draw[red] (G-5) -- node[anchor=south] {$\scriptscriptstyle-0.06$} (G-6);
    \draw (G-6) -- node[anchor=south] {$\scriptscriptstyle-0.10$} (G-7);
    \draw (G-7) -- node[anchor=south] {$\scriptscriptstyle 0.11$} (G-8);
    \draw (G-8) -- node[anchor=south] {$\scriptscriptstyle 0.35$} (G-9);
    \draw (G-9) -- node[anchor=south] {$\scriptscriptstyle 0.09$} (G-10);
    \draw (G-10) -- node[anchor=south] {$\scriptscriptstyle 0.20$} (13.75,0);
    \end{tikzpicture}
    \end{minipage}
    
    \vspace{\baselineskip}
    
    \begin{minipage}[b]{\linewidth}
    \centering
    {\small (b) $\gammab_y$}
    
    \begin{tikzpicture}[shorten >=1pt]
    \foreach \x in {1, 2, ..., 10} \node[vertex, label=below:$\scriptstyle\x$] (G-\x) at (1.25*\x,0) {};
    \foreach \x in {1, 2, 3} \node[dotdotdot] (E-\x) at (13.75+0.2*\x,0) {};
    \draw (G-1) -- node[anchor=south] {$\scriptscriptstyle-0.54$} (G-2);
    \draw (G-2) -- node[anchor=south] {$\scriptscriptstyle-0.85$} (G-3);
    \draw (G-3) -- node[anchor=south] {$\scriptscriptstyle 0.74$} (G-4);
    \draw (G-4) -- node[anchor=south] {$\scriptscriptstyle 0.16$} (G-5);
    \draw[red] (G-5) -- node[anchor=south] {$\scriptscriptstyle 0.14$} (G-6);
    \draw (G-6) -- node[anchor=south] {$\scriptscriptstyle 0.30$} (G-7);
    \draw (G-7) -- node[anchor=south] {$\scriptscriptstyle 0.11$} (G-8);
    \draw (G-8) -- node[anchor=south] {$\scriptscriptstyle 0.35$} (G-9);
    \draw (G-9) -- node[anchor=south] {$\scriptscriptstyle 0.09$} (G-10);
    \draw (G-10) -- node[anchor=south] {$\scriptscriptstyle 0.20$} (13.75,0);
    \path (G-4) edge[bend left = 75] node[anchor=south] {$\scriptscriptstyle-0.20$} (G-6);
    \path (G-5) edge[bend right = 75] node[below] {$\scriptscriptstyle-0.20$} (G-7);
    \end{tikzpicture}
    \end{minipage}
    
    \vspace{\baselineskip}
    
    \begin{minipage}[b]{\linewidth}
    \centering
    {\small (c) difference}
    
    \begin{tikzpicture}[shorten >=1pt]
    \foreach \x in {1, 2, ..., 10} \node[vertex, label=below:$\scriptstyle\x$] (G-\x) at (1.25*\x,0) {};
    \foreach \x in {1, 2, 3} \node[dotdotdot] (E-\x) at (13.75+0.2*\x,0) {};
    \draw (G-4) -- node[anchor=south] {$\scriptscriptstyle 0.4$} (G-5);
    \draw[red] (G-5) -- node[anchor=south] {$\scriptscriptstyle-0.2$} (G-6);
    \draw (G-6) -- node[anchor=south] {$\scriptscriptstyle-0.4$} (G-7);
    \path (G-4) edge[bend left = 75] node[anchor=south] {$\scriptscriptstyle 0.2$} (G-6);
    \path (G-5) edge[bend right = 75] node[below] {$\scriptscriptstyle 0.2$} (G-7);
    \end{tikzpicture}
    \end{minipage}
\end{figure}
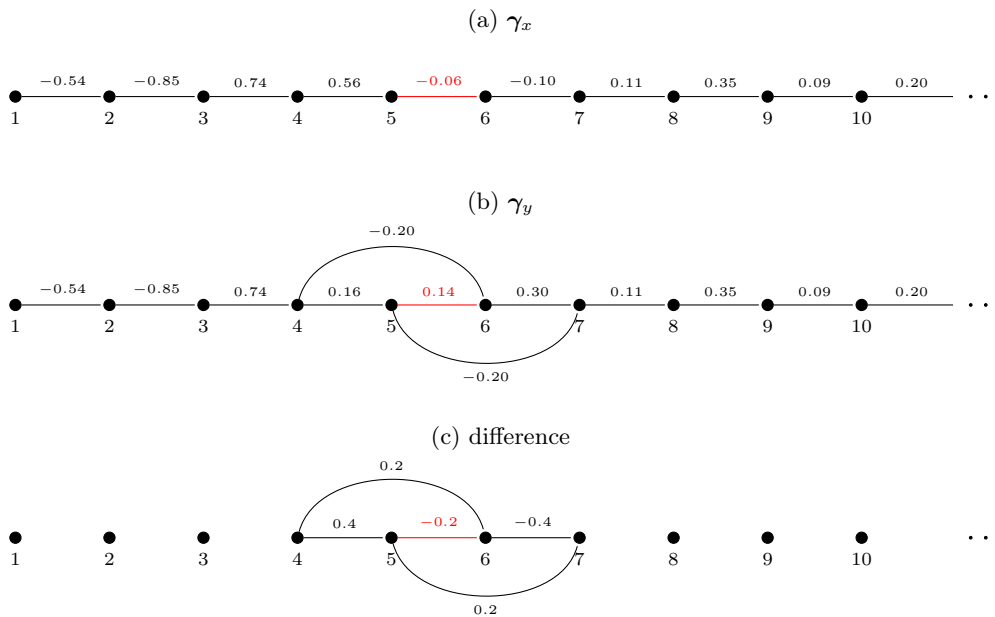

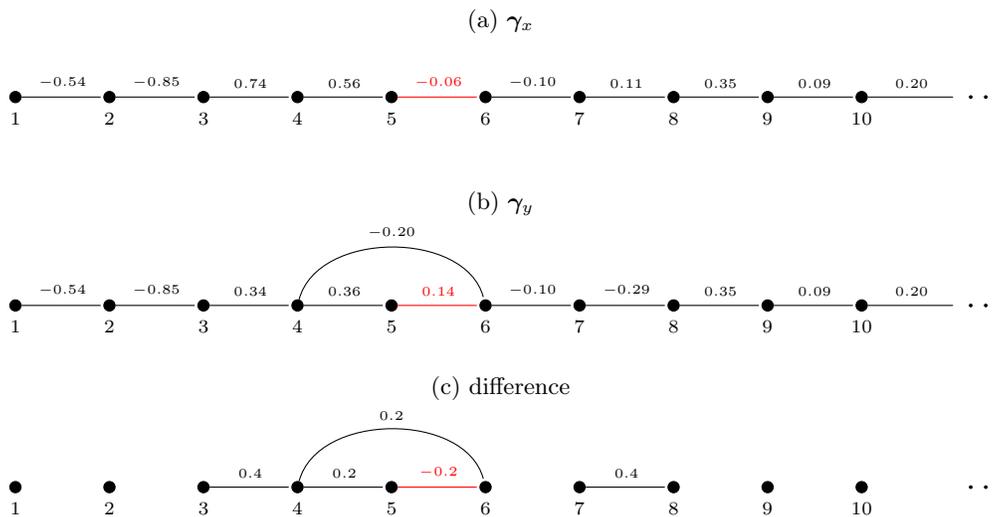
\begin{figure}[bp]
    \caption{\label{fig:chain2}\textbf{Chain 2 pair, realized edge weights.} The displayed weights are the actual values used in the experiments. Excluding the difference graph, all the weights, including the ones not shown here, were generated i.i.d.~$\Unif(-1,1)$. The target of inference is marked in red.}
    \centering
    
    \vspace{\baselineskip}
    
    \begin{minipage}[t][0.1\textheight][c]{\linewidth}
    \centering
    {\small (a) $\gammab_x$}
    
    \vspace{\baselineskip}
    
    \begin{tikzpicture}[shorten >=1pt]
    \foreach \x in {1, 2, ..., 10} \node[vertex, label=below:$\scriptstyle\x$] (G-\x) at (1.25*\x,0) {};
    \foreach \x in {1, 2, 3} \node[dotdotdot] (E-\x) at (13.75+0.2*\x,0) {};
    \draw (G-1) -- node[anchor=south] {$\scriptscriptstyle-0.54$} (G-2);
    \draw (G-2) -- node[anchor=south] {$\scriptscriptstyle-0.85$} (G-3);
    \draw (G-3) -- node[anchor=south] {$\scriptscriptstyle 0.74$} (G-4);
    \draw (G-4) -- node[anchor=south] {$\scriptscriptstyle 0.56$} (G-5);
    \draw[red] (G-5) -- node[anchor=south] {$\scriptscriptstyle-0.06$} (G-6);
    \draw (G-6) -- node[anchor=south] {$\scriptscriptstyle-0.10$} (G-7);
    \draw (G-7) -- node[anchor=south] {$\scriptscriptstyle 0.11$} (G-8);
    \draw (G-8) -- node[anchor=south] {$\scriptscriptstyle 0.35$} (G-9);
    \draw (G-9) -- node[anchor=south] {$\scriptscriptstyle 0.09$} (G-10);
    \draw (G-10) -- node[anchor=south] {$\scriptscriptstyle 0.20$} (13.75,0);
    \end{tikzpicture}
    \end{minipage}
    
    \vspace{\baselineskip}
    
    \begin{minipage}[b]{\linewidth}
    \centering
    {\small (b) $\gammab_y$}
    
    \begin{tikzpicture}[shorten >=1pt]
    \foreach \x in {1, 2, ..., 10} \node[vertex, label=below:$\scriptstyle\x$] (G-\x) at (1.25*\x,0) {};
    \foreach \x in {1, 2, 3} \node[dotdotdot] (E-\x) at (13.75+0.2*\x,0) {};
    \draw (G-1) -- node[anchor=south] {$\scriptscriptstyle-0.54$} (G-2);
    \draw (G-2) -- node[anchor=south] {$\scriptscriptstyle-0.85$} (G-3);
    \draw (G-3) -- node[anchor=south] {$\scriptscriptstyle 0.34$} (G-4);
    \draw (G-4) -- node[anchor=south] {$\scriptscriptstyle 0.36$} (G-5);
    \draw[red] (G-5) -- node[anchor=south] {$\scriptscriptstyle 0.14$} (G-6);
    \draw (G-6) -- node[anchor=south] {$\scriptscriptstyle-0.10$} (G-7);
    \draw (G-7) -- node[anchor=south] {$\scriptscriptstyle-0.29$} (G-8);
    \draw (G-8) -- node[anchor=south] {$\scriptscriptstyle 0.35$} (G-9);
    \draw (G-9) -- node[anchor=south] {$\scriptscriptstyle 0.09$} (G-10);
    \draw (G-10) -- node[anchor=south] {$\scriptscriptstyle 0.20$} (13.75,0);
    \path (G-4) edge[bend left = 75] node[anchor=south] {$\scriptscriptstyle-0.20$} (G-6);
    \end{tikzpicture}
    \end{minipage}
    
    \vspace{\baselineskip}
    
    \begin{minipage}[b]{\linewidth}
    \centering
    {\small (c) difference}
    
    \begin{tikzpicture}[shorten >=1pt]
    \foreach \x in {1, 2, ..., 10} \node[vertex, label=below:$\scriptstyle\x$] (G-\x) at (1.25*\x,0) {};
    \foreach \x in {1, 2, 3} \node[dotdotdot] (E-\x) at (13.75+0.2*\x,0) {};
    \draw (G-3) -- node[anchor=south] {$\scriptscriptstyle 0.4$} (G-4);
    \draw (G-4) -- node[anchor=south] {$\scriptscriptstyle 0.2$} (G-5);
    \draw[red] (G-5) -- node[anchor=south] {$\scriptscriptstyle-0.2$} (G-6);
    \draw (G-7) -- node[anchor=south] {$\scriptscriptstyle 0.4$} (G-8);
    \path (G-4) edge[bend left = 75] node[anchor=south] {$\scriptscriptstyle 0.2$} (G-6);
    \end{tikzpicture}
    \end{minipage}
\end{figure}

\tikzstyle{level 1}=[level distance=50pt, sibling distance = 125pt]
\tikzstyle{level 2}=[level distance=50pt, sibling distance = 40pt]
\tikzstyle{level 3}=[level distance=30pt, sibling distance = 15pt]

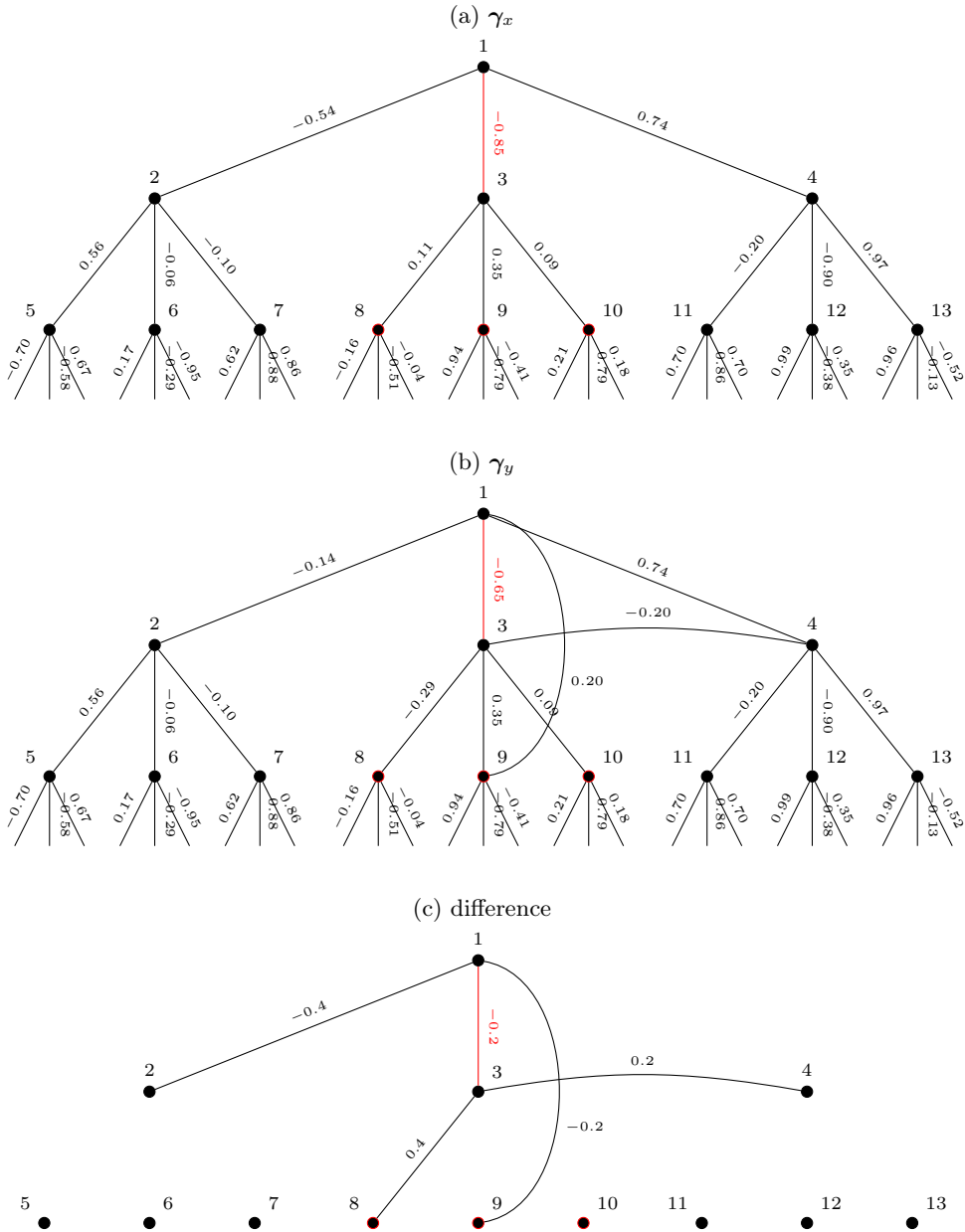
\begin{figure}[p]
    \caption{\label{fig:tree1}\textbf{Tree 1 pair, realized edge weights.} The displayed weights are the actual values used in the experiments. Excluding the difference graph, all the weights, including the ones not shown here, were generated i.i.d.~$\Unif(-1,1)$. The target of inference is marked in red.}
    \centering
    
    \vspace{\baselineskip}
    
    \begin{minipage}[b]{\linewidth}
    \centering
    {\small (a) $\gammab_x$}
    
    \begin{tikzpicture}[>=stealth]
    \node(0)[vertex, label = above:{$\scriptstyle 1$}]{}
    	child {node[vertex, label = above:{$\scriptstyle 2$}]{}
    		child {node[vertex, label = above left:{$\scriptstyle 5$}]{}
    			child{node{} edge from parent node[above, sloped] {$\scriptscriptstyle-0.70$}}
    			child{node{} edge from parent node[above, sloped] {$\scriptscriptstyle-0.58$}}
    			child{node{} edge from parent node[above, sloped] {$\scriptscriptstyle 0.67$}}
    			edge from parent node[above, sloped] {$\scriptscriptstyle 0.56$}}
    		child {node[vertex, label = above right:{$\scriptstyle 6$}]{}
    			child{node{} edge from parent node[above, sloped] {$\scriptscriptstyle 0.17$}}
    			child{node{} edge from parent node[above, sloped] {$\scriptscriptstyle-0.29$}}
    			child{node{} edge from parent node[above, sloped] {$\scriptscriptstyle-0.95$}}
    			edge from parent node[above, sloped] {$\scriptscriptstyle-0.06$}}
    		child {node[vertex, label = above right:{$\scriptstyle 7$}]{}
    			child{node{} edge from parent node[above, sloped] {$\scriptscriptstyle 0.62$}}
    			child{node{} edge from parent node[above, sloped] {$\scriptscriptstyle 0.88$}}
    			child{node{} edge from parent node[above, sloped] {$\scriptscriptstyle 0.86$}}
    			edge from parent node[above, sloped] {$\scriptscriptstyle-0.10$}}
    		edge from parent node[above, sloped] {$\scriptscriptstyle-0.54$}}
    	child {node[vertex, label = above right:{$\scriptstyle 3$}]{}
    		child {node[vertex, label = above left:{$\scriptstyle 8$}]{}
    			child{node{} edge from parent node[above, sloped] {$\scriptscriptstyle-0.16$}}
    			child{node{} edge from parent node[above, sloped] {$\scriptscriptstyle-0.51$}}
    			child{node{} edge from parent node[above, sloped] {$\scriptscriptstyle-0.04$}}
    			edge from parent[black] node[above, sloped] {$\scriptscriptstyle 0.11$}}
    		child {node[vertex, label = above right:{$\scriptstyle 9$}]{}
    			child{node{} edge from parent node[above, sloped] {$\scriptscriptstyle 0.94$}}
    			child{node{} edge from parent node[above, sloped] {$\scriptscriptstyle-0.79$}}
    			child{node{} edge from parent node[above, sloped] {$\scriptscriptstyle-0.41$}}
    			edge from parent[black] node[above, sloped] {$\scriptscriptstyle 0.35$}}
    		child {node[vertex, label = above right:{$\scriptstyle 10$}]{}
    			child{node{} edge from parent node[above, sloped] {$\scriptscriptstyle 0.21$}}
    			child{node{} edge from parent node[above, sloped] {$\scriptscriptstyle 0.79$}}
    			child{node{} edge from parent node[above, sloped] {$\scriptscriptstyle 0.18$}}
    			edge from parent[black] node[above, sloped] {$\scriptscriptstyle 0.09$}}
    		edge from parent[red] node[above, sloped] {$\scriptscriptstyle-0.85$}}
    	child {node[vertex, label = above:{$\scriptstyle 4$},]{}
    		child {node[vertex, label = above left:{$\scriptstyle 11$}]{}
    			child{node{} edge from parent node[above, sloped] {$\scriptscriptstyle 0.70$}}
    			child{node{} edge from parent node[above, sloped] {$\scriptscriptstyle 0.86$}}
    			child{node{} edge from parent node[above, sloped] {$\scriptscriptstyle 0.70$}}
    			edge from parent node[above, sloped] {$\scriptscriptstyle-0.20$}}
    		child {node[vertex, label = above right:{$\scriptstyle 12$}]{}
    			child{node{} edge from parent node[above, sloped] {$\scriptscriptstyle 0.99$}}
    			child{node{} edge from parent node[above, sloped] {$\scriptscriptstyle-0.38$}}
    			child{node{} edge from parent node[above, sloped] {$\scriptscriptstyle 0.35$}}
    			edge from parent node[above, sloped] {$\scriptscriptstyle-0.90$}}
    		child {node[vertex, label = above right:{$\scriptstyle 13$}]{}
    			child{node{} edge from parent node[above, sloped] {$\scriptscriptstyle 0.96$}}
    			child{node{} edge from parent node[above, sloped] {$\scriptscriptstyle-0.13$}}
    			child{node{} edge from parent node[above, sloped] {$\scriptscriptstyle-0.52$}}
    			edge from parent node[above, sloped] {$\scriptscriptstyle 0.97$}}
    		edge from parent node[above, sloped] {$\scriptscriptstyle 0.74$}};
    \end{tikzpicture}
    \end{minipage}
    
    \vspace{\baselineskip}
    
    \begin{minipage}[b]{\linewidth}
    \centering
    {\small (b) $\gammab_y$}
    
    \begin{tikzpicture}[>=stealth]
    \node(0)[vertex, label = above:{$\scriptstyle 1$}]{}
    	child {node[vertex, label = above:{$\scriptstyle 2$}]{}
    		child {node[vertex, label = above left:{$\scriptstyle 5$}]{}
    			child{node{} edge from parent node[above, sloped] {$\scriptscriptstyle-0.70$}}
    			child{node{} edge from parent node[above, sloped] {$\scriptscriptstyle-0.58$}}
    			child{node{} edge from parent node[above, sloped] {$\scriptscriptstyle 0.67$}}
    			edge from parent node[above, sloped] {$\scriptscriptstyle 0.56$}}
    		child {node[vertex, label = above right:{$\scriptstyle 6$}]{}
    			child{node{} edge from parent node[above, sloped] {$\scriptscriptstyle 0.17$}}
    			child{node{} edge from parent node[above, sloped] {$\scriptscriptstyle-0.29$}}
    			child{node{} edge from parent node[above, sloped] {$\scriptscriptstyle-0.95$}}
    			edge from parent node[above, sloped] {$\scriptscriptstyle-0.06$}}
    		child {node[vertex, label = above right:{$\scriptstyle 7$}]{}
    			child{node{} edge from parent node[above, sloped] {$\scriptscriptstyle 0.62$}}
    			child{node{} edge from parent node[above, sloped] {$\scriptscriptstyle 0.88$}}
    			child{node{} edge from parent node[above, sloped] {$\scriptscriptstyle 0.86$}}
    			edge from parent node[above, sloped] {$\scriptscriptstyle-0.10$}}
    		edge from parent node[above, sloped] {$\scriptscriptstyle-0.14$}}
    	child {node[vertex, label = above right:{$\scriptstyle 3$}]{}
    		child {node[vertex, label = above left:{$\scriptstyle 8$}]{}
    			child{node{} edge from parent node[above, sloped] {$\scriptscriptstyle-0.16$}}
    			child{node{} edge from parent node[above, sloped] {$\scriptscriptstyle-0.51$}}
    			child{node{} edge from parent node[above, sloped] {$\scriptscriptstyle-0.04$}}
    			edge from parent[black] node[above, sloped] {$\scriptscriptstyle-0.29$}}
    		child {node[vertex, label = above right:{$\scriptstyle 9$}]{}
    			child{node{} edge from parent node[above, sloped] {$\scriptscriptstyle 0.94$}}
    			child{node{} edge from parent node[above, sloped] {$\scriptscriptstyle-0.79$}}
    			child{node{} edge from parent node[above, sloped] {$\scriptscriptstyle-0.41$}}
    			edge from parent[black] node[above, sloped] {$\scriptscriptstyle 0.35$}}
    		child {node[vertex, label = above right:{$\scriptstyle 10$}]{}
    			child{node{} edge from parent node[above, sloped] {$\scriptscriptstyle 0.21$}}
    			child{node{} edge from parent node[above, sloped] {$\scriptscriptstyle 0.79$}}
    			child{node{} edge from parent node[above, sloped] {$\scriptscriptstyle 0.18$}}
    			edge from parent[black] node[above, sloped] {$\scriptscriptstyle 0.09$}}
    		edge from parent[red] node[above, sloped] {$\scriptscriptstyle-0.65$}}
    	child {node[vertex, label = above:{$\scriptstyle 4$},]{}
    		child {node[vertex, label = above left:{$\scriptstyle 11$}]{}
    			child{node{} edge from parent node[above, sloped] {$\scriptscriptstyle 0.70$}}
    			child{node{} edge from parent node[above, sloped] {$\scriptscriptstyle 0.86$}}
    			child{node{} edge from parent node[above, sloped] {$\scriptscriptstyle 0.70$}}
    			edge from parent node[above, sloped] {$\scriptscriptstyle-0.20$}}
    		child {node[vertex, label = above right:{$\scriptstyle 12$}]{}
    			child{node{} edge from parent node[above, sloped] {$\scriptscriptstyle 0.99$}}
    			child{node{} edge from parent node[above, sloped] {$\scriptscriptstyle-0.38$}}
    			child{node{} edge from parent node[above, sloped] {$\scriptscriptstyle 0.35$}}
    			edge from parent node[above, sloped] {$\scriptscriptstyle-0.90$}}
    		child {node[vertex, label = above right:{$\scriptstyle 13$}]{}
    			child{node{} edge from parent node[above, sloped] {$\scriptscriptstyle 0.96$}}
    			child{node{} edge from parent node[above, sloped] {$\scriptscriptstyle-0.13$}}
    			child{node{} edge from parent node[above, sloped] {$\scriptscriptstyle-0.52$}}
    			edge from parent node[above, sloped] {$\scriptscriptstyle 0.97$}}
    		edge from parent node[above, sloped] {$\scriptscriptstyle 0.74$}};
    \path (0) edge[bend right = 1000] node[draw = none, anchor=west, pos = 0.6] {$\scriptscriptstyle 0.20$} (0-2-2);
    \path (0-2) edge[bend left = 10] node[draw=none, anchor=south] {$\scriptscriptstyle-0.20$} (0-3);
    \end{tikzpicture}
    \end{minipage}
    
    \vspace{\baselineskip}
    
    \begin{minipage}[b]{\linewidth}
    \centering
    {\small (c) difference}
    
    \begin{tikzpicture}[>=stealth]
    \node(0)[vertex, label = above:{$\scriptstyle 1$}]{}
    	child {node[vertex, label = above:{$\scriptstyle 2$}]{}
    		child {node[vertex, label = above left:{$\scriptstyle 5$}]{}
    			edge from parent[draw=none] node[above, sloped] {}}
    		child {node[vertex, label = above right:{$\scriptstyle 6$}]{}
    			edge from parent[draw=none] node[above, sloped] {}}
    		child {node[vertex, label = above right:{$\scriptstyle 7$}]{}
    			edge from parent[draw=none] node[above, sloped] {}}
    		edge from parent node[above, sloped] {$\scriptscriptstyle-0.4$}}
    	child {node[vertex, label = above right:{$\scriptstyle 3$}]{}
    		child {node[vertex, label = above left:{$\scriptstyle 8$}]{}
    			edge from parent[black] node[above, sloped] {$\scriptscriptstyle 0.4$}}
    		child {node[vertex, label = above right:{$\scriptstyle 9$}]{}
    			edge from parent[draw=none] node[above, sloped] {}}
    		child {node[vertex, label = above right:{$\scriptstyle 10$}]{}
    			edge from parent[draw=none] node[above, sloped] {}}
    		edge from parent[red] node[above, sloped] {$\scriptscriptstyle-0.2$}}
    	child {node[vertex, label = above:{$\scriptstyle 4$},]{}
    		child {node[vertex, label = above left:{$\scriptstyle 11$}]{}
    			edge from parent[draw=none] node[above, sloped] {}}
    		child {node[vertex, label = above right:{$\scriptstyle 12$}]{}
    			edge from parent[draw=none] node[above, sloped] {}}
    		child {node[vertex, label = above right:{$\scriptstyle 13$}]{}
    			edge from parent[draw=none] node[above, sloped] {}}
    		edge from parent[draw=none] node[above, sloped] {}};
    \path (0) edge[bend right = 1000] node[draw = none, anchor = west, pos = 0.6] {$\scriptscriptstyle-0.2$} (0-2-2);
    \path (0-2) edge[bend left = 10] node[draw = none, anchor = south] {$\scriptscriptstyle 0.2$} (0-3);
    \end{tikzpicture}
    \end{minipage}
\end{figure}

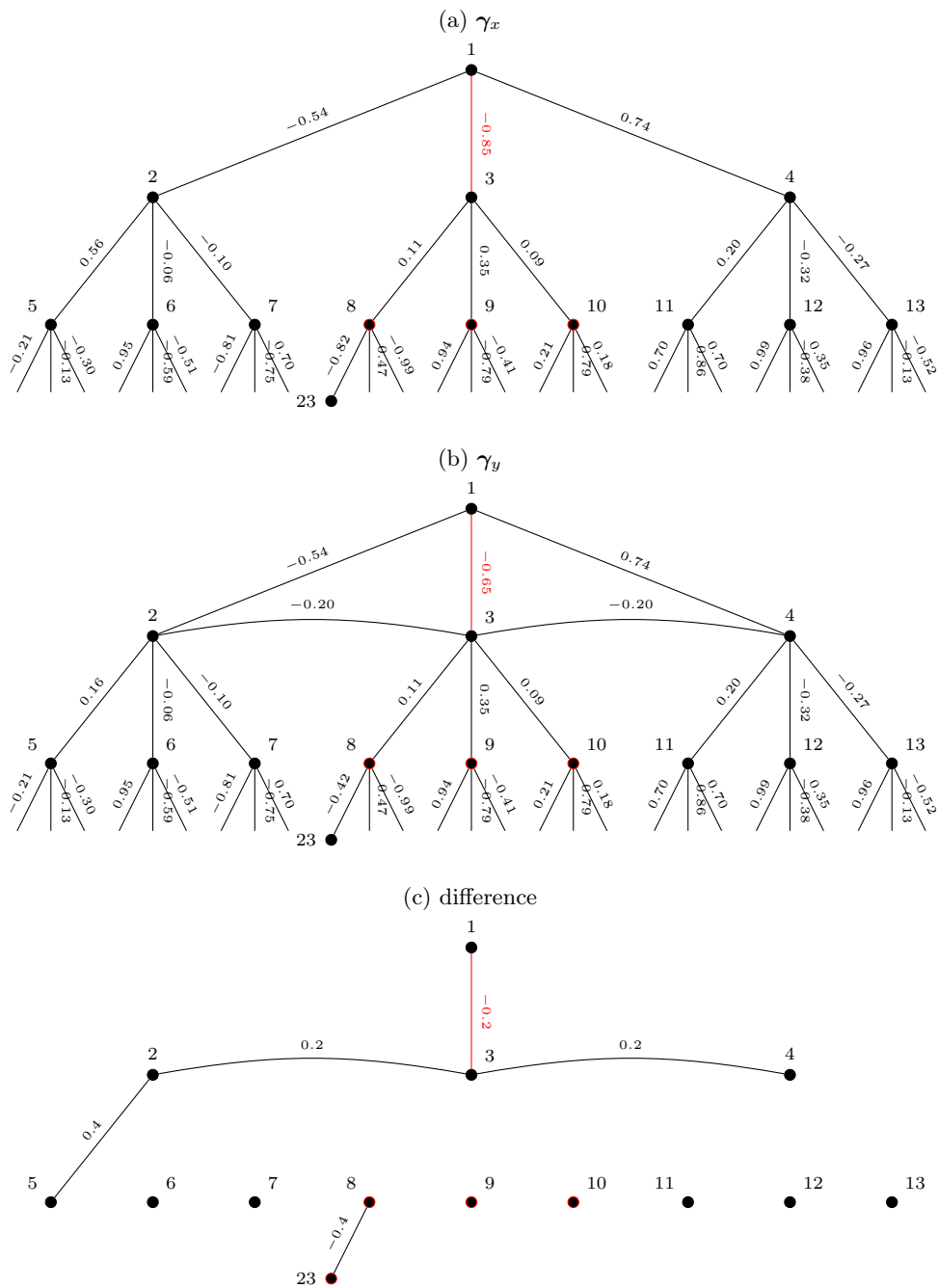
\begin{figure}[p]
    \caption{\label{fig:tree2}\textbf{Tree 2 pair, realized edge weights.} The displayed weights are the actual values used in the experiments. Excluding the difference graph, all the weights, including the ones not shown here, were generated i.i.d.~$\Unif(-1,1)$. The target of inference is marked in red.}
    \centering
    
    \vspace{\baselineskip}
    
    \begin{minipage}[b]{\linewidth}
    \centering
    {\small (a) $\gammab_x$}
    
    \begin{tikzpicture}[>=stealth]
    \node(0)[vertex, label = above:{$\scriptstyle 1$}]{}
    	child {node[vertex, label = above:{$\scriptstyle 2$}]{}
    		child {node[vertex, label = above left:{$\scriptstyle 5$}]{}
    			child{node{} edge from parent node[above, sloped] {$\scriptscriptstyle-0.21$}}
    			child{node{} edge from parent node[above, sloped] {$\scriptscriptstyle-0.13$}}
    			child{node{} edge from parent node[above, sloped] {$\scriptscriptstyle-0.30$}}
    			edge from parent node[above, sloped] {$\scriptscriptstyle 0.56$}}
    		child {node[vertex, label = above right:{$\scriptstyle 6$}]{}
    			child{node{} edge from parent node[above, sloped] {$\scriptscriptstyle 0.95$}}
    			child{node{} edge from parent node[above, sloped] {$\scriptscriptstyle-0.59$}}
    			child{node{} edge from parent node[above, sloped] {$\scriptscriptstyle-0.51$}}
    			edge from parent node[above, sloped] {$\scriptscriptstyle-0.06$}}
    		child {node[vertex, label = above right:{$\scriptstyle 7$}]{}
    			child{node{} edge from parent node[above, sloped] {$\scriptscriptstyle-0.81$}}
    			child{node{} edge from parent node[above, sloped] {$\scriptscriptstyle-0.75$}}
    			child{node{} edge from parent node[above, sloped] {$\scriptscriptstyle 0.70$}}
    			edge from parent node[above, sloped] {$\scriptscriptstyle-0.10$}}
    		edge from parent node[above, sloped] {$\scriptscriptstyle-0.54$}}
    	child {node[vertex, label = above right:{$\scriptstyle 3$}]{}
    		child {node[vertex, label = above left:{$\scriptstyle 8$}]{}
    			child{node[vertex, label = left:{$\scriptstyle 23$}]{}
    				edge from parent node[above, sloped] {$\scriptscriptstyle-0.82$}}
    			child{node{} edge from parent node[above, sloped] {$\scriptscriptstyle 0.47$}}
    			child{node{} edge from parent node[above, sloped] {$\scriptscriptstyle-0.99$}}
    			edge from parent[black] node[above, sloped] {$\scriptscriptstyle 0.11$}}
    		child {node[vertex, label = above right:{$\scriptstyle 9$}]{}
    			child{node{} edge from parent node[above, sloped] {$\scriptscriptstyle 0.94$}}
    			child{node{} edge from parent node[above, sloped] {$\scriptscriptstyle-0.79$}}
    			child{node{} edge from parent node[above, sloped]{$\scriptscriptstyle-0.41$}}
    			edge from parent[black] node[above, sloped] {$\scriptscriptstyle 0.35$}}
    		child {node[vertex, label = above right:{$\scriptstyle 10$}]{}
    			child{node{} edge from parent node[above, sloped] {$\scriptscriptstyle 0.21$}}
    			child{node{} edge from parent node[above, sloped] {$\scriptscriptstyle 0.79$}}
    			child{node{} edge from parent node[above, sloped] {$\scriptscriptstyle 0.18$}}
    			edge from parent[black] node[above, sloped] {$\scriptscriptstyle 0.09$}}
    		edge from parent[red] node[above, sloped] {$\scriptscriptstyle-0.85$}}
    	child {node[vertex, label = above:{$\scriptstyle 4$}]{}
    		child {node[vertex, label = above left:{$\scriptstyle 11$}]{}
    			child{node{} edge from parent node[above, sloped] {$\scriptscriptstyle 0.70$}}
    			child{node{} edge from parent node[above, sloped] {$\scriptscriptstyle 0.86$}}
    			child{node{} edge from parent node[above, sloped] {$\scriptscriptstyle 0.70$}}
    			edge from parent node[above, sloped] {$\scriptscriptstyle 0.20$}}
    		child {node[vertex, label = above right:{$\scriptstyle 12$}]{}
    			child{node{} edge from parent node[above, sloped] {$\scriptscriptstyle 0.99$}}
    			child{node{} edge from parent node[above, sloped] {$\scriptscriptstyle-0.38$}}
    			child{node{} edge from parent node[above, sloped] {$\scriptscriptstyle 0.35$}}
    			edge from parent node[above, sloped] {$\scriptscriptstyle-0.32$}}
    		child {node[vertex, label = above right:{$\scriptstyle 13$}]{}
    			child{node{} edge from parent node[above, sloped] {$\scriptscriptstyle 0.96$}}
    			child{node{} edge from parent node[above, sloped] {$\scriptscriptstyle-0.13$}}
    			child{node{} edge from parent node[above, sloped] {$\scriptscriptstyle-0.52$}}
    			edge from parent node[above, sloped] {$\scriptscriptstyle-0.27$}}
    		edge from parent node[above, sloped] {$\scriptscriptstyle 0.74$}};
    \end{tikzpicture}
    \end{minipage}
    
    \vspace{\baselineskip}
    
    \begin{minipage}[b]{\linewidth}
    \centering
    {\small (b) $\gammab_y$}
    
    \begin{tikzpicture}[>=stealth]
    \node(0)[vertex, label = above:{$\scriptstyle 1$}]{}
    	child {node[vertex, label = above:{$\scriptstyle 2$}]{}
    		child {node[vertex, label = above left:{$\scriptstyle 5$}]{}
    			child{node{} edge from parent node[above, sloped] {$\scriptscriptstyle-0.21$}}
    			child{node{} edge from parent node[above, sloped] {$\scriptscriptstyle-0.13$}}
    			child{node{} edge from parent node[above, sloped] {$\scriptscriptstyle-0.30$}}
    			edge from parent node[above, sloped] {$\scriptscriptstyle 0.16$}}
    		child {node[vertex, label = above right:{$\scriptstyle 6$}]{}
    			child{node{} edge from parent node[above, sloped] {$\scriptscriptstyle 0.95$}}
    			child{node{} edge from parent node[above, sloped] {$\scriptscriptstyle-0.59$}}
    			child{node{} edge from parent node[above, sloped] {$\scriptscriptstyle-0.51$}}
    			edge from parent node[above, sloped] {$\scriptscriptstyle-0.06$}}
    		child {node[vertex, label = above right:{$\scriptstyle 7$}]{}
    			child{node{} edge from parent node[above, sloped] {$\scriptscriptstyle-0.81$}}
    			child{node{} edge from parent node[above, sloped] {$\scriptscriptstyle-0.75$}}
    			child{node{} edge from parent node[above, sloped] {$\scriptscriptstyle 0.70$}}
    			edge from parent node[above, sloped] {$\scriptscriptstyle-0.10$}}
    		edge from parent node[above, sloped] {$\scriptscriptstyle-0.54$}}
    	child {node[vertex, label = above right:{$\scriptstyle 3$}]{}
    		child {node[vertex, label = above left:{$\scriptstyle 8$}]{}
    			child{node[vertex, label = left:{$\scriptstyle 23$}]{}
    				edge from parent node[above, sloped] {$\scriptscriptstyle-0.42$}}
    			child{node{} edge from parent node[above, sloped] {$\scriptscriptstyle 0.47$}}
    			child{node{} edge from parent node[above, sloped] {$\scriptscriptstyle-0.99$}}
    			edge from parent[black] node[above, sloped] {$\scriptscriptstyle 0.11$}}
    		child {node[vertex, label = above right:{$\scriptstyle 9$}]{}
    			child{node{} edge from parent node[above, sloped] {$\scriptscriptstyle 0.94$}}
    			child{node{} edge from parent node[above, sloped] {$\scriptscriptstyle-0.79$}}
    			child{node{} edge from parent node[above, sloped]{$\scriptscriptstyle-0.41$}}
    			edge from parent[black] node[above, sloped] {$\scriptscriptstyle 0.35$}}
    		child {node[vertex, label = above right:{$\scriptstyle 10$}]{}
    			child{node{} edge from parent node[above, sloped] {$\scriptscriptstyle 0.21$}}
    			child{node{} edge from parent node[above, sloped] {$\scriptscriptstyle 0.79$}}
    			child{node{} edge from parent node[above, sloped] {$\scriptscriptstyle 0.18$}}
    			edge from parent[black] node[above, sloped] {$\scriptscriptstyle 0.09$}}
    		edge from parent[red] node[above, sloped] {$\scriptscriptstyle-0.65$}}
    	child {node[vertex, label = above:{$\scriptstyle 4$}]{}
    		child {node[vertex, label = above left:{$\scriptstyle 11$}]{}
    			child{node{} edge from parent node[above, sloped] {$\scriptscriptstyle 0.70$}}
    			child{node{} edge from parent node[above, sloped] {$\scriptscriptstyle 0.86$}}
    			child{node{} edge from parent node[above, sloped] {$\scriptscriptstyle 0.70$}}
    			edge from parent node[above, sloped] {$\scriptscriptstyle 0.20$}}
    		child {node[vertex, label = above right:{$\scriptstyle 12$}]{}
    			child{node{} edge from parent node[above, sloped] {$\scriptscriptstyle 0.99$}}
    			child{node{} edge from parent node[above, sloped] {$\scriptscriptstyle-0.38$}}
    			child{node{} edge from parent node[above, sloped] {$\scriptscriptstyle 0.35$}}
    			edge from parent node[above, sloped] {$\scriptscriptstyle-0.32$}}
    		child {node[vertex, label = above right:{$\scriptstyle 13$}]{}
    			child{node{} edge from parent node[above, sloped] {$\scriptscriptstyle 0.96$}}
    			child{node{} edge from parent node[above, sloped] {$\scriptscriptstyle-0.13$}}
    			child{node{} edge from parent node[above, sloped] {$\scriptscriptstyle-0.52$}}
    			edge from parent node[above, sloped] {$\scriptscriptstyle-0.27$}}
    		edge from parent node[above, sloped] {$\scriptscriptstyle 0.74$}};
    \path (0-1) edge[bend left = 10] node[draw = none, anchor = south] {$\scriptscriptstyle-0.20$} (0-2);
    \path (0-2) edge[bend left = 10] node[draw = none, anchor = south] {$\scriptscriptstyle-0.20$} (0-3);
    \end{tikzpicture}
    \end{minipage}
    
    \vspace{\baselineskip}
    
    \begin{minipage}[b]{\linewidth}
    \centering
    {\small (c) difference}
    
    \begin{tikzpicture}[>=stealth]
    \node(0)[vertex, label = above:{$\scriptstyle 1$}]{}
    	child {node[vertex, label = above:{$\scriptstyle 2$}]{}
    		child {node[vertex, label = above left:{$\scriptstyle 5$}]{}
    			child{node{} edge from parent[draw=none] node[above, sloped] {}}
    			child{node{} edge from parent[draw=none] node[above, sloped] {}}
    			child{node{} edge from parent[draw=none] node[above, sloped] {}}
    			edge from parent node[above, sloped] {$\scriptscriptstyle 0.4$}}
    		child {node[vertex, label = above right:{$\scriptstyle 6$}]{}
    			child{node{} edge from parent[draw=none] node[above, sloped] {}}
    			child{node{} edge from parent[draw=none] node[above, sloped] {}}
    			child{node{} edge from parent[draw=none] node[above, sloped] {}}
    			edge from parent[draw=none] node[above, sloped] {}}
    		child {node[vertex, label = above right:{$\scriptstyle 7$}]{}
    			child{node{} edge from parent[draw=none] node[above, sloped] {}}
    			child{node{} edge from parent[draw=none] node[above, sloped] {}}
    			child{node{} edge from parent[draw=none] node[above, sloped] {}}
    			edge from parent[draw=none] node[above, sloped] {}}
    		edge from parent[draw=none] node[above, sloped] {}}
    	child {node[vertex, label = above right:{$\scriptstyle 3$}]{}
    		child {node[vertex, label = above left:{$\scriptstyle 8$}]{}
    			child{node[vertex, label = left:{$\scriptstyle 23$}]{}
    				edge from parent[black] node[above, sloped] {$\scriptscriptstyle-0.4$}}
    			child{node{} edge from parent[draw=none] node[above, sloped] {}}
    			child{node{} edge from parent[draw=none] node[above, sloped] {}}
    			edge from parent[draw=none] node[above, sloped] {}}
    		child {node[vertex, label = above right:{$\scriptstyle 9$}]{}
    			child{node{} edge from parent[draw=none] node[above, sloped] {}}
    			child{node{} edge from parent[draw=none] node[above, sloped] {}}
    			child{node{} edge from parent[draw=none] node[above, sloped]{}}
    			edge from parent[draw=none]node[above, sloped] {}}
    		child {node[vertex, label = above right:{$\scriptstyle 10$}]{}
    			child{node{} edge from parent[draw=none] node[above, sloped] {}}
    			child{node{} edge from parent[draw=none] node[above, sloped] {}}
    			child{node{} edge from parent[draw=none] node[above, sloped] {}}
    			edge from parent[draw=none]node[above, sloped] {}}
    		edge from parent[red] node[above, sloped] {$\scriptscriptstyle-0.2$}}
    	child {node[vertex, label = above:{$\scriptstyle 4$}]{}
    		child {node[vertex, label = above left:{$\scriptstyle 11$}]{}
    			child{node{} edge from parent[draw=none] node[above, sloped] {}}
    			child{node{} edge from parent[draw=none] node[above, sloped] {}}
    			child{node{} edge from parent[draw=none] node[above, sloped] {}}
    			edge from parent[draw=none] node[above, sloped] {}}
    		child {node[vertex, label = above right:{$\scriptstyle 12$}]{}
    			child{node{} edge from parent[draw=none] node[above, sloped] {}}
    			child{node{} edge from parent[draw=none] node[above, sloped] {}}
    			child{node{} edge from parent[draw=none] node[above, sloped] {}}
    			edge from parent[draw=none] node[above, sloped] {}}
    		child {node[vertex, label = above right:{$\scriptstyle 13$}]{}
    			child{node{} edge from parent[draw=none] node[above, sloped] {}}
    			child{node{} edge from parent[draw=none] node[above, sloped] {}}
    			child{node{} edge from parent[draw=none] node[above, sloped] {}}
    			edge from parent[draw=none] node[above, sloped] {}}
    		edge from parent[draw=none] node[above, sloped] {}};
    \path (0-1) edge[bend left = 10] node[draw = none, anchor = south] {$\scriptscriptstyle 0.2$} (0-2);
    \path (0-2) edge[bend left = 10] node[draw = none, anchor = south] {$\scriptscriptstyle 0.2$} (0-3);
    \end{tikzpicture}
    \end{minipage}
\end{figure}
\FloatBarrier

The advantage of our method is most clearly illustrated in settings in which initial sparse estimates are likely to miss parts of the support that are nonetheless important for inference. That is to say, both SparKLIE+ and the na\"{i}ve procedure described in Appendix~\ref{supp:sec5:competing} are expected to do well when the support is recovered with high probability. However, when this is no longer true, only SparKLIE+ will perform well.

We constructed eight graph pairs to highlight this difference. See Figures~\ref{fig:chain1} to \ref{fig:tree2}. We have four designs, and each design has a 25-node version and a 50-node version. The designs are labeled as Chain~1, Chain~2, Tree~1, and Tree~2, where the first part refers to the structure of $\gammab_x$ and the second, the type of modification used to obtain $\gammab_y$ from $\gammab_x$.

The edge weights were picked in the following manner. First, the weights for $\gammab_x$ were generated i.i.d.~$\Unif(-1,1)$. Next, $\gammab_y$ was obtained from $\gammab_x$ by modifying five edges. Thus, the difference graph always contained \emph{five} nonzero edges.

Each design has a fixed inference target, a.k.a.~the edge of interest. For Chain~1and Chain~2, this was always the edge $(5,6)$. For Tree~1 and Tree~2, this was always the edge $(1,3)$. The magnitude was always fixed at 0.2. By contrast, two of the nuisance edges had magnitude 0.4, while the two others had magnitude 0.2. The signs were chosen so that the none of the edge weights had magnitudes exceeding 1.

For each design, we first generated a 25-node version, and then embedded the 25-node version into a 50-node one.

\subsection{Data generation}

In Experiments~1 -- 5, the data were generated as i.i.d.~draws from an Ising model with zero node potentials.
A Gibbs sampler \citep{Geman1984Stochastic} was used. For Experiments~1, 2, and 5 burn-in was 3000 and thinning was 1000. For Experiments~3 and 4, burn-in was 3000 and thinning was 2000.

\subsection{\label{supp:sec5:res:exper1}Additional figures and tables for Experiment~1}

\begin{figure}[p]
    \caption{\label{fig:qqplots_chain1}\textbf{The quality of Gaussian approximation for Chain 1 pair.}
    We plot the distributions of the na\"{i}ve re-fitted (left), the SparKLIE+1 (middle), and the SparKLIE+2 (right) estimators after studentization (i.e., standardizing by the standard error estimate \eqref{eq:varest}), first as a Normal Q-Q plot (top) and then as a histogram (bottom). In each Q-Q plot, the distribution of the oracle estimator after studentization (gray dots) is also provided for easy comparison. The orange curves in each histogram represents the density of $\Ncal(0,1)$ and is provided for reference.}
    \centering
    
    \vspace{\baselineskip}

    \begin{minipage}[b]{\linewidth}
    \centering
    {\small (a) 25 nodes}

    \includegraphics[height=0.40\textheight]{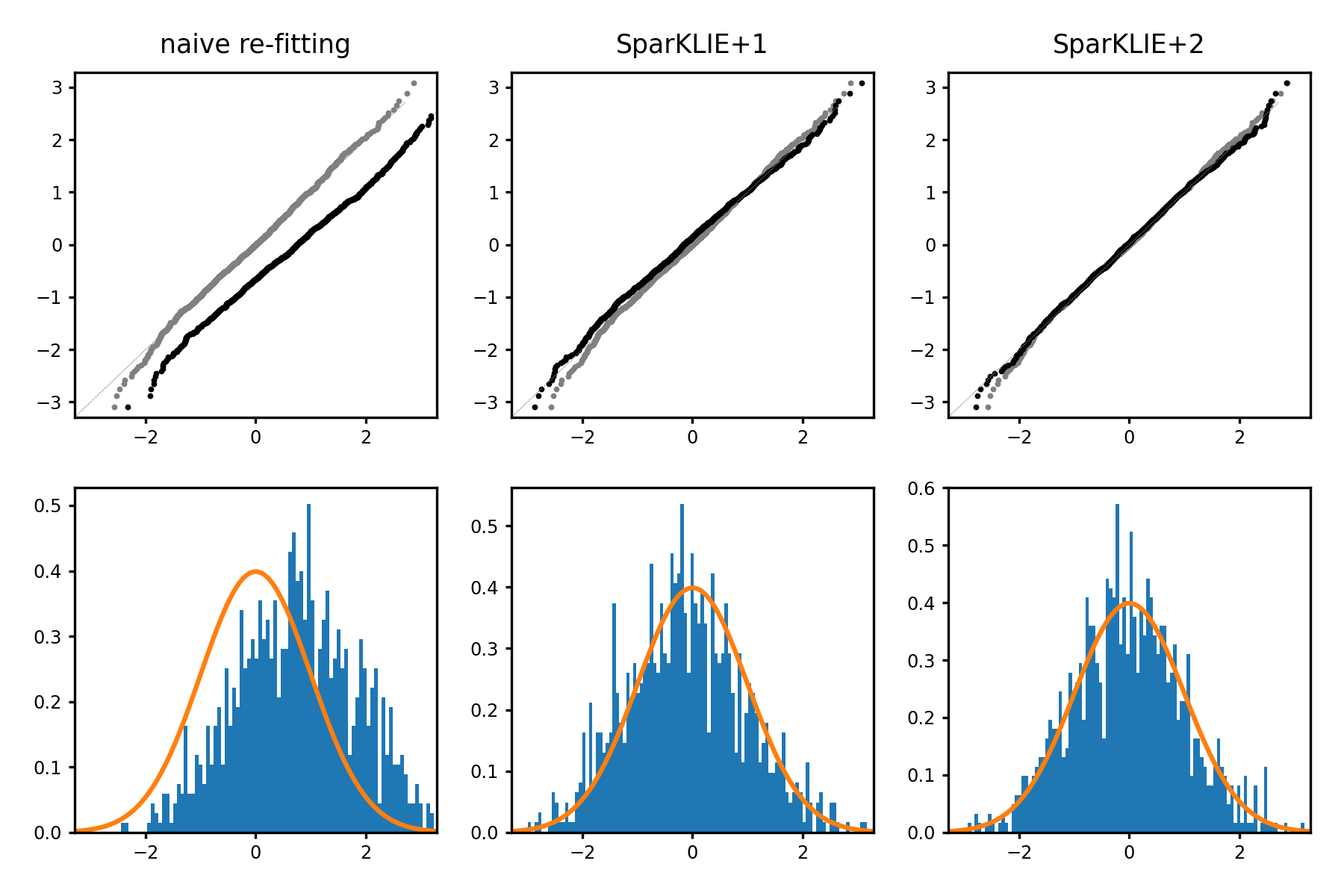}
    \end{minipage}
    
    \vskip\baselineskip

    \begin{minipage}[b]{\linewidth}
    \centering
    {\small (b) 50 nodes}

    \includegraphics[height=0.40\textheight]{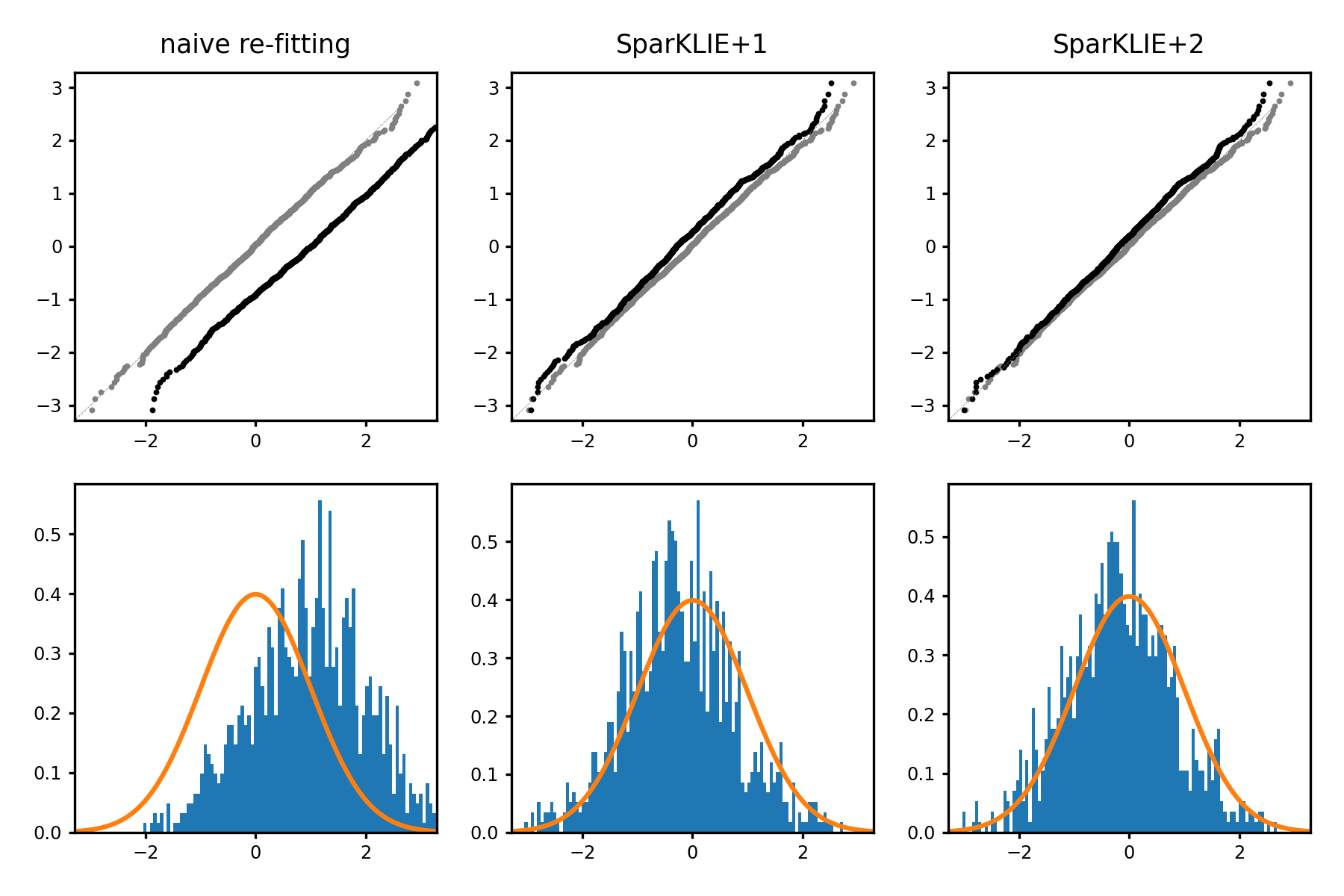}
    \end{minipage}
\end{figure}

\begin{figure}[p]
    \caption{\label{fig:qqplots_chain2}\textbf{The quality of Gaussian approximation for Chain 2 pair.}
    We plot the distributions of the na\"{i}ve re-fitted (left), the SparKLIE+1 (middle), and the SparKLIE+2 (right) estimators after studentization (i.e., standardizing by the standard error estimate \eqref{eq:varest}), first as a Normal Q-Q plot (top) and then as a histogram (bottom). In each Q-Q plot, the distribution of the oracle estimator after studentization (gray dots) is also provided for easy comparison. The orange curves in each histogram represents the density of $\Ncal(0,1)$ and is provided for reference.}
    \centering
    
    \vspace{\baselineskip}

    \begin{minipage}[b]{\linewidth}
    \centering
    {\small (a) 25 nodes}

    \includegraphics[height=0.40\textheight]{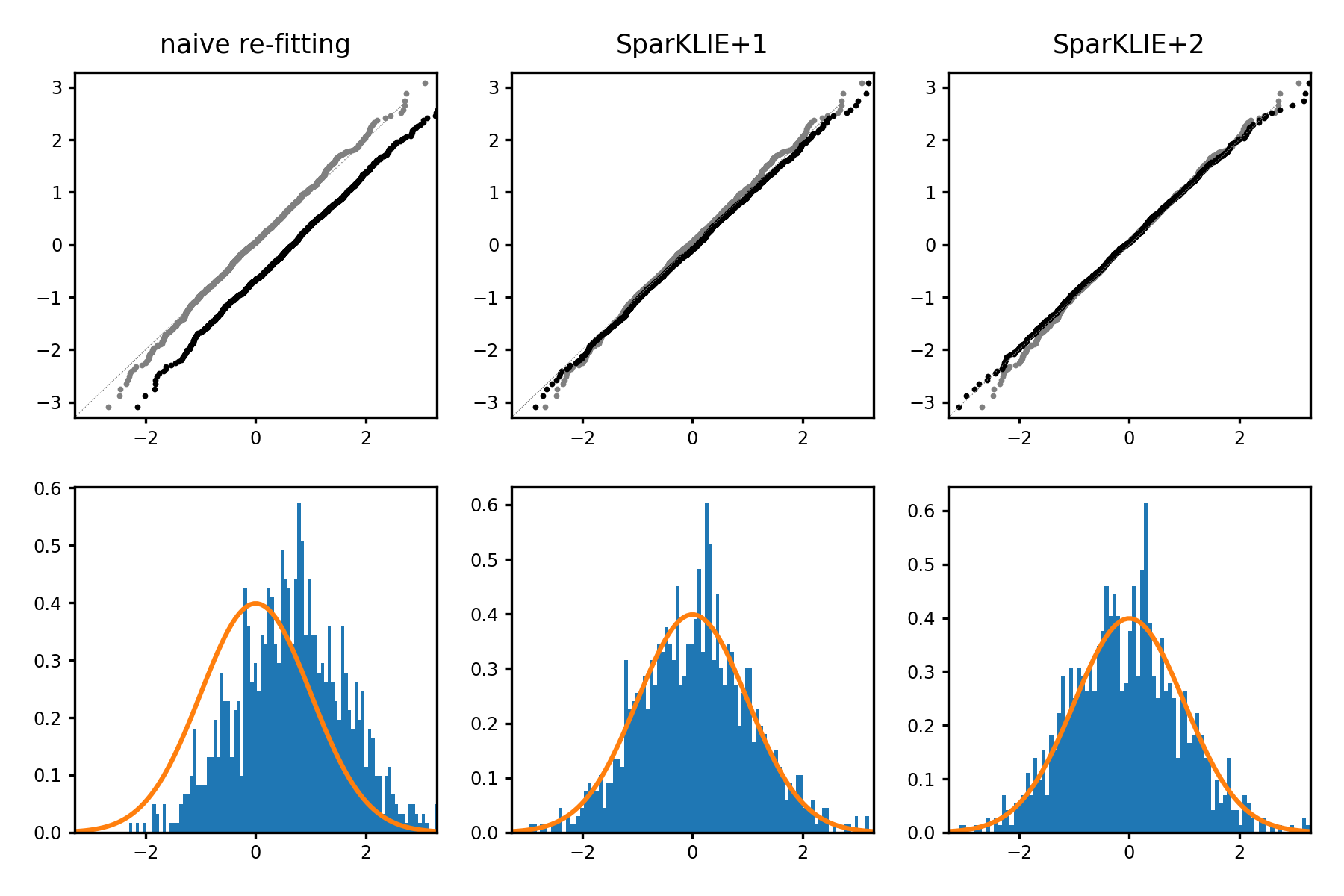}
    \end{minipage}
    
    \vskip\baselineskip

    \begin{minipage}[b]{\linewidth}
    \centering
    {\small (b) 50 nodes}

    \includegraphics[height=0.40\textheight]{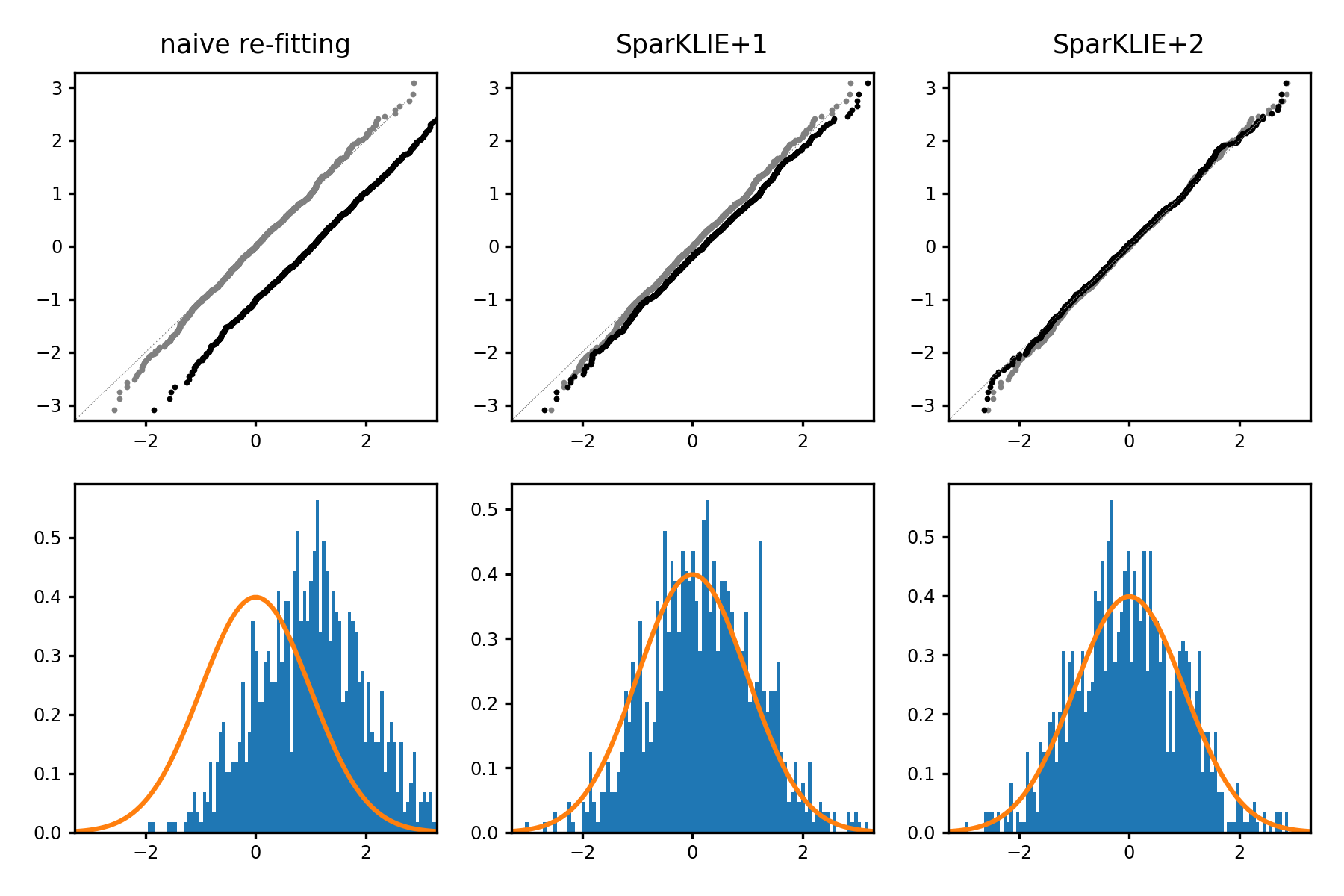}
    \end{minipage}
\end{figure}

\begin{figure}[p]
    \caption{\label{fig:qqplots_tree1}\textbf{The quality of Gaussian approximation for Tree 1 pair.}
    We plot the distributions of the na\"{i}ve re-fitted (left), the SparKLIE+1 (middle), and the SparKLIE+2 (right) estimators after studentization (i.e., standardizing by the standard error estimate \eqref{eq:varest}), first as a Normal Q-Q plot (top) and then as a histogram (bottom). In each Q-Q plot, the distribution of the oracle estimator after studentization (gray dots) is also provided for easy comparison. The orange curves in each histogram represents the density of $\Ncal(0,1)$ and is provided for reference.}
    \centering
    
    \vspace{\baselineskip}

    \begin{minipage}[b]{\linewidth}
    \centering
    {\small (a) 25 nodes}

    \includegraphics[height=0.40\textheight]{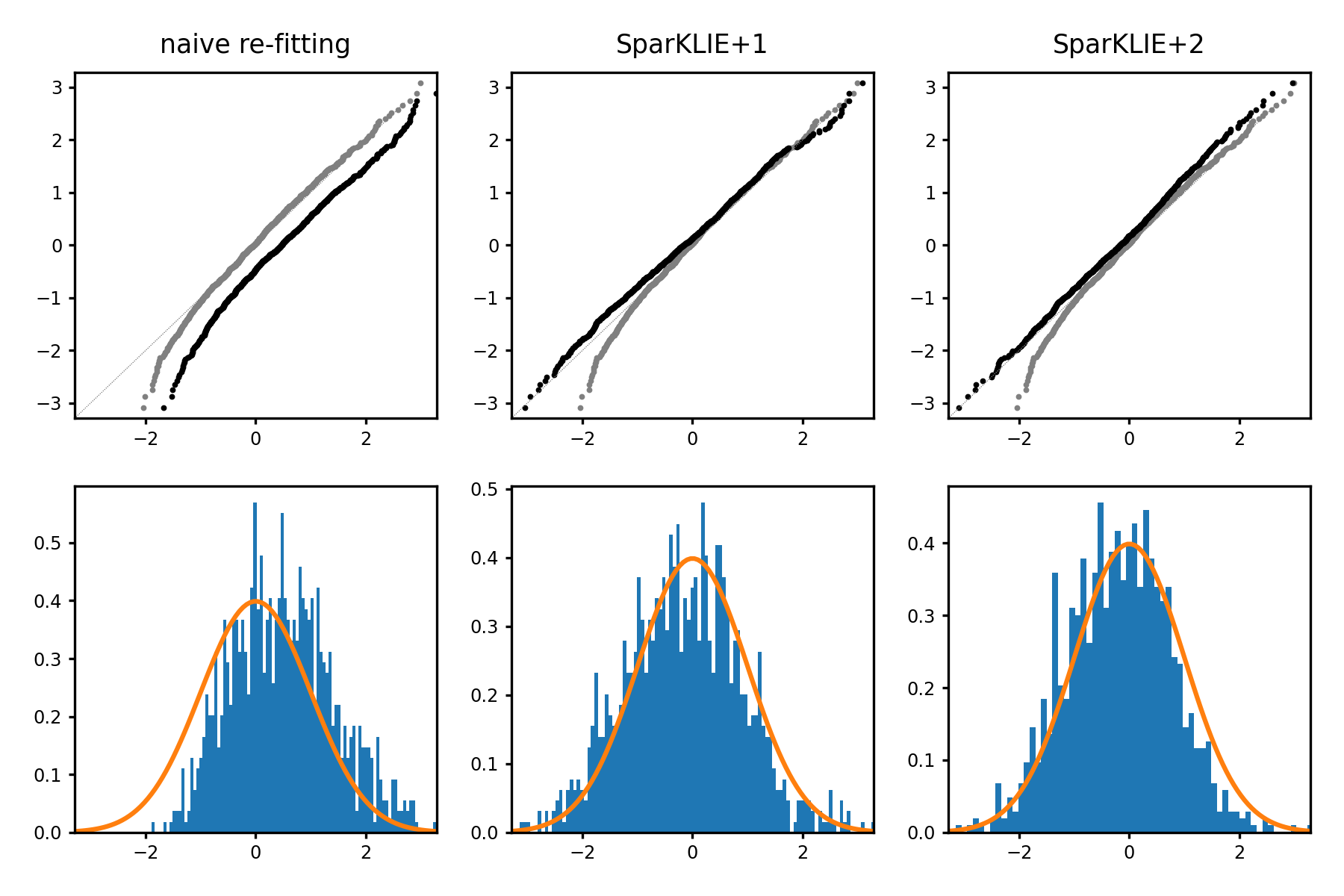}
    \end{minipage}
    
    \vskip\baselineskip

    \begin{minipage}[b]{\linewidth}
    \centering
    {\small (b) 50 nodes}

    \includegraphics[height=0.40\textheight]{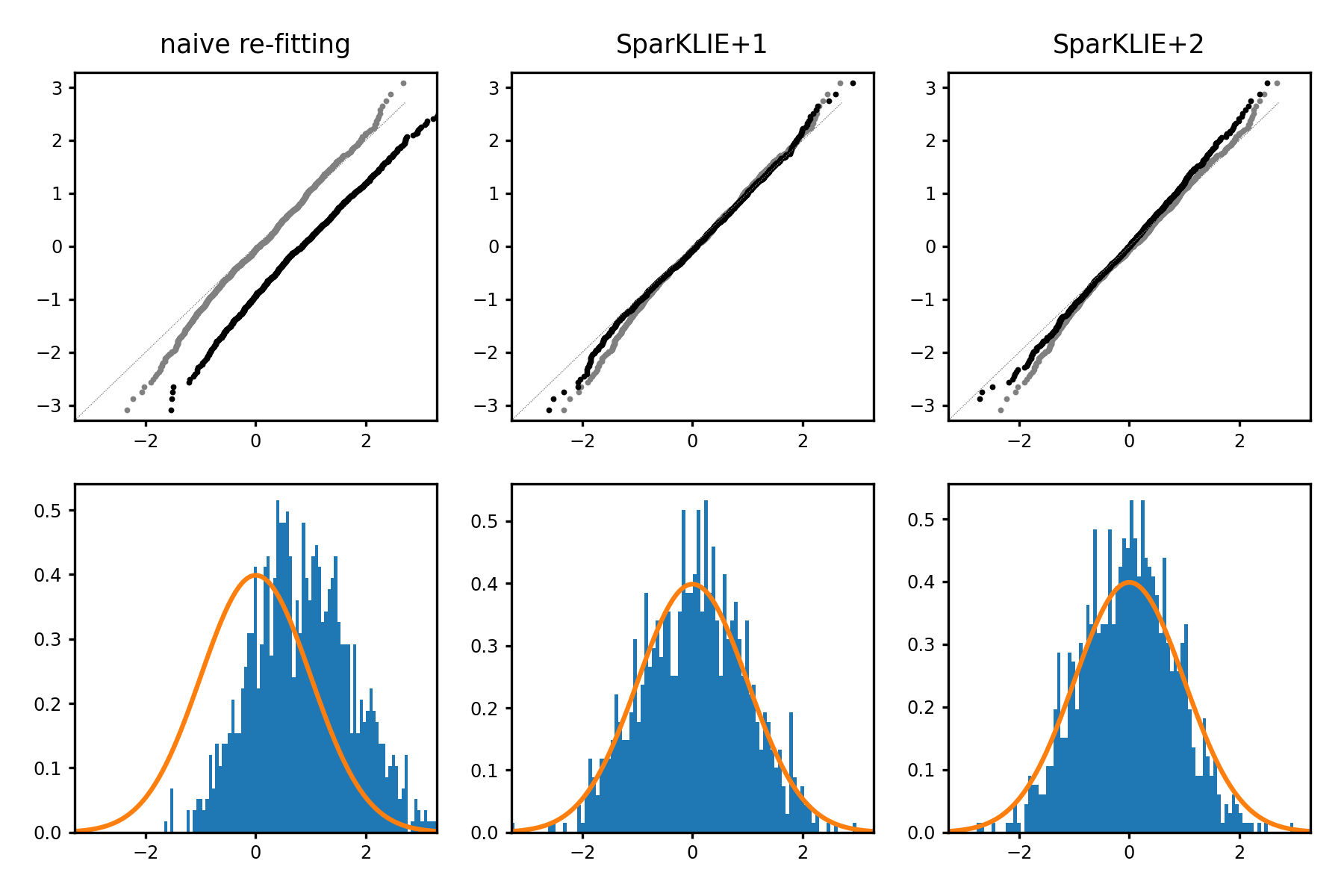}
    \end{minipage}
\end{figure}

\begin{figure}[p]
    \caption{\label{fig:qqplots_tree2}\textbf{The quality of Gaussian approximation for Tree 2 pair.}
    We plot the distributions of the na\"{i}ve re-fitted (left), the SparKLIE+1 (middle), and the SparKLIE+2 (right) estimators after studentization (i.e., standardizing by the standard error estimate \eqref{eq:varest}), first as a Normal Q-Q plot (top) and then as a histogram (bottom). In each Q-Q plot, the distribution of the oracle estimator after studentization (gray dots) is also provided for easy comparison. The orange curves in each histogram represents the density of $\Ncal(0,1)$ and is provided for reference.}
    \centering
    
    \vspace{\baselineskip}

    \begin{minipage}[b]{\linewidth}
    \centering
    {\small (a) 25 nodes}

    \includegraphics[height=0.40\textheight]{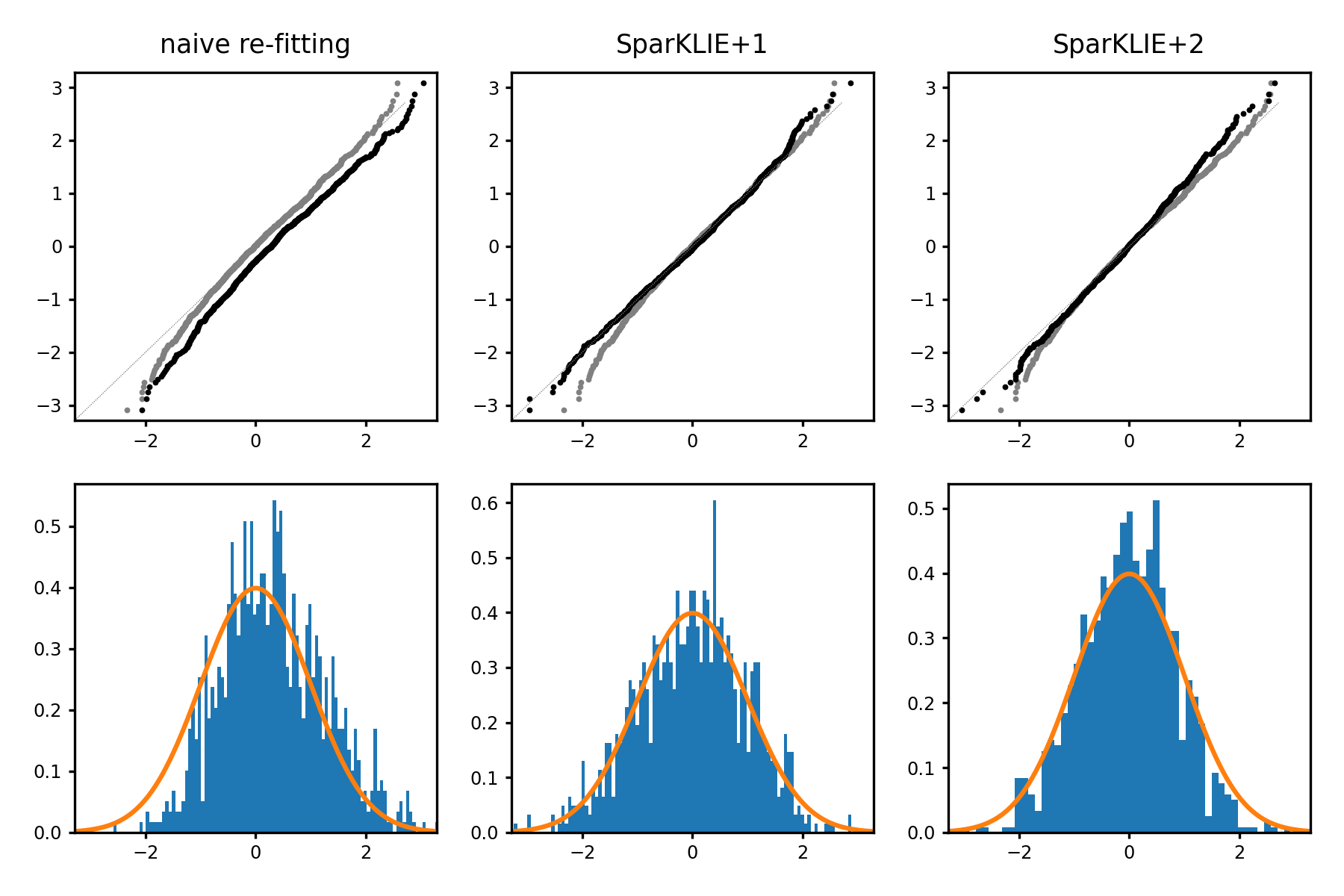}
    \end{minipage}
    
    \vskip\baselineskip

    \begin{minipage}[b]{\linewidth}
    \centering
    {\small (b) 50 nodes}

    \includegraphics[height=0.40\textheight]{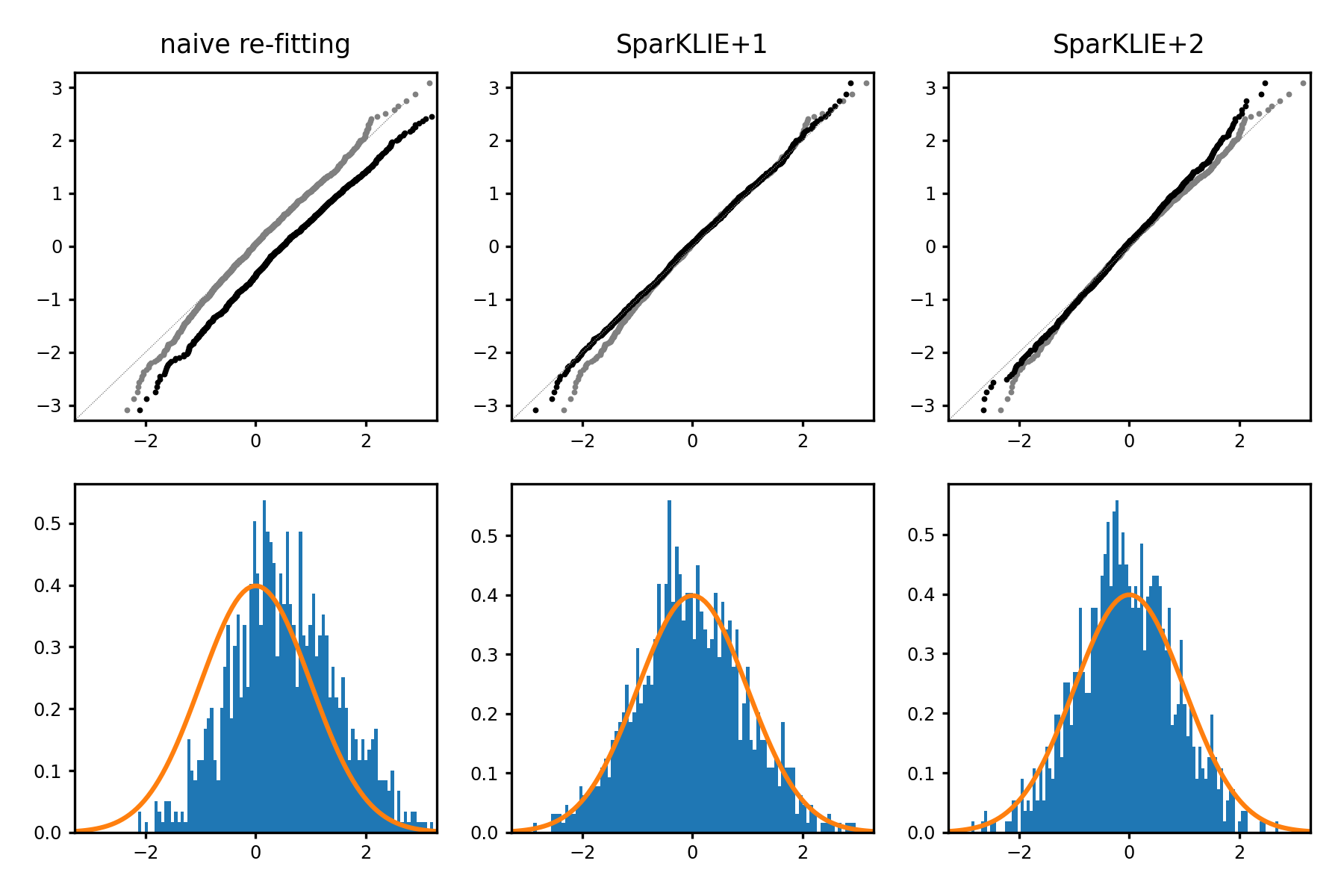}
    \end{minipage}
\end{figure}

\begin{table}
\caption{\textbf{Empirical bias $\times 10^2$.} We estimate the bias $\EE[\hat\theta_k - \theta_k^*] \times 10^2$, where $\hat\theta_k$ is either the oracle, the na\"{i}ve re-fitted, the SparKLIE+1, or the SparKLIE+2 estimator. The results are averages over 1000 independent replications.}
\label{table:bias}
\centering
\fbox{%
\begin{tabular}{c *{4}{c} *{4}{r}}
\multirow{2}{*}{$\gammab_x$} & \multirow{2}{*}{$\gammab_y$} & \multirow{2}{*}{$m$} & \multirow{2}{*}{$n_x$} & \multirow{2}{*}{$n_y$} & \multirow{2}{*}{oracle} & \multirow{2}{*}{na\"{i}ve} & \multicolumn{2}{c}{\scriptsize SparKLIE}\\
& & & & & & & {\scriptsize +1} & {\scriptsize +2}\\
\hline
\multirow{4}{*}{chain}
& \multirow{2}{*}{(1)} & 25 & 150 & 300 & -0.505 & 8.033 & -1.894 & -0.621\\
& & 50 & 300 & 600 & -0.360 & 7.692 & -2.301 & -1.673\\
& \multirow{2}{*}{(2)} & 25 & 150 & 300 & -0.819 & 6.920 & 0.526 & -1.013\\
& & 50 & 300 & 600 & -0.039 & 7.636 & 1.516 & -0.369\\
\multirow{4}{*}{\shortstack{ternary \\ tree}}
& \multirow{2}{*}{(1)} & 25 & 150 & 300 & -1.763 & 6.698 & -2.323 & -4.143\\
& & 50 & 300 & 600 & 0.256 & 8.975 & 0.875 & -0.539\\
& \multirow{2}{*}{(2)} & 25 & 150 & 300 & -0.770 & 3.803 & 1.168 & -0.587\\
& & 50 & 300 & 600 & -0.611 & 5.306 & -0.248 & -0.826
\end{tabular}}
\end{table}
\FloatBarrier

\section{Additional experiments}

\subsection{\label{supp:exper2}Experiment~2: Power of the normal-theory based test}

We study the power of the normal-theory based test with SparKLIE+1 and +2 estimators.
The parameters for this experiment were generated by first fixing $\gammab_y$ at the $\gammab_y$ of the 25-node Chain~1 pair from Experiment~1, and then obtaining 124 distinct graphs for $\gammab_x$ by varying the value of the change of interest over a grid $\delta = -0.75, -0.60, \dots, 0.75$ in one of the four settings described below:

\vskip\baselineskip

\begin{enumerate}[nosep, label={\it Setting \arabic*.}, wide=0pt]

\item (\textsc{None}) the edge of interest is the only edge that changes from $\gammab_y$ to $\gammab_x$,

\item (\textsc{Strong}) there are two additional strong changes of magnitude $0.4$,

\item (\textsc{Weak}) there are two additional weak changes of magnitude $0.2$, or

\item (\textsc{Mixed}) there are both weak and strong changes.

\end{enumerate}

\vskip\baselineskip

\noindent See Figures~\ref{fig:power_none} -- \ref{fig:power_mixed} for illustration.

\begin{figure}[H]
    \caption{\label{fig:power_none}{\bfseries{\scshape None}, realized edge weights.} $\gammab_y$ is the same as the $\gammab_y$ of Chain 1 pair. $\gammab_x$ is obtained from $\gammab_y$ by applying the change $\delta$ to the target edge marked in red.}
    \centering
    
    \vspace{\baselineskip}
    
    \begin{minipage}[b]{\linewidth}
    \centering
    {\small (a) $\gammab_x$}
    
    \begin{tikzpicture}[shorten >=1pt]
    \foreach \x in {1, 2, ..., 10} \node[vertex, label=below:$\scriptstyle\x$] (G-\x) at (1.25*\x,0) {};
    \foreach \x in {1, 2, 3} \node[dotdotdot] (E-\x) at (13.75+0.2*\x,0) {};
    \draw (G-1) -- node[anchor=south] {$\scriptscriptstyle-0.54$} (G-2);
    \draw (G-2) -- node[anchor=south] {$\scriptscriptstyle-0.85$} (G-3);
    \draw (G-3) -- node[anchor=south] {$\scriptscriptstyle 0.74$} (G-4);
    \draw (G-4) -- node[anchor=south] {$\scriptscriptstyle 0.16$} (G-5);
    \draw[red] (G-5) -- node[anchor=south] {$\scriptscriptstyle 0.14 + \theta_k^*$} (G-6);
    \draw (G-6) -- node[anchor=south] {$\scriptscriptstyle 0.30$} (G-7);
    \draw (G-7) -- node[anchor=south] {$\scriptscriptstyle 0.11$} (G-8);
    \draw (G-8) -- node[anchor=south] {$\scriptscriptstyle 0.35$} (G-9);
    \draw (G-9) -- node[anchor=south] {$\scriptscriptstyle 0.09$} (G-10);
    \draw (G-10) -- node[anchor=south] {$\scriptscriptstyle 0.20$} (13.75,0);
    \path (G-4) edge[bend left = 75] node[anchor=south] {$\scriptscriptstyle-0.2$} (G-6);
    \path (G-5) edge[bend right = 75] node[anchor=south] {$\scriptscriptstyle-0.2$} (G-7);
    \end{tikzpicture}
    \end{minipage}
    
    \vspace{\baselineskip}
    
    \begin{minipage}[b]{\linewidth}
    \centering
    {\small (b) $\gammab_y$}
    
    \begin{tikzpicture}[shorten >=1pt]
    \foreach \x in {1, 2, ..., 10} \node[vertex, label=below:$\scriptstyle\x$] (G-\x) at (1.25*\x,0) {};
    \foreach \x in {1, 2, 3} \node[dotdotdot] (E-\x) at (13.75+0.2*\x,0) {};
    \draw (G-1) -- node[anchor=south] {$\scriptscriptstyle-0.54$} (G-2);
    \draw (G-2) -- node[anchor=south] {$\scriptscriptstyle-0.85$} (G-3);
    \draw (G-3) -- node[anchor=south] {$\scriptscriptstyle 0.74$} (G-4);
    \draw (G-4) -- node[anchor=south] {$\scriptscriptstyle 0.16$} (G-5);
    \draw[red] (G-5) -- node[anchor=south] {$\scriptscriptstyle 0.14$} (G-6);
    \draw (G-6) -- node[anchor=south] {$\scriptscriptstyle 0.30$} (G-7);
    \draw (G-7) -- node[anchor=south] {$\scriptscriptstyle 0.11$} (G-8);
    \draw (G-8) -- node[anchor=south] {$\scriptscriptstyle 0.35$} (G-9);
    \draw (G-9) -- node[anchor=south] {$\scriptscriptstyle 0.09$} (G-10);
    \draw (G-10) -- node[anchor=south] {$\scriptscriptstyle 0.20$} (13.75,0);
    \path (G-4) edge[bend left = 75] node[anchor=south] {$\scriptscriptstyle-0.2$} (G-6);
    \path (G-5) edge[bend right = 75] node[below] {$\scriptscriptstyle-0.2$} (G-7);
    \end{tikzpicture}
    \end{minipage}
    
    \begin{minipage}[t][0.1\textheight][c]{\linewidth}
    \centering
    {\small (c) difference}
    
    \vspace{\baselineskip}
    
    \begin{tikzpicture}[shorten >=1pt]
    \foreach \x in {1, 2, ..., 10} \node[vertex, label=below:$\scriptstyle\x$] (G-\x) at (1.25*\x,0) {};
    \foreach \x in {1, 2, 3} \node[dotdotdot] (E-\x) at (13.75+0.2*\x,0) {};
    \path[red] (G-5) edge node[anchor=south] {$\scriptscriptstyle \theta_k^*$} (G-6);
    \end{tikzpicture}
    \end{minipage}
\end{figure}
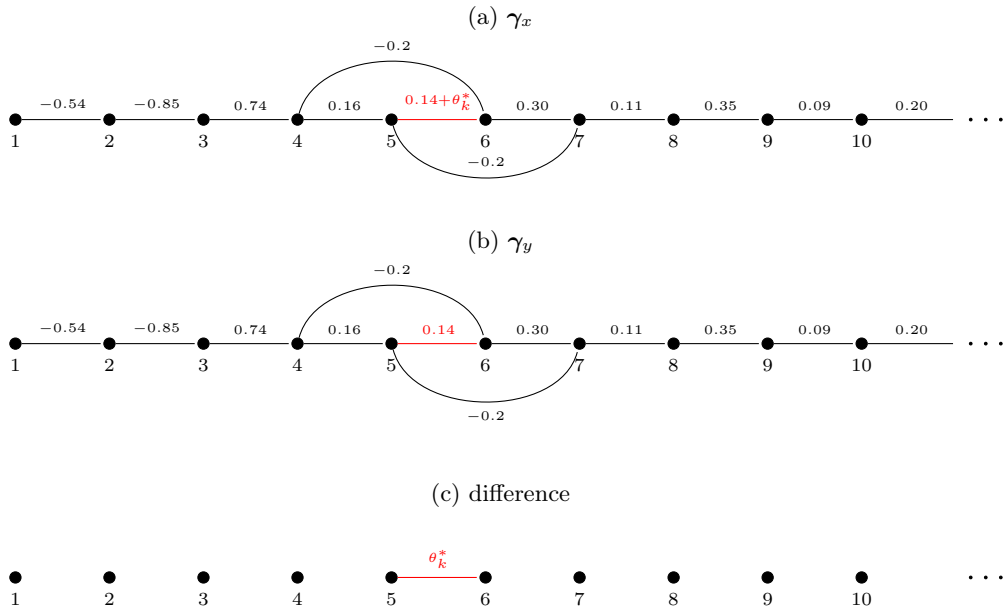

\begin{figure}[tp]
    \caption{\label{fig:power_strong}{\bfseries{\scshape Strong}, realized edge weights.} $\gammab_y$ is the same as the $\gammab_y$ of Chain 1 pair. $\gammab_x$ is obtained from $\gammab_y$ by applying the change $\delta$ to the target edge marked in red, as well as changes of magnitude $0.4$ to two neighboring edges.}
    \centering
    
    \vspace{\baselineskip}
    
    \begin{minipage}[b]{\linewidth}
    \centering
    {\small (a) $\gammab_x$}
    
    \begin{tikzpicture}[shorten >=1pt]
    \foreach \x in {1, 2, ..., 10} \node[vertex, label=below:$\scriptstyle\x$] (G-\x) at (1.25*\x,0) {};
    \foreach \x in {1, 2, 3} \node[dotdotdot] (E-\x) at (13.75+0.2*\x,0) {};
    \draw (G-1) -- node[anchor=south] {$\scriptscriptstyle-0.54$} (G-2);
    \draw (G-2) -- node[anchor=south] {$\scriptscriptstyle-0.85$} (G-3);
    \draw (G-3) -- node[anchor=south] {$\scriptscriptstyle 0.74$} (G-4);
    \draw (G-4) -- node[anchor=south] {$\scriptscriptstyle 0.56$} (G-5);
    \draw[red] (G-5) -- node[anchor=south] {$\scriptscriptstyle 0.14 + \theta_k^*$} (G-6);
    \draw (G-6) -- node[anchor=south] {$\scriptscriptstyle-0.10$} (G-7);
    \draw (G-7) -- node[anchor=south] {$\scriptscriptstyle 0.11$} (G-8);
    \draw (G-8) -- node[anchor=south] {$\scriptscriptstyle 0.35$} (G-9);
    \draw (G-9) -- node[anchor=south] {$\scriptscriptstyle 0.09$} (G-10);
    \draw (G-10) -- node[anchor=south] {$\scriptscriptstyle 0.20$} (13.75,0);
    \path (G-4) edge[bend left = 75] node[anchor=south] {$\scriptscriptstyle-0.2$} (G-6);
    \path (G-5) edge[bend right = 75] node[below] {$\scriptscriptstyle-0.2$} (G-7);
    \end{tikzpicture}
    \end{minipage}
    
    \vspace{\baselineskip}
    
    \begin{minipage}[b]{\linewidth}
    \centering
    {\small (b) $\gammab_y$}
    
    \begin{tikzpicture}[shorten >=1pt]
    \foreach \x in {1, 2, ..., 10} \node[vertex, label=below:$\scriptstyle\x$] (G-\x) at (1.25*\x,0) {};
    \foreach \x in {1, 2, 3} \node[dotdotdot] (E-\x) at (13.75+0.2*\x,0) {};
    \draw (G-1) -- node[anchor=south] {$\scriptscriptstyle-0.54$} (G-2);
    \draw (G-2) -- node[anchor=south] {$\scriptscriptstyle-0.85$} (G-3);
    \draw (G-3) -- node[anchor=south] {$\scriptscriptstyle 0.74$} (G-4);
    \draw (G-4) -- node[anchor=south] {$\scriptscriptstyle 0.16$} (G-5);
    \draw[red] (G-5) -- node[anchor=south] {$\scriptscriptstyle 0.14$} (G-6);
    \draw (G-6) -- node[anchor=south] {$\scriptscriptstyle 0.30$} (G-7);
    \draw (G-7) -- node[anchor=south] {$\scriptscriptstyle 0.11$} (G-8);
    \draw (G-8) -- node[anchor=south] {$\scriptscriptstyle 0.35$} (G-9);
    \draw (G-9) -- node[anchor=south] {$\scriptscriptstyle 0.09$} (G-10);
    \draw (G-10) -- node[anchor=south] {$\scriptscriptstyle 0.20$} (13.75,0);
    \path (G-4) edge[bend left = 75] node[anchor=south] {$\scriptscriptstyle-0.2$} (G-6);
    \path (G-5) edge[bend right = 75] node[below] {$\scriptscriptstyle-0.2$} (G-7);
    \end{tikzpicture}
    \end{minipage}
    
    \begin{minipage}[t][0.1\textheight][c]{\linewidth}
    \centering
    {\small (c) difference}
    
    \vspace{\baselineskip}
    
    \begin{tikzpicture}[shorten >=1pt]
    \foreach \x in {1, 2, ..., 10} \node[vertex, label=below:$\scriptstyle\x$] (G-\x) at (1.25*\x,0) {};
    \foreach \x in {1, 2, 3} \node[dotdotdot] (E-\x) at (13.75+0.2*\x,0) {};
    \draw (G-4) -- node[anchor=south] {$\scriptscriptstyle 0.4$} (G-5);
    \draw (G-6) -- node[anchor=south] {$\scriptscriptstyle-0.4$} (G-7);
    \draw[red] (G-5) -- node[anchor=south] {$\scriptscriptstyle \theta_k^*$} (G-6);
    \end{tikzpicture}
    \end{minipage}
\end{figure}
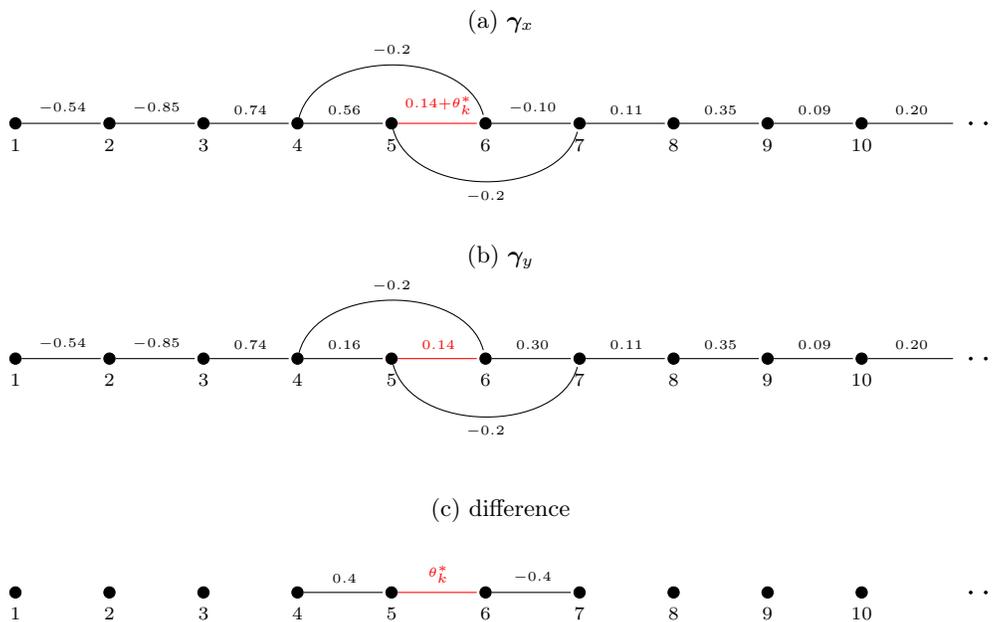

\begin{figure}[bp]
    \caption{\label{fig:power_weak}{\bfseries{\scshape Weak}, realized edge weights.} $\gammab_y$ is the same as the $\gammab_y$ of Chain 1 pair. $\gammab_x$ is obtained from $\gammab_y$ by applying the change $\delta$ to the target edge marked in red, as well as changes of magnitude $0.2$ to two neighboring edges.}
    \centering
    
    \vspace{\baselineskip}
    
    \begin{minipage}[t][0.1\textheight][c]{\linewidth}
    \centering
    {\small (a) $\gammab_x$}
    
    \vspace{\baselineskip}
    
    \begin{tikzpicture}[shorten >=1pt]
    \foreach \x in {1, 2, ..., 10} \node[vertex, label=below:$\scriptstyle\x$] (G-\x) at (1.25*\x,0) {};
    \foreach \x in {1, 2, 3} \node[dotdotdot] (E-\x) at (13.75+0.2*\x,0) {};
    \draw (G-1) -- node[anchor=south] {$\scriptscriptstyle-0.54$} (G-2);
    \draw (G-2) -- node[anchor=south] {$\scriptscriptstyle-0.85$} (G-3);
    \draw (G-3) -- node[anchor=south] {$\scriptscriptstyle 0.74$} (G-4);
    \draw (G-4) -- node[anchor=south] {$\scriptscriptstyle 0.16$} (G-5);
    \draw[red] (G-5) -- node[anchor=south] {$\scriptscriptstyle 0.14 + \theta_k^*$} (G-6);
    \draw (G-6) -- node[anchor=south] {$\scriptscriptstyle 0.30$} (G-7);
    \draw (G-7) -- node[anchor=south] {$\scriptscriptstyle 0.11$} (G-8);
    \draw (G-8) -- node[anchor=south] {$\scriptscriptstyle 0.35$} (G-9);
    \draw (G-9) -- node[anchor=south] {$\scriptscriptstyle 0.09$} (G-10);
    \draw (G-10) -- node[anchor=south] {$\scriptscriptstyle 0.20$} (13.75,0);
    \end{tikzpicture}
    \end{minipage}
    
    \vspace{\baselineskip}
    
    \begin{minipage}[b]{\linewidth}
    \centering
    {\small (b) $\gammab_y$}
    
    \begin{tikzpicture}[shorten >=1pt]
    \foreach \x in {1, 2, ..., 10} \node[vertex, label=below:$\scriptstyle\x$] (G-\x) at (1.25*\x,0) {};
    \foreach \x in {1, 2, 3} \node[dotdotdot] (E-\x) at (13.75+0.2*\x,0) {};
    \draw (G-1) -- node[anchor=south] {$\scriptscriptstyle-0.54$} (G-2);
    \draw (G-2) -- node[anchor=south] {$\scriptscriptstyle-0.85$} (G-3);
    \draw (G-3) -- node[anchor=south] {$\scriptscriptstyle 0.74$} (G-4);
    \draw (G-4) -- node[anchor=south] {$\scriptscriptstyle 0.16$} (G-5);
    \draw[red] (G-5) -- node[anchor=south] {$\scriptscriptstyle 0.14$} (G-6);
    \draw (G-6) -- node[anchor=south] {$\scriptscriptstyle 0.30$} (G-7);
    \draw (G-7) -- node[anchor=south] {$\scriptscriptstyle 0.11$} (G-8);
    \draw (G-8) -- node[anchor=south] {$\scriptscriptstyle 0.35$} (G-9);
    \draw (G-9) -- node[anchor=south] {$\scriptscriptstyle 0.09$} (G-10);
    \draw (G-10) -- node[anchor=south] {$\scriptscriptstyle 0.20$} (13.75,0);
    \path (G-4) edge[bend left = 75] node[anchor=south] {$\scriptscriptstyle-0.2$} (G-6);
    \path (G-5) edge[bend right = 75] node[below] {$\scriptscriptstyle-0.2$} (G-7);
    \end{tikzpicture}
    \end{minipage}
    
    \vspace{\baselineskip}
    
    \begin{minipage}[b]{\linewidth}
    \centering
    {\small (c) difference}
    
    \begin{tikzpicture}[shorten >=1pt]
    \foreach \x in {1, 2, ..., 10} \node[vertex, label=below:$\scriptstyle\x$] (G-\x) at (1.25*\x,0) {};
    \foreach \x in {1, 2, 3} \node[dotdotdot] (E-\x) at (13.75+0.2*\x,0) {};
    \draw[red] (G-5) -- node[anchor=south] {$\scriptscriptstyle \theta_k^*$} (G-6);
    \path (G-4) edge[bend left = 75] node[anchor=south] {$\scriptscriptstyle 0.2$} (G-6);
    \path (G-5) edge[bend right = 75] node[below] {$\scriptscriptstyle 0.2$} (G-7);
    \end{tikzpicture}
    \end{minipage}
\end{figure}
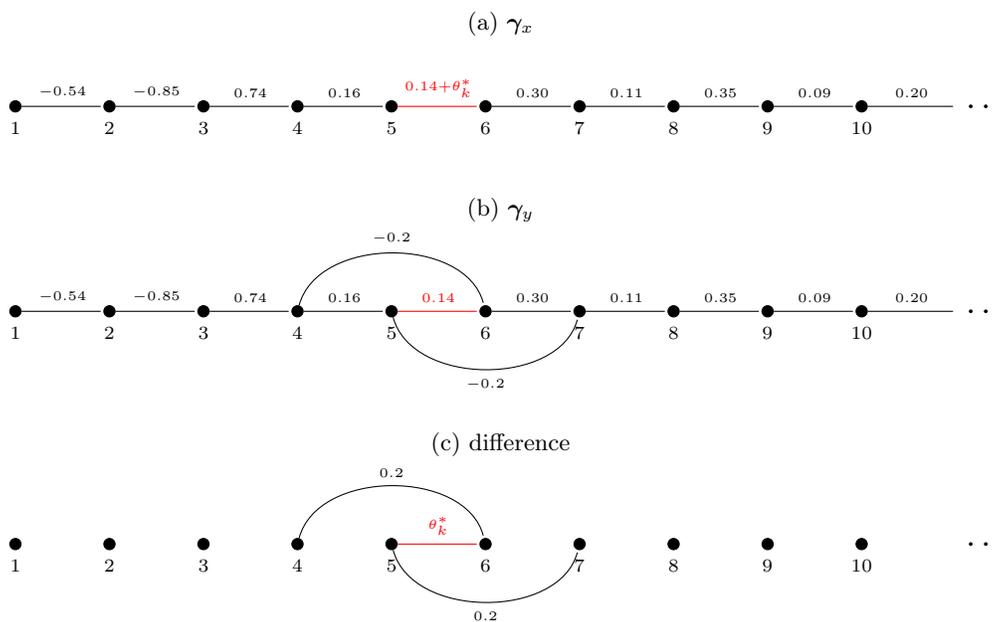

\begin{figure}[tp]
    \caption{\label{fig:power_mixed}{\bfseries{\scshape Mixed}, realized edge weights.} $\gammab_y$ is the same as the $\gammab_y$ of Chain 1 pair. $\gammab_x$ is obtained from $\gammab_y$ by applying the change $\delta$ to the target edge marked in red, as well as both types of nuisance changes in \textsc{Strong} and \textsc{Weak}.}
    \centering
    
    \vspace{\baselineskip}
    
    \begin{minipage}[t][0.1\textheight][c]{\linewidth}
    \centering
    {\small (a) $\gammab_x$}
    
    \vspace{\baselineskip}
    
    \begin{tikzpicture}[shorten >=1pt]
    \foreach \x in {1, 2, ..., 10} \node[vertex, label=below:$\scriptstyle\x$] (G-\x) at (1.25*\x,0) {};
    \foreach \x in {1, 2, 3} \node[dotdotdot] (E-\x) at (13.75+0.2*\x,0) {};
    \draw (G-1) -- node[anchor=south] {$\scriptscriptstyle-0.54$} (G-2);
    \draw (G-2) -- node[anchor=south] {$\scriptscriptstyle-0.85$} (G-3);
    \draw (G-3) -- node[anchor=south] {$\scriptscriptstyle 0.74$} (G-4);
    \draw (G-4) -- node[anchor=south] {$\scriptscriptstyle 0.56$} (G-5);
    \draw[red] (G-5) -- node[anchor=south] {$\scriptscriptstyle 0.14 + \theta_k^*$} (G-6);
    \draw (G-6) -- node[anchor=south] {$\scriptscriptstyle-0.10$} (G-7);
    \draw (G-7) -- node[anchor=south] {$\scriptscriptstyle 0.11$} (G-8);
    \draw (G-8) -- node[anchor=south] {$\scriptscriptstyle 0.35$} (G-9);
    \draw (G-9) -- node[anchor=south] {$\scriptscriptstyle 0.09$} (G-10);
    \draw (G-10) -- node[anchor=south] {$\scriptscriptstyle 0.20$} (13.75,0);
    \end{tikzpicture}
    \end{minipage}
    
    \vspace{\baselineskip}
    
    \begin{minipage}[b]{\linewidth}
    \centering
    {\small (b) $\gammab_y$}
    
    \begin{tikzpicture}[shorten >=1pt]
    \foreach \x in {1, 2, ..., 10} \node[vertex, label=below:$\scriptstyle\x$] (G-\x) at (1.25*\x,0) {};
    \foreach \x in {1, 2, 3} \node[dotdotdot] (E-\x) at (13.75+0.2*\x,0) {};
    \draw (G-1) -- node[anchor=south] {$\scriptscriptstyle-0.54$} (G-2);
    \draw (G-2) -- node[anchor=south] {$\scriptscriptstyle-0.85$} (G-3);
    \draw (G-3) -- node[anchor=south] {$\scriptscriptstyle 0.74$} (G-4);
    \draw (G-4) -- node[anchor=south] {$\scriptscriptstyle 0.16$} (G-5);
    \draw[red] (G-5) -- node[anchor=south] {$\scriptscriptstyle 0.14$} (G-6);
    \draw (G-6) -- node[anchor=south] {$\scriptscriptstyle 0.30$} (G-7);
    \draw (G-7) -- node[anchor=south] {$\scriptscriptstyle 0.11$} (G-8);
    \draw (G-8) -- node[anchor=south] {$\scriptscriptstyle 0.35$} (G-9);
    \draw (G-9) -- node[anchor=south] {$\scriptscriptstyle 0.09$} (G-10);
    \draw (G-10) -- node[anchor=south] {$\scriptscriptstyle 0.20$} (13.75,0);
    \path (G-4) edge[bend left = 75] node[anchor=south] {$\scriptscriptstyle-0.2$} (G-6);
    \path (G-5) edge[bend right = 75] node[below] {$\scriptscriptstyle-0.2$} (G-7);
    \end{tikzpicture}
    \end{minipage}
    
    \vspace{\baselineskip}
    
    \begin{minipage}[b]{\linewidth}
    \centering
    {\small (c) difference}
    
    \begin{tikzpicture}[shorten >=1pt]
    \foreach \x in {1, 2, ..., 10} \node[vertex, label=below:$\scriptstyle\x$] (G-\x) at (1.25*\x,0) {};
    \foreach \x in {1, 2, 3} \node[dotdotdot] (E-\x) at (13.75+0.2*\x,0) {};
    \draw (G-4) -- node[anchor=south] {$\scriptscriptstyle 0.4$} (G-5);
    \draw[red] (G-5) -- node[anchor=south] {$\scriptscriptstyle \theta_k^*$} (G-6);
    \draw (G-6) -- node[anchor=south] {$\scriptscriptstyle-0.4$} (G-7);
    \path (G-4) edge[bend left = 75] node[anchor=south] {$\scriptscriptstyle 0.2$} (G-6);
    \path (G-5) edge[bend right = 75] node[below] {$\scriptscriptstyle 0.2$} (G-7);
    \end{tikzpicture}
    \end{minipage}
\end{figure}
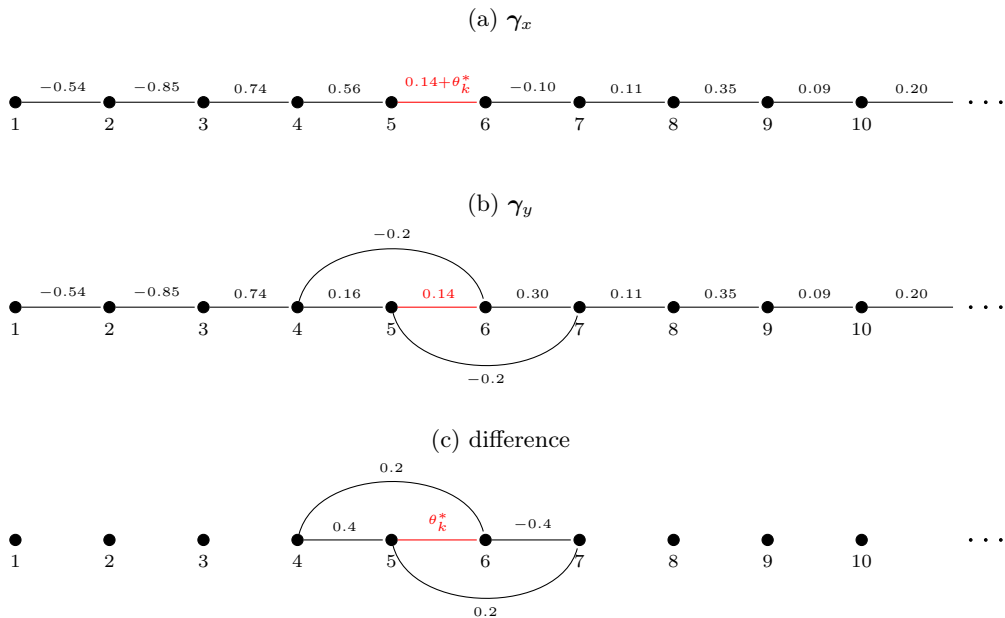
\FloatBarrier

We expect \textsc{None} and \textsc{Strong} to be easy in the sense that all four estimators are projected to perform equally well. By contrast, \textsc{Weak} and \textsc{Mixed} represent hard problems for the na\"{i}ve re-estimation procedure.

Figure \ref{fig:power} gives a summary of the results. The power is estimated as the proportion of rejections out of 1000 independent replications at level $0.05$. As in Experiment~1, both SparKLIE+ estimators behave similarly.

\begin{figure}
    \caption{\label{fig:power}\textbf{Power of the test $|\hat\theta_k| / \hat\sigma_k > z_{0.975} / \sqrt{n}$ for the hypothesis $\Hcal_0: \theta_k^* = 0$.} Here, $\hat\theta_k$ is either the SparKLIE+1 or the SparKLIE+2 estimator, $\hat\sigma_k$ is the estimator of the standard error from \eqref{eq:varest}, and $z_{0.975}$ is the 0.975-quantile of $\Ncal(0,1)$. The blue line with $\bullet$ indicates SparKLIE+1. The orange line with $\blacktriangledown$ indicates SparKLIE+2.}
    \centering
    
    \vspace{\baselineskip}

    \begin{minipage}[b]{0.48\linewidth}
    \centering
    {\small (a) \textsc{None}}

    \includegraphics[width=\linewidth]{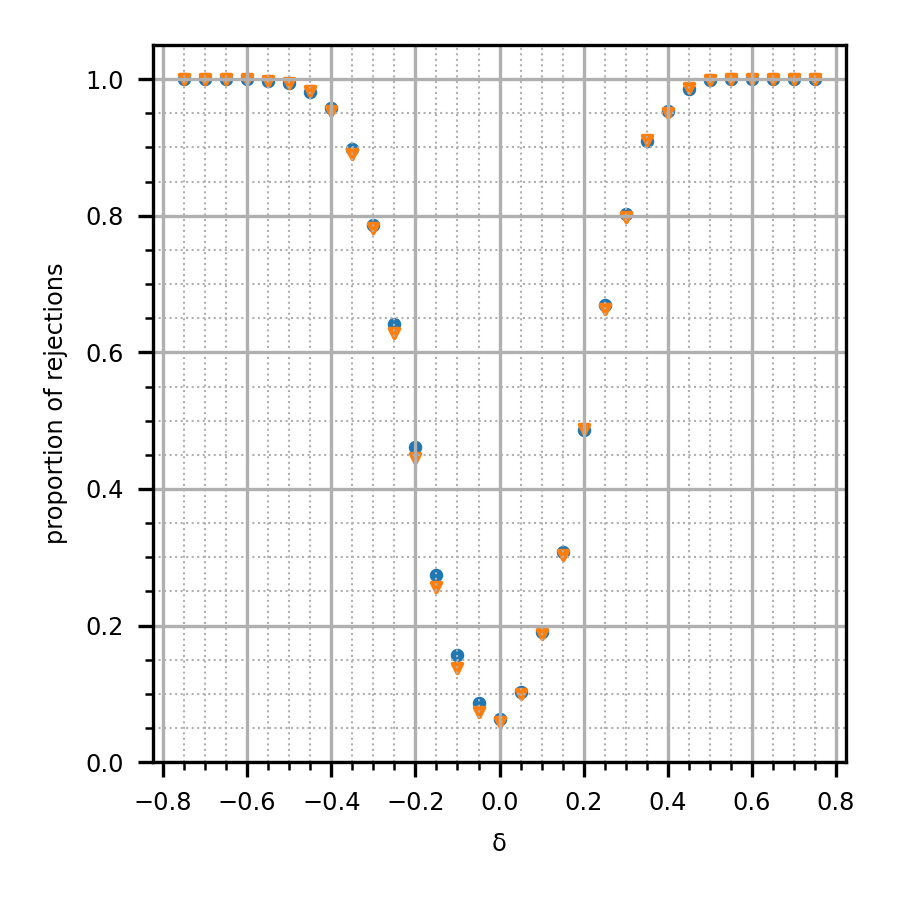}
    \end{minipage}
    \hfill
    \begin{minipage}[b]{0.48\linewidth}
    \centering
    {\small (b) \textsc{Strong}}

    \includegraphics[width=\linewidth]{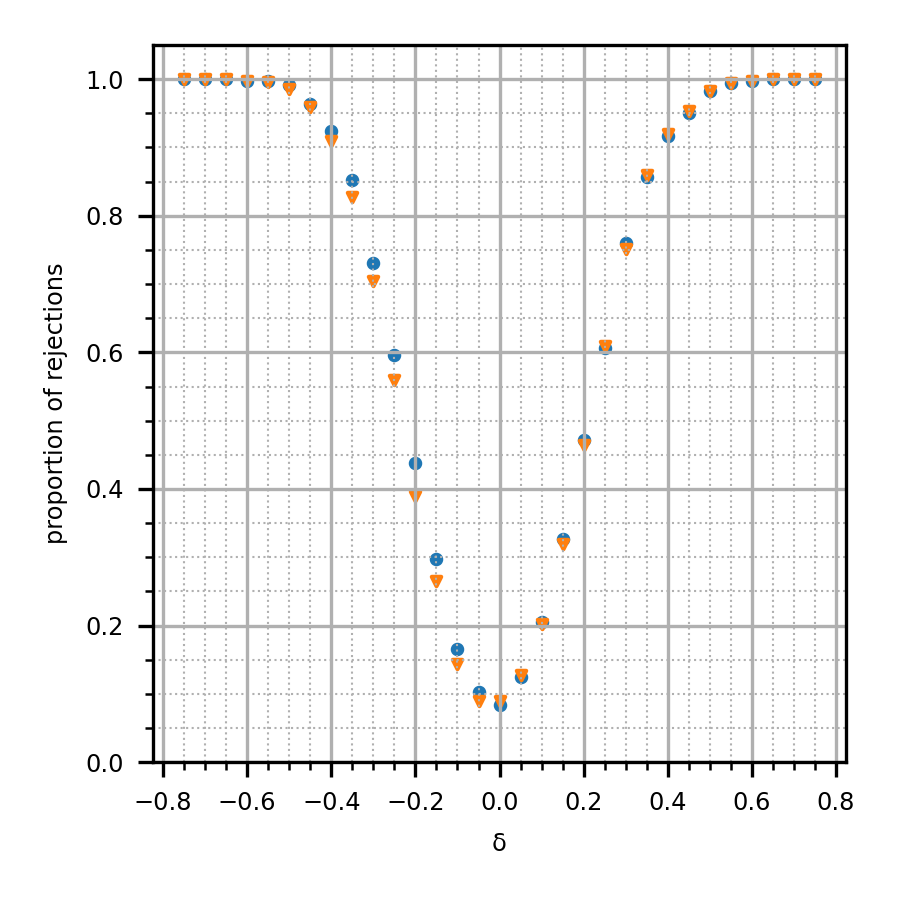}
    \end{minipage}

    \begin{minipage}[b]{0.48\linewidth}
    \centering
    {\small (c) \textsc{Weak}}

    \includegraphics[width=\linewidth]{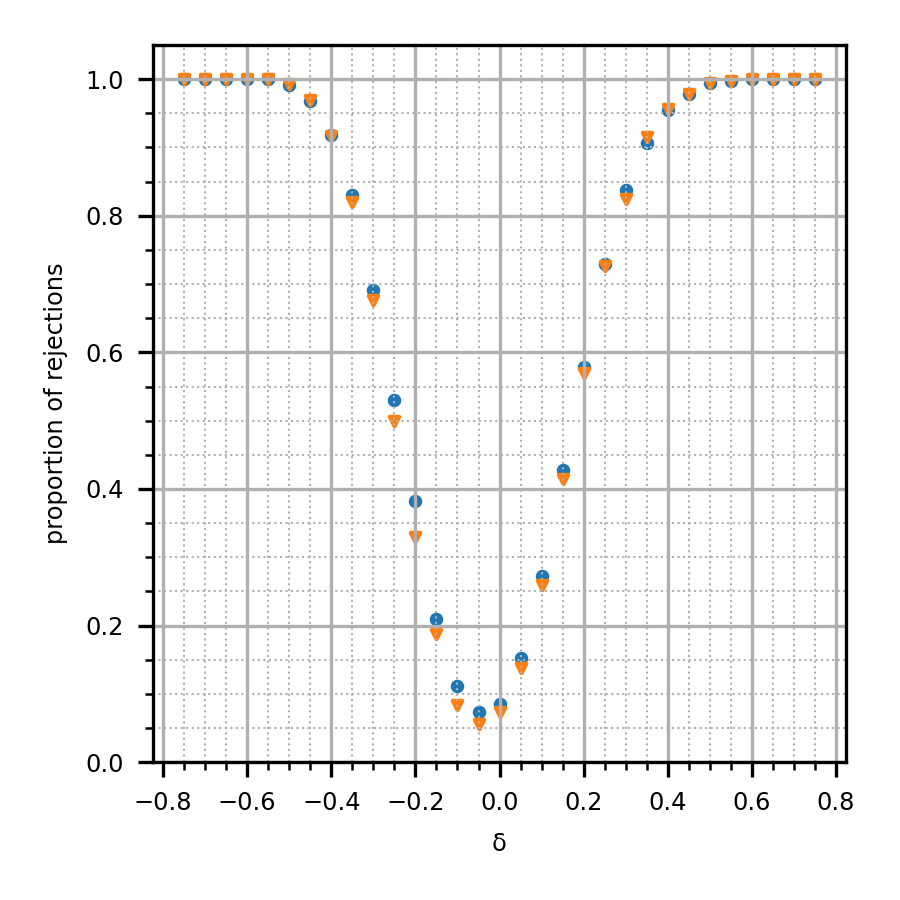}
    \end{minipage}
    \hfill
    \begin{minipage}[b]{0.48\linewidth}
    \centering
    {\small (d) \textsc{Mixed}}

    \includegraphics[width=\linewidth]{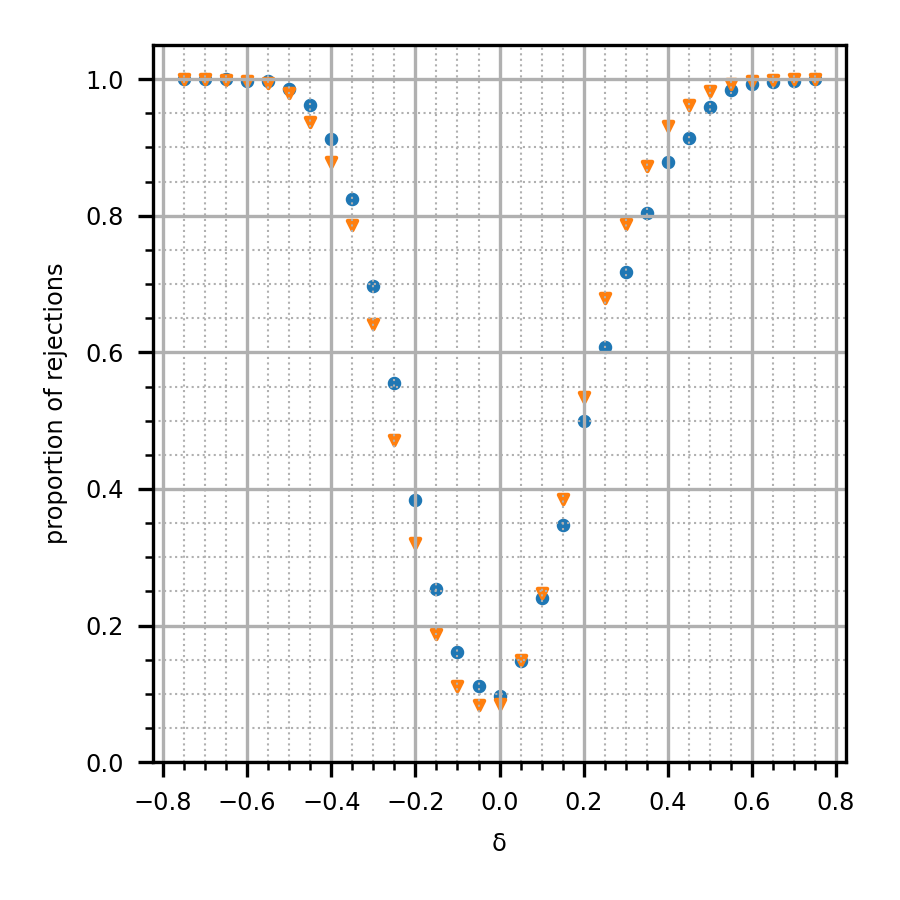}
    \end{minipage}
\end{figure}
\FloatBarrier

\subsection{\label{supp:exper4}Experiment~4: Power of the empirical bootstrap test}

We look at the power of the empirical bootstrap test as a function of the number of the changes and their magnitudes. For each $m \in \{25, 500, 100\}$, we fix $\gammab_x$ at the $\gammab_x$ from Experiment~3, and then modify $\gammab_x$ to obtain $\gammab_y$. This was done by first picking $s_\theta \in \cbr{1, 3, 5}$ edges uniformly at random from the set of all possible edges, next drawing $\delta \sim \Unif(l, l+0.1)$ for $l \in \cbr{0,.05,.10, \dots, .50}$ for each edge in the difference graph independently of everything else, and finally subtracting the chosen $\delta$'s from $\gammab_x$.

Here, we focused on bootstrapping SparKLIE+2 only. Also, we considered the studentized version $W = \max_k \sqrt{n} \, |\hat\theta_k - \theta_k^*| /\hat\sigma_k$, where $\hat\sigma_k$ is the estimate of the standard error \eqref{eq:varest}. $\hat c_{W, \alpha}$ refers to the estimated $(1-\alpha)$-quantile of $W$ (see Appendix~\ref{supp:implementation:studentization}).

The results are summarized in Figure~\ref{fig:power_boot} at level $0.05$. In the plots, the label ``unnormalized" refers to the testing procedure using the unnormalized statistics $T$, and the label ``normalized", to the studentized version $W$. There is a moderate gain in power when the latter is used.

\begin{figure}
    \caption{\label{fig:power_boot}\textbf{Power of the empirical bootstrap test for the global hypothesis $\Hcal_0: \thetab^* = \zero$.} We plot the power curves for $m = 25, 50, 100$ and the number of changes $= 1, 3, 5$ using two different test statistics. The three panels on the left correspond to the test $\max_k |\hat\theta_k| > \hat c_{T, 1-\alpha} / \sqrt{n}$. The three panels on the right correspond to the \emph{studentized} test $\max_k |\hat\theta_k| / \hat\sigma_k > \hat c_{W, 1-\alpha} / \sqrt{n}$. The blue $\bullet$'s correspond to the case of the difference graph with 1 change; the orange $+$'s, 3 changes; and the green x's, 5 changes.}
    \centering
    
    \vspace{\baselineskip}

    \begin{minipage}[b]{\linewidth}
    \centering
    {\small (a) 25 nodes}
    
    \hspace{1cm} \includegraphics[height=0.27\textheight]{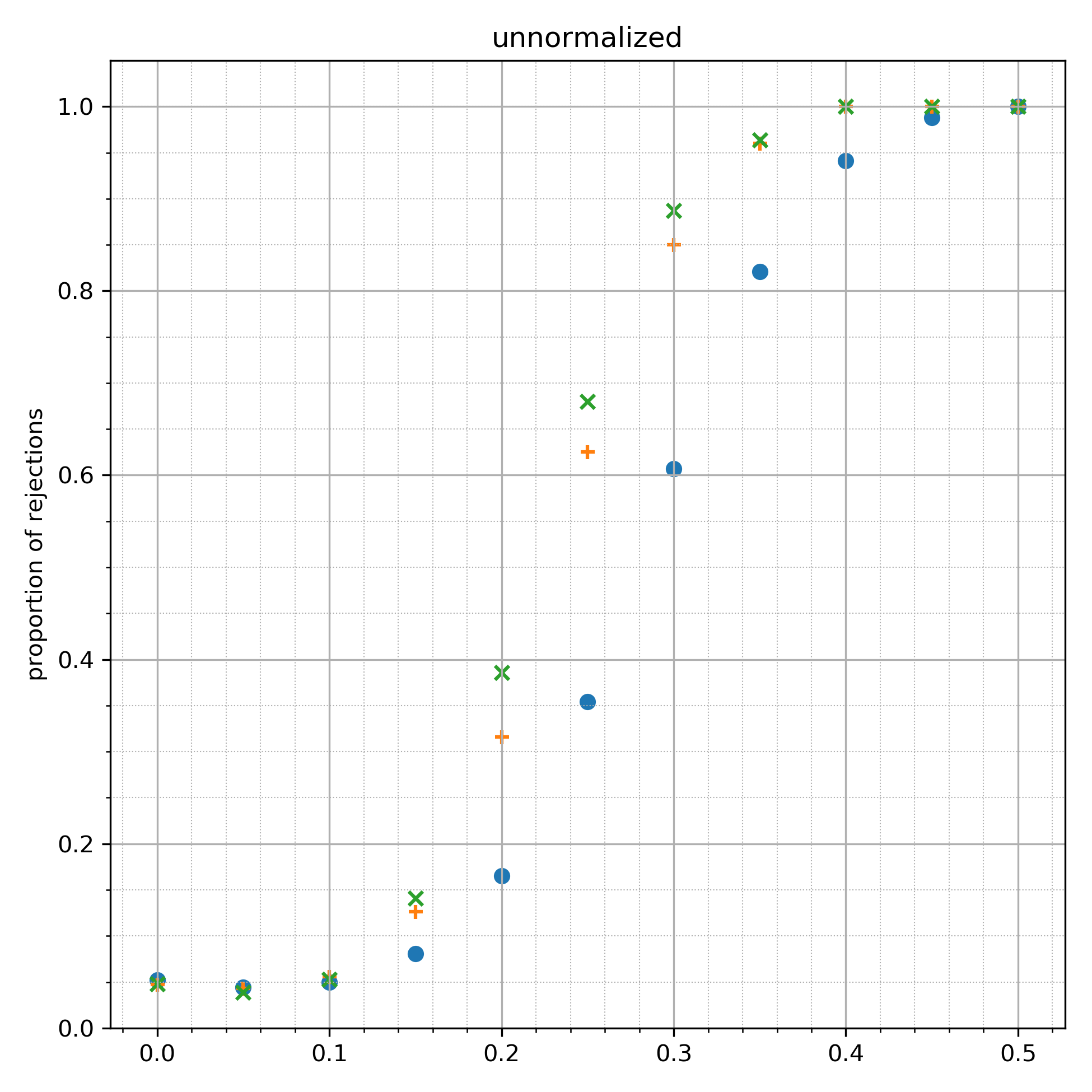} \hfill \includegraphics[height=0.27\textheight]{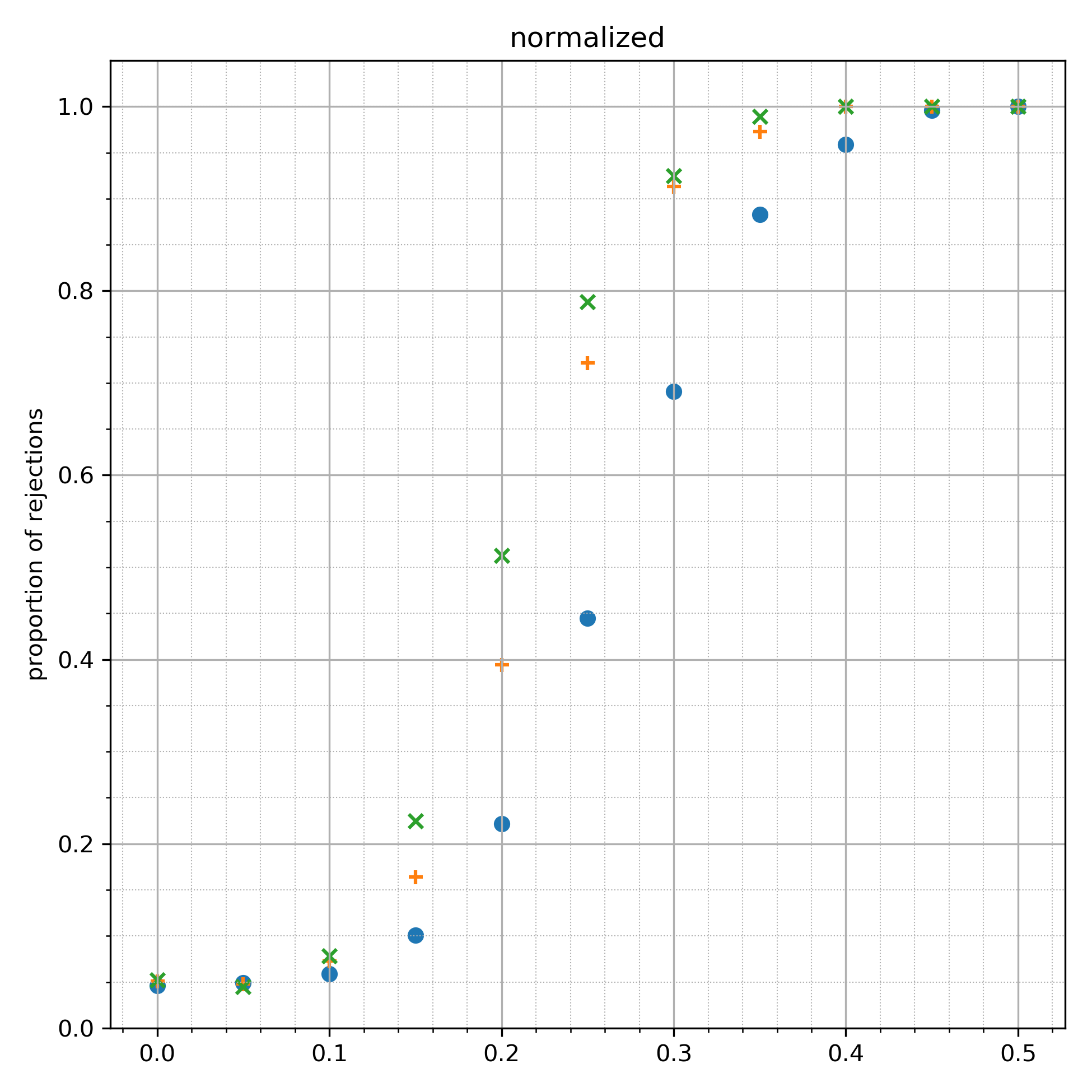} \hspace{1cm}
    \end{minipage}

    \begin{minipage}[b]{\linewidth}
    \centering
    {\small (b) 50 nodes}
    
    \hspace{1cm} \includegraphics[height=0.27\textheight]{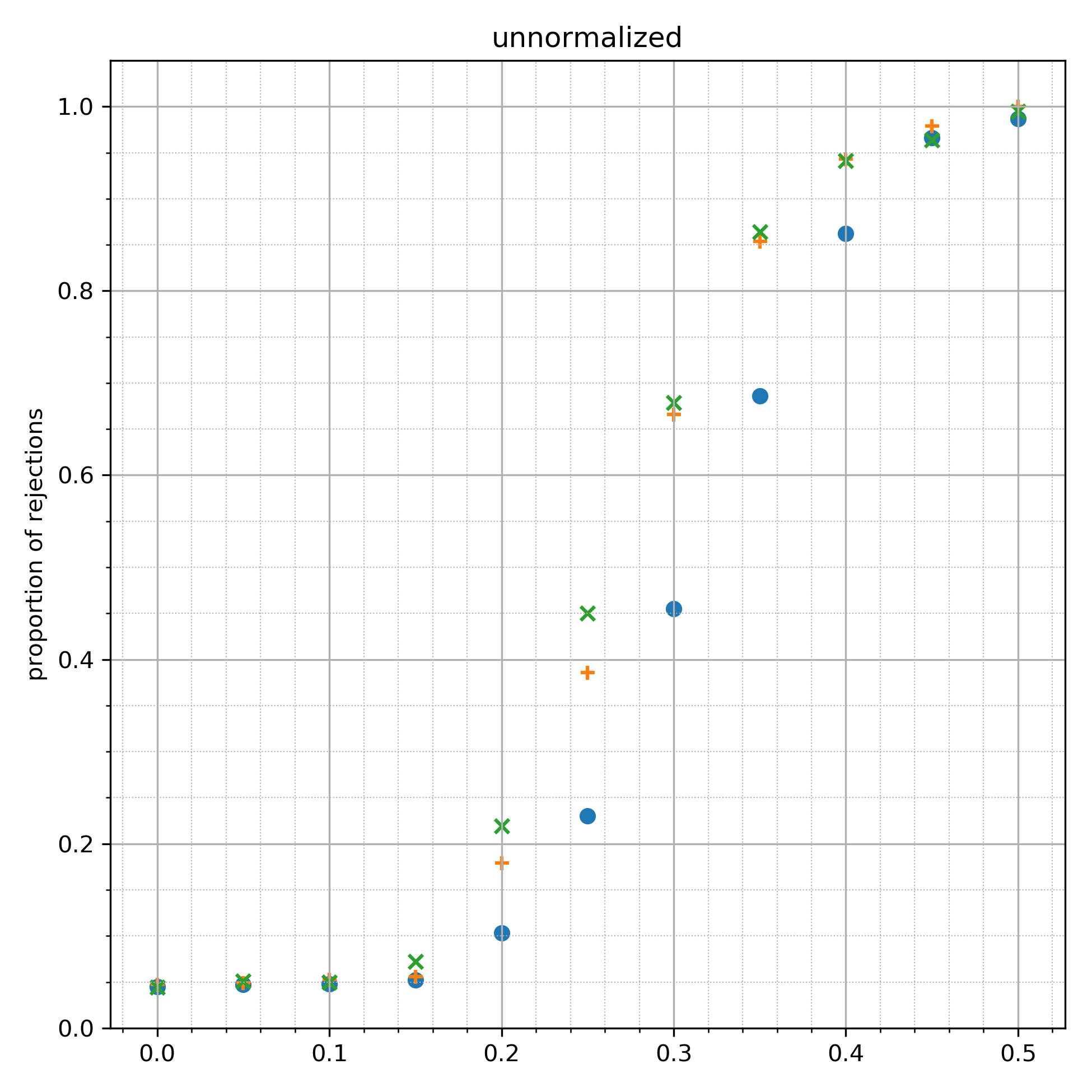} \hfill \includegraphics[height=0.27\textheight]{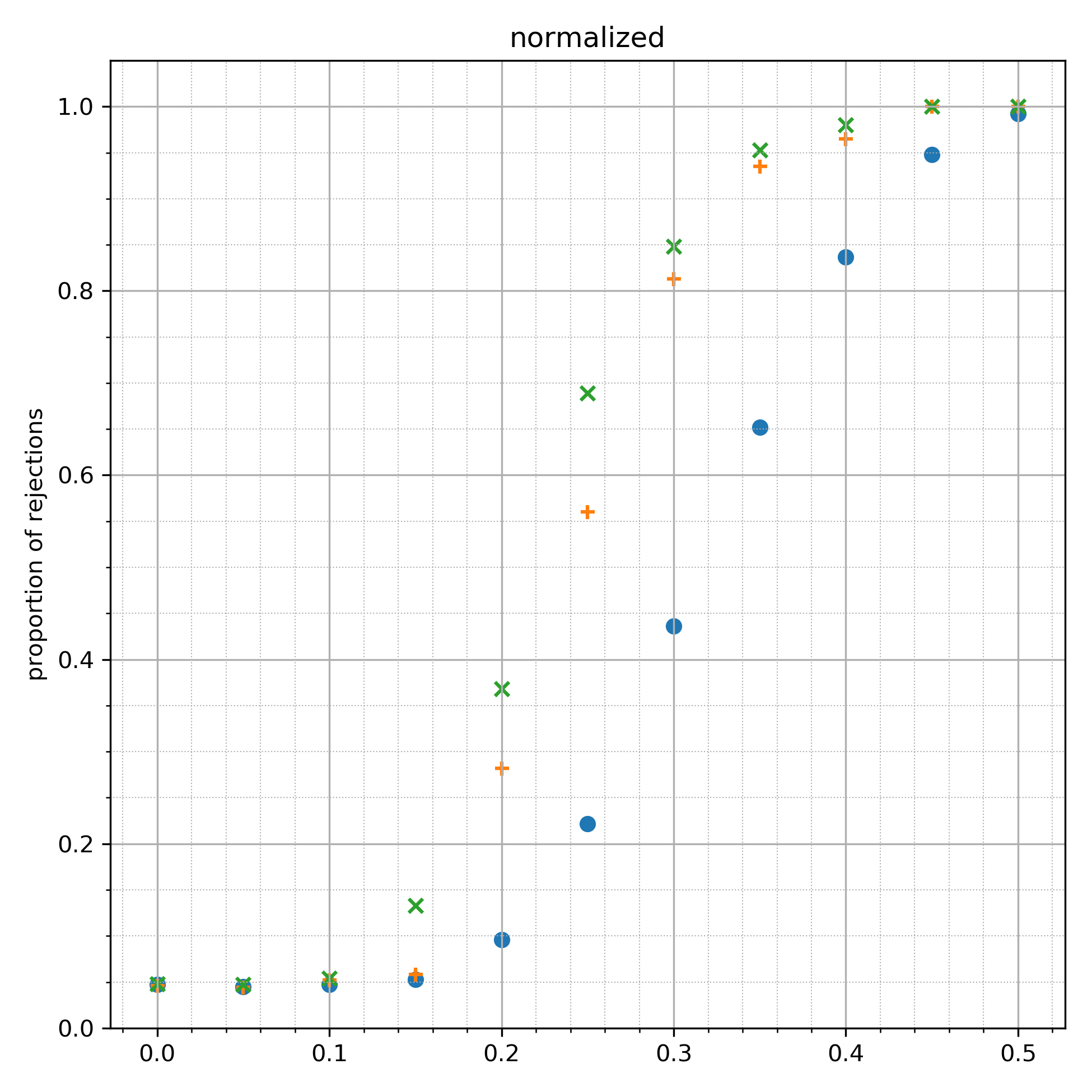} \hspace{1cm}
    \end{minipage}

    \begin{minipage}[b]{\linewidth}
    \centering
    {\small (c) 100 nodes}
    
    \hspace{1cm} \includegraphics[height=0.27\textheight]{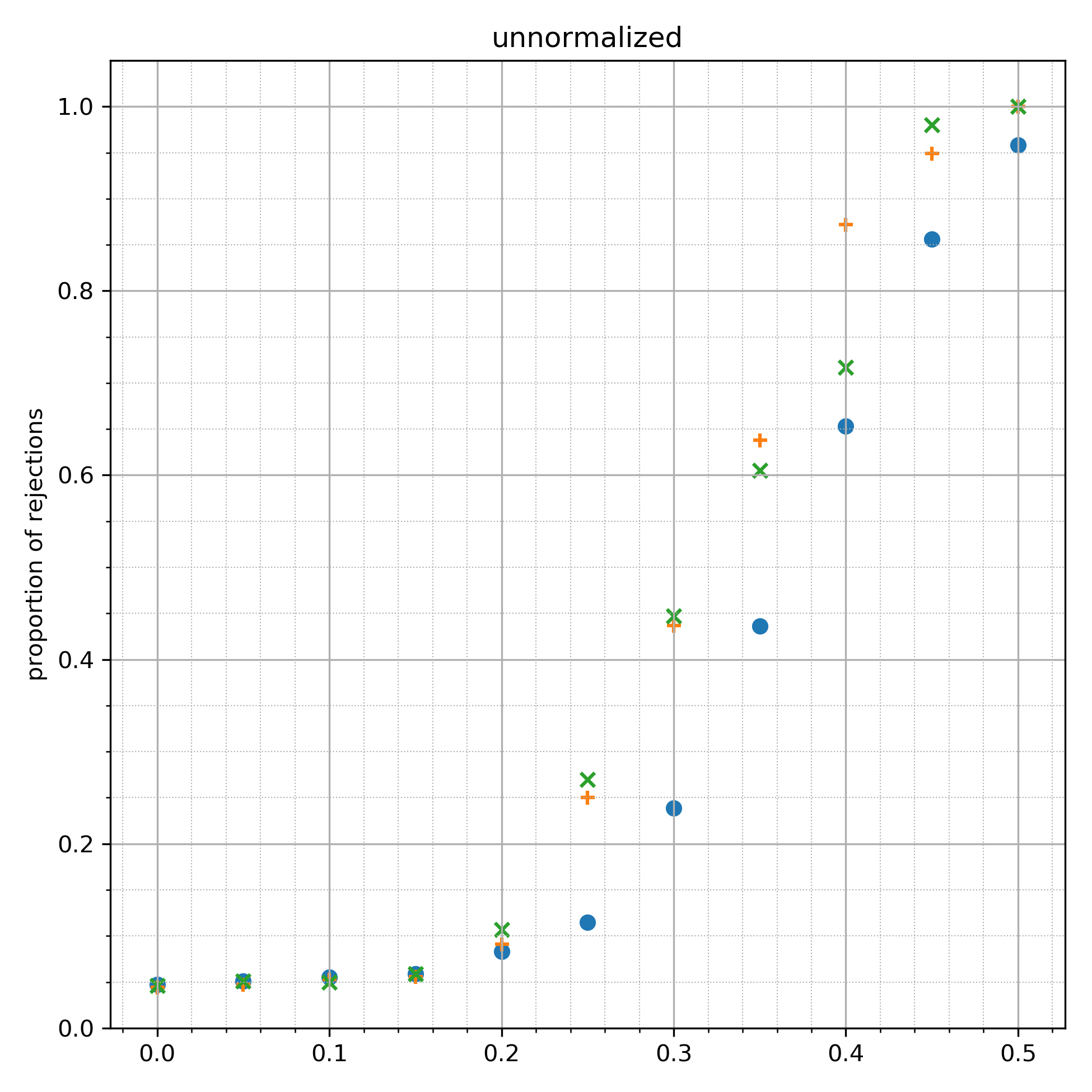} \hfill \includegraphics[height=0.27\textheight]{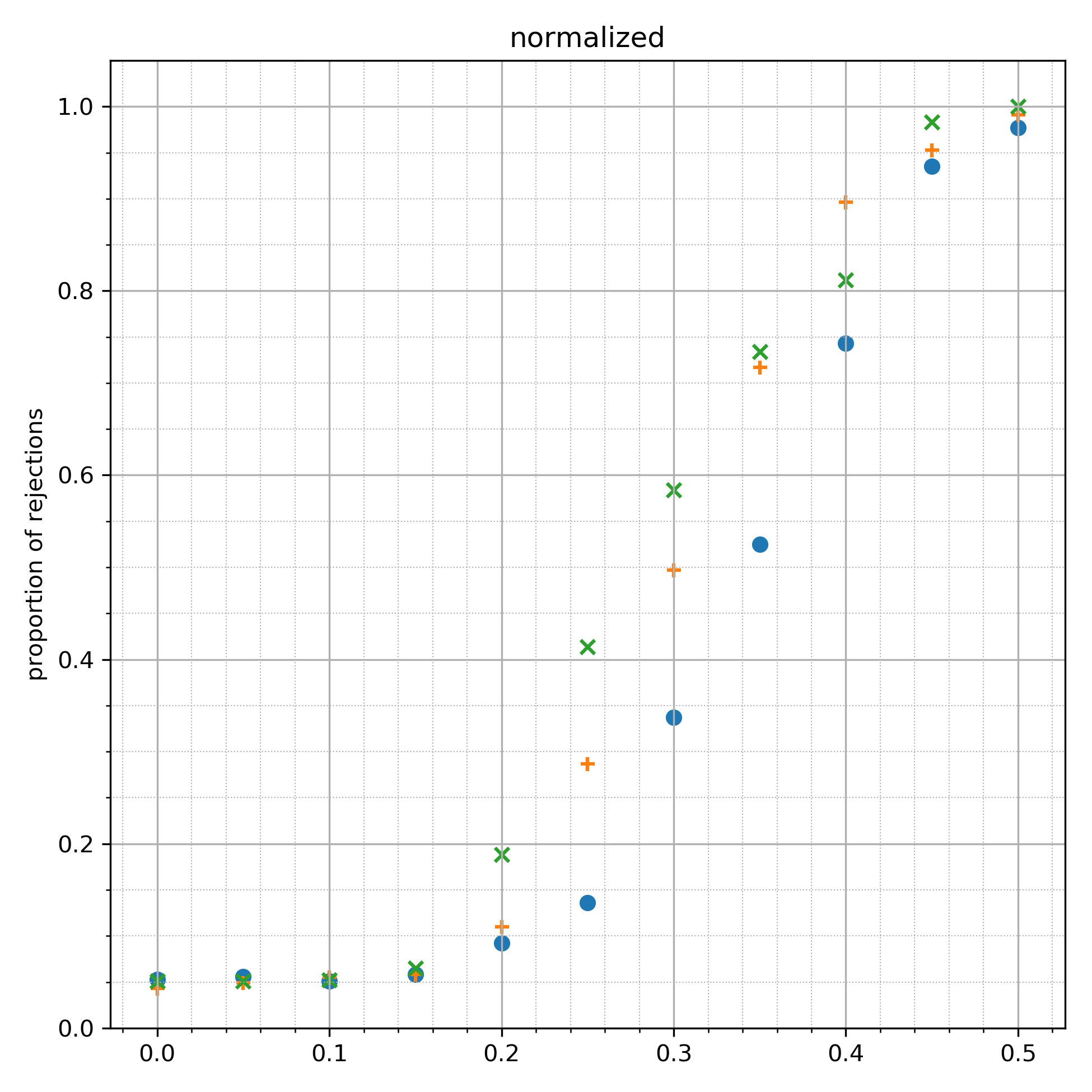} \hspace{1cm}
    \end{minipage}
\end{figure}

\subsection{\label{supp:exper5}Experiment~5: Reversed and symmetrized procedures and sensitivity to $\lambda_\theta$}

We study the performance of the reversed and the symmetrized procedures using the same synthetic data as in Experiment~1 for easier comparison with SparKLIE+.
The reversed procedure is obtained by replacing $\lKLIEP$ with the reversed loss
\[
\ell_\text{RevKLIEP}(\thetab; \Xb_{n_x}, \Yb_{n_y})
= \frac{1}{n_y} \sum_{j=1}^{n_y} \thetab^\top \psib(\yb^{(j)}) + \log\cbr{\frac{1}{n_x} \sum_{i=1}^{n_x} \exp\rbr{-\thetab^\top \psib(\xb^{(i)})}}.
\]
It is easy to see that this is just $\lKLIEP$ with the roles of $\xb$ and $\yb$ switched.
$\ell_\text{RevKLIEP}$ also occurs as a result of minimizing the reverse KL divergence from $f_x / r_\theta$ to $f_y$.
The symmetrized procedure minimizes the sum of $\lKLIEP$ and $\ell_\text{RevKLIEP}$
\begin{multline*}
\ell_\text{SymKLIEP}(\thetab; \Xb_{n_x}, \Yb_{n_y})\\
\begin{aligned}[t]
&= \lKLIEP(\thetab; \Xb_{n_x}, \Yb_{n_y}) + \ell_\text{RevKLIEP}(\thetab; \Xb_{n_x}, \Yb_{n_y})\\
&=
\begin{multlined}[t]
-\frac{1}{n_x} \sum_{i=1}^{n_x} \thetab^\top \psib(\xb^{(i)}) + \frac{1}{n_y} \sum_{j=1}^{n_y} \thetab^\top \psib(\yb^{(j)})\\
+ \log\cbr{\frac{1}{n_x} \sum_{i=1}^{n_x} \exp\rbr{-\thetab^\top \psib(\xb^{(i)})}}
+ \log\cbr{\frac{1}{n_y} \sum_{j=1}^{n_y} \exp\rbr{\thetab^\top \psib(\yb^{(j)})}}.
\end{multlined}
\end{aligned}
\end{multline*}
To measure performance, we looked at the coverage and the median width of 95\% confidence intervals, as well as the bias of the estimator over the same 1000 replications as in Experiment~1.
The results are in Tables~\ref{table:divergences:coverage:chain} to \ref{table:divergences:bias:tree}.
The reversed and the symmetrized procedures are expected to have worse sample complexity compared to SparKLIE+.
This is indeed what we observe.

Also, to study the sensitivity to the regularization parameter choice, we tried five difference values of $\lambda_\theta$ as detailed in \Cref{table:lambda_values:additional}.
The results in Tables~\ref{table:divergences:coverage:chain} to \ref{table:divergences:bias:tree} tell us that all performance measures are quite stable for both SparKLIE+ procedures.
The reversed and the symmetrized procedures do show some instability, but it is likely that this has more to do with the fact that both procedures have larger sample complexity relative to KLIEP.
See \Cref{remark:asymmetry} in \Cref{sec:theory:main}.

\begin{table}
\caption{\label{table:lambda_values:additional}\textbf{Regularization parameter settings for Experiment~5.}}
\centering
\fbox{%
\begin{tabular}{c c l c}
Divergence & \multicolumn{2}{c}{$\lambda_\theta$} & $\lambda_k$\\
\hline
KL & $\sqrt{j \log p / \min\{n_x, n_y\}}$, & $j = 4, 3.5, \dots, 2$ & $\sqrt{2 \log p/ n_y}$\\
Reverse & $\sqrt{j \log p / \min\{n_x, n_y\}}$, & $j = 16, 12.5, \dots, 2$ & $\sqrt{2 \log p / n_x}$\\
Symmetric & $\sqrt{j \log p / \min\{n_x, n_y\}}$, & $j = 16, 12.5, \dots, 2$ & $\tfrac 12 \sqrt{2 \log p / n_x} + \tfrac 12 \sqrt{2 \log p / n_y}$
\end{tabular}}
\end{table}

\begin{table}
\caption{\label{table:divergences:coverage:chain}\textbf{Empirical coverage of the 95\% CI $\hat\theta_k \pm z_{0.975} \hat\sigma_k / \sqrt{n}$ for Chain 1 and Chain 2.}}
\centering
\fbox{%
\begin{tabular}{*{4}{c} *{5}{c}}
$\gammab_y$ & $m$ & Divergence & De-biasing & \multicolumn{5}{c}{Coverage as a function of $\lambda_\theta$}\\
\hline
\multirow{15}{*}{(1)} & \multirow{7}{*}{25} & KL & \multirow{3}{*}{+1} & 0.934 & 0.941 & 0.942 & 0.943 & 0.953\\
&& Reverse && 0.920 & 0.919 & 0.921 & 0.917 & 0.902\\
&& Symmetric && 0.911 & 0.895 & 0.893 & 0.876 & 0.875\\
\\
&& KL & \multirow{3}{*}{+2} & 0.963 & 0.965 & 0.964 & 0.965 & 0.964\\
&& Reverse && 0.967 & 0.965 & 0.956 & 0.936 & 0.915\\
&& Symmetric && 0.940 & 0.930 & 0.897 & 0.781 & 0.567\\
\\
& \multirow{7}{*}{50} & KL & \multirow{3}{*}{+1} & 0.951 & 0.955 & 0.953 & 0.955 & 0.957\\
&& Reverse && 0.888 & 0.876 & 0.859 & 0.919 & 0.891\\
&& Symmetric && 0.909 & 0.914 & 0.887 & 0.868 & 0.708\\
\\
&& KL & \multirow{3}{*}{+2} & 0.970 & 0.972 & 0.974 & 0.970 & 0.964\\
&& Reverse && 0.947 & 0.930 & 0.889 & 0.930 & 0.895\\
&& Symmetric && 0.940 & 0.933 & 0.871 & 0.525 & 0.978\\
\\
\multirow{15}{*}{(2)} & \multirow{7}{*}{25} & KL & \multirow{3}{*}{+1} & 0.956 & 0.951 & 0.947 & 0.948 & 0.957\\
&& Reverse && 0.900 & 0.900 & 0.891 & 0.898 & 0.877\\
&& Symmetric && 0.938 & 0.929 & 0.917 & 0.895 & 0.889\\
\\
&& KL & \multirow{3}{*}{+2} & 0.959 & 0.955 & 0.956 & 0.955 & 0.961\\
&& Reverse && 0.953 & 0.953 & 0.951 & 0.948 & 0.910\\
&& Symmetric && 0.949 & 0.948 & 0.903 & 0.783 & 0.568\\
\\
& \multirow{7}{*}{50} & KL & \multirow{3}{*}{+1} & 0.924 & 0.930 & 0.938 & 0.943 & 0.928\\
&& Reverse && 0.877 & 0.877 & 0.873 & 0.878 & 0.857\\
&& Symmetric && 0.927 & 0.926 & 0.887 & 0.836 & 0.718\\
\\
&& KL & \multirow{3}{*}{+2} & 0.937 & 0.942 & 0.943 & 0.952 & 0.945\\
&& Reverse && 0.926 & 0.925 & 0.927 & 0.920 & 0.883\\
&& Symmetric && 0.936 & 0.935 & 0.859 & 0.487 & 0.987
\end{tabular}}
\end{table}

\begin{table}
\caption{\textbf{Empirical coverage of the 95\% CI $\hat\theta_k \pm z_{0.975} \hat\sigma_k / \sqrt{n}$ for Tree 1 and Tree 2.}}
\centering
\fbox{%
\begin{tabular}{*{4}{c} *{5}{c}}
$\gammab_y$ & $m$ & Divergence & De-biasing & \multicolumn{5}{c}{Coverage as a function of $\lambda_\theta$}\\
\hline
\multirow{15}{*}{(1)} & \multirow{7}{*}{25} & KL & \multirow{3}{*}{+1} & 0.940 & 0.945 & 0.947 & 0.955 & 0.952\\
&& Reverse && 0.798 & 0.801 & 0.831 & 0.862 & 0.893\\
&& Symmetric && 0.880 & 0.892 & 0.925 & 0.946 & 0.895\\
\\
&& KL & \multirow{3}{*}{+2} & 0.977 & 0.976 & 0.974 & 0.972 & 0.977\\
&& Reverse && 0.939 & 0.939 & 0.939 & 0.940 & 0.934\\
&& Symmetric && 0.909 & 0.903 & 0.905 & 0.865 & 0.728\\
\\
& \multirow{7}{*}{50} & KL & \multirow{3}{*}{+1} & 0.954 & 0.957 & 0.961 & 0.961 & 0.959\\
&& Reverse && 0.743 & 0.755 & 0.820 & 0.843 & 0.860\\
&& Symmetric && 0.871 & 0.883 & 0.903 & 0.964 & 0.435\\
\\
&& KL & \multirow{3}{*}{+2} & 0.985 & 0.985 & 0.985 & 0.981 & 0.982\\
&& Reverse && 0.906 & 0.914 & 0.940 & 0.942 & 0.934\\
&& Symmetric && 0.905 & 0.908 & 0.878 & 0.734 & 0.987\\
\\
\multirow{15}{*}{(2)} & \multirow{7}{*}{25} & KL & \multirow{3}{*}{+1} & 0.955 & 0.961 & 0.959 & 0.959 & 0.958\\
&& Reverse && 0.860 & 0.861 & 0.856 & 0.862 & 0.889\\
&& Symmetric && 0.887 & 0.905 & 0.937 & 0.970 & 0.906\\
\\
&& KL & \multirow{3}{*}{+2} & 0.982 & 0.987 & 0.988 & 0.985 & 0.985\\
&& Reverse && 0.941 & 0.941 & 0.939 & 0.927 & 0.929\\
&& Symmetric && 0.925 & 0.918 & 0.917 & 0.896 & 0.731\\
\\
& \multirow{7}{*}{50} & KL & \multirow{3}{*}{+1} & 0.954 & 0.956 & 0.950 & 0.954 & 0.955\\
&& Reverse && 0.859 & 0.859 & 0.855 & 0.860 & 0.873\\
&& Symmetric && 0.903 & 0.910 & 0.932 & 0.972 & 0.435\\
\\
&& KL & \multirow{3}{*}{+2} & 0.990 & 0.988 & 0.982 & 0.980 & 0.980\\
&& Reverse && 0.954 & 0.951 & 0.939 & 0.936 & 0.932\\
&& Symmetric && 0.935 & 0.921 & 0.914 & 0.784 & 0.990
\end{tabular}}
\end{table}

\begin{table}
\caption{\textbf{Median width of the 95\% CI $\hat\theta_k \pm z_{0.975} \hat\sigma_k / \sqrt{n}$ for Chain 1 and Chain 2.}}
\centering
\fbox{%
\begin{tabular}{*{4}{c} *{5}{c}}
$\gammab_y$ & $m$ & Divergence & De-biasing & \multicolumn{5}{c}{Median width as a function of $\lambda_\theta$}\\
\hline
\multirow{15}{*}{(1)} & \multirow{7}{*}{25} & KL & \multirow{3}{*}{+1} & 0.479 & 0.481 & 0.485 & 0.490 & 0.497\\
&& Reverse && 0.500 & 0.500 & 0.494 & 0.478 & 0.503\\
&& Symmetric && 0.420 & 0.438 & 0.503 & 0.701 & 1.467\\
\\
&& KL & \multirow{3}{*}{+2} & 0.511 & 0.517 & 0.519 & 0.523 & 0.532\\
&& Reverse && 0.540 & 0.540 & 0.531 & 0.502 & 0.528\\
&& Symmetric && 0.454 & 0.483 & 0.531 & 0.669 & 1.605\\
\\
& \multirow{7}{*}{50} & KL & \multirow{3}{*}{+1} & 0.347 & 0.347 & 0.346 & 0.347 & 0.351\\
&& Reverse && 0.353 & 0.351 & 0.331 & 0.316 & 0.344\\
&& Symmetric && 0.300 & 0.310 & 0.384 & 0.776 & 766.6\\
\\
&& KL & \multirow{3}{*}{+2} & 0.366 & 0.364 & 0.364 & 0.365 & 0.369\\
&& Reverse && 0.382 & 0.381 & 0.346 & 0.324 & 0.359\\
&& Symmetric && 0.333 & 0.340 & 0.385 & 0.649 & 936.7\\
\\
\multirow{15}{*}{(2)} & \multirow{7}{*}{25} & KL & \multirow{3}{*}{+1} & 0.436 & 0.446 & 0.454 & 0.466 & 0.483\\
&& Reverse && 0.483 & 0.483 & 0.494 & 0.524 & 0.573\\
&& Symmetric && 0.443 & 0.463 & 0.528 & 0.727 & 1.503\\
\\
&& KL & \multirow{3}{*}{+2} & 0.444 & 0.454 & 0.465 & 0.481 & 0.504\\
&& Reverse && 0.521 & 0.522 & 0.537 & 0.568 & 0.630\\
&& Symmetric && 0.458 & 0.480 & 0.535 & 0.680 & 1.569\\
\\
& \multirow{7}{*}{50} & KL & \multirow{3}{*}{+1} & 0.318 & 0.323 & 0.329 & 0.336 & 0.349\\
&& Reverse && 0.341 & 0.344 & 0.362 & 0.380 & 0.410\\
&& Symmetric && 0.319 & 0.328 & 0.390 & 0.787 & 756.2\\
\\
&& KL & \multirow{3}{*}{+2} & 0.322 & 0.327 & 0.336 & 0.348 & 0.363\\
&& Reverse && 0.368 & 0.372 & 0.395 & 0.413 & 0.445\\
&& Symmetric && 0.331 & 0.342 & 0.388 & 0.654 & 953.3
\end{tabular}}
\end{table}

\begin{table}
\caption{\textbf{Median width of the 95\% CI $\hat\theta_k \pm z_{0.975} \hat\sigma_k / \sqrt{n}$ for Tree 1 and Tree 2.}}
\centering
\fbox{%
\begin{tabular}{*{4}{c} *{5}{c}}
$\gammab_y$ & $m$ & Divergence & De-biasing & \multicolumn{5}{c}{Median width as a function of $\lambda_\theta$}\\
\hline
\multirow{15}{*}{(1)} & \multirow{7}{*}{25} & KL & \multirow{3}{*}{+1} & 0.754 & 0.765 & 0.776 & 0.792 & 0.815\\
&& Reverse && 0.711 & 0.712 & 0.740 & 0.781 & 0.865\\
&& Symmetric && 0.707 & 0.772 & 0.969 & 1.467 & 2.925\\
\\
&& KL & \multirow{3}{*}{+2} & 0.845 & 0.865 & 0.881 & 0.903 & 0.940\\
&& Reverse && 0.786 & 0.788 & 0.804 & 0.831 & 0.925\\
&& Symmetric && 0.783 & 0.853 & 1.014 & 1.508 & 4.574\\
\\
& \multirow{7}{*}{50} & KL & \multirow{3}{*}{+1} & 0.581 & 0.578 & 0.575 & 0.575 & 0.584\\
&& Reverse && 0.508 & 0.516 & 0.559 & 0.580 & 0.676\\
&& Symmetric && 0.527 & 0.558 & 0.717 & 1.709 & 2.008\\
\\
&& KL & \multirow{3}{*}{+2} & 0.659 & 0.654 & 0.651 & 0.652 & 0.669\\
&& Reverse && 0.577 & 0.583 & 0.607 & 0.614 & 0.746\\
&& Symmetric && 0.592 & 0.619 & 0.758 & 1.733 & 411.9\\
\\
\multirow{15}{*}{(2)} & \multirow{7}{*}{25} & KL & \multirow{3}{*}{+1} & 0.815 & 0.826 & 0.835 & 0.842 & 0.867\\
&& Reverse && 0.686 & 0.686 & 0.696 & 0.770 & 0.889\\
&& Symmetric && 0.740 & 0.802 & 0.990 & 1.533 & 3.451\\
\\
&& KL & \multirow{3}{*}{+2} & 0.893 & 0.906 & 0.928 & 0.933 & 0.973\\
&& Reverse && 0.726 & 0.726 & 0.738 & 0.814 & 0.948\\
&& Symmetric && 0.783 & 0.852 & 1.014 & 1.514 & 4.893\\
\\
& \multirow{7}{*}{50} & KL & \multirow{3}{*}{+1} & 0.620 & 0.621 & 0.620 & 0.617 & 0.632\\
&& Reverse && 0.485 & 0.486 & 0.524 & 0.599 & 0.735\\
&& Symmetric && 0.539 & 0.579 & 0.755 & 1.848 & 1.954\\
\\
&& KL & \multirow{3}{*}{+2} & 0.687 & 0.684 & 0.679 & 0.679 & 0.693\\
&& Reverse && 0.515 & 0.517 & 0.558 & 0.629 & 0.797\\
&& Symmetric && 0.574 & 0.611 & 0.752 & 1.754 & 416.6
\end{tabular}}
\end{table}

\begin{table}
\caption{\textbf{Empirical bias of $\hat\theta_k$ for Chain 1 and Chain 2.}}
\centering
\fbox{%
\begin{tabular}{*{4}{c} *{5}{r}}
$\gammab_y$ & $m$ & Divergence & De-biasing & \multicolumn{5}{c}{Bias as a function of $\lambda_\theta$}\\
\hline
\multirow{15}{*}{(1)} & \multirow{7}{*}{25} & KL & \multirow{3}{*}{+1} & -0.009 & -0.014 & -0.019 & -0.021 & -0.023\\
&& Reverse && -0.061 & -0.062 & -0.046 & -0.002 & 0.003\\
&& Symmetric && 0.006 & -0.006 & -0.033 & -1.591 & -$1.9 \times 10^{15}$\\
\\
&& KL & \multirow{3}{*}{+2} & 0.009 & -0.001 & -0.012 & -0.017 & -0.021\\
&& Reverse && -0.058 & -0.059 & -0.045 & -0.005 & -0.038\\
&& Symmetric && 0.005 & -0.009 & -0.041 & -0.541 & -12.007\\
\\
& \multirow{7}{*}{50} & KL & \multirow{3}{*}{+1} & -0.018 & -0.017 & -0.017 & -0.017 & -0.017\\
&& Reverse && -0.058 & -0.054 & -0.005 & 0.023 & 0.005\\
&& Symmetric && 0.008 & -0.002 & -0.043 & -0.775 & -96.784\\
\\
&& KL & \multirow{3}{*}{+2} & -0.011 & -0.013 & -0.012 & -0.012 & -0.014\\
&& Reverse && -0.054 & -0.052 & -0.007 & 0.019 & -0.002\\
&& Symmetric && 0.006 & -0.004 & -0.050 & -2.337 & -22.035\\
\\
\multirow{15}{*}{(2)} & \multirow{7}{*}{25} & KL & \multirow{3}{*}{+1} & 0.012 & 0.007 & 0.004 & -0.000 & -0.004\\
&& Reverse && -0.070 & -0.070 & -0.076 & -0.073 & -0.078\\
&& Symmetric && -0.023 & -0.029 & -0.047 & -0.118 & -10.152\\
\\
&& KL & \multirow{3}{*}{+2} & -0.004 & -0.006 & -0.008 & -0.012 & -0.014\\
&& Reverse && -0.067 & -0.067 & -0.073 & -0.140 & -0.237\\
&& Symmetric && -0.023 & -0.031 & -0.054 & -0.282 & -9.502\\
\\
& \multirow{7}{*}{50} & KL & \multirow{3}{*}{+1} & 0.022 & 0.018 & 0.013 & 0.005 & -0.003\\
&& Reverse && -0.066 & -0.067 & -0.073 & -0.069 & -0.074\\
&& Symmetric && -0.019 & -0.022 & -0.054 & -0.696 & -83.982\\
\\
&& KL & \multirow{3}{*}{+2} & -0.006 & -0.007 & -0.008 & -0.010 & -0.014\\
&& Reverse && -0.063 & -0.064 & -0.070 & -0.070 & -0.083\\
&& Symmetric && -0.020 & -0.023 & -0.061 & -2.634 & -18.973
\end{tabular}}
\end{table}

\begin{table}
\caption{\label{table:divergences:bias:tree}\textbf{Empirical bias of $\hat\theta_k$ for Tree 1 and Tree 2.}}
\centering
\fbox{%
\begin{tabular}{*{4}{c} *{5}{r}}
$\gammab_y$ & $m$ & Divergence & De-biasing & \multicolumn{5}{c}{Bias as a function of $\lambda_\theta$}\\
\multirow{15}{*}{(1)} & \multirow{7}{*}{25} & KL & \multirow{3}{*}{+1} & -0.021 & -0.017 & -0.014 & -0.012 & -0.012\\
&& Reverse && -22.828 & -21.619 & -21.573 & -21.307 & -20.354\\
&& Symmetric && -0.042 & -0.085 & -0.129 & -0.300 & -11.936\\
\\
&& KL & \multirow{3}{*}{+2} & -0.030 & -0.031 & -0.031 & -0.034 & -0.039\\
&& Reverse && -4.351 & -3.258 & -4.820 & -4.550 & -3.982\\
&& Symmetric && -3.215 & -3.624 & -3.284 & -3.849 & -11.791\\
\\
& \multirow{7}{*}{50} & KL & \multirow{3}{*}{+1} & 0.001 & -0.000 & -0.003 & -0.007 & -0.011\\
&& Reverse && -0.381 & 0.008 & -2.644 & -1.543 & -2.899\\
&& Symmetric && -0.046 & -0.063 & -0.105 & -0.341 & -56.174\\
\\
&& KL & \multirow{3}{*}{+2} & -0.012 & -0.012 & -0.014 & -0.017 & -0.021\\
&& Reverse && -0.331 & 0.038 & -0.226 & -0.343 & -0.684\\
&& Symmetric && -0.056 & -0.080 & -0.140 & -2.748 & -14.916\\
\\
\multirow{15}{*}{(2)} & \multirow{7}{*}{25} & KL & \multirow{3}{*}{+1} & 0.020 & 0.021 & 0.017 & 0.016 & 0.012\\
&& Reverse && -20.257 & -19.118 & -19.523 & -20.280 & -19.418\\
&& Symmetric && -0.062 & -0.074 & -0.106 & -0.251 & -9.982\\
\\
&& KL & \multirow{3}{*}{+2} & 0.005 & 0.005 & 0.005 & 0.005 & -0.001\\
&& Reverse && -3.518 & -3.371 & -3.643 & -3.835 & -4.016\\
&& Symmetric && -3.006 & -3.024 & -2.678 & -3.294 & -10.106\\
\\
& \multirow{7}{*}{50} & KL & \multirow{3}{*}{+1} & 0.001 & -0.001 & -0.003 & -0.004 & -0.005\\
&& Reverse && -1.360 & -0.999 & -0.880 & -2.011 & -2.479\\
&& Symmetric && -0.046 & -0.052 & -0.084 & -0.756 & -60.579\\
\\
&& KL & \multirow{3}{*}{+2} & -0.007 & -0.008 & -0.007 & -0.007 & -0.008\\
&& Reverse && -0.200 & -0.104 & -0.284 & -0.121 & -0.918\\
&& Symmetric && -0.048 & -0.057 & -0.101 & -2.445 & -13.089
\end{tabular}}
\end{table}
\FloatBarrier

\section{\label{supp:sec6}Supplementary material for Section 6}

\subsection{Preprocessing}
The data were preprocessed in SPM12 (Wellcome Trust Centre for Neuroimaging, \url{http://www.fil.ion.ucl.ac.uk/spm}).
The default SPM12 steps were used, except in normalization, the voxel size was set to $2 \times 2 \times 2$ and the bounding box was changed to match the automated anatomical labelling atlas \citep{Tzourio-Mazoyer2002Automated}.

\subsection{Experiment}
The fMRI measurements were made while the participants were asked to go through four blocks of task sequences, each made up of three types of tasks arranged in some order.
During the experiment, the participants were asked to look at a screen, through which they received instructions about the tasks.
All three tasks involved squeezing and releasing a hand dynamometer while looking at the screen.
For the sensorimotor task (T1), the participants were asked to squeeze and release the hand dynamometer freely at their own pace while paying heed to the images on the screen.
By contrast, in the intrinsic alertness task (T2) or the extrinsic alertness task (T3), the participants were supposed to squeeze the hand dynamometer only after seeing a white square.
In the case of T3, a black screen always preceded each occurrence of the white square.
For T2, there was no forewarning.

Figure~\ref{fig:bristol_paradigm} gives the task sequence used in the pilot study.

\begin{figure}
\caption{\label{fig:bristol_paradigm}Task design (T1 - blue, T2 - green, T3 - red)}
\centering
\includegraphics[width=\linewidth]{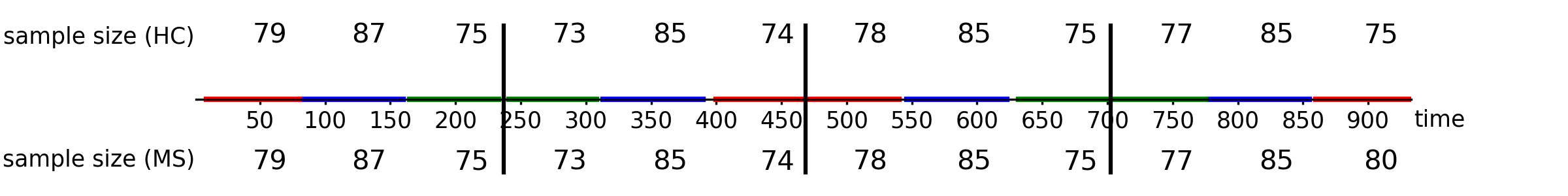}
\end{figure}
\FloatBarrier

\section{\label{supp:add_realdata}Additional real data example: Voting records of the 109th United States Senate}

We apply \Cref{sec:methodology:single} and \Cref{procedure:bootstrap:empirical} to compare the voting records in the 109th US Senate between the first half (January 3, 2005 -- January 16, 2006) and the second half (January 16, 2006 -- January 3, 2007). The data were taken from a larger dataset covering a longer period (1979 -- 2012) originally extracted from the website \url{www.voteview.com} and then processed by the authors of \cite{Roy2017Change}. We are grateful to the authors of \cite{Roy2017Change} for sharing their data with us.

We focus on the two halves of the 109th Senate. This is to ensure a sparse network difference as well as homogeneity of the data. Only one seat changed hands between the two periods from one Democrat to another. On January 16, 2006, Democrat Jon Corzine resigned in order to assume his new position as Governor of New Jersey, naming Democrat Bob Menendez to succeed. In spite of the change in membership, one would not expect there to be significant changes in the overall voting pattern, as the votes tend to split along the party lines, and nothing in our research suggests that the two Democrats were exceptional in this respect. This leads to the hypothesis
\[
	\Hcal_\text{NJ}: \gamma_{\text{1st half, Corzine / Menendez}, v} = \gamma_{\text{2nd half, Corzine / Menendez}, v} \text{ for all } v \neq \text{Corzine / Menendez}.
\]

There were 251 votes in the first half, and 177 votes in the second. Following \cite{Roy2017Change}, we code ``Yea" as $+1$ and ``Nay" as $-1$, and model the votes as independent observations from one of two Ising models with zero node potentials, one for each period. Admittedly, our model is far too simple to capture all the nuances of the complex political process. What we are hoping to observe with this toy example is whether the pattern recovered by SparKLIE+ aligns well with our knowledge of past political events, which in this case corresponds to an empty graph for the neighborhood of the New Jersey seat of interest.

We test $\Hcal_\text{NJ}$ at level 0.05. We use \Cref{procedure:KLIEP+} to estimate the differential network in the neighborhood of the New Jersey seat. We use the version of \Cref{procedure:KLIEP+} employing autoscaling formulations for Steps 1 and 2 with the universal penalty levels, as explained in Remark~\ref{remark:lambda} in Section~\ref{sec:methodology:single}. The rejection threshold for the test statistic
\[
	T_0 = \max_{v \neq \text{Corzine / Menendez}} |\hat\theta_{\text{Corzine / Menendez}, v}|
\]
was estimated using \Cref{procedure:bootstrap:empirical}. Comparing $T_0$ with the estimated rejection threshold yielded no statistically significant edges in this neighborhood differential network. We conclude that Senator Menendez's records did not differ significantly from those of his predecessor, as expected.

\end{document}